\documentclass[11pt, a4paper]{article}

\usepackage{amsmath}
\usepackage{amsfonts}
\usepackage{amsthm}
\usepackage{hyperref}
\usepackage{fullpage}
\usepackage[inline]{enumitem}
\usepackage{amssymb}
\usepackage{graphicx}
\usepackage{overpic}
\usepackage{tikz}

\newcommand{\aend}{\hat{a}}

\newcommand{\bend}{\hat{b}}
\renewcommand{\c}{\widehat{c}}
\newcommand{\C}{\widetilde{C}}
\newcommand{\compC}{\mathbb{C}}
\newcommand{\Ctilde}{\widetilde{C}}
\newcommand{\dhat}{\hat{d}}
\newcommand{\D}{\hat{D}}

\newcommand{\Ehat}{\hat{E}}
\newcommand{\Ehattilde}{\widetilde{\hat{E}}}

\newcommand{\fhat}{\hat{f}}
\newcommand{\ghat}{\hat{g}}
\newcommand{\ghattilde}{\widetilde{\hat{g}}}
\newcommand{\gfn}{\mathbf{g}}
\newcommand{\gfntilde}{\widetilde{\mathbf{g}}}
\newcommand{\gfnhat}{\hat{\mathbf{g}}}
\newcommand{\gfnhattilde}{\widetilde{\hat{\mathbf{g}}}}
\newcommand{\G}{\widetilde{G}}
\newcommand{\Ghat}{\hat{G}}
\newcommand{\Ghattilde}{\widetilde{\hat{G}}}
\newcommand{\hhat}{\hat{h}}
\newcommand{\Hessian}{\mathbf{H}}
\newcommand{\halfH}[1][\theta]{\mathbb{H}_{#1}}
\newcommand{\Iinv}{\mathbf{I}}
\newcommand{\Iinvhat}{\hat{\mathbf{I}}}
\newcommand{\J}{\widetilde{J}_c}
\newcommand{\ellhat}{\hat{\ell}}
\newcommand{\m}{\tilde{m}}
\newcommand{\M}{\widetilde{M}}
\newcommand{\Mcyclic}{M_{\cyclic}}
\newcommand{\n}{\widetilde{n}}
\newcommand{\N}{\widetilde{N}}
\newcommand{\Nhat}{\hat{N}}

\newcommand{\bigO}{\mathcal{O}}

\renewcommand{\P}{\mathcal{P}}
\newcommand{\Phat}{\hat{P}}
\newcommand{\Phattilde}{\widetilde{\hat{P}}}
\newcommand{\Phatsingle}{\hat{\mathcal{P}}}
\newcommand{\Phatsingletilde}{\widetilde{\hat{\mathcal{P}}}}
\newcommand{\Pb}{\widetilde P^{(b)}}
\newcommand{\tP}{\mathcal{P}}
\newcommand{\Pmodeltilde}{\widetilde{\mathsf{P}}}
\newcommand{\Pmodelhat}{\hat{\mathsf{P}}}
\newcommand{\Pmodelhattilde}{\widetilde{\hat{\mathsf{P}}}}
\newcommand{\Pinfty}{\widetilde P^{(\infty)}}

\newcommand{\Q}{\widetilde{Q}}
\newcommand{\Qhat}{\hat{Q}}
\newcommand{\Qhattilde}{\widetilde{\hat{Q}}}
\newcommand{\R}{\widetilde{R}}
\newcommand{\rhat}{\hat{r}}
\newcommand{\Rhat}{\hat{R}}
\newcommand{\Rhattilde}{\widetilde{\hat{R}}}
\newcommand{\realR}{\mathbb{R}}
\newcommand{\tR}{\mathcal{R}}
\newcommand{\tRhat}{\hat{\mathcal{R}}}
\newcommand{\tRhattilde}{\widetilde{\hat{\mathcal{R}}}}
\newcommand{\s}{\widetilde{s}}
\newcommand{\Shat}{\hat{S}}
\newcommand{\Shattilde}{\widetilde{\hat{S}}}
\newcommand{\St}{\widetilde{S}}
\newcommand{\T}{\widetilde{T}}
\newcommand{\That}{\hat{T}}
\newcommand{\Thattilde}{\widetilde{\hat{T}}}

\newcommand{\U}{\widetilde{U}}
\newcommand{\Uhat}{\hat{U}}
\newcommand{\Uhattilde}{\widetilde{\hat{U}}}
\newcommand{\V}{\widetilde{V}}
\newcommand{\Vhat}{\hat{V}}
\newcommand{\Vhattilde}{\widetilde{\hat{V}}}
\renewcommand{\v}{\widetilde{v}}
\newcommand{\Vb}{\widetilde V^{(b)}}
\newcommand{\What}{\hat{W}}
\newcommand{\Y}{\widetilde{Y}}

\newcommand{\intZ}{\mathbb{Z}}
\newcommand{\betatilde}{\tilde{\beta}}
\newcommand{\DeltaR}{\Delta_{\hat{\pmb{\mathit{\Sigma}}}}}
\newcommand{\gammahat}{\hat{\gamma}}
\newcommand{\muhat}{\hat{\mu}}
\newcommand{\phitilde}{\tilde{\phi}}
\newcommand{\Phitilde}{\widetilde{\Phi}}
\newcommand{\Phicheck}{\check{\Phi}}
\newcommand{\phihat}{\hat{\phi}}
\newcommand{\Phihat}{\hat{\Phi}}
\newcommand{\psihat}{\hat{\psi}}

\newcommand{\Psitilde}{\widetilde{\Psi}}

\newcommand{\Sigmahat}{\hat{\Sigma}}
\newcommand{\SigmahatR}{\hat{\pmb{\mathit{\Sigma}}}}
\newcommand{\Xicheck}{\check{\Xi}}

\newcommand{\gfrac}[2]{\genfrac{}{}{0pt}{}{#1}{#2}}
\newcommand{\MeijerG}[5]{G^{#1}_{#2} \left( \left. \gfrac{#3}{#4} \right\rvert {#5} \right)}
\newcommand{\choice}[1]{\overset{\bullet}{#1}}

\newtheorem{theorem}{Theorem}[section]
\newtheorem{lemma}[theorem]{Lemma}
\newtheorem{prop}[theorem]{Proposition}
\theoremstyle{remark}
\newtheorem{rmk}[theorem]{Remark}
\newtheorem{regularitycondition}[theorem]{Regularity Condition}

\DeclareMathOperator{\Ai}{Ai}

\DeclareMathOperator{\cyclic}{cyclic}
\DeclareMathOperator{\diag}{diag}

\DeclareMathOperator{\In}{in}
\DeclareMathOperator{\Mei}{Mei}
\DeclareMathOperator{\Out}{out}

\DeclareMathOperator{\pre}{pre}
\DeclareMathOperator{\supp}{supp}

\newtheorem{RHP}[theorem]{RH Problem}

\title{Hard to soft edge transition for the Muttalib-Borodin ensembles with integer parameter $\theta$}

\author{
  Dong Wang\thanks{School of Mathematical Sciences, University of Chinese Academy of Sciences, Beijing, P.~R.~China 100049 \newline
    email: \href{mailto:wangdong@wangd-math.xyz}{\protect\nolinkurl{wangdong@wangd-math.xyz}}}
   \and   
    Shuai-Xia Xu\thanks{Institut Franco-Chinois de l’Energie Nucl\'eaire, Sun Yat-sen University, Guangzhou, P.~R.~China 510275 \newline
    e-mail: \href{xushx3@mail.sysu.edu.cn}{\protect\nolinkurl{xushx3@mail.sysu.edu.cn}}}
}

\begin{document}
\maketitle
\begin{abstract}  
  We find the universal limiting correlation kernels of the Muttalib-Borodin (MB) ensembles with integer parameter $\theta \geq 2$ at $0$ in the transitive regime between the hard edge regime and the soft edge regime. This generalizes the previously studied hard edge to soft edge transition in unitarily invariant random matrix theory,  which is the $\theta = 1$ special case of our MB ensemble. The derivation is based on the vector Riemann-Hilbert (RH) problems for the biorthogonal polynomials associated with the MB ensemble. In the analysis of the RH problems, we construct matrix-valued model RH problems of size $(\theta + 1) \times (\theta + 1)$, and prove the solvability of the model RH problems by a vanishing lemma. The new limiting correlation kernels are proved to be universal for a large class of potential functions, and they interpolate the Meijer G-kernels for the hard edge regime and the Airy kernel for the soft edge regime. We observe that the new limiting correlation kernels have the integrability that is not seen in previous studies in random matrix theory and determinantal point processes. In the $\theta = 2$ case, we give a detailed analysis of the Lax pair associated with the model RH problem, show that it results in the Chazy I equation and the similarity reduction of the Boussinesq equation, which has a Painlev\'{e} IV reduction, and find that the Lax pair is in the Drinfeld-Sokolov hierarchies.
\end{abstract}
\tableofcontents

\section{Introduction} \label{sec:introduction}

We consider the Laguerre type Muttalib-Borodin (MB) ensemble with integer parameter $\theta \geq 2$ and an analytic potential function that has transition behaviour at the edge $0$. More precisely, we consider the system of $n$ particles $x_1< \dotsb < x_n$ distributed on $[0,\infty)$ with joint probability density function
\begin{equation} \label{eq:bioLag}
  P(x_1, \dotsc, x_n) = \frac{1}{\mathcal{Z}_n} \prod_{1 \leq i < j \leq n} (x_i - x_j)(x^{\theta}_i - x^{\theta}_j) \prod^n_{i = 1} x^{\alpha}_i e^{-N V(x_i)},
\end{equation}
where $\mathcal{Z}_n$ is the normalization constant, $\alpha > -1$ and $V$ is a real analytic function on $[0, +\infty)$ that satisfies
\begin{equation}
  \lim_{x \to +\infty} \frac{V(x)}{\log x} = +\infty.
\end{equation}
We are interested in the limiting behaviour of the particles as $n, N \to \infty$ and $N/n \to 1$. We concentrate on the limiting distribution of the eigenvalues near the hard edge.

We refer the interested readers to \cite{Wang-Zhang21} for a general discussion of this model. See also \cite{Muttalib95}, \cite{Lueck-Sommers-Zirnbauer06}, \cite{Borodin99}, \cite{Cheliotis14}, \cite{Forrester23}, \cite{Forrester-Wang15}, \cite{Zhang15} for MB ensemble with classical potential, \cite{Bloom-Levenberg-Totik-Wielonsky17}, \cite{Butez17}, \cite{Credner-Eichelsbacher15}, \cite{Eichelsbacher-Sommerauer-Stolz11}, \cite{Lubinsky-Sidi22} for the equilibrium and large deviation, \cite{Charlier21}, \cite{Charlier-Lenells-Mauersberger19}, \cite{Claeys-Girotti-Stivigny19}, \cite{Claeys-Romano14}, \cite{Kuijlaars16}, \cite{Kuijlaars-Molag19}, \cite{Zhang17}, \cite{Molag20} for Riemann-Hilbert approach, \cite{Betea-Occelli20a}, \cite{Betea-Occelli20}, \cite{Husson-Mazzuca-Occelli25} for discrete generalizations, \cite{Kazi-Muttalib22}, \cite{Lambert18}, \cite{Alam-Muttalib-Wang-Yadav20} for various generalizations. We note that the MB ensemble \eqref{eq:bioLag} constitutes a biorthogonal ensemble with a rich structure and has been studied as an archetype of biorthogonal ensembles. It is a generalization of the classical Laguerre type Hermitian matrix model eigenvalue distribution \cite{Anderson-Guionnet-Zeitouni10,Forrester10} that is the $\theta = 1$ special case of \eqref{eq:bioLag}. It is known from the general framework of biorthogonal ensembles that the MB ensemble is a determinantal point process, so that there exists a correlation kernel $K_n(x,y) = K^{(V, N)}_n(x, y)$ defined in \eqref{eq:sum_form_K} below such that the density function can be rewritten in the following determinantal form
\begin{equation} \label{eq:kernel_intro}
  \frac{1}{n!}\det\left(K_n(x_i,x_j)\right)_{i,j=1}^n.
\end{equation}
Furthermore, all $k$-point  correlation functions admit a similar determinantal representation in terms of the correlation kernel.

When \(V(x)\) is linear, Borodin's pioneering work \cite{Borodin99} shows that, for each fixed \(\theta\), the scaling limit of \(K_n\) near the origin (the hard edge) exists as \(n\to\infty\) and the limiting kernel depends on \(\theta\). This  family of limiting kernels, which describes the limiting distribution of the left-most particles in \eqref{eq:bioLag}, does not occur in Hermitian matrix models, unless $\theta = 1$. They reduce to the classical Bessel kernel \cite{Forrester93,Tracy-Widom94b} when $\theta = 1$, and reduce to the Meijer G-kernels encountered mainly in products of random matrices and related models (cf. \cite{Akemann-Strahov16,Bertola-Bothner14,Bertola-Gekhtman-Szmigielski14,Kuijlaars-Zhang14,Silva-Zhang20}) when $\theta$ or $1/\theta$ is a positive integer, as observed in \cite{Kuijlaars-Stivigny14}. This family of limiting kernels are proved to be universal, in the sense that if $V$ satisfies $V''(x)x+V'(x)>0$ for all $x > 0$. For the proof, see \cite{Kuijlaars-Molag19} for $\theta = 1/2$, \cite{Molag20} for $\theta^{-1} \in \intZ$, \cite{Wang-Zhang21} for $\theta \in \intZ$ and \cite{Wang23} for general $\theta \in \realR_+$. For this kind of $V$, we say that the model is in the hard edge regime. It is shown in \cite{Claeys-Romano14} that in the hard edge regime mentioned above, the density of the equilibrium measure for the MB ensemble, which is the limiting empirical distribution of the particles \cite[Theorem 2.1 and Corollary 2.2]{Eichelsbacher-Sommerauer-Stolz11} and \cite[Theorem 1.2 and Corollary 1.4]{Butez17}, blows up at the speed of $x^{-1/(\theta + 1)}$ as $x \to 0$.

On the other hand, for a large class of $V(x)$, the equilibrium measure is supported on an interval $[a, b]$ with $a > 0$, and its density function vanishes like a square root at $a$. Although a rigorous proof is absent from the literature, this suggests that the scaling limit of \(K_n\) near \(a\) is the classical Airy kernel defining the Tracy-Widom distribution, as in the case \(\theta = 1\) \cite{Deift-Kriecherbauer-McLaughlin-Venakides-Zhou99}. For this kind of $V$, we say that the model is in the soft edge regime.

In this paper, we consider the scaling limit of $K_n$ at $0$ when $V(x)$ is in the transition regime between the hard edge and the soft edge. We assume that $\theta$ is an integer greater than $1$. In this regime, the equilibrium measure is supported on an interval $[0, b]$, and its density function vanishes at the speed of $x^{(\theta - 1)/(\theta + 1)}$ as $x \to 0$. For example,
\begin{equation} \label{eq:quadratic_potential}
  V(x) = x^2 + \rho x, \quad \text{with} \quad \rho = -\frac{2\sqrt{2}}{\sqrt{\theta}},
\end{equation}
is in the transition regime. If $\theta = 2$, the limiting empirical distribution of particles is shown in Figure \ref{fig:eq_measure} \cite[Figure 6]{Claeys-Romano14} when $\rho$ in \eqref{eq:quadratic_potential} changes around $-2$.

\begin{figure}[htb]
  \centering
  \begin{minipage}[b]{0.24\linewidth}
    \begin{overpic}[width=\linewidth]{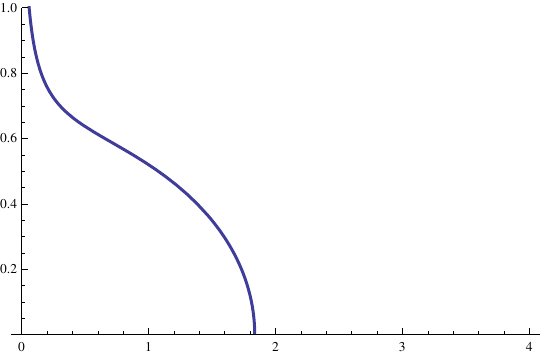}
      \put(45,40){$\rho = 0$}
    \end{overpic}
  \end{minipage}
  \begin{minipage}[b]{0.24\linewidth}
    \begin{overpic}[width=\linewidth]{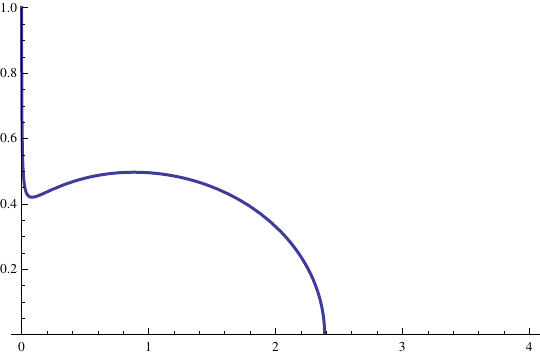}
      \put(45,40){$\rho = -1.8$}
    \end{overpic}
  \end{minipage}
  \begin{minipage}[b]{0.24\linewidth}
    \begin{overpic}[width=\linewidth]{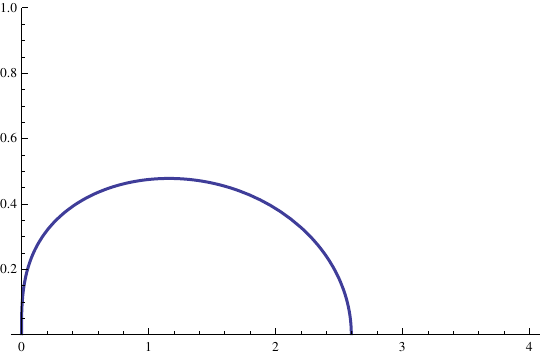}
      \put(45,40){$\rho = -2$}
    \end{overpic}
  \end{minipage}
  \begin{minipage}[b]{0.24\linewidth}
    \begin{overpic}[width=\linewidth]{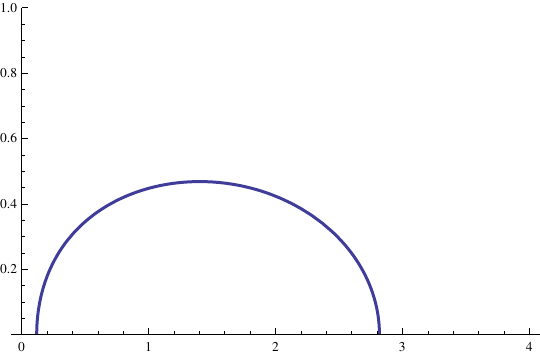}
      \put(45,40){$\rho = -2.5$}
    \end{overpic}
  \end{minipage}
  \caption{The equilibrium measure of the MB ensemble with $\theta = 2$ and $V(x) = x^2 + \rho x$ with $\rho = 0$, $-1.8$, $-2$ and $-2.5$. The figures are plotted over the x-range \([0,4]\) and the y-range \([0,1]\). When $\rho = 0$ and $-1.8$, the ensemble is in the hard edge regime in the sense that the density of the equilibrium measure blows up like $x^{-1/3}$. When $\rho = 0$, it is proved in \cite{Wang-Zhang21} that the left-most particles are around $0$ and their limiting distribution is given by the Meijer G-kernel. It is expected to hold also when $\rho = -1.8$. When $\rho = -2.5$, the density of the equilibrium measure vanishes like a square root at the left-end point of the support of the equilibrium measure, and we expect the left-most particle to be there and its limiting distribution to be the Tracy-Widom distribution given by the Airy kernel. When $\rho = -2$, the density of the equilibrium measure vanishes at the new speed of $x^{1/3}$ at $0$. Courtesy of Tom Claeys.}
  \label{fig:eq_measure}
\end{figure}

In this paper, we show that the correlation kernel of the MB ensemble converges to a new limit in the transition regime, and study the properties of the new limiting kernel.

\subsection{Statement of main results}

Before stating the main results, we need to impose some conditions on the potential function $V$ and introduce some other functions. The equilibrium measure associated with $V$ is the probability measure on $[0, +\infty)$ that minimizes the energy functional
\begin{equation} \label{eq:eq_measure_functional}
  I^{(V)}(\nu) := \frac{1}{2} \iint \log \frac{1}{\lvert x - y \rvert} d\nu(x) d\nu(y) + \frac{1}{2} \iint \log \frac{1}{\lvert x^{\theta} - y^{\theta} \rvert} d\nu(x) d\nu(y) + \int V(x) d\nu(x).
\end{equation}
By an argument similar to that in \cite[Section 6.6]{Deift99}, we have that the equilibrium measure $\mu = \mu^{(V)}$ is characterized by the following Euler-Lagrange conditions:
\begin{align}
  \int \log \lvert x - y \rvert d\nu(y) + \int \log \lvert x^{\theta} - y^{\theta} \rvert d\nu(y) - V(x) = {}& \ell, & x \in {}& \supp(\nu), \label{eq:E-L_1} \\
  \int \log \lvert x - y \rvert d\nu(y) + \int \log \lvert x^{\theta} - y^{\theta} \rvert d\nu(y) - V(x) \leq {}& \ell, & x \in {}& [0, +\infty) \setminus \supp(\nu), \label{eq:E-L_2}
\end{align}
where $\ell$ is a real constant.

Throughout this paper, we assume that  $V$ is real analytic and  the equilibrium measure satisfies the following regularity condition:
\begin{regularitycondition} (One-cut regularity condition away from the left-end point) \label{reg:primary}
  The equilibrium measure $\mu = \mu^{(V)}$ is supported on an interval $[a, b]$ with continuous density, where $0 \leq a < b < \infty$, such that
  \begin{equation} \label{eq:density_function}
    d\mu(x) = \psi(x) dx, \quad x \in (a, b),
  \end{equation}
  where $\psi(x)$ is a continuous function on $(a, b)$, and for small positive numbers $\epsilon_1$ and $\epsilon_2$, the following hold:
  \begin{enumerate}
  \item \label{enu:reg:primary:1}
    $\psi(x) > 0$ strictly for all $x \in (a + \epsilon_1, b - \epsilon_2)$.
  \item \label{enu:reg:primary:2_b}
    There is a positive number $d_2 = d^{(V)}_2$ and $\delta_2 > 0$ depending on $\epsilon_2$, such that
    \begin{equation} \label{eq:behaviour_of_density_b}
      \psi(x) = d_2 (b - x)^{\frac{1}{2}} (1 + h_2(x)), \quad x \in (b - \epsilon_2, b),
    \end{equation}
    $h_2(x) \to 0$ as $x \to b_-$ and $\lvert h_2(x) \rvert \leq \delta_2$ for all $x \in (b - \epsilon_2, b)$. 
  \item \label{enu:reg:primary:3}
    With $d\mu(x) = \psi(x) dx$ supported on $[a, b]$, inequality \eqref{eq:E-L_2} holds strictly for $x \in (b, +\infty)$ and, if $a > 0$, also for $x \in (0, a)$.
  \end{enumerate}
\end{regularitycondition}
Here we avoid the discussion on $\psi(x)$ for $x \in (a, a + \epsilon_1)$.

Like the Laguerre type Hermitian matrix model, whose eigenvalue distribution is the $\theta = 1$ specialization of \eqref{eq:bioLag}, the equilibrium measure of the Muttalib-Borodin ensemble has two typical behaviours at the left end point of its support:
\begin{regularitycondition} (Hard edge regime condition) \label{reg:hard-edge}
  The equilibrium measure $\mu = \mu^{(V)}$ satisfies Regularity Condition \ref{reg:primary} with $a = 0$, and the density function $\psi(x)$ blows up like $x^{-\frac{1}{\theta + 1}}$ as $x \to 0_+$.
\end{regularitycondition}

\begin{regularitycondition} (Soft edge regime condition) \label{reg:soft-edge}
  The equilibrium measure $\mu = \mu^{(V)}$ satisfies Regularity Condition \ref{reg:primary} with $a > 0$, and the density function $\psi(x)$ vanishes like $(x - a)^{\frac{1}{2}}$ as $x \to a_+$. 
\end{regularitycondition}


In this paper, we require that $V$ is one-cut regular in the transition regime. We split the regularity condition into two parts, the one-cut regularity (away from the left-end point) given in Regularity Condition \ref{reg:primary} and the transition regime condition below:

\begin{regularitycondition} (Transition regime condition) \label{enu:reg:primary:2}
  The equilibrium measure $\mu = \mu^{(V)}$ satisfies the one-cut regular condition \ref{reg:primary} with $a = 0$, and on an interval $[0, \epsilon_1]$
  \begin{equation} \label{eq:behaviour_of_density}
    \psi(x) = d_1 x^{\frac{\theta - 1}{\theta + 1}} (1 + h_1(x)), \quad x \in (0, \epsilon_1), 
  \end{equation}
  such that $d_1 = d^{(V)}_1$ is a positive number, $h_1(x) \to 0$ as $x \to 0_+$ and $\lvert h_1(x) \rvert \leq \delta_1$, a positive constant, for all $x \in (0, \epsilon_1)$.
\end{regularitycondition}

This is the transition regime between the hard-edge and the soft edge  universality classes. One concrete example is given by \eqref{eq:quadratic_potential}. If $\rho < -2\sqrt{2}/\sqrt{\theta}$, $V(x)$ in \eqref{eq:quadratic_potential} is in the soft edge regime, and if $\rho > -2\sqrt{2}/\sqrt{\theta}$, $V(x)$ in \eqref{eq:quadratic_potential} is in the hard edge regime, see \cite[Section 4.5.2]{Claeys-Romano14} and Figure \ref{fig:eq_measure} for $\theta = 2$.  When $\theta = 1$, the MB ensemble becomes the Larguerre type unitary ensemble, and the transition regime has been studied in \cite{Claeys-Kuijlaars08} and is described  in terms of a solution of the Painlev\'e XXXIV equation.  A different  Painlev\'e XXXIV solution also arises in the studies of the Gaussian type unitary ensemble with a spectral singularity near the  soft edge   \cite{Its-Kuijlaars-Ostensson08, Its-Kuijlaars-Ostensson09}.  We find and analyse new limiting correlation kernels in the transition regime for $\theta \geq 2$. 

Let $t \in (0, \infty)$ and define
\begin{equation} \label{eq:defn_Vt}
  V_t(x) = \frac{1}{t} V(x), \quad t > 0.
\end{equation}
We define two sequences of monic polynomials $\{p_j(x)=p^{(V_t)}_{n, j}(x) \}_{j=0}^{\infty}$ and $\{q_k(x)=q^{(V_t)}_{n, k}(x)\}_{k=0}^\infty$ by the biorthogonal conditions
\begin{equation} \label{eq:biorthogonality}
  \int^{\infty}_0 p_{j}(x) q_{k}(x^\theta) x^{\alpha} e^{-nV_t(x)} dx = \kappa_j \delta_{j, k},
\end{equation}
where $p_j(x)$ is of degree $j$, $q_k(x)$ is of degree $k$, and $\kappa_j=\kappa^{(V_t)}_{n,j} \neq 0$. The existence and uniqueness of the $\{ p_j(x) \}$ and $\{ q_k(x) \}$ are guaranteed by the fact that functions $\{ x^{k\theta} \}^{n - 1}_{k = 1}$ form a Chebyshev system in the sense of \cite[First definition in Section 4.4]{Nikishin-Sorokin91}. The functions $p_n(z)$ and $q_n(z^{\theta})$ can be interpreted as the averages over the Muttalib-Borodin ensemble, see \cite[Equation (1.14)]{Wang-Zhang21}. Also they are building blocks of the correlation kernel of the Muttalib-Borodin ensemble, since the correlation kernel introduced in \eqref{eq:kernel_intro} is
\begin{equation} \label{eq:sum_form_K}
  K^{(V; N)}_n(x, y) = x^{\alpha} e^{-nV_t(x)}\sum^{n - 1}_{j = 0} \frac{p^{(V_t)}_{n, j}(x)q^{(V_t)}_{n, j}(y^{\theta})}{\kappa^{(V_t)}_{n, j}}, \quad t = \frac{n}{N}.
\end{equation}

Let $V$ be a potential that is one-cut regular in the transition regime, that is, its equilibrium measure $\mu^{(V)}$ satisfies both Regularity Conditions \ref{reg:primary} and \ref{enu:reg:primary:2}. Let
\begin{equation} \label{eq:defn_c}
  c = b\theta/(1 + \theta)^{1 + 1/\theta}
\end{equation}
be a positive number depending on $V$, where $b$ is the right end of the support of $\mu^{(V)}$ as in \eqref{eq:density_function}. We denote
\begin{align} \label{eq:defn_rho}
  \rho = \rho^{(V)} = {}& \left. \pi d_1 \middle/ \sin \left( \frac{2\pi}{\theta + 1} \right) \right., & A_1 = {}& c^{\frac{2\theta}{\theta + 1}} \rho,
\end{align}
where $d_1$ is the same as in \eqref{eq:behaviour_of_density}.


Our first main result of this paper is the large degree asymptotics of the biorthogonal polynomials. To state the result in Theorems \ref{thm:main}, \ref{thm:universality} and \ref{thm:kernel_asy}, we need some model-dependent constants that can only be specified along the course of analysis. However, their quantitative values do not affect the qualitative understanding of the results.

\begin{theorem} \label{thm:main}
  Suppose $V$ satisfies Regularity conditions \ref{reg:primary} and \ref{enu:reg:primary:2}. In addition, suppose $A_2 \neq 0$, where $A_2$ is determined by $V$ through \eqref{eq:defn_A2}. Let $\tau$ be in a compact subset of $\realR$, $x$ be in a compact subset of $\compC$, and $y$ be in a compact subset of $\overline{\halfH}$, where 
  \begin{equation}
    \halfH = \{z \in \compC : z = 0 \text{ or } \arg z \in ( -\frac{\pi}{\theta}, \frac{\pi}{\theta}) \}.
  \end{equation}
  Let
  \begin{equation} \label{eq:t_and_tau}
    t = t_n = 1 - \sqrt{\frac{A_1}{n}} \tau
  \end{equation}
  that depends on $n$. Let $p_j(x) = p^{(V_t)}_{n, k}(x)$, $q_k(x) = q^{(V_t)}_{n, k}(x)$ and $\kappa_j = \kappa^{(V_t)}_{n,j}$ be defined by \eqref{eq:biorthogonality}. As $n \to \infty$ and $j = k = n$, we have the asymptotic formulas
\begin{align}
  p_n \left( \frac{x}{(\rho n)^{\frac{\theta + 1}{2\theta}}} \right) = {}& \sqrt{\frac{\theta}{\theta + 1}} c^{\frac{2(\alpha + 1) - \theta}{2(\theta + 1)}} \left( \rho n \right)^{\frac{1}{2}(\frac{\alpha + 1}{\theta} - \frac{1}{2})} (-1)^n \exp \left( nV_t \left( \frac{x}{(\rho n)^{\frac{\theta + 1}{2\theta}}} \right) \right) e^{n (\ell_t - \Re \gfntilde_t(0))} \notag \\
  & \times \left( \phi^{(\tau)}(x) + \bigO(n^{-\frac{1}{2(2\theta + 1)}}) \right), \label{eq:asy_p_n} \\
  q_n \left( \frac{y^{\theta}}{(\rho n)^{\frac{\theta + 1}{2}}} \right) = {}& \sqrt{\theta + 1} c^{\frac{(\alpha + 1/2) \theta}{\theta + 1}} \left( \rho n \right)^{\frac{1}{2}(\alpha +  \frac{1}{2})} (-1)^n e^{n \Re \gfntilde_t(0)} (\phitilde^{(\tau)}(y) + \bigO(n^{-\frac{1}{2(2\theta + 1)}})), \label{eq:asy_q_n} \\
  \kappa_n = {}& 2\pi \theta^{-1/2} c^{\alpha + 1} e^{n \ell_t} \left( 1 + \bigO(n^{-\frac{1}{2(2\theta + 1)}}) \right), \label{eq:asy_kappa}
\end{align}
where the limit functions $\phi^{(\tau)}(x)$ and $\phitilde^{(\tau)}(y)$ are defined in \eqref{eq:defn_phi^tau} and \eqref{eq:defn_phitilde^tau} below, the constants $\Re \gfntilde_t(0)$ and $\ell_t$ are defined in \eqref{eq:gfn_gfntilde_t} and \eqref{eq:defn_phi} below, and the error terms are uniform in $\tau$, $x$ and $y$, respectively in their compact domains.
\end{theorem}

We remark that the technical condition $A_2 \neq 0$ is probably unnecessary for Theorem \ref{thm:main} to hold. But the proof in our paper requires it.

Our next main result is the limit of $K^{(V; N)}_n(x, y)$.
\begin{theorem} \label{thm:universality}
  Suppose $V$ satisfies Regularity Conditions \ref{reg:primary} and \ref{enu:reg:primary:2} and also $A_2 \neq 0$ as in Theorem \ref{thm:main}. Furthermore, suppose for all $t \in (0, 1)$, the potential $V_t$ satisfies Regularity Conditions \ref{reg:primary} and \ref{reg:soft-edge}.

  Let $\tau$ be in a compact subset of $\realR$ and let $x$ and $y$ be in a compact subset of $[0, +\infty)$. Let $n \to \infty$ and $N = N(n) \to \infty$, in such a way that $\lim_{n \to \infty} (N - n)/\sqrt{n} = -\sqrt{A_1} \tau$ where $A_1$ is defined in \eqref{eq:defn_rho}. Then we have the limit of the correlation kernel defined in \eqref{eq:sum_form_K}
  \begin{equation} \label{eq:kernel_limit_formula}
    \lim_{n \to \infty} \frac{1}{(\rho n)^{\frac{\theta + 1}{2\theta}}} K^{(V; N)}_n \left( \frac{x}{(\rho n)^{\frac{\theta + 1}{2\theta}}}, \frac{y}{(\rho n)^{\frac{\theta + 1}{2\theta}}} \right) = K^{(\tau)}(x, y) := \frac{\theta}{2\pi} x^{\alpha} \int^{\infty}_{\tau} \phi^{(\sigma)}(x) \phitilde^{(\sigma)}(y) d\sigma.
  \end{equation}
\end{theorem}

We remark that all conditions in Theorem \ref{thm:universality} are satisfied for the quadratic $V$ in \eqref{eq:quadratic_potential}. In Theorem \ref{thm:universality}, we additionally require the regularity condition for for all $V_t$ ($t \in (0, 1)$), and it might be avoided if we could make use of the Christoffel-Darboux type formula \cite[Theorem 1.1]{Claeys-Romano14} for the biorthogonal polynomials. This improvement will be investigated in a further publication. The advantage of our proof of Theorem \ref{thm:universality} that does not rely on the Christoffel-Darboux type formula is its potential to be generalized to all $\theta \in \realR_+$.

The theorem below shows that the limiting correlation kernel in Theorem \ref{thm:universality} interpolates the limiting correlation kernels in the hard edge regime and that in the soft edge regime.

\begin{theorem} \label{thm:kernel_asy}
  \begin{enumerate}
  \item \label{enu:thm:kernel_asy:1}
    As $\tau \to -\infty$, for any $x, y \in [0, +\infty)$, we have
    \begin{equation}
      \lim_{\tau \to -\infty} (-\tau)^{-\frac{\theta + 1}{\theta}} K^{(\tau)} \left( (-\tau)^{-\frac{\theta + 1}{\theta}} x, (-\tau)^{-\frac{\theta + 1}{\theta}} y \right) = K^{(\Mei)}(x, y),
    \end{equation}
    where
    \begin{equation}
      K^{(\Mei)}(x, y) = \theta^2 \int^1_0 (ux)^{\alpha} \phi^{(\Mei)}(ux) \phitilde^{(\Mei)}(uy) du
    \end{equation}
    is the limiting kernel $K^{(\alpha, \theta)}(x, y)$ in \cite[Theorem 1.3]{Wang-Zhang21} that defines the hard edge universality of the MB ensemble. Here $\phi^{(\Mei)}$ and $\phitilde^{(\Mei)}$ are expressed in terms of Meijer G-functions, as in \eqref{eq:expr_phi^Mei} and \eqref{eq:expr_phitilde^Mei}.
  \item \label{enu:thm:kernel_asy:2}
    As $\tau \to +\infty$, for any $x, y \in \realR$, we have
    \begin{multline} \label{eq:thm:kernel_asy:2}
      \lim_{\tau \to +\infty} \frac{f(x; \tau)}{f(y; \tau)} (c_1 \tau)^{\frac{\theta + 1}{\theta}} c_2 \tau^{-\frac{4}{3}} K^{(\tau)} \left( (c_1 \tau)^{\frac{\theta + 1}{\theta}} (1 - c_2 \tau^{-\frac{4}{3}}x), (c_1 \tau)^{\frac{\theta + 1}{\theta}} (1 - c_2 \tau^{-\frac{4}{3}}y) \right) \\
      = K^{(\Ai)}(x, y) = \int^{\infty}_0 \Ai(x + u) \Ai(y + u) du,
    \end{multline}
    where $K^{(\Ai)}$ is the Airy kernel that defines the soft edge universality of Hermitian matrix models \cite{Deift99},
    \begin{align} \label{eq:defn_c_1_c_2}
      c_1 = {}& \frac{\theta^2}{\theta^2-1}\theta^{-\frac{1}{\theta+1}}, & c_2 = {}& \frac{(\theta - 1)^{\frac{2}{3}} (\theta + 1)^{\frac{5}{3}}}{2^{\frac{1}{3}} \theta^2}, & f(x; \tau) = e^{\frac{\tau^2}{2} (g_0(\xi) + g_1(\xi))},
    \end{align}
    such that $g_0$ and $g_1$ are defined in \eqref{eq: g_another}, and $\xi = 1 - \theta c_2 \tau^{-4/3} x$.
  \end{enumerate}
\end{theorem}


Here we remark that the the multiplication of $K^{(\tau)}$ by $f(x; \tau)/f(y; \tau)$,  for any positive function $f$, is also a correlation kernel for the same point process.

The limit functions $\phi^{(\tau)}$ and $\phitilde^{(\tau)}$ in Theorems \ref{thm:main} and \ref{thm:universality} are expressed by a model Riemann-Hilbert (RH) problem that is stated in the next subsection.

\subsection{Model Riemann-Hilbert problem and the limit functions $\phi^{(\tau)}$ and $\phitilde^{(\tau)}$}

We denote for $\theta\in \mathbb{N}$  and $\theta\geq 2$, the $(\theta + 1) \times (\theta + 1)$ matrix
\begin{equation} \label{eq:defn_Mcyclic}
  \Mcyclic = I_1 \oplus
  \begin{pmatrix}
    0 & 0 & \cdots & 0 & 1 \\
    1 & 0 & \cdots & 0 & 0\\
    0 & 1 & \ddots & 0 & 0 \\
    \vdots& \ddots & \ddots & \ddots & \vdots\\
    0& \cdots & 0 & 1 & 0
  \end{pmatrix}.
\end{equation}
Here and below, by $I_m$, $m\in \mathbb{N}$, we mean the identity matrix of size $m$, and for two matrices $A$ and $B$, $A \oplus B$ is the matrix $(
\begin{smallmatrix}
  A & 0 \\
  0 & B
\end{smallmatrix}
)$. Below we state the following model RH problem: (with $\beta = \frac{\theta - 2(\alpha + 1)}{1 + \theta}$)
\begin{RHP} \label{RHP:general_model}
  $\Phi(\xi) := \Phi^{(\tau)}(\xi)$ is a $(\theta + 1) \times (\theta + 1)$ matrix-valued function on $\compC$ except for $\realR$ and $i\realR$. 
   \begin{enumerate} 
  \item \label{enu:RHP:general_model:1}
    $\Phi(\xi)$ satisfies the following jump conditions: with $\realR_+$ and $\realR_-$ oriented from left to right, $\{ it : t \in \realR_+ \}$ oriented upwards, and $\{ -it : t \in \realR_+ \}$ oriented downwards, as shown in Figure \ref{fig:jumps-MeiG} in Section \ref{subsec:local_para},
    \begin{equation} \label{eq:defn_J_hard_to_soft}
      \Phi_+(\xi) = \Phi_-(\xi) J^{(\theta)}_{\Phi}(\xi), \quad \text{where} \quad J^{(\theta)}_{\Phi}(\xi) =
      \begin{cases}
        \begin{pmatrix}
          1 & e^{\beta \pi i} \\
          0 & 1
        \end{pmatrix}
        \oplus I_{\theta - 1},
        & \arg \xi = \frac{\pi}{2}, \\[10pt]
        \begin{pmatrix}
          1 & -e^{-\beta \pi i} \\
          0 & 1
        \end{pmatrix}
        \oplus I_{\theta - 1},
        & \arg \xi = -\frac{\pi}{2}, \\[10pt]
        \begin{pmatrix}
          0 & 1 \\
          1 & 0
        \end{pmatrix}
        \oplus I_{\theta - 1},
        & \xi \in \realR_+, \\[10pt]
        \Mcyclic,
            & \xi \in \realR_-.
      \end{cases}
    \end{equation}
  \item \label{enu:RHP:general_model:2}
    $\Phi(\xi)$ has the following asymptotic condition at infinity: as $\xi \to \infty$ for $\xi \in \compC_{\pm}$
    \begin{equation} \label{eq:asyNew}
      \Phi(\xi) = \left( I + \bigO(\xi^{-1}) \right) \Upsilon(\xi) \Omega_{\pm} e^{-\Theta(\xi)},
    \end{equation}
    where
         
   \begin{equation}\label{eq:asyold}
     \Theta(\xi) =
     \begin{cases}
       \diag \left( \frac{\theta + 1}{2\theta} e^{\frac{-2\pi i}{\theta + 1}} \xi^{\frac{2}{\theta + 1}} - \tau e^{\frac{-\pi i}{\theta + 1}} \xi^{\frac{1}{\theta + 1}}, \frac{\theta + 1}{2\theta} e^{\frac{2\pi i}{\theta + 1}} \xi^{\frac{2}{\theta + 1}} - \tau e^{\frac{ \pi i}{\theta + 1}}\xi^{\frac{1}{\theta + 1}} \right) & \\
       \oplus \diag \left( \frac{\theta + 1}{2\theta} e^{\frac{2(2k - 1) \pi i}{\theta + 1}} \xi^{\frac{2}{\theta + 1}} - \tau e^{\frac{(2k - 1)\pi i}{\theta + 1}} \xi^{\frac{1}{\theta + 1}} \right)^{\theta}_{k = 2}, & \xi \in \compC_+, \\
       \diag \left(  \frac{\theta + 1}{2\theta} e^{\frac{2\pi i}{\theta + 1}} \xi^{\frac{2}{\theta + 1}} - \tau e^{\frac{ \pi i}{\theta + 1}} \xi^{\frac{1}{\theta + 1}}, \frac{\theta + 1}{2\theta} e^{\frac{-2\pi i}{\theta + 1}} \xi^{\frac{2}{\theta + 1}} - \tau e^{\frac{-\pi i}{\theta + 1}} \xi^{\frac{1}{\theta + 1}} \right) & \\
       \oplus \diag \left( \frac{\theta + 1}{2\theta} e^{\frac{2(2k - 1) \pi i}{\theta + 1}} \xi^{\frac{2}{\theta + 1}} - \tau e^{\frac{(2k - 1)\pi i}{\theta + 1}}  \xi^{\frac{1}{\theta + 1}} \right)^{\theta}_{k = 2}, & \xi \in \compC_-,
     \end{cases}
   \end{equation}
    \begin{align}\label{def:Lpm}
      \Upsilon(\xi) = {}& \diag \left( e^{-\frac{k}{\theta + 1} \pi i} \xi^{\frac{k}{\theta + 1}} \right)^{\theta}_{k = 0}, & \Omega_+ = {}& \left( e^{\frac{2jk}{\theta + 1} \pi i} \right)^{\theta}_{j, k = 0}, & \Omega_-= \Omega_+ \times \left(
                                                                                            \begin{pmatrix}
                                                                                              0 & 1 \\
                                                                                              1 & 0
                                                                                            \end{pmatrix}
                                                                                            \oplus I_{\theta - 1} \right).
    \end{align}
  \item \label{enu:RHP:general_model:3}
    $\Phi(\xi)$ has the following local growth condition as $\xi \to 0$ 
    \begin{equation} \label{eq:defn_N(xi)}
      \Phi(\xi) = N(\xi) \xi^{-\frac{\beta}{2\theta}} \diag \left( \xi^{-\frac{\alpha + 1 - \theta}{\theta}}, 1, \xi^{\frac{1}{\theta}}, \xi^{\frac{2}{\theta}}, \dotsc, \xi^{\frac{\theta - 1}{\theta}} \right)E,
    \end{equation}
    where $N(\xi) = N^{(\tau)}(\xi)$ is analytic at $0$, and $E$ is a piecewise constant (except for possibly one $\log \xi$ entry) matrix in each region between the rays. 
   In the sector $\arg \xi \in (\pi/2, \pi)$, the matrix $E$ is given by
    \begin{equation} \label{eq:defn_E_region_II}
      E =E_0
      \begin{pmatrix}
        1 & \vec{0}^T \\
        \vec{u} & C_{\theta \times \theta}
      \end{pmatrix},
      \quad \text{where} \quad
      C=\left(e^{\frac{2j - 2 - \beta}{\theta} (k - 1)\pi i}\right)_{j,k=1}^{\theta}.
    \end{equation}
    Here the vector $\vec{u}$ depends on the parameter $\alpha$.  For $\alpha \notin \intZ$, the entries of $\vec{u}$ are defined as 
    
    \begin{equation}
      u_j  =  \frac{1}{e^{\beta \pi i} + e^{\frac{2j - 2 - \beta}{\theta} \pi i}}, \quad j = 1, \dotsc, \theta.
    \end{equation}
For $\alpha\in\intZ$ and  $\alpha \equiv -k \pmod{\theta}$ where $k \in \{ 1, 2, \dotsc, \theta \}$,  we have  $e^{\beta \pi i} + e^{\frac{2k - 2 - \beta}{\theta} \pi i} = 0$. Then 
 the entry $u_k$ is replaced by
 \begin{equation} \label{eq:E_log}
   u_k = -\frac{e^{-\beta \pi i}}{2\pi i} \log \xi,
 \end{equation}
 and all other entries of  $\vec{u}$  remain the same. In the remaining sectors,  the matrix $E$ is defined via \eqref{eq:defn_E_region_II} and the jump condition \eqref{eq:defn_J_hard_to_soft} as 
 \begin{equation}
   E=E_0 \prod_{l=1}^k J_l^{-1}, \quad  \arg \xi \in \left(\frac{(1-k)\pi}{2}, \frac{(2-k)\pi}{2}\right), \quad k=1,2,3,
 \end{equation}
 where  $J_1$, $J_2$ and $J_3$ are the jump matrices in \eqref{eq:defn_J_hard_to_soft} along $ \arg \xi=\pi/2, 0, -\pi/2 $, respectively.  
  \end{enumerate} 
\end{RHP}


We have the following result:
\begin{theorem} \label{thm:unique_solvability}
  RH problem \ref{RHP:general_model} has a unique solution for all $\tau \in \realR$.
\end{theorem}
We also note that since
\begin{equation}
  \det \Phi(\xi) = \pm (\theta + 1)^{\frac{\theta + 1}{2}} e^{-\frac{\theta(\theta + 1)}{4} \pi i} \xi^{\frac{\theta}{2}}, \quad \xi \in \compC_{\pm},
\end{equation}
$\Phi^{-1}(\xi)$ is well defined.

We denote the analytic function $\phi^{(\tau)}(z)$ on $\compC$ and the analytic function $\phitilde^{(\tau)}(z)$ on $\halfH$ by
\begin{align}
  \phi^{(\tau)}(z) = {}& (-z)^{\frac{\beta}{2}} \times
  \begin{cases}
    \Phi_{0, k}(z^{\theta}), & k = 2, \dotsc, \theta, \text{ and } \arg z \in (\frac{2k - 3}{\theta} \pi, \frac{2k - 1}{\theta} \pi) \\
                             & \text{or } k = 1 \text{ and } \arg z \in (\frac{\pi}{2\theta}, \frac{\pi}{\theta}) \cup (-\frac{\pi}{\theta}, -\frac{\pi}{2\theta}), \\
    \Phi_{0, 1}(z^{\theta}) + e^{\beta \pi i} \Phi_{0, 0}(z^{\theta}), & \arg z \in (0, \frac{\pi}{2\theta}), \\
    \Phi_{0, 1}(z^{\theta}) + e^{-\beta \pi i} \Phi_{0, 0}(z^{\theta}), & \arg z \in (-\frac{\pi}{2\theta}, 0),
  \end{cases} \label{eq:defn_phi^tau} \\
  \phitilde^{(\tau)}(z) = {}& z^{\frac{\theta}{2}(1 + \beta)} \times
  \begin{cases}
    e^{-\frac{1}{2}(1 + \beta) \pi i} (\Phi(z^{\theta}))^{-1}_{0, \theta}, & \arg z \in (\frac{\pi}{2\theta}, \frac{\pi}{\theta}), \\
    e^{-\frac{1}{2}(1 + \beta) \pi i} (\Phi(z^{\theta}))^{-1}_{0, \theta} + e^{\frac{1}{2}(1 + \beta) \pi i} (\Phi(z^{\theta}))^{-1}_{1, \theta}, & \arg z \in (0, \frac{\pi}{2\theta}), \\
    e^{\frac{1}{2}(1 + \beta) \pi i} (\Phi(z^{\theta}))^{-1}_{0, \theta}, & \arg z \in (-\frac{\pi}{\theta}, -\frac{\pi}{2\theta}), \\
    e^{\frac{1}{2}(1 + \beta) \pi i} (\Phi(z^{\theta}))^{-1}_{0, \theta} + e^{-\frac{1}{2}(1 + \beta) \pi i} (\Phi(z^{\theta}))^{-1}_{1, \theta}, & \arg z \in (-\frac{\pi}{2\theta}, 0),
  \end{cases} \label{eq:defn_phitilde^tau}
\end{align}
where power functions take the principal branch.

We have the following limit results for $\phi^{(\tau)}$ and $\phitilde^{(\tau)}$:
\begin{lemma} \label{lem:phi_phitilde_sy}
  Suppose $x$ is in a compact subset of $\compC$ and $y$ is in a compact subset of $\overline{\halfH}$.
  \begin{enumerate}
  \item \label{enu:lem:phi_phitilde_sy:1}
    As $\tau \to -\infty$, we have
    \begin{align}
      \phi^{(\tau)}((-\tau)^{-\frac{\theta + 1}{\theta}} x) = {}& (-\tau)^{-\frac{\theta + 1}{2\theta} \beta} (2\pi)^{1 - \frac{\theta}{2}} \sqrt{\theta + 1} \phi^{(\Mei)}(x) (1 + \bigO((-\tau)^{-\frac{2}{\theta + 1}})), \label{eq:expr_phi^Mei} \\
      \phitilde^{(\tau)}((-\tau)^{-\frac{\theta + 1}{\theta}} y) = {}& (-\tau)^{\theta - \frac{\theta + 1}{2} (1 + \beta)} (2\pi)^{\frac{\theta}{2}} \sqrt{\theta + 1} \phitilde^{(\Mei)}(y) (1 + \bigO((-\tau)^{-\frac{2}{\theta + 1}})), \label{eq:expr_phitilde^Mei}
    \end{align}
    where $\phi^{(\Mei)}$ and $\phitilde^{(\Mei)}$ are the limit functions in \cite[Theorem 1.1]{Wang-Zhang21}
    \begin{align}
      \phi^{(\Mei)}(x) = {}& z^{\theta - \alpha - 1} \MeijerG{\theta, 0}{0, \theta + 1}{-}{\frac{\alpha - \theta + 1}{\theta}, \frac{\alpha - \theta + 2}{\theta}, \dotsc, \frac{\alpha - 1}{\theta}, \frac{\alpha}{\theta}, 0}{x^{\theta}}, \\
      \phitilde^{(\Mei)}(y) = {}& \MeijerG{1, 0}{0, \theta + 1}{-}{0, -\frac{\alpha}{\theta}, \frac{1 - \alpha}{\theta}, \dotsc, \frac{\theta - 1 - \alpha}{\theta}}{y^{\theta}},
    \end{align}
    where $\MeijerG{\theta, 0}{0, \theta + 1}{}{}{\cdot}$ and $\MeijerG{1, 0}{0, \theta + 1}{}{}{\cdot}$ are Meijer G-functions defined in \cite[Chapter V]{Luke69}.
  \item \label{enu:lem:phi_phitilde_sy:2}
    As $\tau \to +\infty$, we have (with $\betatilde = -\frac{2\alpha + 1}{1 + \theta}$)
    \begin{multline} \label{eq:phi_asy_airy}
      \phi^{(\tau)} \left( (c_1 \tau)^{\frac{\theta + 1}{\theta}} (1 - c_2 \tau^{-\frac{4}{3}} x) \right) = (c_1 \tau)^{\frac{\theta + 1}{\theta} \frac{\beta}{2}} 2\sqrt{\pi} \theta^{-\frac{\beta + 1}{2}} \left( \frac{\theta + 1}{2} \right)^{\frac{1}{4}} c^{-\frac{1}{4}}_2 \tau^{\frac{1}{3}} f(x; \tau) \\
      \times (\Ai(x) + \bigO(\tau^{-\frac{2}{3}})),
    \end{multline}
    \begin{multline} \label{eq:phitilde_asy_airy}
      \phitilde^{(\tau)} \left( (c_1 \tau)^{\frac{\theta + 1}{\theta}} (1 - c_2 \tau^{-\frac{4}{3}} x) \right) = (c_1 \tau)^{(\theta + 1)\betatilde} \frac{2 \sqrt{\pi}}{\theta + 1} \theta^{\frac{\betatilde}{2}} \left( \frac{\theta + 1}{2} \right)^{\frac{1}{4}} c^{-\frac{1}{4}}_2 \tau^{\frac{1}{3}} f(x; \tau)^{-1} \\
      \times (\Ai(x) + \bigO(\tau^{-\frac{2}{3}})),
    \end{multline}
    where $c_1$, $c_2$ and $f(x; \tau)$ are defined in \eqref{eq:defn_c_1_c_2}.
  \end{enumerate}
\end{lemma}

RH problem \ref{RHP:general_model} has interesting integrable properties. In this paper, we consider only the $\theta = 2$ case and leave the general case to a further publication.

We have the following Lax pair for  $\Phi = \Phi^{(\tau)}$, whose compatibility is described by the Chazy I equation \cite[Equation (A4)]{Cosgrove00}. 

\begin{theorem}\label{Pro:LaxPair}
Let $\Phi$ be the solution of RH problem \ref{RHP:general_model} with $\theta = 2$. We have the Lax pair 
\begin{equation} \label{eq:LaxPair}
  \frac{d}{d\xi} \Phi = A \Phi, \quad  \frac{d}{d\tau} \Phi = B \Phi,
\end{equation}
where, with $D = \diag(1, \sqrt{2}, 2)$,
\begin{align}\label{eq:ABExpand}
  A = {}& D( 2^{-\frac{3}{2}} A_0 + \xi^{-1} A_{-1}) D^{-1}, & B = {}& D( \frac{1}{2} \xi B_1 + \sqrt{2} B_0) D^{-1},
\end{align}
with 
\begin{align}\label{eq:Coefficients}
  A_0 = {}&
            \begin{pmatrix}
              0 & 0 & 0 \\
              1 & 0 & 0 \\
              -\frac{\sqrt{2}}{3} \tau & 1 & 0
            \end{pmatrix},
  & A_{-1} = {}&
                 \begin{pmatrix}
                   b & c + \frac{\sqrt{2}}{3} \tau & -1 \\
                   a & -b - f + \frac{1}{3} & -c + \frac{\sqrt{2}}{3} \tau \\
                   d & k & f + \frac{2}{3}
                 \end{pmatrix}, \\
  B_1 = {}&
            \begin{pmatrix}
              0 & 0 & 0 \\
              0 & 0 & 0 \\
              -1 & 0 & 0 
            \end{pmatrix},
  & B_0 = {}&
              \begin{pmatrix}
                -c & 1 & 0 \\
                f - \frac{\sqrt{2}}{3} \tau c & 0 & 1 \\
                \frac{\sqrt{2}}{3} \tau(b + f) - (a + k) & b + \frac{\sqrt{2}}{3} \tau c & c
              \end{pmatrix}.
\end{align}
Here $a,b,c,d,f,k$ are analytic for $\tau\in\mathbb{R}$ and  satisfy the relations
\begin{equation}\label{eq:Constant_Equation_1}
f-b-\frac{\sqrt{2}}{3} c\tau+\gamma=0,
\end{equation}
\begin{equation}\label{eq:Constant_Equation_2}
a+k-c(\gamma+\frac{1}{3})=0,
\end{equation}
and
\begin{equation}\label{eq:Consatant_Equation_3}
  d+c(k-a)-\frac{\sqrt{2}}{3} \tau(a+k)-(b^2+bf+f^2)+\frac{1}{3}(b-f)-\gamma= 0,
\end{equation}
with $\gamma= \frac{1}{36} + \frac{\alpha}{12} - \frac{\alpha^2}{12}$.
The  compatibility of the Lax pair is equivalent to the following system of differential equations
\begin{align}
  \frac{c'}{\sqrt{2}} = {}& -c^2 - b - f, \label{eq:NonEqs_c} \\
  \frac{b'}{\sqrt{2}} = {}& -cf - k + \frac{\sqrt{2}}{3} \tau(b + c^2) + \frac{2}{9} \tau^2 c, \label{eq:NonEqs_b} \\
  \frac{f'}{\sqrt{2}} = {}& -bc + a - \frac{\sqrt{2}}{3} \tau(f + c^2) + \frac{2}{9} \tau^2 c, \label{eq:NonEqs_f} \\
  \frac{a'}{\sqrt{2}} = {}& 2bf + f^2 - ck + d - \frac{1}{3} f - \frac{\sqrt{2}}{3} \tau(bc - a - k - \frac{1}{3} c) - \frac{2}{9} \tau^2(b + f),\label{eq:NonEqs_a} \\
  \frac{k'}{\sqrt{2}} = {}& -b^2 - 2bf - ac - d - \frac{1}{3} b - \frac{\sqrt{2}}{3} \tau(cf + a + k + \frac{1}{3} c) + \frac{2}{9} \tau^2(b + f), \label{eq:NonEqs_k}\\
  \frac{d'}{\sqrt{2}} = {}& 2cd - bk + af + \frac{2}{3} a + \frac{2}{3} k - \frac{\sqrt{2}}{3} \tau(f^2 - b^2 - ac - ck + \frac{2}{3} b + \frac{2}{3} f).\label{eq:NonEqs_d}
 \end{align}
 Using \eqref{eq:Constant_Equation_1}-\eqref{eq:NonEqs_f} to express $a$, $b$,  $d$, $f$ and $k$ in terms of $c$, $c'$ and $c''$,  we obtain from 
\eqref{eq:NonEqs_a} and \eqref{eq:NonEqs_k} a third order differential equation 
\begin{equation}\label{eq:Third_order_Eqs_c}
c'''+3 \cdot 2^{\frac{3}{2}} c'^2+ \frac{4}{3} \tau^2c'+4 \tau c+\frac{\sqrt{2}}{9}(1+3\alpha-3\alpha^2)=0.
\end{equation}
Let $y(\tau)=c(\frac{\tau}{\sqrt{2}})+\frac{\tau^3}{108}$, then $y$ satisfies the  Chazy I equation
\begin{equation}\label{eq:Chazy_Eqs_y}
  y'''+6y'^2+ \tau y-\frac{1}{72}\tau^4+\frac{1}{6}(\alpha-\alpha^2)=0.
\end{equation}
Furthermore, let $u(\tau)=\sqrt{2}c(\tau)+\frac{4}{27}\tau^3$, we have 
\begin{equation}\label{eq:Chazy_Eqs_u} (u'')^2+4(u')^3-4(\tau u'-u)^2+\frac{4}{3}(\alpha-\alpha^2-1)u'+\frac{4}{27}(\alpha+1)(2\alpha-1)(\alpha-2)=0, \end{equation}
which is a second-degree Chazy equation \cite [Equation (5.6)]{Cosgrove93}   and can be solved  in terms of the Painlev\'e IV equation via the correspondence given in \cite[Equations (5.41)-(5.44)]{Cosgrove93}. 
\end{theorem}


\begin{rmk} \label{rmk:integrability}
Let $v(\tau)=3\sqrt{6}c'\left(3^{\frac{1}{4}}\tau/2\right)$, by differentiating \eqref{eq:Third_order_Eqs_c} twice, we have
 \begin{equation}\label{eq:fourth_order_Equation}
  v^{(4)}+v'^2+vv''+\frac{\tau^2v''}{4}+ \frac{7\tau v'}{4} +2v=0,
\end{equation}
which is  the similarity reduction of the Boussinesq equation \cite[Equations (2.9)]{Clarkson-Kruskal89}.  
This equation \eqref{eq:fourth_order_Equation} has also been derived from  the similarity reduction of the Drinfeld-Sokolov  hierarchies associated to the affine Kac-Moody algebra of type $A_2^{(1)}$ in \cite[Equation (5.18)]{Liu-Wu-Zhang22}. After some gauge transformation, the Lax pair \eqref{eq:LaxPair} can be expressed in the form presented in \cite[Example 5.4]{Liu-Wu-Zhang22}. It is worth mentioning that the similarity reductions of type $A_1^{(1)}$ give the Painlev\'{e} II and  Painlev\'{e} XXXIV equations  \cite[Example 5.2]{Liu-Wu-Zhang22}, and in the transition regime of the MB ensemble with $\theta = 1$, it has been shown in \cite{Claeys-Kuijlaars08}, \cite{Its-Kuijlaars-Ostensson08} and \cite{Its-Kuijlaars-Ostensson09} that the Lax pair of the model RH problem yields the Painlev\'{e} XXXIV equation. In \cite{Liu-Wu-Zhang22}, the similarity reductions of the Drinfeld-Sokolov hierarchies associated to the other types of affine Kac-Moody algebra are also considered; see also \cite{Noumi02}. For integer $\theta > 2$, we expect that the solution $\Phi(\xi)$ for the model RH problem \ref{RHP:general_model} would be related to similarity reductions of the Drinfeld-Sokolov hierarchies of type $A^{(1)}_{\theta}$, and it will be the subject of further work.
 \end{rmk}

\subsection{Strategy and organization}\label{sec:Strategy}

The main technique in the proof of Theorem \ref{thm:main} is the analysis of vector-valued Riemann-Hilbert problems, like the proof of the hard edge universality for integer $\theta$ in \cite{Wang-Zhang21}. The new feature in our paper is that for each $\theta \geq 2$, the $(\theta + 1) \times (\theta + 1)$ model RH problem for the local parametrix at $0$, 
  is new and shown to  connect to Chazy I equation and the  Drinfeld-Sokolov hierarchies, see  Theorem \ref{Pro:LaxPair} and Remark \ref{rmk:integrability}. 
The solvability of the new model RH problem is non-trivial, since it does not have an explicit construction as for the hard edge model RH problems in \cite{Wang-Zhang21} by Meijer G-functions. We prove a vanishing lemma to guarantee the solvability. We remark that the vanishing lemmas for RH problems of dimension greater than $2$ are rare. To our limited knowledge, we only know that Charlier and Lenells have worked out one vanishing lemma for a $3 \times 3$ RH problem in a quite different setting \cite{Charlier-Lenells22}.



This paper is organized as follows. Section \ref{sec:equilibrium} provides results about the equilibrium measure that will be used in Sections \ref{sec:RH_transitive} and \ref{sec:RH_soft_edge}. Section \ref{sec:RH_transitive} contains the RH problems for the proof of Theorem \ref{thm:main}. Section \ref{sec:RH_soft_edge} contains the additional RH problems for the proof of Theorem \ref{thm:universality}. Then in Section \ref{sec:proof_main} the two theorems are proved. Theorem \ref{thm:unique_solvability} is proved in Section \ref{sec:vanishing_lemma}. Theorem \ref{Pro:LaxPair} is proved in Section \ref{sec:Lax_pair}. Finally, Theorem \ref{thm:kernel_asy} is proved in Sections \ref{sec:asy_tau_to_negative} and \ref{sec:asy_tao_to_+infty}.


\paragraph{Notations}

Throughout this paper, the following notations are frequently used. We denote by $\compC_{\pm}=\{z\in\compC \mid \pm \Im z>0\}$, by $D(z_0, \delta)$ the open disc centred at $z_0$ with radius $\delta>0$, i.e.,
\begin{equation} \label{eq:size_U_b}
  D(z_0,\delta) := \{z\in \mathbb{C} : \lvert z - z_0 \rvert < \delta \},
\end{equation}
We use $J_n$ to denote the $n$ dimensional reversed identity matrix $(\delta_{j, n - 1 - k})^{n - 1}_{j, k = 0}$.

\subsection*{Acknowledgements}

We thank the anonymous referees for their careful reading and constructive suggestions.
Dong Wang was partially supported by the National Natural Science Foundation of China under grant number 12271502, and the University of Chinese Academy of Sciences start-up grant 118900M043. Shuai-Xia Xu was partially supported by the National Natural Science Foundation of China under grant numbers  12431008, 12371257 and 11971492, and by  Guangdong Basic and Applied Basic Research Foundation (Grant No. 2022B1515020063). We thank Lun Zhang for helpful discussion at the early stage of the project, Xiang-Ke Chang for helpful discussion on Chazy equations and Chao-Zhong Wu for drawing  our attention to the relations with Drinfeld-Sokolov hierarchies.

\section{Equilibrium measure in transition regime and soft edge regime} \label{sec:equilibrium}

For any $u, v, w \in \realR_+$, let, with the power function taking the principal branch,
\begin{align} \label{eq:J_function}
  J_{u, v}(s) = {}& (us + v) \left( \frac{s + 1}{s} \right)^{\frac{1}{\theta}}, & \text{and} && J_w(s) = {}& J_{w, w}(s).
\end{align}

\subsection{Equilibrium measure in transition regime without deformation} \label{subsec:exact_equilibrium_measure}

We assume that the potential $V$ satisfies Regularity Conditions \ref{reg:primary} and \ref{enu:reg:primary:2}, such that the support of its equilibrium measure is $[0, b]$. Recall $c$ defined by \eqref{eq:defn_c}.

The equilibrium measure for $V$ is described by $J_c(s)$, as shown in \cite[Theorem 1.11]{Claeys-Romano14}. Here we recall the results there and set up the notation to be used later. Let $\gamma_1 \subseteq \compC_+$ be the curve connecting $-1$ to $1/\theta$ such that $J_c$ maps $\gamma_1$ to the interval $[0, b]$. Then let $\gamma_2 = \overline{\gamma_1} \subseteq \compC_-$ and $\gamma = \gamma_1 \cup \gamma_2$ be the closed contour oriented counterclockwise and $D$ be the region enclosed by $\gamma$. Recall that $J_c$ maps $\compC \setminus \overline{D}$ to $\compC \setminus [0, b]$ and $D \setminus [-1, 0]$ to $\halfH \setminus [0, b]$ univalently. Then we define $\Iinv_1(z)$ and $\Iinv_2(z)$ to be the inverses of $J_c(z)$ on $\compC \setminus \overline{D}$ and $D \setminus [-1, 0]$ respectively. Also we define $\Iinv_+(x)$ (resp.~$\Iinv_-(x)$) to be the inverse of the map $J_c(s)$ from $\gamma_1$ (resp.~$\gamma_2$) to $[0, b]$. We have as $z\to 0$ \cite[Equations (3.26) and (3.27)]{Wang-Zhang21}
\begin{align}
  \Iinv_1(z) = {}& -1 + c^{-\frac{\theta}{\theta + 1}} \times
                   \begin{cases}
                     e^{\frac{\pi i}{\theta + 1}} z^{\frac{\theta}{\theta + 1}} (1 + \bigO(z^{\frac{\theta}{\theta + 1}})), & \arg z \in (0, \pi), \\
                     e^{-\frac{\pi i}{\theta + 1}} z^{\frac{\theta}{\theta + 1}} (1 + \bigO(z^{\frac{\theta}{\theta + 1}})), & \arg z \in (-\pi, 0),
                   \end{cases} \label{eq:I_1_at_0} \\
  \Iinv_2(z) = {}& -1 + c^{-\frac{\theta}{\theta + 1}} \times
                   \begin{cases}
                     e^{-\frac{\pi i}{\theta + 1}} z^{\frac{\theta}{\theta + 1}} (1 + \bigO(z^{\frac{\theta}{\theta + 1}})), & \arg z \in (0, \frac{\pi}{\theta}), \\
                     e^{\frac{\pi i}{\theta + 1}} z^{\frac{\theta}{\theta + 1}} (1 + \bigO(z^{\frac{\theta}{\theta + 1}})), & \arg z \in (-\frac{\pi}{\theta}, 0).
                   \end{cases} \label{eq:I_2_at_0}
\end{align}

We denote the so-called $\gfn$-functions and their derivatives \cite[Equations (4.4), (4.5) and (4.9)]{Claeys-Romano14}
\begin{align} \label{eq:defn_g_gtilde_G_Gtilde}
  \gfn(z) = {}& \int^b_0 \log(z - y) \psi(y) dy, & \gfntilde(z) = {}& \int^b_0 \log(z^{\theta} - y^{\theta}) \psi(y) dy, & G(z) = {}& \gfn'(z), & \G(z) = {}& \gfntilde'(z),
\end{align}
 where $\psi(x)$ is the density of the equilibrium
measure and the functions $\gfn(z)$ and $\gfntilde(z)$ are analytic in $\compC \setminus (-\infty, b]$ and $\halfH \setminus [0, b]$ respectively.
We let
\begin{align} \label{eq:defn_phi_t=1}
  \phi(z) = {}& \gfn(z) + \gfntilde(z) - V(z) - \ell, & \ell = {}& (\gfn)_{\pm}(0) + (\gfntilde)_{\mp}(0) -V(0),
\end{align}
where $\ell$ is the same as in \eqref{eq:E-L_1} and \eqref{eq:E-L_2}.

We define \cite[Equation (4.12)]{Claeys-Romano14}
\begin{equation} \label{eq:formula_N_in_t=1}
  N(s) = 
  \begin{cases}
    N_{\In}(s) = \frac{1}{2\pi i} \oint_{\gamma} \frac{U(J_c(\xi))}{\xi - s} d\xi - 1, & s \in D \setminus [-1, 0], \\
    N_{\Out}(s) = -\frac{1}{2\pi i} \oint_{\gamma} \frac{U(J_c(\xi))}{\xi - s} d\xi + 1, & s \in \compC \setminus \overline{D},
  \end{cases}
  \quad \text{where} \quad U(z) = zV'(z).
\end{equation}
Since $J_c(s)$ is analytic in $\compC \setminus \overline{D}$, we can enlarge $\gamma$ to a contour $\gamma_{\Out}$ that encloses $\gamma$ and use this contour instead of $\gamma$ to evaluate the integral for $N_{\In}(s)$ in \eqref{eq:formula_N_in_t=1}, and we find in particular that $N_{\In}(s)$ is analytic in neighbourhoods of $-1$ and $\theta^{-1}$.

By \cite[Equations (4.10) and (4.11)]{Claeys-Romano14}, $G(z)$ and $\G(z)$ in \eqref{eq:defn_g_gtilde_G_Gtilde} can be defined by $N(s)$ by
\begin{equation} \label{eq:N(s)_t=1}
  N(s) =
  \begin{cases}
    J_c(s) G(J_c(s)), & s \in \compC \setminus \overline{D}, \\
    J_c(s) \G(J_c(s)), & s \in D \setminus [-1, 0],
  \end{cases}
\end{equation}
and then the density function $\psi(x)$ in \eqref{eq:density_function} satisfies
\begin{equation} \label{eq:psi(x)_in_G_tilde_G}
  \begin{split}
    \psi(x) = {}& \frac{1}{2\pi i}(\G_-(x) - \G_+(x)) = \frac{x^{-1}}{2\pi i} (N_{\In}(\Iinv_+(x)) - N_{\In}(\Iinv_-(x))) \\
    = {}& \frac{1}{2\pi i}(G_-(x) - G_+(x)) = \frac{x^{-1}}{2\pi i} (N_{\Out}(\Iinv_-(x)) - N_{\Out}(\Iinv_+(x))).
  \end{split}
\end{equation}
The limit behaviour of $\psi(x)$ as $x \to 0$ stated in \eqref{eq:behaviour_of_density} implies 
\begin{align} \label{eq:critical_in_N_in}
  N'_{\In}(-1) = \frac{1}{2\pi i} \oint_{\gamma_{\Out}} \frac{U(J_c(s))}{(s + 1)^2} ds = {}& 0, & N''_{\In}(-1) = {}& \frac{2}{2\pi i} \oint_{\gamma_{\Out}} \frac{U(J_c(s))}{(s + 1)^3} ds = 2A_1,
\end{align}
where $A_1$ is a positive number determined by $d_1$ in \eqref{eq:defn_rho}.

Since the density $\psi(x)$ and related functions are expressed by $N(s)$, Regularity Condition \ref{reg:primary} can be expressed as conditions on $N(s)$. For example, part \ref{enu:reg:primary:2_b} of Regularity Condition \ref{reg:primary} is equivalent to
\begin{equation} \label{eq:N'(theta^-1}
  N'_{\In}(\theta^{-1}) = \frac{1}{2\pi i} \oint_{\gamma_{\Out}} \frac{U(J_c(s))}{(s - \theta^{-1})^2} ds > 0,
\end{equation}
and part \ref{enu:reg:primary:3} is equivalent to
\begin{equation}
  \phi(x) > 0, \quad x \in (b, +\infty),
\end{equation}
which is a property of $N(s)$ since $\phi(z)$ is related to $N(s)$ by \eqref{eq:defn_g_gtilde_G_Gtilde}, \eqref{eq:defn_phi_t=1} and \eqref{eq:N(s)_t=1}.

\subsection{Deformation of the equilibrium measure: algebraic setup} \label{subsec:deformation_algebraic}

Now we consider the deformation of the equilibrium measure when the potential function $V(x)$ is replaced by $V_t(x)$ such that $t$ is close to $1$.

Recall the contour $\gamma_{\Out}$ defined in Section \ref{subsec:exact_equilibrium_measure}. We define real-valued functions
\begin{align} \label{eq:defn_F_t(x)}
  E(u, v) = {}& \frac{1}{2\pi i} \oint_{\gamma_{\Out}} \frac{U(J_{u, v}(s))}{s + 1} ds, & F(u, v) = {}& \frac{1}{2\pi i} \oint_{\gamma_{\Out}} \frac{U(J_{u, v}(s))}{s} ds, 
\end{align}
We want to solve, for $t$ in a small neighbourhood of $1$, the equation that
\begin{align} \label{eq:F_E}
  F(u, v) = {}& (1 + \theta)t, & E(u, v) = {}& t,
\end{align}
such that $u, v$ are in the vicinity of $c$.

To solve \eqref{eq:F_E}, we need some preliminary results. By \cite[Equations (1.27), (1.41) and (1.42)]{Claeys-Romano14}, we have
\begin{align} \label{eq:quadratic_equilibrium_hard_edge}
  F(c, c) = {}& 1 + \theta, & E(c, c) = {}& 1.
\end{align}
Next, we denote
\begin{align}
  A_2 = {}& \frac{1}{2\pi i} \oint_{\gamma_{\Out}} J^2_c(s) V''(J_c(s)) \frac{1}{(s + 1)^2} ds = \frac{1}{2\pi i} \oint_{\gamma_{\Out}} J_c(s) U'(J_c(s)) \frac{1}{(s + 1)^2} ds, \label{eq:defn_A2} \\
  A_3 = {}& \frac{1}{2\pi i} \oint_{\gamma_{\Out}} J^2_c(s) V''(J_c(s)) \frac{1}{s + 1} ds = \frac{1}{2\pi i} \oint_{\gamma_{\Out}} J_c(s) U'(J_c(s)) \frac{1}{s + 1} ds - 1, \label{eq:defn_A3}
\end{align}
where the second identity in \eqref{eq:defn_A2} is due to \eqref{eq:critical_in_N_in}, and the second identity in \eqref{eq:defn_A3} is due to \eqref{eq:quadratic_equilibrium_hard_edge}. Using the integral representations of $A_2$ and $A_3$ involving $U'(J_c(s))$, we have that
\begin{equation} \label{eq:A_2_A_3_relation}
  \begin{split}
    A_3 + 1 - \frac{\theta + 1}{\theta} A_2 = {}& \frac{1}{2\pi i} \oint_{\gamma_{\Out}} U'(J_c(s)) J_c(s)\frac{s - \theta^{-1}}{(s + 1)^2} ds = \frac{1}{2\pi i} \oint_{\gamma_{\Out}} U'(J_c(s)) J'_c(s) \frac{s}{s + 1} ds \\
    = {}& -\frac{1}{2\pi i} \oint_{\gamma_{\Out}} U(J_c(s)) \frac{1}{(s + 1)^2} ds = 0,
  \end{split}
\end{equation}
where the last step is due to \eqref{eq:critical_in_N_in}.

Also we define
\begin{equation} \label{eq:defn_H_in_EF}
  H(u, v) = (1 + \theta^{-1}) E(u, v) - \theta^{-1} F(u, v).
\end{equation}
For any $x \in \realR_+$,
\begin{equation}
  H(x, x) = \frac{\theta^{-1}}{2\pi i} \oint_{\gamma_{\Out}} U(J_x(s)) \left( \frac{\theta + 1}{s + 1} - \frac{1}{s} \right) ds     = \frac{1}{x} \frac{1}{2\pi i} \oint_{\gamma_{\Out}} V'(J_x(s)) J'(s) ds = 0.
\end{equation}
By \eqref{eq:defn_H_in_EF}, the solution $(u, v)$ of \eqref{eq:F_E} also satisfies $H(u, v) = 0$. We have that, for any $x \in \realR_+$,
\begin{equation} \label{eq:H_along_xx}
  \left. \frac{\partial H}{\partial u} \right\rvert_{u = v = x} + \left. \frac{\partial H}{\partial v} \right\rvert_{u = v = x} = \frac{1}{x} \frac{1}{2\pi i} \oint_{\gamma_{\Out}} U'(J_x(s)) J'(s) ds = 0,
\end{equation}
and also by \eqref{eq:critical_in_N_in}
\begin{equation}
  \left. \frac{\partial H}{\partial v} \right\rvert_{u = v = c} = \frac{1}{c} \frac{1}{2\pi i} \oint_{\gamma_{\Out}} U'(J_c(s)) J'(s) \frac{1}{s + 1} ds = \frac{1}{c} \frac{1}{2\pi i} \oint_{\gamma_{\Out}} U(J_c(s)) \frac{1}{(s + 1)^2} ds = 0.
\end{equation}
Hence, $(u, v) = (c, c)$ is a critical point of $H(u, v)$. The Hessian of $H(u, v)$ at $(c, c)$ is
\begin{equation}
  \Hessian_H =
  \begin{pmatrix}
    \frac{\partial^2 H}{\partial u^2} & \frac{\partial^2 H}{\partial u \partial v} \\
    \frac{\partial^2 H}{\partial u \partial v} & \frac{\partial^2 H}{\partial v^2} 
  \end{pmatrix}
  = \frac{1}{c^2}
  \begin{pmatrix}
    h_2 & h_1 \\
    h_1 & h_0
  \end{pmatrix},
\end{equation}
where
\begin{equation}
  h_k = \frac{1}{2\pi i} \oint_{\gamma_{\Out}} \frac{d}{ds} \left( J^2_c(s) V''(J_c(s)) \right) \frac{s^k}{(s + 1)^2} ds, \quad k = 0, 1, 2.
\end{equation}
Using \eqref{eq:critical_in_N_in}, we have
\begin{equation}
  \begin{split}
    & A_2 - \frac{\theta + 1}{\theta} (A_1 + \frac{1}{2} h_0) \\
    = {}& \frac{1}{2\pi i} \oint_{\gamma_{\Out}} U'(J_c(s)) J_c(s) \frac{1}{(s + 1)^2} ds - \frac{\theta + 1}{\theta} \frac{1}{2\pi i} \oint_{\gamma_{\Out}} U'(J_c(s)) J_c(s) \frac{1}{(s + 1)^3} ds \\
    = {}& \frac{1}{2\pi i} \oint_{\gamma_{\Out}} \frac{d}{ds} \left( U(J_c(s)) \right) \frac{s}{(s + 1)^2} ds \\
    = {}& \frac{1}{2\pi i} \oint_{\gamma_{\Out}} U(J_c(s)) \frac{1}{(s + 1)^2} ds - \frac{1}{2\pi i} \oint_{\gamma_{\Out}} U(J_c(s)) \frac{2}{(s + 1)^3} ds \\
    = {}& 0 - 2A_1 = -2A_1,
  \end{split}
\end{equation}
Hence, we derive that 
\begin{equation}
  h_0 = \frac{2\theta}{\theta + 1} A_2 + \frac{2(\theta - 1)}{\theta + 1} A_1.
\end{equation}
Noting that
\begin{gather}
  h_2 + 2h_1 + h_0 = \frac{1}{2\pi i} \oint_{\gamma_{\Out}} \frac{d}{ds} \left( J^2_c(s) V''(J_c(s)) \right) ds = 0, \\
  h_1 + h_0 = \frac{1}{2\pi i} \oint_{\gamma_{\Out}} \frac{d}{ds} \left( J^2_c(s) V''(J_c(s)) \right) \frac{1}{s + 1} ds = \frac{1}{2\pi i} \oint_{\gamma_{\Out}} J^2_c(s) V''(J_c(s)) \frac{1}{(s + 1)^2} ds = A_2, 
\end{gather}
we have that
\begin{align}
  h_2 = {}& -\frac{2}{\theta + 1} A_2 + \frac{2(\theta - 1)}{\theta + 1} A_1, & h_1 = {}& -\frac{\theta - 1}{\theta + 1} A_2 - \frac{2(\theta - 1)}{\theta + 1} A_1.
\end{align}
We find that $\det(\Hessian_H) = -c^{-4} A^2_2$. Assuming the technical condition in Theorem \ref{thm:main} that $A_2 \neq 0$, we see that there are exactly two directions, along which the vectors $\vec{v}_i$ satisfy $\vec{v}_i \Hessian_H \vec{v}^T_i = 0$. It is easy to find that these two directions are represented by
\begin{align}
  \vec{v}_1 = {}& (1, 1), & \vec{v}_2 = {}& ((\theta - 1)A_1 + \theta A_2, (\theta - 1)A_1 - A_2).
\end{align}
Then for $u, v$ around $c$, the solutions to $H(u, v) = 0$ lie on two curves $C_1, C_2$ which pass through $(c, c)$ and they have tangent vectors $\vec{v}_1$ and $\vec{v}_2$, respectively. Moreover, by \eqref{eq:H_along_xx}, we have that $C_1$ is the straight line $\{ u \vec{v}_1: u\in\mathbb{R} \}$.

Along $C_1$, let us denote the solution to \eqref{eq:F_E} locally around $(c, c)$ by $(c(t), c(t))$. Then $c(t)$ is analytic around $t = 1$, and its derivative at $t = 1$ satisfies
\begin{equation}
  1 = \left. \frac{d E(c(t), c(t))}{dt} \right\rvert_{t = 1} = \frac{c'(1)}{c} \frac{1}{2\pi i} \oint_{\gamma_{\Out}} \frac{U'(J_c(s)) J_c(s)}{s + 1} ds.
\end{equation}
Since by \eqref{eq:defn_A3} and \eqref{eq:A_2_A_3_relation}
\begin{equation}
  \frac{1}{2\pi i} \oint_{\gamma_{\Out}} \frac{U'(J_c(s)) J_c(s)}{s + 1} ds = 1 + A_3 = \frac{\theta + 1}{\theta} A_2,
\end{equation}
we conclude that
\begin{equation} \label{eq:c'(1)}
  c'(1) = \frac{c}{1 + A_3} = \frac{\theta c}{(\theta + 1)A_2}.
\end{equation}

Along $C_2$, let us denote the solution to \eqref{eq:F_E} locally around $(c, c)$ by $(c_1(t), c_0(t))$. Then $c_1(t)$ and $c_0(t)$ are analytic around $t = 1$, and their derivatives at $t = 1$ satisfy
\begin{align} \label{eq:derivative_c_1_c_0}
  c'_1(1) = {}& ((\theta - 1)A_1 + \theta A_2) w, & c'_0(t) = {}& ((\theta - 1)A_1 - A_2)) w
\end{align}
for some $w$. We have
\begin{align} \label{eq:w_for_c_1_c_0}
  w = \frac{\theta c}{(\theta^2 - 1) A_1 A_2},
\end{align}
because we have, with the help of \eqref{eq:defn_A3}, \eqref{eq:A_2_A_3_relation} and \eqref{eq:defn_A2}
\begin{equation}
  \begin{split}
    1 = {}& \left. \frac{d E(c_1(t), c_0(t))}{dt} \right\rvert_{t = 1} = \frac{w}{c} \frac{1}{2\pi i} \oint_{\gamma_{\Out}} U'(J_c(s)) J_c(s)\frac{((\theta - 1)A_1 + \theta A_2)s + ((\theta - 1)A_1 - A_2)}{(s + 1)^2} ds \\
    = {}& \frac{w}{c} \left[ \frac{(\theta - 1)A_1 + \theta A_2}{2\pi i} \oint_{\gamma_{\Out}} U'(J_c(s)) J_c(s)\frac{1}{s + 1} ds - \frac{(\theta + 1) A_2}{2\pi i} \oint_{\gamma_{\Out}} U'(J_c(s)) J_c(s)\frac{1}{(s + 1)^2} ds \right] \\
    = {}& \frac{w}{c} \frac{\theta^2 - 1}{\theta} A_1 A_2.
  \end{split}
\end{equation}

\subsection{Deformation along $C_1$ when $\lvert t - 1 \rvert$ is small} \label{subsec:deformation_C1}

In this subsection and the next, we assume that $t$ is in the vicinity of $1$, and denote $u = t - 1$.

Recall $c(t)$ defined in Section \ref{subsec:deformation_algebraic}. We let $b(t) = c(t)(1 + \theta)^{1 + 1/\theta}/\theta$. Then $J_{c(t)}$ maps $\compC \setminus \overline{D}$ to $\compC \setminus [0, b(t)]$ and $D \setminus [-1, 0]$ to $\halfH \setminus [0, b(t)]$ univalently. We define $\Iinv_{1, t}(z)$ and $\Iinv_{2, t}(z)$ to be the inverses of $J_{c(t)}(z)$ on $\compC \setminus \overline{D}$ and $D \setminus [-1, 0]$ respectively. Also we define $\Iinv_{+, t}(x)$ (resp.~$\Iinv_{-, t}(x)$) to be the inverse map of $J_{c(t)}(s)$ from $\gamma_1$ (resp.~$\gamma_2$) to $[0, b(t)]$. We have
\begin{align}
  J_{c(t)}(s) = {}& \frac{c(t)}{c} J_c(s), & \Iinv_{*, t}(z) = {}& \Iinv_* \left( \frac{c}{c(t)} z \right), \quad * = 1, 2, +, -.
\end{align}

Let, analogous to \eqref{eq:formula_N_in_t=1},
\begin{equation} \label{eq:formula_N_in}
  N_t(s) = 
  \begin{cases}
    N_{t, \In}(s) = \frac{1}{2\pi i} \oint_{\gamma} \frac{U_t(J_{c(t)}(\xi))}{\xi - s} d\xi - 1, & s \in D \setminus [-1, 0], \\
    N_{t, \Out}(s) = -\frac{1}{2\pi i} \oint_{\gamma} \frac{U_t(J_{c(t)}(\xi))}{\xi - s} d\xi + 1, & s \in \compC \setminus \overline{D},
  \end{cases}
  \quad \text{where} \quad U_t(z) = zV'_t(z).
\end{equation}
Like in the formula \eqref{eq:formula_N_in_t=1} for $N_{\In}(s)$ in Section \ref{subsec:exact_equilibrium_measure}, we can use $\gamma_{\Out}$ instead of $\gamma$ to evaluate $N_{t, \In}(s)$. We find that $N_{t, \In}(s)$ is analytic in neighbourhoods of $-1$ and $\theta^{-1}$.

Since $(c(t), c(t))$ satisfies \eqref{eq:F_E}, we have that
\begin{align} \label{eq:N_t_in(-1)}
  N_{t, \In}(-1) = {}& \frac{1}{t} E(c(t), c(t)) - 1 = 0, & N_{t, \In}(0) = {}& \frac{1}{t} F(c(t), c(t)) - 1 = \theta. 
\end{align}
We define the functions $G_t(z)$ and $\G_t(z)$, which are analytic functions on $\compC \setminus [0, b(t)]$ and $\halfH \setminus [0, b(t)]$ respectively, by
\begin{equation} 
  N_t(s) =
  \begin{cases}
    J_{c(t)}(s) G_t(J_{c(t)}(s)), & s \in \compC \setminus \overline{D_t}, \\
    J_{c(t)}(s) \G_t(J_{c(t)}(s)), & s \in D_t \setminus [-1, 0].
  \end{cases}
\end{equation}
Like \eqref{eq:psi(x)_in_G_tilde_G}, for all $x \in (0, b(t))$,
\begin{equation} \label{eq:psi(x)_in_G_tilde_G_extend}
  \begin{split}
    (\G_t)_-(x) - (\G_t)_+(x) = {}& x^{-1}(N_{t, \In}(\Iinv_{+, t}(x)) - N_{t, \In}(\Iinv_{-, t}(x))) \\
    = {}& x^{-1}(N_{t, \Out}(\Iinv_{-, t}(x)) - N_{t, \Out}(\Iinv_{+, t}(x))) = (G_t)_-(x) - (G_t)_+(x),
  \end{split}
\end{equation}
and then we define
\begin{equation}
  \psi_t(x) = \frac{1}{2\pi i} ((G_t)_-(x) - (G_t)_+(x)) = \frac{1}{2\pi i} ((\G_t)_-(x) - (\G_t)_+(x)),
\end{equation}
and define the functions $\gfn_t(z)$ and $\gfntilde_t(z)$, which are analytic functions on $\compC \setminus (-\infty, b(t)]$ and $\halfH \setminus [0, b(t)]$ respectively, by
\begin{align} \label{eq:gfn_gfntilde_t}
  \gfn'_t(z) = {}& G_t(z), & \gfntilde'_t(z) = \G_t(z), && \text{and} &&  \lim_{z \to \infty} \gfn_t(z)-\log z = \lim_{z \to \infty} \gfntilde_t(z)-\log z^{\theta} = {}& 0.
\end{align}
Then let
\begin{equation} \label{eq:defn_phi}
  \phi_t(z) = \gfn_t(z) + \gfntilde_t(z) - V_t(z) - \ell_t, \quad \ell_t = (\gfn_t)_{\pm}(x) + (\gfntilde_t)_{\mp}(x) -V_t(x) \quad \text{for  $x \in (0, b(t))$}.
\end{equation}
We note that
\begin{equation} \label{eq:total_measure_1}
  \begin{split}
    \int^{b(t)}_0 \psi_t(x) dx = & \frac{1}{2\pi i} \int^{b(t)}_0 (G_t)_-(x) - (G_t)_+(x) dx = \frac{1}{2\pi i} \int^{b(t)}_0 (\G_t)_-(x) - (\G_t)_+(x) dx \\
    = {}& \frac{-1}{2\pi i} \oint_{\gamma} \G_t(J_{c(t)}(s)) d J_{c(t)}(s) \\
    = {}& \theta^{-1} \frac{1}{2\pi i} \oint_{\gamma_{\Out}}\frac{N_{t, \In}(s)}{s} ds - (1 + \theta^{-1}) \frac{1}{2\pi i} \oint_{\gamma_{\Out}} \frac{N_{t, \In}(s)}{s + 1} ds \\
    = {}& \theta^{-1} N_{t, \In}(0) - (1 + \theta^{-1}) N_{t, \In}(-1) = 1.
  \end{split}
\end{equation}

When $t = 1$, all these notions above coincide with those defined in Section \ref{subsec:exact_equilibrium_measure} without $t$. Since $J_{c(t)}(s)$ depends on $t$ analytically, we find that if $\lvert u \rvert = \lvert t - 1 \rvert$ is small enough, the requirement for $\psi(x)$ in Regularity Condition \ref{reg:primary} still holds for $\psi_t(x)$ here. Although we have not verified that $\psi_t(x)$ defines an equilibrium measure, by the discussion at the end of Section \ref{subsec:exact_equilibrium_measure}, Regularity Condition \ref{reg:primary} can be expressed as conditions on $N(s)$, and we find that these conditions hold on $N_t(s)$ if $\lvert u \rvert = \lvert t - 1 \rvert$ is small enough. To be precise, we have
\begin{regularitycondition} \label{reg:one-cut}
  Let $\epsilon_1$ and $\epsilon_2$ be small enough positive numbers. There is $\delta_1 > 0$, such that for all $t \in [1 - \delta_1, 1 + \delta_1]$,
  \begin{enumerate} 
  \item 
    $\psi_t(x) > 0$ strictly for all $x \in [\epsilon_1, b(t) - \epsilon_2]$.
  \item
    \begin{equation} \label{eq:N_t_in_at_s_b}
      N'_{t, \In}(\theta^{-1}) = \frac{1}{2\pi i} \oint_{\gamma_{\Out}} \frac{U_t(J_{c(t)}(s))}{(s - \theta^{-1})^2} ds > 0,
    \end{equation}
    or equivalently, there is a positive number $d_{2, t} = d^{(V_t)}_{2, t}$ and $\delta_2 > 0$ depending on $\epsilon_2$, such that
    \begin{equation} \label{eq:psi_t_at_b}
      \psi_t(x) = d_{2, t} (b(t) - x)^{\frac{1}{2}} (1 + h_{2, t}(x)), \quad x \in (b(t) - \epsilon_2, b(t)),
    \end{equation}
    such that $h_{2, t}(x) \to 0$ as $x \to b(t)_-$ and $\lvert h_{2, t}(x) \rvert \leq \delta_2$ for all $x \in (b(t) - \epsilon_2, b(t))$. 
  \item
    \begin{equation}
      \phi_t(x) > 0, \quad x \in (b(t), +\infty).
    \end{equation}
  \end{enumerate}
\end{regularitycondition}

The function $N_{t, \In}(s)$ depends on $t$ analytically. From \eqref{eq:critical_in_N_in} we have that $N'_{1, \In}(-1) = 0$. We also have, by \eqref{eq:critical_in_N_in}, \eqref{eq:defn_A2} and \eqref{eq:c'(1)},
\begin{equation}
  \begin{split}
    \left. \frac{d}{dt} N'_{t, \In}(-1) \right\rvert_{t = 1} = {}& \frac{1}{2\pi i} \oint_{\gamma_{\Out}} \frac{\left. \frac{d}{dt} U_t(J_{c(t)}(s)) \right\rvert_{t = 1}}{(s + 1)^2} ds \\
    = {}& \left( \frac{c'(1)}{c} - 1 \right) \frac{1}{2\pi i} \oint_{\gamma_{\Out}} \frac{U(J_c(s))}{(s + 1)^2} ds + \frac{c'(1)}{c} \frac{1}{2\pi i} \oint_{\gamma_{\Out}} J_c(s)^2 V''(J_c(s)) \frac{ds}{(s + 1)^2} \\
    = {}& \left( \frac{c'(1)}{c} - 1 \right) \cdot 0 + \frac{c'(1)}{c} A_2\\
     = {}& \frac{\theta}{\theta + 1}.  \label{eq:N'_t_in(-1)_dt}
  \end{split}
\end{equation}
By \eqref{eq:N_t_in(-1)}, \eqref{eq:critical_in_N_in} and \eqref{eq:N'_t_in(-1)_dt}, we have that when $\lvert u \rvert = \lvert t - 1 \rvert$ is small,
\begin{equation}
  N_{t, \In}(s) = \frac{\theta}{\theta + 1} u (s + 1) + A_1(s + 1)^2 + \bigO(u^2(s + 1)) +\bigO(u(s + 1)^2)+ \bigO((s + 1)^3), \quad s \to -1.
\end{equation}
and $N_{t, \Out}(s)$ has a similar asymptotic expansion since for $s$ in the domain of $N_{t, \Out}(s)$ and $s$ is close to $-1$, 
\begin{equation} \label{eq:relation_U_in_U_out}
  N_{t, \Out}(s) = U_t(J_{c(t)}(s)) - N_{t, \In}(s).
\end{equation}

We have that as $z \to 0$, the limits of $\Iinv_{1, t}(z)$ and $\Iinv_{2, t}(z)$ are analogous to those of $\Iinv_1(z)$ and $\Iinv_2(z)$ given in \eqref{eq:I_1_at_0} and \eqref{eq:I_2_at_0}, with $c$ replaced by $c(t)$ there.
Noting that $c(t) = c + \bigO(u)$ as $t\to 1$, we have ($\rho$ is defined in \eqref{eq:defn_rho}) 
\begin{equation} \label{eq:formula_of_G(z)}  
  \begin{split}
    G_t(z) = {}& G_t(J_{c(t)}(\Iinv_{1, t}(z))) = \frac{1}{z} N_{t, \Out}(\Iinv_{1, t}(z)) \\
    = {}& V'_t(z) -
             \begin{cases}
               \rho e^{\frac{2\pi i}{\theta + 1}} z^{\frac{\theta - 1}{\theta + 1}} (1 + \bigO(u) + \bigO(z^{\frac{\theta}{\theta + 1}})) \\
               + \frac{\theta}{\theta + 1} u c^{-\frac{\theta}{\theta + 1}} e^{\frac{\pi i}{\theta + 1}} z^{\frac{-1}{\theta + 1}} (1 + \bigO(u) + \bigO(z^{\frac{\theta}{\theta + 1}})), & \arg z \in (0, \pi), \\
               \rho e^{-\frac{2\pi i}{\theta + 1}} z^{\frac{\theta - 1}{\theta + 1}} (1 + \bigO(u) + \bigO(z^{\frac{\theta}{\theta + 1}})) & \\
               + \frac{\theta}{\theta + 1} u c^{-\frac{\theta}{\theta + 1}} e^{-\frac{\pi i}{\theta + 1}} z^{\frac{-1}{\theta + 1}} (1 + \bigO(u) + \bigO(z^{\frac{\theta}{\theta + 1}})), & \arg z \in (-\pi, 0), 
             \end{cases}
  \end{split}
\end{equation}
\begin{equation} \label{eq:formula_of_Gtilde(z)}
  \begin{split}
    \G_t(z) = {}& \G_t(J_{c(t)}(\Iinv_{2, t}(z))) = \frac{1}{z} N_{t, \In}(\Iinv_{2, t}(z)) \\
    = {}& 
        \begin{cases}
          \rho e^{-\frac{2\pi i}{\theta + 1}} z^{\frac{\theta - 1}{\theta + 1}} (1 + \bigO(u) + \bigO(z^{\frac{\theta}{\theta + 1}})) & \\
          + \frac{\theta}{\theta + 1} u c^{-\frac{\theta}{\theta + 1}} e^{-\frac{\pi i}{\theta + 1}} z^{\frac{-1}{\theta + 1}} (1 + \bigO(u) + \bigO(z^{\frac{\theta}{\theta + 1}})), & \arg z \in (0, \frac{\pi}{\theta}), \\
          \rho e^{\frac{2\pi i}{\theta + 1}} z^{\frac{\theta - 1}{\theta + 1}} (1 + \bigO(u) + \bigO(z^{\frac{\theta}{\theta + 1}})) \\
          + \frac{\theta}{\theta + 1} u c^{-\frac{\theta}{\theta + 1}} e^{\frac{\pi i}{\theta + 1}} z^{\frac{-1}{\theta + 1}} (1 + \bigO(u) + \bigO(z^{\frac{\theta}{\theta + 1}})), & \arg z \in (-\frac{\pi}{\theta}, 0).
        \end{cases}
  \end{split}
\end{equation}
From the asymptotic expansions of $G_t$ and $\G_t$, we have that for $x$ close to $0$,
\begin{equation} \label{eq:limit_psi_t_at_0}
    \psi_t(x) = d_1 x^{\frac{\theta - 1}{\theta + 1}} (1 + \bigO(u) + \bigO(x^{\frac{\theta}{\theta + 1}})) + \frac{\theta u}{(\theta + 1)\pi} c^{\frac{-\theta}{\theta + 1}} \sin \left( \frac{\pi}{\theta + 1} \right) x^{\frac{-1}{\theta + 1}} (1 + \bigO(u) + \bigO(x^{\frac{\theta}{\theta + 1}})).
\end{equation}
\begin{rmk} \label{rmk:eq_measure_or_not}
  Given a small enough $\epsilon_1 > 0$, if $u = t - 1$ is a small enough positive number, we have that $\psi_t(x)$ is positive on $(0, \epsilon_1)$ and $\psi_t(x)$ blows up as $x^{-\theta/(\theta + 1)}$ at $0$. Hence, by \cite[Theorem 1.11]{Claeys-Romano14}, $d\mu^{(V_t)} = \psi_t(x) dx$ supported on $[0, b(t)]$ is the equilibrium measure of $V_t$, and it satisfies Regularity Condition \ref{reg:hard-edge}. However, if $u = t - 1$ is a negative number with small magnitude, $\psi_t(x)$ is not the equilibrium measure for $V_t(x)$. It is not even a probability measure since it is negative in the vicinity of $0$. 
\end{rmk}

Equation \eqref{eq:total_measure_1} implies $(\gfn_t)_+(0) - (\gfn_t)_-(0) = (\gfntilde_t)_+(0) - (\gfntilde_t)_-(0) = 2 \pi i$. Then we have the asymptotic expansions for $\gfn_t(z)$ and $\gfntilde_t(z)$ defined in \eqref{eq:gfn_gfntilde_t}
\begin{multline}
  \gfn_t(z) = \Re (\gfn_t)_+(0) + V_t(z)-V_t(0) \\ -
  \begin{cases}
    -\pi i + \frac{\theta + 1}{2\theta} \rho e^{\frac{2\pi i}{\theta + 1}} z^{\frac{2\theta}{\theta + 1}} (1 + \bigO(u) + \bigO(z^{\frac{\theta}{\theta + 1}})) \\
    + u c^{-\frac{\theta}{\theta + 1}} e^{\frac{\pi i}{\theta + 1}} z^{\frac{\theta}{\theta + 1}} (1 + \bigO(u) + \bigO(z^{\frac{\theta}{\theta + 1}})), & \arg z \in (0, \pi), \\
    \pi i + \frac{\theta + 1}{2\theta} \rho e^{-\frac{2\pi i}{\theta + 1}} z^{\frac{2\theta}{\theta + 1}} (1 + \bigO(u) + \bigO(z^{\frac{\theta}{\theta + 1}})) & \\
    + u c^{-\frac{\theta}{\theta + 1}} e^{-\frac{\pi i}{\theta + 1}} z^{\frac{\theta}{\theta + 1}} (1 + \bigO(u) + \bigO(z^{\frac{\theta}{\theta + 1}})), & \arg z \in (-\pi, 0), 
  \end{cases}
\end{multline}
\begin{equation}
  \gfntilde_t(z) = \Re (\gfntilde_t)_+(0) +
  \begin{cases}
    \pi i + \frac{\theta + 1}{2\theta} \rho e^{-\frac{2\pi i}{\theta + 1}} z^{\frac{2\theta}{\theta + 1}} (1 + \bigO(u) + \bigO(z^{\frac{\theta}{\theta + 1}})) & \\
    + u c^{-\frac{\theta}{\theta + 1}} e^{-\frac{\pi i}{\theta + 1}} z^{\frac{\theta}{\theta + 1}} (1 + \bigO(u) + \bigO(z^{\frac{\theta}{\theta + 1}})), & \arg z \in (0, \frac{\pi}{\theta}), \\
    -\pi i + \frac{\theta + 1}{2\theta} \rho e^{\frac{2\pi i}{\theta + 1}} z^{\frac{2\theta}{\theta + 1}} (1 + \bigO(u) + \bigO(z^{\frac{\theta}{\theta + 1}})) \\
    + u c^{-\frac{\theta}{\theta + 1}} e^{\frac{\pi i}{\theta + 1}} z^{\frac{\theta}{\theta + 1}} (1 + \bigO(u) + \bigO(z^{\frac{\theta}{\theta + 1}})), & \arg z \in (-\frac{\pi}{\theta}, 0),
  \end{cases}
\end{equation}
where $(\gfn_t)_{\pm}(0)$ (resp.~$(\gfntilde_t)_{\pm}(0)$) is the limit of $\gfn_t(z)$ (resp.~$\gfntilde_t(z)$) as $z \to 0$ from $\compC_{\pm}$ (resp.~$(\halfH)_{\pm}$). Then $\phi_t(z)$ in \eqref{eq:defn_phi} has the asymptotic expansions
\begin{equation}
  \phi_t(z) = 
  \begin{cases}
    2\pi i - \frac{\theta + 1}{\theta} \pi d_1 i z^{\frac{2\theta}{\theta + 1}} (1 + \bigO(u) + \bigO(z^{\frac{\theta}{\theta + 1}})) & \\
    - 2u c^{-\frac{\theta}{\theta + 1}} \sin \left( \frac{\pi}{\theta + 1} \right) i z^{\frac{\theta}{\theta + 1}} (1 + \bigO(u) + \bigO(z^{\frac{\theta}{\theta + 1}})), & \arg z \in (0, \frac{\pi}{\theta}), \\
    -2\pi i + \frac{\theta + 1}{\theta} \pi d_1 i z^{\frac{2\theta}{\theta + 1}} (1 + \bigO(u) + \bigO(z^{\frac{\theta}{\theta + 1}})) \\
    + 2u c^{-\frac{\theta}{\theta + 1}} \sin \left( \frac{\pi}{\theta + 1} \right) i z^{\frac{\theta}{\theta + 1}} (1 + \bigO(u) + \bigO(z^{\frac{\theta}{\theta + 1}})), & \arg z \in (-\frac{\pi}{\theta}, 0).
  \end{cases}
\end{equation}

Similarly, from \eqref{eq:N_t_in_at_s_b} and \eqref{eq:psi_t_at_b}, let $\epsilon, \delta$ be small enough positive numbers, we have for $t \in [1 - \delta, 1 + \delta]$ and $z \in D(b(t), \epsilon) \setminus [0, b(t)]$,
\begin{equation} \label{eq:behaviour_phi_t_at_b}
  \phi_t(z) = -\frac{4\pi}{3} d_{2,t} (z - b(t))^{\frac{3}{2}} (1 + \bigO(z - b(t))),
\end{equation}
where the term $\bigO(z - b(t))$ is analytic there.

\begin{rmk} \label{rmk:about_theta=1_transitive}
  This result of the equilibrium measure holds for all $\theta > 1$, and essentially also holds for $\theta \in (0, 1)$. However, it does not hold for $\theta = 1$. For $\theta = 1$, in the hard-soft edge transition, the behaviour of the density near the edge is not $\bigO(x^{(\theta - 1)/(\theta + 1)}) = \bigO(1)$, but $\bigO(x^{1/2})$, see \cite{Its-Kuijlaars-Ostensson08}, \cite{Its-Kuijlaars-Ostensson09}. 
\end{rmk}

\subsection{Deformation along $C_2$  when $1 - t$ is small positive} \label{subsec:deformation_along_C2_small_t}
 


For $V_t(x)$ defined in \eqref{eq:defn_Vt} with $u = t - 1 < 0$, and $\lvert u \rvert$ small, by \eqref{eq:derivative_c_1_c_0} and \eqref{eq:w_for_c_1_c_0} we have that 
\begin{align}
  c_1(t) = {}& c + \frac{\theta c}{\theta^2 - 1} \frac{(\theta - 1)A_1 + \theta A_2}{A_1 A_2} u + \bigO(u^2), & c_0(t) = {}& c + \frac{\theta c}{\theta^2 - 1} \frac{(\theta - 1)A_1 - A_2}{A_1 A_2} u + \bigO(u^2).
\end{align}
We define, like \cite[Equations (1.38) and (1.39) and Theorem 1.12]{Claeys-Romano14},
\begin{align}
  s_a(t) = {}& -\frac{\theta - 1}{2\theta} - \frac{1}{2\theta c_1(t)} \sqrt{4 \theta c_0(t) c_1(t) + (\theta - 1)^2 c^2_1(t)}, & \aend(t) = {}& J_{c_1(t), c_0(t)}(s_a(t)), \label{eq:sat_and_at} \\
  s_b(t) = {}& -\frac{\theta - 1}{2\theta} + \frac{1}{2\theta c_1(t)} \sqrt{4 \theta c_0(t) c_1(t) + (\theta - 1)^2 c^2_1(t)}, & \bend(t) = {}& J_{c_1(t), c_0(t)}(s_a(t)), \label{eq:sbt_and_bt}
\end{align}
and have
\begin{align}
  s_a(t) = {}&
               -1 + \frac{\theta u}{(\theta^2 - 1)A_1} + \bigO(u^2), & \aend(t) = {}& \theta c \left( \frac{\theta}{\theta^2 - 1} \frac{-u}{A_1} \right)^{1 + \frac{1}{\theta}} (1 + \bigO(u)), \\
  s_b(t) = {}&
               \frac{1}{\theta} - \frac{\theta u}{(\theta^2 - 1)A_1} + \bigO(u^2) & \bend(t) = {}& b + (\theta + 1)^{\frac{1}{\theta}} \frac{c}{A_2}u + \bigO(u^2).
\end{align}

We define, for such $t$, $\gammahat_{1, t}$ to be the curve connecting $s_a(t)$ to $s_b(t)$ such that $J_{c_1(t), c_0(t)}$ maps $\gammahat_{1, t}$ to the interval $[\aend(t), \bend(t)]$, and then $\gammahat_{2, t} = \overline{\gammahat_{1, t}} \subseteq \compC_-$ and $\gammahat_t = \gammahat_{1, t} \cup \gammahat_{2, t}$ to be the closed contour oriented counterclockwise and $\D_t$ to be the region enclosed by $\gammahat_t$. Then $J_{c_1(t), c_0(t)}$ maps $\compC \setminus \overline{\D_t}$ to $\compC \setminus [\aend(t), \bend(t)]$ and $\D_t \setminus [-1,0]$ to $\halfH \setminus [\aend(t), \bend(t)]$ univalently. Then we define $\Iinvhat_{1, t}(z)$, $\Iinvhat_{2, t}(z)$ to be the inverses of $J_{c_1(t), c_0(t)}(z)$ on $\compC \setminus \overline{\D_t}$ and $\D_t \setminus [-1, 0]$ respectively. Also we define $\Iinvhat_{+, t}(x)$ (resp.~$\Iinvhat_{-, t}(x)$) to be the inverse map of $J_{c_1(t), c_0(t)}(s)$ from $\gammahat_{1, t}$ (resp.~$\gammahat_{2, t}$) to $[\aend(t), \bend(t)]$.

Let, analogous to \eqref{eq:formula_N_in_t=1} and \eqref{eq:formula_N_in}, with $U_t$ defined in \eqref{eq:formula_N_in},
\begin{equation} \label{eq:formula_Nhat_in}
  \Nhat_t(s) = 
  \begin{cases}
    \Nhat_{t, \In}(s) = \frac{1}{2\pi i} \oint_{\gammahat_t} \frac{U_t(J_{c_1(t), c_0(t)}(\xi))}{\xi - s} d\xi - 1, & s \in \D_t \setminus [-1, 0], \\
    \Nhat_{t, \Out}(s) = -\frac{1}{2\pi i} \oint_{\gammahat_t} \frac{U_t(J_{c_1(t), c_0(t)}(\xi))}{\xi - s} d\xi + 1, & s \in \compC \setminus \overline{\D_t}.
  \end{cases}
\end{equation}
Like in Sections \ref{subsec:exact_equilibrium_measure} and \ref{subsec:deformation_C1}, we can deform $\gammahat_t$ to a slightly larger contour $\gamma_{\Out}$ that encloses $\gammahat_t$ in the evaluation of $\Nhat_{t, \In}(s)$ in \eqref{eq:formula_Nhat_in}. For $t$ close to $1$, we can take $\gamma_{\Out}$ the same as in Sections \ref{subsec:exact_equilibrium_measure} and \ref{subsec:deformation_C1}, independent of $t$. In this way we find that  $\Nhat_{t, \In}(s)$ is analytic in a neighbourhood of $-1$ including $s_a(t)$, and a neighbourhood of $\theta^{-1}$ including $s_b(t)$. 

Since $c_1(t), c_0(t)$ satisfy \eqref{eq:F_E}, we have that, like \eqref{eq:N_t_in(-1)},
\begin{align} \label{eq:N_tilde_t_in(-1)}
  \Nhat_{t, \In}(-1) = {}& \frac{1}{t} E(c_1(t), c_0(t)) - 1 = 0, & \Nhat_{t, \In}(0) = {}& \frac{1}{t} F(c_1(t), c_0(t)) - 1 = \theta. 
\end{align}
By the identity
\begin{equation}
  \frac{1}{2\pi i} \oint_{\gamma_{\Out}} U_t(J_{c_1(t), c_0(t)}(\xi)) \left( \frac{1}{\xi + \frac{c_0(t)}{c_1(t)}} + \frac{\theta^{-1}}{s + 1} - \frac{\theta^{-1}}{s} \right) d\xi = \frac{1}{2\pi i} \oint_{\gamma_{\Out}} \frac{d}{d\xi} V_t(J_{c_1(t), c_0(t)}(\xi)) d\xi = 0,
\end{equation}
we derive that
\begin{equation}
  \begin{split}
    \Nhat_{t, \In}(-\frac{c_0(t)}{c_1(t)}) = {}& \frac{1}{2\pi i} \oint_{\gamma_{\Out}} U_t(J_{c_1(t), c_0(t)}(\xi)) \frac{1}{\xi + \frac{c_0(t)}{c_1(t)}} d\xi - 1\\
    = {}& \frac{1}{2\pi i} \oint_{\gamma_{\Out}} U_t(J_{c_1(t), c_0(t)}(\xi)) \left( -\frac{\theta^{-1}}{s + 1} + \frac{\theta^{-1}}{s} \right) d\xi - 1 \\
    = {}& \frac{1}{\theta t} \left( -E(c_1(t), c_0(t)) + F(c_1(t), c_0(t)) \right) - 1\\
    = {}& 0.
  \end{split}
\end{equation}
We define the functions $\Ghat_t(z)$ and $\Ghattilde_t(z)$, which are analytic functions on $\compC \setminus [\aend(t), \bend(t)]$ and $\halfH \setminus [\aend(t), \bend(t)]$ respectively, by
\begin{equation} \label{eq:N(s)}
  \Nhat_t(s) =
  \begin{cases}
    J_{c_1(t), c_0(t)}(s) \Ghat(J_{c_1(t), c_0(t)}(s)), & \text{outside $\gammahat_t$}, \\
    J_{c_1(t), c_0(t)}(s) \Ghattilde(J_{c_1(t), c_0(t)}(s)), & \text{inside $\gammahat_t$}.
  \end{cases}
\end{equation}

Like \eqref{eq:psi(x)_in_G_tilde_G} and \eqref{eq:psi(x)_in_G_tilde_G_extend}, for all $x \in (\aend(t), \bend(t))$,
\begin{equation}
  \begin{split}
    (\Ghat_t)_-(x) - (\Ghat_t)_+(x) = {}& x^{-1}(\Nhat_{t, \In}(\Iinvhat_{+, t}(x)) - \Nhat_{t, \In}(\Iinvhat_{-, t}(x))) \\
    = {}& x^{-1}(\Nhat_{t, \Out}(\Iinvhat_{-, t}(x)) - \Nhat_{t, \Out}(\Iinvhat_{+, t}(x))) = (\Ghattilde_t)_-(x) - (\Ghattilde_t)_+(x),
  \end{split}
\end{equation}
and then we define
\begin{equation} \label{eq:defn_psihat_small_u}
  \psihat_t(x) = \frac{1}{2\pi i} ((\Ghat_t)_-(x) - (\Ghat_t)_+(x)) = \frac{1}{2\pi i} ((\Ghattilde_t)_-(x) - (\Ghattilde_t)_+(x)).
\end{equation}
Then we define the functions $\gfnhat_t(z)$ and $\gfnhattilde_t(z)$, which are analytic functions on $\compC \setminus (-\infty, \bend(t)]$ and $\halfH \setminus [0, \bend(t)]$ respectively, by
\begin{align} \label{eq:gfnhat_gfnhattilde_t}
  \gfnhat'_t(z) = {}& \Ghat_t(z), & \gfnhattilde'_t(z) = \Ghat_t(z), && \text{and} && \lim_{z \to \infty} \gfnhat_t(z) -\log z = \lim_{z \to \infty} \gfnhattilde_t(z) -\log z^{\theta} = {}& 0.
\end{align}
From the formulas of $\Ghat_t(z)$ and $\Ghattilde_t(z)$, we can find the expressions of $\gfnhat_t(z)$ and $\gfnhattilde_t(z)$ by integration. At last, we define
\begin{equation} \label{eq:defn_ellha_t}
  \phihat_t(z) = \gfnhat_t(z) + \gfnhattilde_t(z) - V_t(z) - \ellhat_t, \quad \ellhat_t = (\gfnhat_t)_{\pm}(x) + (\gfnhattilde_t)_{\mp}(x) -V_t(x) \quad \text{for $x \in (\aend(t), \bend(t))$}.
\end{equation}
We note that, like \eqref{eq:total_measure_1},
\begin{equation} \label{eq:total_measure_1_soft}
  \begin{split}
    \int^{\bend(t)}_{\aend(t)} \psihat_t(x) dx = {}& \frac{1}{2\pi i} \int^{\bend(t)}_{\aend(t)} (\Ghat_t)_-(x) - (\Ghat_t)_+(x) dx = \frac{1}{2\pi i} \int^{\bend(t)}_{\aend(t)} (\Ghattilde_t)_-(x) - (\Ghattilde_t)_+(x) dx \\
    = {}& \frac{-1}{2\pi i} \oint_{\gammahat_t} \Ghattilde_t(J_{c_1(t), c_0(t)}(s)) d J_{c_1(t), c_0(t)}(s) \\
    = {}& \theta^{-1} \frac{1}{2\pi i} \oint_{\gamma_{\Out}}\frac{\Nhat_{t, \In}(s)}{s} ds - \theta^{-1} \frac{1}{2\pi i} \oint_{\gamma_{\Out}} \frac{\Nhat_{t, \In}(s)}{s + 1} ds - \frac{1}{2\pi i} \oint_{\gamma_{\Out}} \frac{\Nhat_{t, \In}(s)}{s + \frac{c_0(t)}{c_1(t)}} ds \\
    = {}& \theta^{-1} \Nhat_{t, \In}(0) - \theta^{-1} \Nhat_{t, \In}(-1) - \Nhat\left(-\frac{c_0(t)}{c_1(t)}\right) = 1.
  \end{split}
\end{equation}

When $t = 1$, the notions $s_{a(t)}$, $s_{b(t)}$, $\aend(t)$, $\bend(t)$, $\gammahat_{1, t}$, $\gammahat_{2, t}$,$\gammahat_t$, $\D_t$, $\Nhat_{t, \In}$, $\Nhat_{t, \Out}$, $\Nhat_t$, $\Ghat_t$, $\Ghattilde_t$ and $\psihat_t$ above coincide with $-1$, $\theta^{-1}$, $0$, $b$, $\gamma_1$, $\gamma_2$,$\gamma$, $D$, $N_{\In}$, $N_{\Out}$, $N$, $G$, $\G$ and $\psi$ occurring in Section \ref{subsec:exact_equilibrium_measure} respectively. Since $J_{c_1(t), c_0(t)}(s)$ depends on $t$ analytically, we find that if $-u = 1 - t$ is a small enough positive number, the Regularity Condition \ref{reg:primary} still holds, except for the strict negativity of \eqref{eq:E-L_2} on $(0, \aend(t))$ which will be discussed in part \ref{enu:prop:phihat_1} of Proposition \ref{prop:phihat}. To be precise, we have
\begin{regularitycondition} \label{reg:one-cut_soft}
  Let $\epsilon_1$ and $\epsilon_2$ be small enough positive numbers. There is $\delta_1 > 0$, such that for all $t \in [1 - \delta_1, 0]$,
  \begin{enumerate} 
  \item 
    $\psihat_t(x) > 0$ strictly for all $x \in [\aend(t) + \epsilon_1, \bend(t) - \epsilon_2]$.
  \item
    \begin{equation}
      \Nhat'_{t, \In}(s_{b(t)}) = \frac{1}{2\pi i} \oint_{\gamma_{\Out}} \frac{U_t(J_{c_1(t), c_2(t)}(s))}{(s - s_{b(t)})^2} ds > 0,
    \end{equation}
    or equivalently, there is a positive number $\dhat_{2, t} = \dhat^{(V_t)}_{2, t}$ and $\delta_2 > 0$ depending on $\epsilon_2$, such that
    \begin{equation} 
      \psihat_t(x) = \dhat_{2, t} (\bend(t) - x)^{\frac{1}{2}} (1 + \hhat_{2, t}(x)), \quad x \in (\bend(t) - \epsilon_2, \bend(t)),
    \end{equation}
    such that $\hhat_{2, t}(x) \to 0$ as $x \to \bend(t)_-$ and $\lvert \hhat_{2, t}(x) \rvert \leq \delta_2$ for all $x \in (\bend(t) - \epsilon_2, \bend(t))$. 
  \item
    \begin{equation}
      \phihat_t(x) > 0, \quad x \in (\bend(t), +\infty).
    \end{equation}
  \end{enumerate}
\end{regularitycondition}

Since the function $\Nhat_{t, \In}(s)$ depends on $t$ analytically, we have  $\Nhat_{t, \In}(s) = N_{\In}(s)$  if $t = 1$. Also we have (noting that $J_{c(1), c(1)}(s) = J_c(s)$)
\begin{equation} \label{eq:Ntilde'_t_in(-1)_dt}
  \begin{split}
    \left. \frac{d}{dt} \Nhat'_{t, \In}(-1) \right\rvert_{t = 1} = {}& \frac{1}{2\pi i} \oint_{\gamma_{\Out}} \frac{\left. \frac{d}{dt} U_t(J_{c_1(t), c_0(t)}(s)) \right\rvert_{t = 1}}{(s + 1)^2} ds \\
    = {}& \frac{-1}{2\pi i} \oint_{\gamma_{\Out}} \frac{U(J_c(s))}{(s + 1)^2} ds + \frac{c'_1(1)}{c} \frac{1}{2\pi i} \oint_{\gamma_{\Out}} \frac{J_c(s) U'(J_c(s))}{(s + 1)^2} ds \\
    & + \frac{c'_0(1) - c'_1(1)}{c} \frac{1}{2\pi i} \oint_{\gamma_{\Out}} \frac{J_c(s) U'(J_c(s))}{(s + 1)^3} ds \\
    = {}& \frac{c_1'(1)}{c} A_2 + \frac{c'_0(1) - c'_1(1)}{c} \left( \frac{\theta}{\theta + 1} A_2 + \frac{2\theta}{\theta + 1} A_1 \right) \\
    = {}& -\frac{\theta}{\theta - 1}.
  \end{split}
\end{equation}
By \eqref{eq:N_tilde_t_in(-1)}, \eqref{eq:critical_in_N_in} and \eqref{eq:Ntilde'_t_in(-1)_dt}, we have that
\begin{equation} \label{eq:limit_Nhat_t_In}
  \Nhat_{t, \In}(s) = -\frac{\theta}{\theta - 1} u (s + 1) + A_1(s + 1)^2 + \bigO(u^2(s + 1)) +\bigO(u(s + 1)^2)+  \bigO((s + 1)^3), \quad s \to -1.
\end{equation}
when $t$ is small, and $\Nhat_{t, \Out}(s)$ has a similar asymptotic expansion since for $s$ in the domain of $N_{t, \Out}(s)$ and $s$ is close to $-1$, 
\begin{equation} \label{eq:relation_U_in_U_out_tilde}
  \Nhat_{t, \Out}(s) = U_t(J_{c_1(t), c_0(t)}(s)) - \Nhat_{t, \In}(s).
\end{equation}

Now we consider the behaviour of $\Iinvhat_{1, t}(z)$ and $\Iinvhat_{2, t}(z)$ around $0$ and $\aend(t)$. To this end, we use the functions defined in Appendix \ref{sec:limiting_J}, namely $J^{(\pre)}$ defined in \eqref{eq:J_mapping} and $\Iinvhat^{(\pre)}_1$ and $\Iinvhat^{(\pre)}_2$ as the inverse functions of $J^{(\pre)}$. We denote, for $u = t - 1 < 0$, the functions depending on $u$ and equivalently $t$
\begin{align}
  k(u) = {}& \frac{\theta^2 - 1}{\theta^{1 + \frac{\theta}{\theta + 1}}} \frac{A_1}{-u}, & K(u) = {}& \theta c \left( \frac{\theta}{\theta^2 - 1} \frac{-u}{A_1} \right)^{\frac{\theta + 1}{\theta}}.
\end{align}
We have that if $(-u)$ is small and $\lvert s + 1 \rvert$ is small, then $K(u) J^{(\pre)}(k(u) (s + 1))$ is an approximation of $J_{c_1(t), c_0(t)}(s)$. Hence, $-1 + k(u)^{-1} \Iinvhat^{(\pre)}_1(K(u)^{-1}z)$ and $-1 + k(u)^{-1} \Iinvhat^{(\pre)}_1(K(u)^{-1}z)$ are approximations of $\Iinvhat_{1, t}(z)$ and $\Iinvhat_{2, t}(z)$, respectively. To be precise, if $\epsilon' < 0$ is a small positive number, $C$ be a large positive number, and for all $t$ such that $(1 - t) C < \epsilon$, we have
\begin{align}
  \Iinvhat_{1, t}(z) = {}& -1 + k(u)^{-1} f^{(1)}_0 \left( \Iinv^{(\pre)}_1 \left( g^{(1)}_0(K(u)^{-1} z) \right) \right), \quad z \in \compC \setminus (\aend(t), \bend(t)) \text{ and } \lvert z \rvert < \epsilon', \label{eq:limit_Ihat_1} \\
  \Iinvhat_{2, t}(z) = {}& -1 + k(u)^{-1} f^{(2)}_0 \left( \Iinv^{(\pre)}_2 \left( g^{(2)}_0(K(u)^{-1} z) \right) \right), \quad z \in \halfH \setminus (\aend(t), \bend(t)) \text{ and } \lvert z \rvert < \epsilon', \label{eq:limit_Ihat_2}
\end{align}
where $g^{(1)}_0$ is a conformal mapping from $D(0, K(u)^{-1} \epsilon')$ to $\compC$, $g^{(2)}_0(z)$ is a conformal mapping from $D(0, K(u)^{-1} \epsilon') \cap \halfH$ to $\halfH$, $f^{(1)}_0(s)$ is a conformal mapping from $D(0, k(u) \epsilon'') \cap (\compC \setminus \overline{\D^{(\pre)}})$ to $\compC$, and $f^{(2)}_0(s)$ is a conformal mapping from $D(0, k(u) \epsilon'') \cap \D^{(\pre)}$ to $\compC$, where $\D^{(\pre)}$ is defined in Appendix \ref{sec:limiting_J}, and $\epsilon'' > 0$ is a constant depending on $\epsilon'$. All the four conformal mappings are close to the identity mapping, in the sense that for $i = 1, 2$, there is $C' > 0$ and $\delta' \in (0, 1)$ that for $z, s$ in the domains of $g^{(i)}_0$ and $f^{(i)}_0$ stated above,
\begin{align}
  \lvert g^{(i)}_0(z) - z \rvert \leq {}& (-u) C' \lvert z \rvert, & \lvert f^{(i)}_0(s) - s \rvert \leq {}& (-u) C' \lvert s \rvert, & \lvert z \rvert, \lvert s \lvert \leq C, \\
  \lvert g^{(i)}_0(z) - z \rvert \leq {}& \delta' \lvert z \rvert, & \lvert f^{(i)}_0(s) - s \rvert \leq {}& \delta' \lvert s \rvert, & \lvert z \rvert, \lvert s \lvert \geq C.
\end{align}

From the asymptotics of $\Iinvhat_{1, t}(z)$ and $\Iinvhat_{2, t}(z)$ in \eqref{eq:limit_Ihat_1} and \eqref{eq:limit_Ihat_2} and the asymptotics of $\Nhat_{t, \In}(s)$ in \eqref{eq:limit_Nhat_t_In}, we derive the asymptotics of $\Ghat_t(z)$, $\Ghattilde_t(z)$, $\gfnhat_t(z)$ and $\gfnhattilde_t(z)$ from \eqref{eq:gfnhat_gfnhattilde_t} and \eqref{eq:N(s)}, and also the asymptotics of $\phihat_t(z)$ in \eqref{eq:defn_ellha_t}, if $z$ is close to $0$. We are not going to state the detailed results of their asymptotic expansions, and only state the following rough estimates:

\begin{prop} \label{prop:phihat}
  Let $\epsilon, \epsilon'$ be two small enough positive constants. For all $t \in (1 - \epsilon, 1)$, $d\muhat^{(V_t)}(x) = \psihat_t(x) dx$ with support $[\aend(t), \bend(t)]$ is the equilibrium measure for $V_t$, and satisfies Regularity Conditions \ref{reg:primary} and \ref{reg:soft-edge}. Also there is $\delta > 0$ depending on $\epsilon$ and $\epsilon'$, such that
  \begin{enumerate}
  \item  \label{enu:prop:phihat_1}
    For all $x \in [0, \aend(t) - \epsilon' (-u)^{\frac{\theta + 1}{\theta}}]$,
    \begin{equation}
      \Re \phihat_t(x) < -\delta (-u)^2.
    \end{equation}
  \item
    For all $x, y \in [0, \epsilon' (-u)^{\frac{\theta + 1}{\theta}}]$,
    \begin{equation}
      \gfnhat_t(x) + \gfnhattilde_t(y) - V_t(x) - \ellhat_t < -\delta (-u)^2.
    \end{equation}
  \item
    For $z \in D(\bend(t), \epsilon) \setminus [\aend(t), \bend(t)]$ and $z \in D(\aend(t), \epsilon (-u)^{(\theta + 1)/\theta}) \setminus [\aend(t), \bend(t)]$,
    \begin{align}
      \phihat_t(z) = {}& -\frac{4\pi}{3} \dhat_{2,t} (z - \bend(t))^{\frac{3}{2}} (1 + \bigO(z - \bend(t))), \label{eq:phihat_est_b} \\
      \phihat_t(z) = {}& \pm i \frac{4\pi}{3} (-u)^2 \dhat_{1, t} ((-u)^{-\frac{\theta + 1}{\theta}}(z - \aend(t)))^{\frac{3}{2}} (1 + \bigO((-u)^{-\frac{\theta + 1}{\theta}}(z - \aend(t))) \pm 2\pi i, \label{eq:phihat_est_a}
    \end{align}
    where $\dhat_{1, t}$ and $\dhat_{2, t}$ are positive for all $t$, the term $\bigO(z - \bend(t))$ is analytic in $D(\bend(t), \epsilon)$, and the term $\bigO((-u)^{-\frac{\theta + 1}{\theta}}(z - \aend(t)))$ is analytic in $D(\aend(t), \epsilon (-u)^{(\theta + 1)/\theta})$. 
  \end{enumerate}
\end{prop}

\subsection{Deformation along $C_2$  when $t \in (0, 1)$} \label{subsec:deformation_along_C2_large_t}

Generally, the local argument in Section \ref{subsec:deformation_along_C2_small_t} cannot say much about the equilibrium measure for $V_t$ if $t < 1$ is not close to $1$. In this subsection, we take the assumption that $V_t$ is one-cut regular in the soft edge regime, and its equilibrium measure $\muhat^{(V_t)}$ satisfies Regularity Conditions \ref{reg:primary} and \ref{reg:soft-edge}, as in Theorem \ref{thm:universality}. In other words, the equilibrium measure $\muhat^{(V_t)}$ constructed in Section \ref{subsec:deformation_along_C2_small_t} for $t \in (1 - \epsilon, 1)$ can be extended to $t \in (0, 1)$. We note that by \cite[Theorem 1.11]{Claeys-Romano14}, the equilibrium measure $\muhat^{(V_t)}$ is determined by $c_1(t), c_0(t)$ that are solutions to \eqref{eq:F_E} for $u, v$ respectively. The assumption implies the following result.

\begin{prop} \label{prop:extension_t}
  Suppose $V_t$ satisfies Regularity Conditions \ref{reg:primary} and \ref{reg:soft-edge}. For all $t \in (0, 1)$, there are $c_1(t)$ and $c_0(t)$ that depend on $t$ continuously and coincide with $c_1(t)$ and $c_0(t)$ obtained in Section \ref{subsec:deformation_along_C2_small_t} when $1 - t$ is small positive, such that if we use \eqref{eq:sat_and_at}, \eqref{eq:sbt_and_bt}, \eqref{eq:formula_Nhat_in}, \eqref{eq:N(s)}, \eqref{eq:defn_psihat_small_u}, \eqref{eq:gfnhat_gfnhattilde_t}, \eqref{eq:defn_ellha_t} to define $s_a(t), s_b(t), \aend(t), \bend(t), \Ghat_t, \Ghattilde_t, \psihat_t, \gfnhat_t, \gfnhattilde_t, \ellhat_t$, and then define $d\muhat^{(V_t)} = \psihat_t(x) dx$ supported on $[\aend(t), \bend(t)]$, then $\muhat^{(V_t)}$ is the equilibrium measure associated to $V_t$. 

  Furthermore, Proposition \ref{prop:phihat} and \eqref{eq:phihat_est_b}, \eqref{eq:phihat_est_a} hold uniformly for all $t \in (\epsilon, 1)$, where $\epsilon$ is any small positive number.
\end{prop}

\begin{proof}
By  assumption,  for each $V_t$ with $t \in (0, 1)$, its equilibrium measure $\muhat^{(V_t)}$ satisfies the  Regularity Conditions \ref{reg:primary} and \ref{reg:soft-edge}. Therefore,  \cite[Theorem 1.12]{Claeys-Romano14} applies to $V_t$ and $\muhat^{(V_t)}$ in place of  $V$ and $\mu$ therein, and $c_0(t)$ and $c_1(t)$ are uniquely determined by \eqref{eq:F_E} (which is \cite[Equations (1.41) and (1.42)]{Claeys-Romano14}). Hence the definitions of $c_0(t)$ and $c_1(t)$  extend to $t \in (0, 1)$, and the other parts of Proposition \ref{prop:extension_t} follow. 
\end{proof}

\section{Asymptotic analysis of the RH problems for $p_n$ and $q_n$ in the transition regime} \label{sec:RH_transitive}

As mentioned in Section \ref{sec:Strategy}, we prove Theorem \ref{thm:main}  by studying the asymptotics of  vector-valued Riemann-Hilbert problems for the biorthogonal polynomials defined in \eqref{eq:biorthogonality}; see  \cite{Wang-Zhang21} for a similar  strategy for the hard edge universality of  Muttalib-Borodin ensemble.  In the  transition regime considered in this section, we focus on the construction of  the local parametrix near the origin by using the solution of the new model RH problem \ref{RHP:general_model}. 

In this section, we assume that $(t - 1) = \bigO(n^{-1/2})$, and let $t$ be expressed by $\tau \in \realR$ in \eqref{eq:t_and_tau}. We abuse the notation to write $c$, $\gfn$, $\gfntilde$, $\ell$, $\phi$, $b$, $\Iinv_1$, $\Iinv_2$, $V$ and $J_c$ to mean $c(t)$, $\gfn_t$, $\gfntilde_t$, $\ell_t$, $\phi_t$, $b(t)$, $\Iinv_{1, t}$, $\Iinv_{2, t}$, $V_t$ and $J_{c(t)}$ that are defined in Sections \ref{subsec:deformation_algebraic} and \ref{subsec:deformation_C1}. When $t = 1$, these notations become identical to those in Section \ref{subsec:exact_equilibrium_measure}. We suppress the $t$ dependence in these symbols, not only to shorten the expressions, but also to make our presentation parallel to that in \cite{Wang-Zhang21} so that we can refer arguments there more easily.

\subsection{Asymptotic analysis of the RH problem for $p_n$} \label{subsec:asy_p_n}

\subsubsection{Transformations $Y \to T \to S \to Q$, and local parametrix around $b$} \label{subsec:trans_YTSQ}

This section is parallel to \cite[Sections 3.1--3.5]{Wang-Zhang21}. When we implement ideas from \cite{Wang-Zhang21}, we only give details when it differs from their counterparts in \cite{Wang-Zhang21}, and refer the reader to corresponding sections and equations in \cite{Wang-Zhang21} for details.

Let $p_n(x)$  be the monic biorthogonal polynomial  defined in \eqref{eq:biorthogonality}, we define its Cauchy-like transform \cite[Equation (2.24)]{Wang-Zhang21}
\begin{equation} \label{eq:defn_Cp_n}
  C p_n(z) = \frac{1}{2\pi i} \int^{\infty}_0 \frac{p_n(x)}{x^{\theta} - z^{\theta}} x^{\alpha} e^{-n V(x)} dx, \quad z\in \halfH \setminus [0,+\infty),
\end{equation}
and denote the vector form $Y = (Y_1, Y_2) := (p_n, C p_n)$ as in \cite[Equation (2.23)]{Wang-Zhang21}. Then it is known that $Y$ satisfies  the following vector RH problem; see \cite[RH problem 2.2]{Wang-Zhang21}.

\begin{RHP} \label{RHP:OPs} Let 
 $Y = (Y_1, Y_2)$ be a vector-valued function  defined and analytic in $(\compC, \mathbb{H}_{\theta} \setminus [0, +\infty))$,  that is, $Y_1$ is entire and  $Y_2$ is analytic in $\mathbb{H}_\theta \setminus [0, \infty)$. 
 \begin{enumerate} 
  \item $Y$ satisfies the  jump condition   \begin{equation}
      Y_+(x) =Y_-(x)
      \begin{pmatrix}
        1 & \frac{1}{\theta  }x^{\alpha+1-\theta}e^{-nV_t(x)}  \\
        0 & 1
      \end{pmatrix},
      \qquad x>0.
    \end{equation}
   
  \item
    As $z \to \infty$ in $\compC$, we have $Y_1(z) = z^n + \bigO(z^{n - 1})$, and as $z \to \infty$ in $\mathbb{H}_{\theta}$,  we have $Y_2(z) = \bigO(z^{-(n + 1)\theta})$.  
  \item As $z \to 0$ in $\compC$, we have $Y_1(z)=\bigO(1)$, and as $z \to 0$ in $\mathbb{H}_{\theta}$,
    we have
    \begin{equation}
      Y_2(z) =
      \begin{cases}
        \bigO(1), & \text{$\alpha+1-\theta > 0$,} \\
        \bigO(\log z), & \text{$\alpha+1-\theta = 0$,} \\
        \bigO(z^{\alpha+1-\theta}), & \text{$\alpha+1-\theta < 0$.}
      \end{cases}
    \end{equation}
  \item  For $x > 0$, we have the  boundary condition $Y_2(e^{\frac{\pi i}{\theta} }x) = Y_2(e^{-\frac{\pi i}{\theta}}x)$.
    
  \end{enumerate}
\end{RHP}

Recall the functions $\gfn$, $\gfntilde$ and constant $\ell$ ($\gfn_t$, $\gfntilde_t$ and $\ell_t$ before the $t$-dependence is suppressed) defined in \eqref{eq:gfn_gfntilde_t} and \eqref{eq:defn_phi}. Let \cite[Equation (3.1)]{Wang-Zhang21}
\begin{equation} \label{eq:defn_T}
  T(z) = (T_1(z), T_2(z)) := (Y_1(z) e^{-n\gfn(z)}, Y_2(z) e^{n(\gfntilde(z) - \ell)}).
\end{equation}
Recall number $b$ and function $\phi$ ($b(t)$ and $\phi_t$ before the $t$-dependence is suppressed) defined in the beginning of Section \ref{subsec:deformation_C1} and equation \eqref{eq:defn_phi}, respectively. Let $\Sigma_1 \subseteq  (\halfH \cap \compC_+)$ be a contour connecting $0$ and $b$, $\Sigma_2 = \overline{\Sigma_1} \subseteq (\halfH \cap \compC_-)$ like in \cite[Figure 2]{Wang-Zhang21}, and let $\Sigma = [0, +\infty) \cup \Sigma_1 \cup \Sigma_2$, as shown in Figure \ref{fig:Sigma}, like in \cite[Equation (3.7)]{Wang-Zhang21}. Moreover, we let the ``lens'' be the region enclosed by $\Sigma_1$ and $\Sigma_2$, and the upper/lower part of the lens be the intersection of the lens with $\compC_{\pm}$. We then define, analogous to \cite[Equation (3.6)]{Wang-Zhang21},
\begin{equation} \label{eq:S_from_T}
  S(z) =
  \begin{cases}
    T(z), & \text{outside the lens}, \\
    T(z)
    \begin{pmatrix}
      1 & 0 \\
      \frac{\theta}{z^{\alpha + 1 - \theta}} e^{-n\phi(z)} & 1
    \end{pmatrix},
    & \text{lower part of the lens}, \\
    T(z)
    \begin{pmatrix}
      1 & 0 \\
      -\frac{\theta}{z^{\alpha + 1 - \theta}} e^{-n\phi(z)} & 1
    \end{pmatrix},
    & \text{upper part of the lens}.
  \end{cases}
\end{equation}

We have that $T(z)$ and $S(z)$ satisfy Riemann-Hilbert problems, as stated in \cite[RH problems 3.1 and 3.2]{Wang-Zhang21}.  We remark that in \cite[RH problems 3.1 and 3.2]{Wang-Zhang21} the $\gfn$, $\gfntilde$ and $\phi$ are associated to the equilibrium measure, and our $\gfn$, $\gfntilde$ and $\phi$ are associated to the equilibrium measure only if $\tau < 0$, as explained in Remark \ref{rmk:eq_measure_or_not}. However, the derivation for \cite[RH problems 3.1 and 3.2]{Wang-Zhang21} does not rely on the relation to equilibrium measure.

Next, we recall the functions $\Iinv_1$, $\Iinv_2$ and parameter $c$ ($\Iinv_{1, t}$, $\Iinv_{2, t}$ and $c(t)$ before the $t$-dependence is suppressed) defined in Section \ref{subsec:deformation_algebraic} and the beginning of Section \ref{subsec:deformation_C1}.  We construct the global parametrix $P^{(\infty)}(z) = (P^{(\infty)}_1(z), P^{(\infty)}_2(z))$ such that \cite[Equations (3.19) and (3.20)]{Wang-Zhang21}
\begin{align}
  P^{(\infty)}_1(z) = {}& \P(\Iinv_1(z)), &z \in {}& \compC \setminus [0, b], \label{eq:Pinfty_1} \\
  P^{(\infty)}_2(z) = {}& \P(\Iinv_2(z)), &z \in {}& \halfH \setminus [0, b], \label{eq:Pinfty_2}
\end{align}
where, with $D$ and $\gamma_1$ defined in Section \ref{subsec:exact_equilibrium_measure} \cite[Equation (3.18)]{Wang-Zhang21}
\begin{equation} \label{eq:Pinfty_in_one}
  \P(s) =
  \begin{cases}
    \frac{s}{\sqrt{(s + 1)(s - s_b)}} \left( \frac{s + 1}{s} \right)^{\frac{\theta - \alpha - 1}{\theta}}, & s \in \compC \setminus \overline{D}, \\
    \frac{c^{\alpha + 1 - \theta} s(s + 1)^{\alpha + 1 - \theta}}{\theta \sqrt{(s + 1)(s - s_b)}}, & s \in D,
  \end{cases}
\end{equation}
such that $s_b = 1/\theta$, the branch cuts of $\sqrt{(s + 1)(s - s_b)}$, $\left( \frac{s + 1}{s} \right)^{\frac{\theta - \alpha - 1}{\theta}}$ and $(s + 1)^{\alpha + 1 - \theta}$ are taken along $\gamma_1$, $[-1, 0]$ and $(-\infty, -1]$ respectively.

Then the function \cite[Equation (3.29)]{Wang-Zhang21}
\begin{equation} \label{eq:defn_Q}
  Q(z) = (Q_1(z), Q_2(z)) := \left( \frac{S_1(z)}{P^{(\infty)}_1(z)}, \frac{S_2(z)}{P^{(\infty)}_2(z)} \right)
\end{equation}
satisfies the Riemann-Hilbert problem stated in \cite[RH problem 3.7]{Wang-Zhang21}.

Analogous to \cite{Wang-Zhang21},  the local parametrix near the right ending point $b$ can be constructed in terms of the  Airy parametrix.
Let $\epsilon$ be a small enough positive number. We define, for $z$ in a neighbourhood $D(b, \epsilon)$ of $b$, \cite[Equation (3.35)]{Wang-Zhang21}
\begin{equation} \label{def:fb_hard}
  f_b(z) = \left( -\frac{3}{4} \phi(z) \right)^{\frac{2}{3}}.
\end{equation}
Due to \eqref{eq:behaviour_phi_t_at_b}, we have that it is a conformal mapping with $f_b(b) = 0$ and $f'_b(b) > 0$.

In Appendix \ref{app:Airy}, the $2 \times 2$ matrix-valued function $\Psi^{(\Ai)}$, which is the well-known Airy parametrix is defined on $\compC \setminus \Gamma_{\Ai}$, where $\Gamma_{\Ai} = e^{-\frac{2\pi i}{3}} [0, +\infty) \cup \realR \cup e^{\frac{2\pi i}{3}} [0, +\infty)$ is the jump contour specified in RH problem \ref{rhp:Ai}. We define \cite[Equation (3.38)]{Wang-Zhang21}
\begin{equation} \label{eq:defn_P^b}
  P^{(b)}(z) = E^{(b)}(z) \Psi^{(\Ai)}(n^{\frac{2}{3}} f_b(z))
  \begin{pmatrix}
    e^{-\frac{n}{2} \phi(z)} g^{(b)}_1(z) & 0 \\
    0 & e^{\frac{n}{2} \phi(z)} g^{(b)}_2(z)
  \end{pmatrix},
\end{equation}
where \cite[Equations (3.37) and (3.43)]{Wang-Zhang21}
\begin{align}
  g^{(b)}_1(z) = {}& \frac{z^{(\theta - \alpha - 1)/2}}{P^{(\infty)}_1(z)}, & g^{(b)}_2(z) = {}& \frac{z^{(\alpha + 1 - \theta)/2}}{\theta P^{(\infty)}_2(z)},
\end{align}
\begin{equation}
  E^{(b)}(z) = \frac{1}{\sqrt{2}}
  \begin{pmatrix}
    g^{(b)}_1(z) & 0 \\
    0 & g^{(b)}_2(z)
  \end{pmatrix}^{-1}
  e^{\frac{\pi i}{4} \sigma_3}
  \begin{pmatrix}
    1 & -1 \\
    1 & 1
  \end{pmatrix}
  \begin{pmatrix}
    n^{\frac{1}{6}} f_b(z)^{\frac{1}{4}} & 0 \\
    0 & n^{-\frac{1}{6}} f_b(z)^{-\frac{1}{4}}
  \end{pmatrix}.
\end{equation}
We have that $P^{(b)}(z)$ satisfies a RH problem; see \cite[RH problem 3.9]{Wang-Zhang21}.
Now we specify the shape of $\Sigma_1$ (and then $\Sigma_2 = \overline{\Sigma_1}$) in $D(b, \epsilon)$ as $f^{-1}_b(\{ e^{\frac{2\pi i}{3}} [0, +\infty) \}) \cap D(b, \epsilon)$. Then we define the vector-valued function $V^{(b)}$ by \cite[Equation (3.46)]{Wang-Zhang21}
\begin{equation} \label{eq:defn_V^b}
  V^{(b)}(z) = Q(z) P^{(b)}(z)^{-1}, \quad z \in D(b, \epsilon) \setminus \Sigma,
\end{equation}
which satisfies the same RH problem as \cite[RH problem 3.10]{Wang-Zhang21}.

\subsubsection{Local parametrix around $0$} \label{subsec:local_para}

\paragraph{A local version of the RH problem for $Q$}

Let $r = r_n$ be a small positive number that depends on $n$ as in \eqref{def:rn} below. We specify the shape of $\Sigma_1$ (and then $\Sigma_2 = \overline{\Sigma_1}$) in $D(0, r)$ as $\{ z \in \compC : \arg z = \pi/(2\theta) \} \cap D(0, r)$. Then we define $\theta + 1$ functions $U_0(z), U_1(z), \dotsc, U_{\theta}(z)$ \cite[Equations (3.51) and (3.52)]{Wang-Zhang21}, see also \cite[Equations (3.48)--(3.50)]{Wang-Zhang21} \footnote{In \cite[Section 3.6]{Wang-Zhang21} there is a notational mismatch: $r = r_n$ in the beginning of \cite[Section 3.6]{Wang-Zhang21} and \cite[RH problem 3.11]{Wang-Zhang21} is equivalent to $r^{\theta}_n$ in \cite[Equation (3.116) and later]{Wang-Zhang21}. We follow the latter usage of $r = r_n$.}
\begin{align}
  U_0(z) = {}& Q_2(z^{\frac{1}{\theta}}), & z \in {}& D(0, r^{\theta}) \setminus \{ (-r^{\theta}, r^{\theta}) \cup (-ir^{\theta}, ir^{\theta}) \}, \label{eq:U_0_def} \\
  U_k(z) = {}& Q_1(z^{\frac{1}{\theta}} e^{\frac{2(k - 1)}{\theta} \pi i}), & z \in {}& D(0, r^{\theta}) \setminus \{ (-r^{\theta}, r^{\theta}) \cup (-ir^{\theta}, ir^{\theta}) \}, \quad k = 1, 2, \dotsc, \theta. \label{eq:U_0_def_2}
\end{align}
Then $U(z) = (U_0(z), \dotsc, U_{\theta}(z))$ satisfies \cite[RH problem 3.11]{Wang-Zhang21}, which is stated  below for the completeness.
\begin{RHP}\label{rhp:U} \hfill
  \begin{enumerate} 
  \item
    $ U=(U_0,U_1,\ldots,U_\theta)$ is defined and analytic in $D(0,r^{\theta}) \setminus \{(-r^{\theta},r^{\theta}) \cup (-ir^{\theta},ir^{\theta})\}$.
  \item
    For $z\in (-r^{\theta},r^{\theta})\cup(-ir^{\theta},ir^{\theta})\setminus\{0\}$, we have
    \begin{equation}
      U_+(z) = U_-(z) J_U(z),
    \end{equation}
    where
    \begin{equation} \label{eq:defn_J_U}
      J_U(z) =
      \begin{cases}
        \begin{pmatrix}
          1 & \theta z^{\frac{\theta - 1 - \alpha}{\theta}} e^{-n \phi(z^{1/\theta})} \frac{P^{(\infty)}_2(z^{1/\theta})}{P^{(\infty)}_1(z^{1/\theta})} \\
          0 & 1
        \end{pmatrix}
        \oplus I_{\theta - 1}, & z \in (0,ir^{\theta})\cup (0,-ir^{\theta}), \\
        \begin{pmatrix}
          0 & 1 \\
          1 & 0
        \end{pmatrix}
        \oplus I_{\theta - 1}, & z \in (0, r^{\theta}), \\
        \Mcyclic, & z \in (-r^{\theta}, 0),
      \end{cases}
    \end{equation}
    with $\Mcyclic$ defined in \eqref{eq:defn_Mcyclic}, and the orientations of the rays are shown in Figure \ref{fig:jumps-U}.
  \item
    As $z\to 0$ from $D(0,r^{\theta}) \setminus \{(-r^{\theta},r^{\theta}) \cup (-ir^{\theta},ir^{\theta})\}$, we have
    \begin{enumerate} 
    \item
      \begin{equation}
        U_0(z) =
        \begin{cases}
          \bigO (z^{1 - \frac{\alpha + 3/2}{\theta + 1}}), & \alpha > \theta-1, \\
          \bigO (z^{\frac{1}{2(1+\theta)}} \log z ), & \alpha = \theta - 1, \\
          \bigO (z^{\frac{\alpha + 1}{\theta} - \frac{\alpha + 3/2}{\theta + 1}}), & -1 < \alpha < \theta - 1,
        \end{cases}
      \end{equation}
    \item
      for $\arg z \in (0, \pi/2) \cup (-\pi/2, 0)$,
      \begin{equation}
        U_1(z) =
        \begin{cases}
          \bigO (z^{1 - \frac{\alpha + 3/2}{\theta + 1}}), & \alpha > \theta - 1 , \\
          \bigO (z^{\frac{1}{2(1+\theta)}} \log z ), & \alpha = \theta-1, \\
          \bigO (z^{\frac{\alpha + 1}{\theta} - \frac{\alpha + 3/2}{\theta + 1}}), & -1 < \alpha < \theta-1,
        \end{cases}
      \end{equation}
    \item
      for $k = 2, \dotsc, \theta$, or $k=1$ and $\arg z \in (\pi/2,\pi)\cup (-\pi,-\pi/2)$,
      \begin{equation}
        U_k(z)=\bigO (z^{\frac{\alpha + 1}{\theta} - \frac{\alpha + 3/2}{\theta + 1}}).
      \end{equation}
    \end{enumerate}
  \end{enumerate}
\end{RHP}
\begin{figure}[htb]
  \begin{minipage}[b]{0.45\linewidth}
    \centering
    \includegraphics{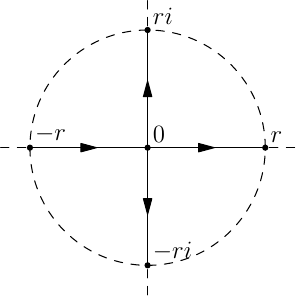}
    \caption{The jump contour for the RH problem \ref{rhp:U} for $U$ and for the RH problem \ref{rhp:tildeU} for $\U$.}
    \label{fig:jumps-U}
  \end{minipage}
  \hspace{\stretch{1}}
  \begin{minipage}[b]{0.45\linewidth}
    \centering
    \includegraphics{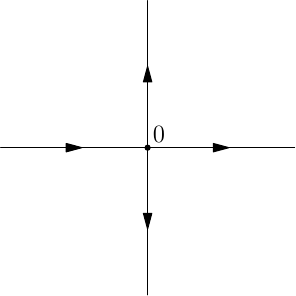}
    \caption{The jump contour for RH problem \ref{RHP:general_model} for $\Phi$, for RH problem \ref{RHP:general_model_tilde} for $\Phitilde$, and for the RH problem \ref{RHP:general_model_tilde_check} for $\Phicheck$.}
    \label{fig:jumps-MeiG}
  \end{minipage}
\end{figure}

\paragraph{Construction of the local parametrix around $0$}

As in \cite[Equations (3.104)--(3.109)]{Wang-Zhang21}, we define
\begin{equation} \label{def:M}
  M(z) = \diag(m_0(z), m_1(z), \dotsc, m_{\theta}(z)), \quad z \in \compC \setminus \realR,
\end{equation}
where (noting that our $V$ below means $V_t$ defined in \eqref{eq:defn_Vt})
\begin{align}
  m_0(z) = {}&
               \begin{cases}
                 \gfntilde(z^{\frac1 \theta}) - \gfntilde_+(0), & z \in \compC_+, \\
                 2\pi i + \gfntilde(z^{\frac 1 \theta}) - \gfntilde_+(0), & z \in \compC_-,
               \end{cases}
               \label{def:m0}\\
  m_1(z) = {}&
               \begin{cases}
                 -\gfn(z^{\frac1 \theta}) + V(z^{\frac1 \theta}) + \ell - \gfntilde_-(0), & z \in \compC_+, \\
                 -2\pi i - \gfn(z^{\frac1 \theta}) + V(z^{\frac1 \theta}) + \ell - \gfntilde_-(0), & z \in \compC_-,
               \end{cases} \\
  \intertext{and for $j = 2, \dotsc, \theta$, $\arg z \in (-\pi, 0) \cup (0, \pi)$,}
  m_j(z) = {}&
               \begin{cases}
                 -\gfn(z^{\frac1 \theta} e^{\frac{2(j - 1)}{\theta}\pi i}) + V(z^{\frac1 \theta} e^{\frac{2(j - 1)}{\theta}\pi i}) + \ell - \gfntilde_-(0), & z^{\frac1 \theta} e^{\frac{2(j - 1)}{\theta}\pi i} \in \compC_+, \\
                 -2\pi i - \gfn(z^{\frac1 \theta} e^{\frac{2(j - 1)}{\theta}\pi i}) + V(z^{\frac1 \theta} e^{\frac{2(j - 1)}{\theta}\pi i}) + \ell - \gfntilde_-(0), & z^{\frac1 \theta} e^{\frac{2(j - 1)}{\theta}\pi i} \in \compC_-,
               \end{cases}
                                                                                                                                                             \label{def:mj}
\end{align}
and also define
\begin{equation} \label{def:N}
  N(z) = \diag (n_0(z), n_1(z), \dotsc, n_{\theta}(z)), \quad z\in\compC \setminus (-\infty,b^{\theta}],
\end{equation}
where
\begin{align}\label{def:ni}
  n_0(z) = {}& P^{(\infty)}_2(z^{\frac1 \theta}), & n_j(z) = {}& P^{(\infty)}_1(e^{\frac{2(j - 1)}{\theta} \pi i} z^{\frac1 \theta}), \quad j = 1, \dotsc, \theta.
\end{align}
We also define
\begin{equation} \label{eq:Npre}
  N^{(\pre)}(z) = \diag (n^{(\pre)}_0(z), n^{(\pre)}_1(z), \dotsc, n^{(\pre)}_{\theta}(z)), \quad z\in\compC \setminus \realR,
\end{equation}
where
\begin{align}
  n^{(\pre)}_0(z) = {}&
               \begin{cases}
                 \frac{c^{\frac{2(\alpha+1)-\theta}{2(1+\theta)}}}{\sqrt{\theta(1+\theta)}}e^{\frac{\theta-2(\alpha+1)}{2(1+\theta)}\pi i} z^{\frac{\alpha + 3/2}{1+\theta} - 1}, & z \in \compC_+, \\
                 -\frac{c^{\frac{2(\alpha+1)-\theta}{2(1+\theta)}}}{\sqrt{\theta(1+\theta)}}e^{\frac{2(\alpha+1) - \theta}{2(1+\theta)}\pi i} z^{\frac{\alpha + 3/2}{1+\theta} - 1}, & z \in \compC_-,
               \end{cases} && \\
  n^{(\pre)}_1(z) = {}&
               \begin{cases}
                 \frac{\theta c^{\frac{2(\alpha+1)-\theta}{2(1+\theta)}}}{\sqrt{\theta(1+\theta)}} e^{\frac{2(\alpha+1)-\theta}{2(1+\theta)}\pi i}z^{\frac{\alpha + 3/2}{1+\theta} - \frac{\alpha + 1}{\theta}}, & z \in \compC_+, \\
                 \frac{\theta c^{\frac{2(\alpha+1)-\theta}{2(1+\theta)}}}{\sqrt{\theta(1+\theta)}} e^{\frac{\theta - 2(\alpha+1)}{2(1+\theta)}\pi i}z^{\frac{\alpha + 3/2}{1+\theta} - \frac{\alpha + 1}{\theta}}, & z \in \compC_-,
               \end{cases} && \\
  n^{(\pre)}_j(z) = {}& \frac{\theta c^{\frac{2(\alpha+1)-\theta}{2(1+\theta)}}}{\sqrt{\theta(1+\theta)}} e^{\frac{2(\alpha+1)-\theta}{2(1+\theta)}\pi i} e^{\frac{(j - 1)(\theta - 2(\alpha + 1))}{\theta(1 + \theta)} \pi i} z^{\frac{\alpha + 3/2}{1+\theta} - \frac{\alpha + 1}{\theta}}, && j=2,\ldots,\theta.
\end{align}
Recall $\rho$ defined in \eqref{eq:defn_rho} that is independent of $t$ (and $\tau$). Analogous to \cite[Proposition 3.14]{Wang-Zhang21}, we have the following estimates. 

\begin{prop} \label{prop:m_and_n}
  Let $m_i(z)$ and $n_i(z)$, $i = 0, 1, \dotsc, \theta$, be the functions defined in \eqref{def:m0}--\eqref{def:mj} and \eqref{def:ni}. As $z \to 0$, we have, with $\rho$ defined in \eqref{eq:defn_rho} and $c$ given in \eqref{eq:defn_c}, 
\begin{align}
  m_0(z) = {}&
               \begin{cases}
                 \frac{\theta + 1}{2\theta} \rho e^{-\frac{2\pi i}{\theta + 1}} z^{\frac{2}{\theta + 1}} (1 + \bigO(z^{\frac{1}{\theta + 1}})) & \\
                 + (t - 1) c^{-\frac{\theta}{\theta + 1}} e^{-\frac{\pi i}{\theta + 1}} z^{\frac{1}{\theta + 1}} (1 + \bigO(z^{\frac{1}{\theta + 1}})) \\
                 + \bigO((t - 1)^2 z^{\frac{1}{\theta + 1}}) + \bigO(z^{\frac{3}{\theta + 1}}), & z \in \compC_+, \\
                 \frac{\theta + 1}{2\theta} \rho e^{\frac{2\pi i}{\theta + 1}} z^{\frac{2}{\theta + 1}} (1 + \bigO(z^{\frac{1}{\theta + 1}})) \\
                 + (t - 1) c^{-\frac{\theta}{\theta + 1}} e^{\frac{\pi i}{\theta + 1}} z^{\frac{1}{\theta + 1}} (1 + \bigO(z^{\frac{1}{\theta + 1}})) \\
                 + \bigO((t - 1)^2 z^{\frac{1}{\theta + 1}}) + \bigO(z^{\frac{3}{\theta + 1}}), & z \in \compC_-,
               \end{cases} \\
  m_1(z) = {}&
               \begin{cases}
                 \frac{\theta + 1}{2\theta} \rho e^{\frac{2\pi i}{\theta + 1}} z^{\frac{2}{\theta + 1}} (1 + \bigO(z^{\frac{1}{\theta + 1}})) \\
                 + (t - 1) c^{-\frac{\theta}{\theta + 1}} e^{\frac{\pi i}{\theta + 1}} z^{\frac{1}{\theta + 1}} (1 + \bigO(z^{\frac{1}{\theta + 1}})) \\
                 + \bigO((t - 1)^2 z^{\frac{1}{\theta + 1}}) + \bigO(z^{\frac{3}{\theta + 1}}), & z \in \compC_+, \\
                 \frac{\theta + 1}{2\theta} \rho e^{-\frac{2\pi i}{\theta + 1}} z^{\frac{2}{\theta + 1}} (1 + \bigO(z^{\frac{1}{\theta + 1}})) & \\
                 + (t - 1) c^{-\frac{\theta}{\theta + 1}} e^{-\frac{\pi i}{\theta + 1}} z^{\frac{1}{\theta + 1}} (1 + \bigO(z^{\frac{1}{\theta + 1}})) \\
                 + \bigO((t - 1)^2 z^{\frac{1}{\theta + 1}}) + \bigO(z^{\frac{3}{\theta + 1}}), & z \in \compC_-, 
               \end{cases} \\
  m_j(z) = {}& \textstyle \frac{\theta + 1}{2\theta} \rho e^{\frac{2(2j - 1) \pi i}{\theta + 1}} z^{\frac{2}{\theta + 1}} (1 + \bigO(z^{\frac{1}{\theta + 1}})) + (t - 1) c^{-\frac{\theta}{\theta + 1}} e^{\frac{(2j - 1) \pi i}{\theta + 1}} z^{\frac{1}{\theta + 1}} (1 + \bigO(z^{\frac{1}{\theta + 1}})) \notag \\
             & \textstyle + \bigO((t - 1)^2 z^{\frac{1}{\theta + 1}}) + \bigO(z^{\frac{3}{\theta + 1}}), 
\end{align}
for $j=2,3, \dotsc, \theta$,  and
\begin{equation}
  n_k(z) = n^{(\pre)}_k(z)  (1 + \bigO(t - 1) + \bigO(z^{\frac{1}{1+\theta}})), \quad k = 0, 1, \dotsc, \theta.
\end{equation}
\end{prop}

From now on, we assume, as mentioned before, that $r=r_n$ shrinks with $n$, namely,
\begin{equation}\label{def:rn}
  r=r_n=n^{-\frac{\theta+1}{2\theta + 1}}.
\end{equation}
\begin{rmk}
  One could actually take the rate of shrinking to be $n^{-\kappa}$ with $(\theta + 1)/(3\theta) < \kappa < (\theta + 1)/(2\theta)$. Here, we choose $\kappa = \theta + 1/(2\theta + 1)$ for convenience. 
\end{rmk} 

For $z \in D(0,r_n^{\theta})\setminus (\mathbb{R} \cup i \mathbb{R})$, analogous to \cite[Equation (3.117)]{Wang-Zhang21}, we first introduce a $(\theta + 1) \times (\theta + 1)$ matrix-valued function
\begin{equation}\label{def:sfP0}
  \mathsf{P}^{(0)}(z) = \diag((\rho n)^{-\frac{k}{2}})^{\theta}_{k = 0} \Phi((\rho n)^{\frac{\theta + 1}{2}} z) N^{(\pre)}(z) N(z)^{-1} e^{nM(z)},
\end{equation}
where $\rho$ is defined in \eqref{eq:defn_rho}, $M(z)$, $N(z)$ and $N^{(\pre)}(z)$ are defined in \eqref{def:M}, \eqref{def:N} and \eqref{eq:Npre}, respectively, and $\Phi$ is the solution of the RH problem \ref{RHP:general_model}. We then have the following proposition regarding the RH problem for $\mathsf{P}^{(0)}$ \cite[Proposition 3.15]{Wang-Zhang21}.
\begin{prop}\label{RHP:sfP0}
The function $\mathsf{P}^{(0)}(z)$ defined in \eqref{def:sfP0} has the following properties.
  \begin{enumerate} 
  \item
    $ \mathsf{P}^{(0)}(z)$ is analytic in $D(0,r_n^{\theta}) \setminus (\mathbb{R} \cup i \mathbb{R})$.
  \item
    For $z\in D(0,r_n^{\theta}) \cap (\mathbb{R} \cup i \mathbb{R})$, we have
    \begin{equation}
      \mathsf{P}^{(0)}_{+}(z) = \mathsf{P}^{(0)}_{-}(z) J_U(z),
    \end{equation}
    where $J_U(z)$ is defined in \eqref{eq:defn_J_U}.
  \item
    For $z\in \partial D(0,r_n^{\theta})$, we have, as $n\to \infty$,
    \begin{equation}\label{eq:sfP0asy}
      \mathsf{P}^{(0)}(z) = \Upsilon(z) \Omega_{\pm} (I+\bigO(n^{-\frac{1}{2(2\theta + 1)}})), \qquad z \in \compC_{\pm},
    \end{equation}
    where $\Upsilon$ and $\Omega_{\pm}$ are defined in \eqref{def:Lpm}.
  \end{enumerate}
\end{prop}

With $\mathsf{P}^{(0)}$ given in \eqref{def:sfP0}, we next define, analogous to \cite[Equation (3.127)]{Wang-Zhang21}
\begin{equation}\label{def:P0}
  P^{(0)}(z) = \Omega_{\pm}^{-1}\Upsilon(z)^{-1}\mathsf{P}^{(0)}(z), \quad z \in \compC_{\pm}.
\end{equation}
In view of \eqref{def:P0}, Proposition \ref{RHP:sfP0} and
\begin{equation}\label{eq:UpsilonOmegajump}
  \Upsilon_+(\zeta) \Omega_+ =  \Upsilon_-(\zeta) \Omega_- \times
  \begin{cases}
    \begin{pmatrix}
      0 & 1 \\
      1 & 0
    \end{pmatrix}
    \oplus I_{\theta - 1}, & \zeta >0, \\
    \Mcyclic, & \zeta <0,
  \end{cases}
\end{equation}
the following RH problem for $P^{(0)}$ is then immediate \cite[Proposition 3.16]{Wang-Zhang21}.
\begin{RHP}\label{rhp:P0}
\hfill
  \begin{enumerate} 
  \item
    $P^{(0)}(z)$ is analytic in $D(0,r_n^{\theta}) \setminus (\mathbb{R} \cup i \mathbb{R})$.
  \item
    For $z\in D(0,r_n^{\theta}) \cap (\mathbb{R} \cup i \mathbb{R})$, we have
    \begin{equation}
      P^{(0)}_+(z) =
      \begin{cases}
        P^{(0)}_-(z) J_U(z), & z \in (0, ir_n^{\theta}) \cup (0, -ir_n^{\theta}), \\
        \left(\begin{pmatrix}
            0 & 1
            \\
            1 & 0
          \end{pmatrix} \oplus I_{\theta - 1} \right)
        P^{(0)}_-(z) J_U(z), & z \in (0, r_n^{\theta}), \\
        \Mcyclic^{-1} P^{(0)}_-(z) J_U(z), & z \in (-r_n^{\theta}, 0),
      \end{cases}
    \end{equation}
    where $J_U(z)$ is defined in \eqref{eq:defn_J_U}.
  \item
    For $z\in \partial D(0,r_n^{\theta})$, we have, as $n\to \infty$,
    \begin{equation}\label{eq:P0asy}
      P^{(0)}(z) = I+\bigO(n^{-\frac{1}{2(2\theta + 1)}}).
    \end{equation}
  \end{enumerate}
\end{RHP}

Finally, we set, analogous to \cite[Equation (3.130)]{Wang-Zhang21}
\begin{equation} \label{def:V0}
  V^{(0)}(z) = U(z) P^{(0)}(z)^{-1}, \quad z \in D(0,r_n^{\theta}) \setminus (\mathbb{R} \cup i \mathbb{R}),
\end{equation}
where $U$ is defined in \eqref{eq:U_0_def} and \eqref{eq:U_0_def_2} and satisfies RH problem \ref{rhp:U}. Then, by a similar argument as  \cite[Proposition 3.17]{Wang-Zhang21},  we have the following proposition.
\begin{prop} \label{rhp:V0}
The function $V^{(0)}(z)$ defined in \eqref{def:V0} has the following properties.
  \begin{enumerate} 
  \item \label{enu:rhp:V0:1}
    $V^{(0)}=(V^{(0)}_0, V^{(0)}_1, \dotsc, V^{(0)}_{\theta})$ is analytic in $D(0,r_n^{\theta}) \setminus \mathbb{R} $, or equivalently, its definition can be analytically extended onto $(0, ir^{\theta}_n) \cup (-ir^{\theta}_n, 0)$.
  \item \label{enu:rhp:V0:2}
    For $z\in (-r_n^\theta,r_n^{\theta})\setminus \{0\}$, we have
    \begin{equation} \label{eq:enu:rhp:V0:2}
      V^{(0)}_+(z) = V^{(0)}_-(z)
      \begin{cases}
        \begin{pmatrix}
          0 & 1 \\
          1 & 0
        \end{pmatrix}
        \oplus I_{\theta - 1}, & z \in (0, r_n^{\theta}), \\
     \Mcyclic, & z \in (-r_n^{\theta}, 0). 
      \end{cases}
    \end{equation}
  \item \label{enu:rhp:V0:3}
    For $z\in \partial D(0,r_n^{\theta})$, we have, as $n \to \infty$,
    \begin{equation}\label{eq:V0asy}
      V^{(0)}(z) = U(z)(I+\bigO(n^{-\frac{1}{2(2\theta + 1)}})).
    \end{equation}
  \item
    As $z \to 0$, we have
    \begin{equation}\label{eq:V0kzero}
      V^{(0)}_k(z)=\bigO(1), \qquad k=0,1,\ldots,\theta.
    \end{equation}
  \end{enumerate}
\end{prop}


\subsubsection{Final transformation} \label{subsubsec:final_trans_p}


Before performing the final transformation, we complete the description of the shape of $\Sigma_1$ (and then $\Sigma_2 = \overline{\Sigma_1}$). Let $\sigma_0$ and $\sigma_b$ be the intersections of $\Sigma_1$ with $\partial D(0, r_n)$ and $\partial D(b, \epsilon)$, respectively, and we let the undetermined part of $\Sigma_1$ be a contour $\Sigma^R_1$ connecting $\sigma_0$ and $\sigma_b$ such that $\Re \phi(z) > 0$ for $z \in \Sigma^R_1$. We may choose $\Sigma^R_1$ to be close to the real axis, since on $(0, b)$, $\Re \phi_+(x) = 0$ and $\phi'_+(x) = -2\pi i \psi(x)$, and $\psi(x)$ has the (almost) positivity on $(0, b)$ as stated in Regularity Condition \ref{reg:one-cut} and \eqref{eq:limit_psi_t_at_0}. The shape of $\Sigma_1$ and $\Sigma_2$ are shown in Figure \ref{fig:Sigma}.

\begin{figure}[htb]
  \begin{minipage}[t]{0.45\linewidth}
    \includegraphics{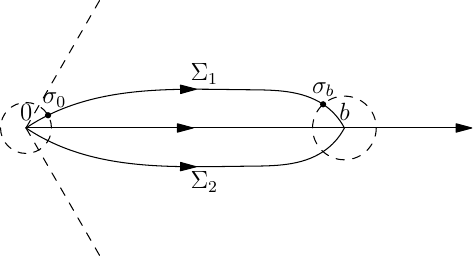}
    \caption{The contour $\Sigma$ ($\theta = 3$ in this instance).    \label{fig:Sigma}}
  \end{minipage}
  \hspace{\stretch{1}}
  \begin{minipage}[t]{0.45\linewidth}
    \includegraphics{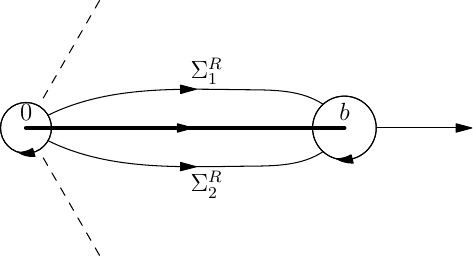}
    \caption{The contour $\Sigma_R$.     \label{fig:Sigma_R}}
  \end{minipage}
\end{figure}

We set, as in \cite[Equation 3.137)]{Wang-Zhang21} and shown in Figure \ref{fig:Sigma_R},
\begin{equation}\label{def:sigmaR}
 \Sigma^{R}:=[0,b]\cup [b+\epsilon, \infty) \cup \partial D(0,r_n) \cup \partial D(b,\epsilon) \cup \Sigma_1^R \cup \Sigma_2^R,
\end{equation}
where $\Sigma^R_1$ is specified above, and $\Sigma^R_2 = \overline{\Sigma^R_1}$.
We also divide the open disk $D(0,r_n)$ into $\theta$ parts by setting, as in \cite[Equation 3.139)]{Wang-Zhang21}
\begin{equation}\label{def:Wk}
  W_k:=\{z\in D(0,r_n)\setminus\{0\} \mid  \arg z \in (\frac{(2k - 3)\pi}{\theta}, \frac{(2k - 1)\pi}{\theta}) \},\qquad k=1,\ldots, \theta.
\end{equation}
We then define a $1\times 2$ vector-valued function $R(z) = (R_1(z), R_2(z))$ such that $R_1(z)$ is analytic in $\compC \setminus \Sigma^R$ and $R_2(z)$ is analytic in $\halfH \setminus \Sigma^R$, as follows \cite[Equation (3.141) and (3.142)]{Wang-Zhang21}:
\begin{align}
  R_1(z) = {}&
  \begin{cases}
    Q_1(z), & \hbox{$z\in \mathbb{C} \setminus \{D(b,\epsilon) \cup D(0,r_n) \cup \Sigma^R \}$,} \\
    V_1^{(b)}(z), & \hbox{$z\in D(b,\epsilon) \setminus [b-\epsilon, b]$,} \\
    V_1^{(0)}(z^\theta), & \hbox{$z\in W_1 \setminus [0,r_n]$,} \\
    V_k^{(0)}(z^\theta), & \hbox{$z\in W_k$, $k=2,\ldots,\theta$,}
  \end{cases} \label{def:R1} \\
  R_2(z) = {}&
  \begin{cases}
    Q_2(z), & \hbox{$z\in \halfH \setminus \{D(b,\epsilon) \cup D(0,r_n) \cup \Sigma^R \}$,} \\
    V_2^{(b)}(z), & \hbox{$z\in D(b,\epsilon) \setminus [b-\epsilon, b]$,} \\
    V_0^{(0)}(z^\theta), & \hbox{$z \in W_1 \setminus [0,r_n]$.}
  \end{cases} \label{def:R2}
\end{align}
Recall the function $J_c$ ($J_{c(t)}$ before the $t$-dependence is suppressed) defined in \eqref{eq:J_function} with $c(t)$ specified in Section \ref{subsec:deformation_algebraic}. Then, we define, as in \cite[Equation 3.146)]{Wang-Zhang21}
\begin{equation} \label{eq:scalar_R_defn}
  \tR(s) =
  \begin{cases}
    R_1(J_c(s)), & \text{$s \in \compC \setminus \overline{D}$ and $s \notin \Iinv_1(\Sigma^R )$,} \\
    R_2(J_c(s)), & \text{$s \in D \setminus[-1, 0]$ and $s \notin \Iinv_2(\Sigma^R)$.}
  \end{cases}
\end{equation}
We have that $\tR(s)$ satisfies a shifted RH problem, as described in \cite[RH problem 3.19]{Wang-Zhang21}. We omit the full statement of this RH problem.

By arguments analogous to those in \cite[Propositions 3.20 and 3.21, and Lemma 3.22]{Wang-Zhang21},  we derive the following estimate for $\tR(s)$.
\begin{lemma} \label{lem:tRest}
  As $n\to \infty$, we have
  \begin{equation} \label{eq:esttR}
    \tR(s)=1+\bigO(n^{-\frac{1}{2(2\theta + 1)}}),
  \end{equation}
  uniformly for $\lvert s+1 \rvert < \epsilon r_n^{\theta/(1+\theta)}$ or $\lvert s \rvert < \epsilon$, where $\epsilon$ is a small positive constant and $r_n$ is given in \eqref{def:rn}.
\end{lemma}

The proof of the lemma is detailed in \cite[Section 3.7]{Wang-Zhang21}, and the idea is as follows: First, we show that $\tR(s)$ satisfies a scalar shifted RH problem \cite[RH problem 3.19]{Wang-Zhang21}. Next, we show that $\tR(s)$ is the unique solution to this RH problem \cite[Proposition 3.20]{Wang-Zhang21}. Then by a small norm argument, we show that the solution to the RH problem that $\tR(s)$ satisfies is close to $1$ in the sense of \cite[Proposition 3.21]{Wang-Zhang21}. Finally, we derive the lemma, as the counterpart of \cite[Proposition 3.22]{Wang-Zhang21}. Since it is lengthy and is identical to the proof in \cite[Section 3.7]{Wang-Zhang21} for a slightly different setting, we omit the proof. In Section \ref{sec:RH_soft_edge}, we give details of a similar small norm argument in the proof of Lemma \ref{lem:tRest_soft} .
\subsection{Asymptotic analysis of the RH problem for $q_n$}

In this subsection, we use the same notational setting as in Section \ref{subsec:asy_p_n}, especially the shape of $\Sigma$ is the same.

\subsubsection{Transformations $\Y \to \T \to \St \to \Q$, and local parametrix around $b$} \label{subsec:trans_YTSQ_tilde}

Let $q_n$ be the monic biorthogonal polynomial defined in \eqref{eq:biorthogonality}. We define its Cauchy transform \cite[Equation (2.28)]{Wang-Zhang21}
\begin{equation}\label{eq:defn_of_Cq_n}
  \widetilde C q_n(z) = \frac{1}{2\pi i}\int_0^{+\infty} \frac{q_n(x^{\theta})}{x-z}x^{\alpha}e^{-nV(x)}dx, \quad z\in \compC \setminus [0,+\infty),
\end{equation}
and denote the vector form $\Y = (\Y_1, \Y_2) := (q_n, \C q_n)$ as in \cite[Equation (2.29)]{Wang-Zhang21}. Then, similar to $Y$ that satisfies RH problem \ref{RHP:OPs},  the vector-valued function $\Y$ also satisfies a vector RH problem. We omit the details since it is the same as that in \cite[Theorem 1.5]{Claeys-Romano14} and \cite[RH problem 2.3]{Wang-Zhang21}.

Analogous to \eqref{eq:defn_T}, we let \cite[Equation (4.1)]{Wang-Zhang21}
\begin{equation}\label{def:tildeT}
  \T(z) = (\T_1(z), \T_2(z)) = (\Y_1(z) e^{-n \gfntilde(z)}, \Y_2(z) e^{n (\gfn(z) - \ell)}),
\end{equation}
and analogous to \eqref{eq:S_from_T}, we let \cite[Equation (4.5)]{Wang-Zhang21}
\begin{equation}\label{def:tildeS}
  \St(z) =
  \begin{cases}
    \T(z), & \text{outside the lens}, \\
    \T(z)
    \begin{pmatrix}
      1 & 0
      \\
   z^{-\alpha} e^{-n \phi(z)} & 1
    \end{pmatrix},
    & \text{lower part of the lens}, \\
    \T(z)
    \begin{pmatrix}
      1 & 0 \\
      -z^{-\alpha} e^{-n \phi(z)} & 1
    \end{pmatrix},
    & \text{upper part of the lens}.
  \end{cases}
\end{equation}
We have that $\T(z)$ and $\St(z)$ satisfy the Riemann-Hilbert problems stated in \cite[RH problems 4.1 and 4.2]{Wang-Zhang21}.

Next, analogous to \eqref{eq:Pinfty_1} -- \eqref{eq:Pinfty_in_one}, we construct the global parametrix $\Pinfty(z) = (\Pinfty_1(z), \Pinfty_2(z))$ such that \cite[Equations (4.15) and (4.16)]{Wang-Zhang21}
\begin{align}\label{eq:tildeP1infty}
\widetilde P_1^{(\infty)}(z)&=\widetilde \tP(\Iinv_2(z)), & z\in {}& \mathbb{H}_\theta \setminus [0,b],
\\
\widetilde P_2^{(\infty)}(z)&=\widetilde \tP(\Iinv_1(z)), & z\in {}& \mathbb{C} \setminus [0,b], \label{eq:tildeP2infty}
\end{align}
where \cite[Equation (4.14)]{Wang-Zhang21} 
\begin{equation}\label{eq:tildeFs}
  \widetilde \tP(s)=
  \begin{cases}
    \frac{c^{\alpha}\sqrt{s_b}i}{\sqrt{(s + 1)(s - s_b)}} \left( \frac{s+1}{s} \right)^{\frac{\alpha}{\theta}}, & \text{$s \in \compC \setminus \overline{D}$,} \\
     \frac{\sqrt{s_b}i }{\sqrt{(s + 1)(s - s_b)}} (s+1)^{-\alpha}, & \text{$s\in D$.}
  \end{cases}
\end{equation}
Here $s_b = 1/\theta$, the branch cuts of $\sqrt{(s + 1)(s - s_b)}$, $\left( \frac{s + 1}{s} \right)^{\frac{\alpha}{\theta}}$ and $(s + 1)^{-\alpha}$ are taken along $\gamma_2$, $[-1, 0]$ and $(-\infty, -1]$ respectively, where the contour $\gamma_2$ is defined at the beginning of Section \ref{subsec:exact_equilibrium_measure}.

Then the function, analogous to \eqref{eq:defn_Q}, \cite[Equation (4.20)]{Wang-Zhang21}
\begin{equation} \label{eq:thirdtransform_tilde}
\Q(z)=(\Q_1(z),\widetilde{ Q}_2(z))=\left(\frac{\St_1(z)}{\Pinfty_1(z)},\frac{\St_2(z)}{\Pinfty_2(z)} \right)
\end{equation}
satisfies the Riemann-Hilbert problem stated in \cite[RH problem 3.7]{Wang-Zhang21}.

Analogous to \cite{Wang-Zhang21},  the local parametrix near the right ending point $b$ can be constructed in terms of the  Airy parametrix. With $f_b(z)$ defined in \eqref{def:fb_hard}, we define, analogous to \eqref{eq:defn_P^b} \cite[Equation (4.25)]{Wang-Zhang21}
\begin{equation}\label{def:tildePb}
 \Pb (z) :=\widetilde E^{(b)}(z)
  \Psi^{(\Ai)}(n^{\frac23}f_b(z))
  \begin{pmatrix}
    e^{-\frac{n}{2} \phi(z)} \widetilde g^{(b)}_1(z) & 0
    \\
    0 & e^{\frac{n}{2} \phi(z)} \widetilde g^{(b)}_2(z)
  \end{pmatrix},  \quad
  z \in D(b,\epsilon) \setminus \Sigma,
\end{equation}
where \cite[Equation (4.26)]{Wang-Zhang21}
\begin{equation}\label{def:tildegib}
  \widetilde g^{(b)}_1(z) = \frac{z^{-\frac{\alpha }{2}}}{\widetilde P_1^{(\infty)}(z)}, \qquad \widetilde g^{(b)}_2(z) = \frac{z^{\frac{\alpha }{2}}}{\widetilde P_2^{(\infty)}(z)},
\end{equation}
and \cite[Equation (4.25)]{Wang-Zhang21}
\begin{equation}\label{def:tildeEb}
\widetilde E^{(b)}(z)=\frac{1}{\sqrt{2}}
  \begin{pmatrix}
    \widetilde g^{(b)}_1(z) & 0  \\
    0 & \widetilde g^{(b)}_2(z)
  \end{pmatrix}^{-1}
  e^{\frac{\pi i}{4} \sigma_3}
  \begin{pmatrix}
    1 & -1 \\
    1 & 1
  \end{pmatrix}
  \begin{pmatrix}
    n^{\frac16} f_b(z)^{\frac14} & 0
    \\
    0 & n^{-\frac16} f_b(z)^{-\frac14}
  \end{pmatrix}.
\end{equation}
We have that $\Pb(z)$ satisfies an RH problem; see \cite[RH problem 4.8]{Wang-Zhang21}. Then analogous to \eqref{eq:defn_V^b}, we define the vector-valued function $\Vb(z)$ by \cite[Equation (4.29)]{Wang-Zhang21}
\begin{equation} \label{eq:defn_V^(b)_tilde}
    \Vb(z) = \widetilde Q(z) \Pb (z)^{-1}, \qquad z \in D(b,\epsilon) \setminus \Sigma.
\end{equation}
We have that $\Vb$ satisfies the same RH problem as \cite[RH problem 4.9]{Wang-Zhang21}.

\subsubsection{Local parametrix around $0$} \label{subsubsec:Phi_q}

\paragraph{A local version of the RH problem for $\Q$}

Let $r = r_n$ be the same as in Section \ref{subsec:local_para}. Analogous to \eqref{eq:U_0_def} and \eqref{eq:U_0_def_2}, we define $\theta + 1$ functions $\U_0(z), \U_1(z), \dotsc, \U_{\theta}(z)$ \cite[Equations (4.34) and (4.35)]{Wang-Zhang21}, see also \cite[Equations (4.31)--(4.33)]{Wang-Zhang21}
\begin{align}
  \U_0(z) = {}& \Q_1(z^{\frac{1}{\theta}}), & z \in {}& D(0, r^{\theta}) \setminus \{ (-r^{\theta}, r^{\theta}) \cup (-ir^{\theta}, ir^{\theta}) \}, \label{eq:Utilde_0_def} \\
  \U_k(z) = {}& \Q_2(z^{\frac{1}{\theta}} e^{\frac{2(k - 1)}{\theta} \pi i}), & z \in {}& D(0, r^{\theta}) \setminus \{ (-r^{\theta}, r^{\theta}) \cup (-ir^{\theta}, ir^{\theta}) \}, \quad k = 1, 2, \dotsc, \theta. \label{eq:Utilde_0_def_2}
\end{align}
Then, $\U(z) = (\U_0(z), \dotsc, \U_{\theta}(z))$ satisfies \cite[RH problem 4.10]{Wang-Zhang21}, which is stated  below for completeness.

\begin{RHP}\label{rhp:tildeU} \hfill
  \begin{enumerate} 
  \item
    $ \U=(\U_0,\U_1,\ldots,\U_\theta)$ is defined and analytic in $D(0,r^{\theta}) \setminus \{(-r^{\theta},r^{\theta}) \cup (-ir^{\theta},ir^{\theta})\}$.
  \item
    For $z\in (-r^{\theta},r^{\theta})\cup(-ir^{\theta},ir^{\theta})\setminus\{0\}$, we have
    \begin{equation}
      \U_+(z) = \U_-(z) J_{\U}(z),
    \end{equation}
    where
    \begin{equation} \label{eq:defn_J_tildeU}
      J_{\U}(z) =
      \begin{cases}
        \begin{pmatrix}
          1 & 0 \\
          z^{-\frac{\alpha}{\theta}}e^{-n \phi(z^{1/\theta})}\frac{\Pinfty_2(z^{1/\theta})}{\Pinfty_1(z^{1/\theta})} & 1
        \end{pmatrix}
        \oplus I_{\theta - 1}, & z \in (0,ir^{\theta})\cup (0,-ir^{\theta}), \\
        \begin{pmatrix}
          0 & 1 \\
          1 & 0
        \end{pmatrix}
        \oplus I_{\theta - 1}, & z \in (0, r^{\theta}), \\
        \Mcyclic, & z \in (-r^{\theta}, 0),
      \end{cases}
    \end{equation}
    with $\Mcyclic$ being defined in \eqref{eq:defn_Mcyclic}, and the orientations of the rays are shown in Figure \ref{fig:jumps-U}.
  \item
    As $z\to 0$ from $D(0,r^{\theta}) \setminus \{(-r^{\theta},r^{\theta}) \cup (-ir^{\theta},ir^{\theta})\}$, we have
    \begin{enumerate} 
    \item for $k=1,\ldots,\theta$, or $k=0$ and $\arg z\in (0, \pi/2) \cup (-\pi/2, 0)$,
      \begin{equation}
      \U_k(z) =
      \begin{cases}
        \bigO (z^{\frac{\theta-2\alpha}{2\theta(1+\theta)}}), & \alpha  > 0, \\
        \bigO ( z^{\frac{1}{2(1+\theta)}} \log z ), &  \alpha = 0, \\
        \bigO ( z^{\frac{\alpha+ 1/2}{1+\theta}} ), & -1< \alpha  < 0,
      \end{cases}
    \end{equation}

    \item
      for $\arg z \in (\pi/2,\pi)\cup (-\pi,-\pi/2)$,
      \begin{equation}
        \U_0(z)=\bigO (z^{\frac{\alpha+ 1/2}{1+\theta}}).
      \end{equation}
    \end{enumerate}
\end{enumerate}
\end{RHP}

\paragraph{Parametrix $\Phitilde$}

Below we state the following model RH problem: (with $\betatilde = -\frac{2\alpha + 1}{1 + \theta}$)
\begin{RHP} \label{RHP:general_model_tilde}
  $\Phitilde(\xi) = \Phitilde^{(\tau)}(\xi)$ is a $(\theta + 1) \times (\theta + 1)$ matrix-valued function on $\compC$ except for $\realR$ and $i\realR$. 
\begin{enumerate} 
  \item It satisfies the jump condition with the  orientation of rays shown in Figure \ref{fig:jumps-MeiG}
    \begin{equation} \label{eq:defn_J_hard_to_soft_q}
      \Phitilde_+(\xi) = \Phitilde_-(\xi) J^{(\theta)}_{\Phitilde}(\xi), \quad \text{where} \quad J^{(\theta)}_{\Phitilde}(\xi) =
      \begin{cases}
        \begin{pmatrix}
          1 & 0 \\
          e^{\betatilde \pi i} & 1
        \end{pmatrix}
        \oplus I_{\theta - 1},
        & \arg \xi = \frac{\pi}{2}, \\
        \begin{pmatrix}
          1 & 0 \\
          -e^{-\betatilde \pi i} & 1
        \end{pmatrix}
        \oplus I_{\theta - 1},
        & \arg \xi = -\frac{\pi}{2}, \\
        \begin{pmatrix}
          0 & 1 \\
          1 & 0
        \end{pmatrix}
        \oplus I_{\theta - 1},
        & \xi \in \realR_+, \\
        \Mcyclic,
            & \xi \in \realR_-.
      \end{cases}
    \end{equation}
  \item
    $\Phitilde(\xi)$ has the following asymptotic condition at infinity as $\xi \to \infty$ for $\xi \in \compC_{\pm}$
    \begin{equation} \label{eq:asyNew_q}
      \Phitilde(\xi) = \left( I + \bigO(\xi^{-1}) \right) \Upsilon(\xi) \Omega_{\pm} e^{\Theta(\xi)},
    \end{equation}
    where $\Theta$, $\Upsilon$ and $\Omega_{\pm}$  are defined in \eqref{eq:asyold} and \eqref{def:Lpm}.
  \item 
    $\Phitilde(\xi)$ has the following local growth condition as $\xi \to 0$, for $\xi \in \compC_{\pm}$,
    \begin{equation}
      \Phitilde(\xi) = \N(\xi) \diag \left( \xi^{-\frac{\betatilde}{2}}, \xi^{\frac{\betatilde + 2\theta - 1}{2\theta}}, \dotsc, \xi^{\frac{\betatilde + 5}{2\theta}} \xi^{\frac{\betatilde + 3}{2\theta}}, \xi^{\frac{\betatilde + 1}{2\theta}}, \right) (E^{-1})^T \Xicheck^{-1}_{\pm},
    \end{equation}
    where $\N(\xi)$ is analytic at $0$, $E$ is defined in \eqref{eq:defn_E_region_II}, and 
    \begin{equation} \label{eq:defn_Xicheck_real}
      \begin{aligned}
        \Xicheck_+ = {}& \diag \left( e^{-\frac{\theta j}{\theta + 1} 2\pi i} \right)^{\theta}_{j = 0}, & \text{on $\compC_+$}, \\
        \Xicheck_- = {}& \diag \left( e^{-\frac{\theta}{\theta + 1} 2\pi i}, 1 \right) \oplus \diag \left( e^{-\frac{\theta j}{\theta + 1} 2\pi i} \right)^{\theta}_{j = 2}, & \text{on $\compC_-$}.
      \end{aligned}
    \end{equation}
  \end{enumerate}
\end{RHP}

Consider
\begin{equation} \label{eq:defn_Xicheck}
  \Phicheck(\xi) = \frac{1}{\theta + 1} e^{\frac{\theta}{\theta + 1} \pi i} \xi^{-\frac{\theta}{\theta + 1}} J_{\theta + 1} \Phitilde(\xi) \Xicheck_{\pm},
\end{equation}
where $J_{\theta + 1} = (\delta_{i, \theta - j})^{\theta}_{i, j = 0}$ is defined at the end of Section \ref{sec:introduction}. Then $\Phicheck(\xi)$ satisfies the following RH problem:
\begin{RHP}\label{RHP:general_model_tilde_check}
  $\Phicheck(\xi)$ is a $(\theta + 1) \times (\theta + 1)$ matrix-valued function on $\compC$ except for $\realR$ and $i\realR$.
  \begin{enumerate} 
  \item  $ \Phicheck(\xi)$ satisfies the jump condition with the  orientation of rays shown in Figure \ref{fig:jumps-MeiG}

    \begin{equation}
      \Phicheck_+(\xi) = \Phicheck_-(\xi) J^{(\theta)}_{\Phicheck}(\xi), \quad \text{where} \quad J^{(\theta)}_{\Phicheck}(\xi) =
      \begin{cases}
        \begin{pmatrix}
          1 & 0 \\
          -e^{\beta \pi i} & 1
        \end{pmatrix}
        \oplus I_{\theta - 1},
        & \arg \xi = \frac{\pi}{2}, \\
        \begin{pmatrix}
          1 & 0 \\
          e^{-\beta \pi i} & 1
        \end{pmatrix}
        \oplus I_{\theta - 1},
        & \arg \xi = -\frac{\pi}{2}, \\
        \begin{pmatrix}
          0 & 1 \\
          1 & 0
        \end{pmatrix}
        \oplus I_{\theta - 1},
        & \xi \in \realR_+, \\
        \Mcyclic,
            & \xi \in \realR_-.
      \end{cases}
    \end{equation}
  \item
    $\Phicheck(\xi)$ has the following asymptotic condition at infinity as $\xi \to \infty$ for $\xi \in \compC_{\pm}$
    \begin{equation}
      \Phicheck(\xi) = \left( I + \bigO(\xi^{-1}) \right) \Upsilon(\xi)^{-1} \Omega^{-1}_{\pm} e^{\Theta(\xi)},
    \end{equation}
    where $\Theta$, $\Upsilon$ and $\Omega_{\pm}$ are defined in \eqref{eq:asyold} and \eqref{def:Lpm}.
  \item 
    $\Phicheck(\xi)$ has the following local growth condition as $\xi \to 0$
    \begin{equation}
      \Phicheck(\xi) = \frac{1}{\theta + 1} e^{\frac{\theta}{\theta + 1} \pi i} J_{\theta + 1} \N(\xi) \xi^{\frac{\beta}{2\theta}} \diag \left( \xi^{\frac{\alpha + 1 - \theta}{\theta}}, 1, \xi^{-\frac{1}{\theta}}, \xi^{-\frac{2}{\theta}}, \dotsc, \xi^{-\frac{\theta - 1}{\theta}} \right) (E^{-1})^T.
    \end{equation}
  \end{enumerate}
\end{RHP}

We can check that $(\Phicheck(\xi)^{-1})^T$, the transpose of the inverse of $\Phicheck(\xi)$, satisfies RH problem \ref{RHP:general_model}. Hence, by the existence and uniqueness of RH problem \ref{RHP:general_model} in Theorem \ref{thm:unique_solvability}, we conclude that, with $\Phi(\xi) = \Phi^{(\tau)}(\xi)$ defined by RH problem \ref{RHP:general_model} and $N(\xi) = N^{(\tau)}(\xi)$ given in \eqref{eq:defn_N(xi)}, 
$(\theta + 1)^{-1} e^{\frac{\theta}{\theta + 1} \pi i} J_{\theta + 1} \N(\xi) = (N(\xi)^{-1})^T$, and
\begin{align}
  \Phicheck(\xi) = {}& (\Phi(\xi)^{-1})^T , & \Phitilde(\xi) = (\theta + 1) e^{-\frac{\theta}{\theta + 1} \pi i} \xi^{\frac{\theta}{\theta + 1}} J_{\theta + 1} (\Phi(\xi)^{-1})^T  \Xicheck^{-1}_{\pm}. 
\end{align} 

\paragraph{Construction of the local parametrix around $0$}

Analogous to the definition of $N(z)$ given in \eqref{def:N}, we set \cite[Equations (4.75) and (4.76)]{Wang-Zhang21}
\begin{equation} \label{def:tildeN}
  \N(z) = \diag (\n_0(z), \n_1(z), \dotsc, \n_{\theta}(z)), \quad z\in\compC \setminus (-\infty,b^{\theta}],
\end{equation}
where
\begin{equation}\label{def:tildeni}
  \n_0(z) = \widetilde{P}^{(\infty)}_1(z^{\frac1 \theta}), \quad \n_j(z) = \widetilde{P}^{(\infty)}_2(e^{\frac{2(j - 1)}{\theta} \pi i} z^{\frac1 \theta}), \quad j = 1, \dotsc, \theta,
\end{equation}
with $\widetilde{P}^{(\infty)} =(\widetilde{P}^{(\infty)}_1, \widetilde{P}^{(\infty)}_2)$ defined in \eqref{eq:tildeP1infty}--\eqref{eq:tildeP2infty}.
We also define
\begin{equation} \label{eq:Ntildepre}
  \N^{(\pre)}(z) = \diag (\n^{(\pre)}_0(z), \n^{(\pre)}_1(z), \dotsc, \n^{(\pre)}_{\theta}(z)), \quad z\in\compC \setminus \realR,
\end{equation}
where
\begin{align}
  \n^{(\pre)}_0(z) = {}&
                         \begin{cases}
                           \frac{c^{\frac{(\alpha+1/2)\theta}{1+\theta}}}{\sqrt{1+\theta}}e^{\frac{\alpha+1/2}{1+\theta}\pi i}z^{-\frac{\alpha+1/2}{1+\theta}}, & z \in \compC_+, \\
                           \frac{c^{\frac{(\alpha+1/2)\theta}{1+\theta}}}{\sqrt{1+\theta}}e^{-\frac{\alpha+1/2}{1+\theta}\pi i}z^{-\frac{\alpha+1/2}{1+\theta}}, & z \in \compC_-, \\
                         \end{cases} && \\
  \n^{(\pre)}_1(z) = {}&
                         \begin{cases}
                           \frac{c^{\frac{(\alpha+1/2)\theta}{1+\theta}}}{\sqrt{1+\theta}}e^{-\frac{\alpha+1/2}{1+\theta}\pi i} z^{\frac{\alpha - \theta/2}{\theta(1+\theta)}}, & z \in \compC_+, \\
                           -\frac{c^{\frac{(\alpha+1/2)\theta}{1+\theta}}}{\sqrt{1+\theta}}e^{\frac{\alpha+1/2}{1+\theta}\pi i} z^{\frac{\alpha - \theta/2}{\theta(1+\theta)}}, & z \in \compC_-,
                         \end{cases} && \\
  \n^{(\pre)}_j(z) = {}& \frac{c^{\frac{(\alpha+1/2)\theta}{1+\theta}}}{\sqrt{1+\theta}}e^{-\frac{\alpha+1/2}{1+\theta}\pi i} e^{\frac{(2\alpha - 1)(j - 1)}{\theta(1 + \theta)} \pi i} z^{\frac{\alpha - \theta/2}{\theta(1+\theta)}}, && j=2,\ldots,\theta.
\end{align}

Analogous to Proposition \ref{prop:m_and_n} and \cite[Proposition 4.13]{Wang-Zhang21}, we have the following asymptotics. 
\begin{prop}\label{prop:tildeNzero}
  Let $\n_i(z)$, $i=0,1,\ldots,\theta$, be the functions defined in \eqref{def:tildeni}. As $z \to 0$, we have, with $\rho$ defined in \eqref{eq:defn_rho}, $c$ given in \eqref{eq:defn_c}, and $j = 2, \dotsc, \theta$
  \begin{equation}
    \n_k(z) = \n^{(\pre)}_k(z)  (1 + \bigO(t - 1) + \bigO(z^{\frac{1}{1+\theta}})), \quad k = 0, 1, \dotsc, \theta,
  \end{equation}
\end{prop}

With $r=r_n$ given in \eqref{def:rn}, we then define, analogous to \cite[Equation (4.80)]{Wang-Zhang21}, a $(\theta + 1) \times (\theta + 1)$ matrix-valued function for $z \in D(0,r_n^{\theta})\setminus (\mathbb{R} \cup i \mathbb{R})$ as follows:
\begin{equation}\label{def:tildesfP0}
  \widetilde{\mathsf{P}}^{(0)}(z) = \diag((\rho n)^{-\frac{k}{2}})^{\theta}_{k = 0} \Phitilde((\rho n)^{\frac{\theta + 1}{2}} z) \N^{(\pre)}(z) \N(z)^{-1} e^{-nM(z)},
\end{equation}
where $\rho$ is defined in \eqref{eq:defn_rho}, $M(z)$, $\N(z)$ and $\N^{(\pre)}$ are defined in \eqref{def:M}, \eqref{def:tildeN} and \eqref{eq:Ntildepre}, respectively, and $\Phitilde$ is the solution of the RH problem \ref{RHP:general_model_tilde}. We then have the following proposition regarding the RH problem for $\widetilde{\mathsf{P}}^{(0)}$ \cite[Proposition 4.14]{Wang-Zhang21}.

\begin{prop}\label{RHP:tildesfP0}
The function $\widetilde{\mathsf{P}}^{(0)}(z)$ defined in \eqref{def:tildesfP0} has the following properties.
  \begin{enumerate} 
  \item
    $ \widetilde{\mathsf{P}}^{(0)}(z)$ is analytic in $D(0,r_n^{\theta}) \setminus (\mathbb{R} \cup i \mathbb{R})$.
  \item
    For $z\in D(0,r_n^{\theta}) \cap (\mathbb{R} \cup i \mathbb{R})$, we have
    \begin{equation}
      \widetilde{\mathsf{P}}^{(0)}_{+}(z) = \widetilde{\mathsf{P}}^{(0)}_{-}(z) J_{\U}(z),
    \end{equation}
    where $J_{\U}(z)$ is defined in \eqref{eq:defn_J_tildeU}.
  \item
    For $z\in \partial D(0,r_n^{\theta})$, we have, as $n\to \infty$,
    \begin{equation}\label{eq:tildesfP0asy}
      \widetilde{\mathsf{P}}^{(0)}(z) = \Upsilon(z)\Omega_{\pm} (I+\bigO(n^{-\frac{1}{2(2\theta + 1)}})), \qquad z \in \compC_{\pm},
    \end{equation}
    where $\Upsilon(z)$ and $\Omega_{\pm}$ are defined in \eqref{def:Lpm}.
  \end{enumerate}
\end{prop}

With $\widetilde{\mathsf{P}}^{(0)}$ given in \eqref{def:tildesfP0}, we next define, analogous to \eqref{def:P0} and \cite[Equation (4.86)]{Wang-Zhang21}
\begin{equation}\label{def:tildeP0}
  \widetilde{P}^{(0)}(z) = \Omega_{\pm}^{-1}\Upsilon(z)^{-1}\widetilde{\mathsf{P}}^{(0)}(z), \quad z \in \compC_{\pm}.
\end{equation}
In view of \eqref{def:tildeP0}, Proposition \ref{RHP:tildesfP0} and \eqref{eq:UpsilonOmegajump}, the following RH problem for $\widetilde{P}^{(0)}$ is then immediate \cite[Proposition 4.15]{Wang-Zhang21}.
\begin{RHP}\label{rhp:tildep0}
\hfill
  \begin{enumerate} 
  \item
    $\widetilde{P}^{(0)}(z)$ is analytic in $D(0,r_n^{\theta}) \setminus (\mathbb{R} \cup i \mathbb{R})$.
  \item
    For $z\in D(0,r_n^{\theta}) \cap (\mathbb{R} \cup i \mathbb{R})$, we have
    \begin{equation}
      \widetilde{P}^{(0)}_+(z) =
      \begin{cases}
        \widetilde{P}^{(0)}_-(z) J_{\U}(z), & z \in (0, ir_n^{\theta}) \cup (0, -ir_n^{\theta}), \\
        \left(\begin{pmatrix}
            0 & 1
            \\
            1 & 0
          \end{pmatrix} \oplus I_{\theta - 1} \right)
        \widetilde{P}^{(0)}_-(z) J_{\U}(z), & z \in (0, r_n^{\theta}), \\
        \Mcyclic^{-1} \widetilde{P}^{(0)}_-(z) J_{\U}(z), & z \in (-r_n^{\theta}, 0),
      \end{cases}
    \end{equation}
    where $J_{\U}(z)$ is defined in \eqref{eq:defn_J_tildeU}.
  \item
    For $z\in \partial D(0,r_n^{\theta})$, we have, as $n\to \infty$,
    \begin{equation}\label{eq:tildeP0asy}
      \widetilde{P}^{(0)}(z) = I+\bigO(n^{-\frac{1}{2(2\theta + 1)}}).
    \end{equation}
  \end{enumerate}
\end{RHP}

Finally, as in the definition of $V^{(0)}$ in \eqref{def:V0}, we set, analogous to \cite[Equation (4.89)]{Wang-Zhang21}
\begin{equation} \label{def:tildeV0}
  \V^{(0)}(z) = \U(z) \widetilde{P}^{(0)}(z)^{-1}, \quad z \in D(0,r_n^{\theta}) \setminus (\mathbb{R} \cup i \mathbb{R}),
\end{equation}
where $\U$ is defined in \eqref{eq:Utilde_0_def} and \eqref{eq:Utilde_0_def_2} and satisfies RH problem \ref{rhp:tildeU}. Then, analogous to Proposition \ref{rhp:V0} and \cite[RH problem 4.16]{Wang-Zhang21}, we have the following RH problem for $\V^{(0)}$.
\begin{RHP}\label{rhp:tildeV0}
\hfill
\begin{enumerate} 
  \item \label{enu:rhp:tildeV0:1}
    $\V^{(0)}=(\V^{(0)}_0, \V^{(0)}_1, \dotsc, \V^{(0)}_{\theta})$ is analytic in $D(0,r_n^{\theta}) \setminus \mathbb{R} $, or equivalently, its definition can be analytically extended onto $(0, ir^{\theta}_n) \cup (-ir^{\theta}_n, 0)$.
  \item \label{enu:rhp:tildeV0:2}
    For $z\in (-r_n^\theta,r_n^{\theta})\setminus \{0\}$, we have
    \begin{equation} \label{eq:enu:rhp:tildeV0:2}
      \V^{(0)}_+(z) = \V^{(0)}_-(z)
      \begin{cases}
        \begin{pmatrix}
          0 & 1 \\
          1 & 0
        \end{pmatrix}
        \oplus I_{\theta - 1}, & z \in (0, r_n^{\theta}), \\
      \Mcyclic, & z \in (-r_n^{\theta}, 0). 
      \end{cases}
    \end{equation}
  \item \label{enu:rhp:tildeV0:3}
    For $z\in \partial D(0,r_n^{\theta})$, we have, as $n \to \infty$,
    \begin{equation}\label{eq:tildeV0asy}
      \V^{(0)}(z) = \U(z)(I+\bigO(n^{-\frac{1}{2(2\theta+1)}})).
    \end{equation}
  \item
    As $z \to 0$, we have
    \begin{equation}\label{eq:tildeV0kzero}
      \V^{(0)}_k(z)=\bigO(1), \qquad k=0,1,\ldots,\theta.
    \end{equation}
  \end{enumerate}
\end{RHP}

\subsubsection{Final transformation}

The argument in this subsubsection is parallel to that in \cite[Section 4.7]{Wang-Zhang21} and essentially the same as that in Section \ref{subsubsec:final_trans_p}.

Recall $\Sigma^R$ defined in \eqref{def:sigmaR} and $W_k$ defined in \eqref{def:Wk}. We define a $1\times 2$ vector-valued function $\R(z) = (\R_1(z), \R_2(z))$ such that $\R_1(z)$ is analytic in $\halfH \setminus \Sigma^R$ and $\R_2(z)$ is analytic in $\compC \setminus \Sigma^R$ \cite[Equations (4.93) and (4.94)]{Wang-Zhang21}:
\begin{align}
   \R_1(z) = {}&
  \begin{cases}
    \Q_1(z), & \hbox{$z\in \halfH \setminus \{D(b,\epsilon) \cup D(0,r_n) \cup \Sigma^R \}$,} \\
    \V_1^{(b)}(z), & \hbox{$z\in D(b,\epsilon) \setminus [b-\epsilon, b]$,} \\
    \V_0^{(0)}(z^\theta), & \hbox{$z \in W_1 \setminus [0,r_n]$,} \\
  \end{cases} \label{def:tildeR1}
  \\
  \R_2(z) = {}&
  \begin{cases}
    \Q_2(z), & \hbox{$z\in \mathbb{C} \setminus \{D(b,\epsilon) \cup D(0,r_n) \cup \Sigma^R \}$,} \\
    \V_2^{(b)}(z), & \hbox{$z\in D(b,\epsilon) \setminus [b-\epsilon, b]$,} \\
    \V_1^{(0)}(z^\theta), & \hbox{$z\in W_1 \setminus [0,r_n]$,} \\
    \V_k^{(0)}(z^\theta), & \hbox{$z\in W_k$, $k=2,\ldots,\theta$.}
  \end{cases} \label{def:tildeR2}
\end{align}
Then analogous to \eqref{eq:scalar_R_defn}, we define, as in \cite[Equation (4.100)]{Wang-Zhang21}, with $\omega(s) = s^{-1}$,
\begin{equation} \label{eq:scalar_tildeR_defn}
  \widetilde{\tR}(s) =
  \begin{cases}
  \R_1(\J(s)), & \text{$s \in \omega(D \setminus[-1, 0])$ and $s \notin \omega(\Iinv_2(\Sigma^R))$}, \\
  \R_2(\J(s)), & \text{$s \in \omega(\compC \setminus \overline{D})$ and $s \notin \omega(\Iinv_1(\Sigma^R ))$,}
  \end{cases},
  \quad \J(s) = J_c(1/s).
\end{equation}
We have that $\widetilde{\tR}(s)$ satisfies a shifted RH problem, as described in \cite[RH problem 4.18]{Wang-Zhang21}. We omit the full statement of this RH problem.

By the arguments in \cite[Lemma 4.19]{Wang-Zhang21}, we derive the following property for $\widetilde{\tR}(s)$. 
\begin{lemma}\label{lem:tildetRest}
As $n\to \infty$, we have
\begin{equation}\label{eq:esttildetR}
\widetilde{\tR}(s)=1+\bigO(n^{-\frac{1}{2(2\theta + 1)}}),
\end{equation}
uniformly for $\lvert s+1 \rvert < \epsilon r_n^{\theta/(1+\theta)}$, where $\epsilon$ is a small positive constant and $r_n$ is defined in \eqref{def:rn}.
\end{lemma}

\section{Asymptotic analysis of the RH problems for $p_n$ and $q_n$  in the soft edge regime} \label{sec:RH_soft_edge}

In this section, we assume that $C^{-1} < t < 1 - C n^{-1/2}$ for some large positive number $C$. We also assume that $V_t$ satisfies Regularity Conditions \ref{reg:primary} and \ref{reg:soft-edge}. Without this assumption, the results in this section hold for $t \in (1 - c, 1 - Cn^{-1/2})$, where $c$ is a small positive constant.

 
In Sections \ref{subsec:transformations_soft_edge}--\ref{subsec:small_norm_soft_edge}, we denote 
  $\gfnhat$, $\gfnhattilde$, $\phihat$, $\ellhat$, $\Iinvhat_1$, $\Iinvhat_2$, $J_{c_1, c_0}$ as short-hand notations for $\gfnhat_t$, $\gfnhattilde_t$, $\phihat_t$, $\ellhat_t$, $\Iinvhat_{1, t}$, $\Iinvhat_{2, t}$, $J_{c_1(t), c_0(t)}$.

\subsection{Transformations $Y \to \That \to \Shat \to \Qhat$ and $\Y \to \Thattilde \to \Shattilde \to \Qhattilde$} \label{subsec:transformations_soft_edge}

Recall the functions $\gfnhat$, $\gfnhattilde$ and $\phihat$ and constant $\ellhat$ ($\gfnhat_t$, $\gfnhattilde_t$, $\phihat_t$ and $\ellhat_t$ before the $t$-dependence is suppressed) defined in \eqref{eq:gfnhat_gfnhattilde_t} and \eqref{eq:defn_ellha_t} for $1 - t$ small and then extended by Proposition \ref{prop:extension_t} to all $t \in (0, 1)$. Let $C p_n(z)$ be defined as in \eqref{eq:defn_Cp_n} and $Y = (Y_1, Y_2) := (p_n, C p_n)$ as in Section \ref{subsec:trans_YTSQ}. Let $\C q_n(z)$ be defined as in \eqref{eq:defn_of_Cq_n} and $\Y = (\Y_1, \Y_2) := (q_n, \C q_n)$ as in Section \ref{subsec:trans_YTSQ_tilde}. Then analogous to \eqref{eq:defn_T} and \eqref{def:tildeT}, we define
\begin{align} \label{eq:defn_That}
  \That(z) = {}& (\That_1(z), \That_2(z)) := (Y_1(z) e^{-n\gfnhat(z)}, Y_2(z) e^{n(\gfnhattilde(z) - \ellhat)}), \\
  \Thattilde(z) = {}& (\Thattilde_1(z), \Thattilde_2(z)) := (\Y_1(z) e^{-n\gfnhattilde(z)}, \Y_2(z) e^{n(\gfnhat(z) - \ellhat)}).
\end{align}
Furthermore, let $\Sigmahat_1 \in (\halfH \cap \compC_+)$ be a contour connecting $\aend(t)$ and $\bend(t)$ whose exact shape will be specified in Section \ref{subsubsec:three_parametrices}, $\Sigmahat_2 = \overline{\Sigmahat_1} \subseteq (\halfH \cap \compC_-)$, and let $\Sigmahat = [0, +\infty) \cup \Sigmahat_1 \cup \Sigmahat_2$. Moreover, we let the ``lens'' be the region enclosed by $\Sigmahat_1$ and $\Sigmahat_2$, and the upper/lower part of the lens be the intersection of the lens with $\compC_{\pm}$. Then let, analogous to \eqref{eq:S_from_T} and \eqref{def:tildeS},
\begin{align}
  \Shat(z) = {}&
  \begin{cases}
    \That(z), & \text{outside the lens}, \\
    \That(z)
    \begin{pmatrix}
      1 & 0 \\
      \frac{\theta}{z^{\alpha + 1 - \theta}} e^{-n\phihat(z)} & 1
    \end{pmatrix},
    & \text{lower part of the lens}, \\
    \That(z)
    \begin{pmatrix}
      1 & 0 \\
      -\frac{\theta}{z^{\alpha + 1 - \theta}} e^{-n\phihat(z)} & 1
    \end{pmatrix},
    & \text{upper part of the lens},
  \end{cases} \\
  \Shattilde(z) = {}&
  \begin{cases}
    \Thattilde(z), & \text{outside the lens}, \\
    \Thattilde(z)
    \begin{pmatrix}
      1 & 0
      \\
   z^{-\alpha} e^{-n \phihat(z)} & 1
    \end{pmatrix},
    & \text{lower part of the lens}, \\
    \Thattilde(z)
    \begin{pmatrix}
      1 & 0 \\
      -z^{-\alpha} e^{-n \phihat(z)} & 1
    \end{pmatrix},
    & \text{upper part of the lens}.
  \end{cases}
\end{align}

Next, we need to construct the global parametrices $\Phat^{(\infty)}(z) = (\Phat^{(\infty)}_1(z), \Phat^{(\infty)}_2(z))$ and $\Phattilde^{(\infty)}(z) = (\Phattilde^{(\infty)}_1(z), \Phattilde^{(\infty)}_2(z))$. Recall $\Iinvhat_1$ and $\Iinvhat_2$ ($\Iinvhat_{1, t}$ and $\Iinvhat_{2, t}$ before the $t$-dependence is suppressed) defined in Section \ref{subsec:deformation_along_C2_small_t} for $1 - t$ small and then extended by Proposition \ref{prop:extension_t} to all $t \in (0, 1)$. Then let, analogous to \eqref{eq:Pinfty_1}, \eqref{eq:Pinfty_2}, \eqref{eq:tildeP1infty} and \eqref{eq:tildeP1infty},
\begin{align}
  \Phat^{(\infty)}_1(z) = {}& \Phatsingle(\Iinvhat_1(z)), & \Phattilde^{(\infty)}_2(z) = {}& \Phatsingletilde(\Iinvhat_1(z)), &z \in {}& \compC \setminus [\aend(t), \bend(t)], \label{eq:Phat^infty_1} \\
  \Phat^{(\infty)}_2(z) = {}& \Phatsingle(\Iinvhat_2(z)), & \Phattilde^{(\infty)}_1(z) = {}& \Phatsingletilde(\Iinvhat_2(z)), & z \in {}& \halfH \setminus [\aend(t), \bend(t)], 
\end{align}
where
\begin{align}
  \Phatsingle(s) = {}&
  \begin{cases}
    \frac{s}{\sqrt{(s - s_{a(t)})(s - s_{b(t)})}} \left( \frac{s + 1}{s} \right)^{\frac{\theta - \alpha - 1}{\theta}}, & s \in \compC \setminus \overline{D}, \\
    \frac{s(c_1(t) s + c_0(t))^{\alpha + 1 - \theta}}{\theta \sqrt{(s - s_{a(t)})(s - s_{b(t)})}}, & s \in D, 
  \end{cases} \\
  \Phatsingletilde(s) = {}&
  \begin{cases}
    \frac{c^{\alpha}_0 \lvert s_{a(t)} s_{b(t)} \rvert^{1/2} i}{\sqrt{(s - s_{a(t)})(s - s_{b(t)})}} \left( \frac{s + 1}{s} \right)^{\frac{\alpha}{\theta}}, & s \in \compC \setminus \overline{D}, \\
    \frac{c^{\alpha}_0 \lvert s_{a(t)} s_{b(t)} \rvert^{1/2} i}{\sqrt{(s - s_{a(t)})(s - s_{b(t)})}} (c_1(t) s + c_0(t))^{-\alpha}, & s \in D.
  \end{cases}
\end{align}
Here $s_{a(t)}$ and $s_{b(t)}$ are defined in \eqref{eq:sat_and_at} and \eqref{eq:sbt_and_bt}, $\left( \frac{s + 1}{s} \right)^{\frac{\theta - \alpha - 1}{\theta}}$ and $\left( \frac{s + 1}{s} \right)^{\frac{\alpha}{\theta}}$ take the principal branch on $\compC \setminus [-1, 0]$, $(c_1(t) s + c_0(t))^{\alpha + 1 - \theta}$ and $(c_1(t) s + c_0(t))^{-\alpha}$ take the principal branch on $\compC \setminus (-\infty, -c_0(t)/c_1(t)]$, and the branch cut of $\sqrt{(s - s_{a(t)})(s - s_{b(t)})}$ is $\gamma_1$ for $\Phatsingle$ and $\gamma_2$ for $\Phatsingletilde$.

Then the functions
\begin{align}
  \Qhat(z) = {}& (\Qhat_1(z), \Qhat_2(z)) := \left( \frac{\Shat_1(z)}{\Phat^{(\infty)}_1(z)}, \frac{\Shat_2(z)}{\Phat^{(\infty)}_2(z)} \right), \\
  \Qhattilde(z) = {}& (\Qhattilde_1(z), \Qhattilde_2(z)) := \left( \frac{\Shattilde_1(z)}{\Phattilde^{(\infty)}_1(z)}, \frac{\Shattilde_2(z)}{\Phattilde^{(\infty)}_2(z)} \right),
\end{align}
satisfy the following Riemann-Hilbert problems:
\begin{RHP} \label{RHP:Svar} \hfill
  \begin{enumerate} 
  \item \label{enu:RHP:Svar:1}
    $\Qhat = (\Qhat_1, \Qhat_2)$ is analytic in $(\compC \setminus \Sigmahat, \mathbb{H}_\theta \setminus \Sigmahat)$.
  \item \label{enu:RHP:Svar:2}
    For $z \in \Sigmahat$, we have
    \begin{equation} \label{def:Jcals}
      \Qhat_+(z) = \Qhat_-(z) J_{\Qhat}(z),
      \end{equation}
      where
      \begin{equation} \label{def:JQ}
      J_{\Qhat}(z) =
      \begin{cases}
        \begin{pmatrix}
          1 & 0 \\
          \frac{\theta \Phat^{(\infty)}_2(z)}{z^{\alpha+1-\theta} \Phat^{(\infty)}_1(z)}e^{-n \phihat(z)} & 1
        \end{pmatrix},
        & z \in \Sigmahat_1 \cup \Sigmahat_2, \\
        \begin{pmatrix}
          0 & 1 \\
          1 & 0
        \end{pmatrix},
        & z \in (\aend(t), \bend(t)), \\
        \begin{pmatrix}
          1 & \frac{z^{\alpha+1-\theta} \Phat^{(\infty)}_1(z)}{\theta \Phat^{(\infty)}_2(z)}e^{n \phihat(z)}  \\
          0 & 1
        \end{pmatrix},
        & z \in (0, \aend(t)) \cup (\bend(t), +\infty).
      \end{cases}
    \end{equation}

  \item \label{enu:RHP:Svar:3}
    As $z \to \infty$ in $\compC$, $\Qhat_1$ behaves as $\Qhat_1(z)=1+\bigO(z^{-1})$.

  \item \label{enu:RHP:Svar:4}
    As $z \to \infty$ in $\mathbb{H}_\theta $, $\Qhat_2$ behaves as $\Qhat_2(z)=\bigO(1)$.

  \item
  \label{enu:RHP:Svar:5}
    $\Qhat_1(z)$ is analytic at $0 \in \compC$, and as $z \to 0$ in $\halfH \setminus \Sigmahat$, we have
    \begin{equation}
      \Qhat_2(z)=
      \begin{cases}
        \bigO (1), & \alpha > \theta - 1, \\
        \bigO (\log z), & \alpha = \theta - 1, \\
        \bigO (z^{\alpha+1-\theta}), & -1 < \alpha < \theta - 1.
      \end{cases}
    \end{equation}
  \item
    As $z \to \aend(t)$, we have $\Qhat_1(z) = \bigO((z - \aend(t))^{1/4})$ and $\Qhat_2(z) = \bigO((z - \aend(t))^{1/4})$, and as $z \to \bend(t)$, we have $\Qhat_1(z) = \bigO((z - \bend(t))^{1/4})$ and $\Qhat_2(z) = \bigO((z - \bend(t))^{1/4})$.
  \item \label{enu:RHP:Svar:7}
    For $x>0$, we have the boundary condition $\Qhat_2(e^{\pi i/\theta}x) = \Qhat_2(e^{-\pi i/\theta}x)$.
  \end{enumerate}
\end{RHP}

\begin{RHP} \label{RHP:tildeSvar}
\hfill
  \begin{enumerate} 
  \item
    $\Qhattilde = (\Qhattilde_1, \Qhattilde_2)$ is analytic in $(\mathbb{H}_\theta  \setminus \Sigmahat, \compC \setminus \Sigmahat)$.
  \item
   For $z \in \Sigmahat$, we have
    \begin{equation}
      \Qhattilde_+(z) = \Qhattilde_-(z) J_{\Qhattilde}(z),
      \end{equation}
      where
      \begin{equation} \label{def:Jtildecals}
      J_{\Qhattilde}(z) =
      \begin{cases}
        \begin{pmatrix}
          1 & 0 \\
          \frac{\Phattilde_2(z)}{\Phattilde_1(z)}z^{-\alpha}e^{-n \phihat(z)} & 1
        \end{pmatrix},
        & \text{$z \in \Sigmahat_1 \cup \Sigmahat_2$}, \\
        \begin{pmatrix}
          0 & 1 \\
          1 & 0
        \end{pmatrix},
        & z \in (\aend(t), \bend(t)), \\
        \begin{pmatrix}
          1 & \frac{\Phattilde_1(z)}{\Phattilde_2(z)}z^{\alpha}e^{n \phihat(z)}  \\
          0 & 1
        \end{pmatrix},
        & z \in (0, \aend(t)) \cup (\bend(t), +\infty).
      \end{cases}
    \end{equation}

    \item As $z \to \infty$ in $\mathbb{H}_\theta $, $\Qhattilde_1$ behaves as $\Qhattilde_1(z)=1+\bigO(z^{-\theta})$.

    \item  As $z \to \infty$ in $\compC$, $\Qhattilde_2$ behaves as $\Qhattilde_2(z)=\bigO(1)$.

    \item
    As $z \to 0$ in $\compC$, $\Qhattilde_1(z) = \bigO(1)$, and as $z \to 0$ in $\mathbb{H}_\theta \setminus \Sigmahat$, we have
    \begin{equation} \label{eq:tildeasy_Q_1_in_lens}
      \Qhattilde_2(z) =
      \begin{cases}
        \bigO(1), & \alpha  > 0, \\
        \bigO(\log z), & \alpha  = 0, \\
        \bigO(z^{\alpha}), & \alpha < 0.
      \end{cases}
    \end{equation}

    \item
      As $z \to \aend(t)$, we have $\Qhattilde_1(z) = \bigO((z - \aend(t))^{1/4})$ and $\Qhattilde_2(z) = \bigO((z - \aend(t))^{1/4})$, and as $z \to \bend(t)$, we have $\Qhattilde_1(z) = \bigO((z - \bend(t))^{1/4})$ and $\Qhattilde_2(z) = \bigO((z - \bend(t))^{1/4})$.

    \item For $x>0$, we have the boundary condition $\Qhattilde_1(e^{\pi i/\theta}x) = \Qhattilde_1(e^{-\pi i/\theta}x)$.
  \end{enumerate}
\end{RHP}

\subsection{Local parametrices and the shape of $\Sigmahat_1$} \label{subsubsec:three_parametrices}

\paragraph{Local parametrices at $\bend(t)$}

Let $\epsilon > 0$ be a small enough positive number, and, analogous to $f_b(z)$ in \eqref{def:fb_hard}, define $\fhat_b(z)$ on $D(\bend(t), \epsilon)$ by
\begin{equation} \label{def:fb}
  \fhat_b(z) = \left( -\frac{3}{4} \phihat(z) \right)^{\frac{2}{3}},
\end{equation}
which is a conformal mapping with $\fhat_b(b) = 0$ and $\fhat'_b(b) > 0$. We specify the shape of $\Sigmahat_1$ in $D(\bend(t), \epsilon)$ as $\fhat^{-1}_b(\{ e^{\frac{2\pi i}{3}} [0, +\infty) \}) \cap D(\bend(t), \epsilon)$, and define
\begin{align}
  \Phat^{(b)}(z) = {}& \Ehat^{(b)}(z) \Psi^{(\Ai)}(n^{\frac{2}{3}} \fhat_b(z))
  \begin{pmatrix}
    e^{-\frac{n}{2} \phihat(z)} \ghat^{(a, b)}_1(z) & 0 \\
    0 & e^{\frac{n}{2} \phihat(z)} \ghat^{(a, b)}_2(z)
  \end{pmatrix}, \label{def:Pb} \\
  \Phattilde^{(b)}(z) = {}& \Ehattilde^{(b)}(z) \Psi^{(\Ai)}(n^{\frac{2}{3}} \fhat_b(z))
  \begin{pmatrix}
    e^{-\frac{n}{2} \phihat(z)} \ghattilde^{(a, b)}_1(z) & 0 \\
    0 & e^{\frac{n}{2} \phihat(z)} \ghattilde^{(a, b)}_2(z)
  \end{pmatrix},
\end{align}
where $\Psi^{(\Ai)}$ is defined in Appendix \ref{app:Airy}, the same as in \eqref{eq:defn_P^b} and \eqref{def:tildePb}
\begin{align}
  \ghat^{(a, b)}_1(z) = {}& \frac{z^{(\theta - \alpha - 1)/2}}{\Phat^{(\infty)}_1(z)}, & \ghat^{(a, b)}_2(z) = {}& \frac{z^{(\alpha + 1 - \theta)/2}}{\theta \Phat^{(\infty)}_2(z)}, & \ghattilde_1(z) = {}& \frac{z^{-\alpha/2}}{\Phattilde^{(\infty)}_1(z)}, & \ghattilde^{(a, b)}_2(z) = {}& \frac{z^{\alpha/2}}{\Phattilde^{(\infty)}_2(z)},
\end{align}
and with $\star = \hat{} \text{ or } \widetilde{\hat{}}$,
\begin{equation}
  \overset{\star}{E}^{(b)}(z) = \frac{1}{\sqrt{2}}
  \begin{pmatrix}
    \overset{\star}{g}^{(a, b)}_1(z) & 0 \\
    0 & \overset{\star}{g}^{(a, b)}_2(z)
  \end{pmatrix}^{-1}
  e^{\frac{\pi i}{4} \sigma_3}
  \begin{pmatrix}
    1 & -1 \\
    1 & 1
  \end{pmatrix}
  \begin{pmatrix}
    n^{\frac{1}{6}} \fhat_b(z)^{\frac{1}{4}} & 0 \\
    0 & n^{-\frac{1}{6}} \fhat_b(z)^{-\frac{1}{4}}
  \end{pmatrix}.
\end{equation}
By \eqref{eq:phihat_est_b} and its generalization stated in Proposition \ref{prop:extension_t}, we have that for all $t \in [C^{-1}, 1 - C n^{-1/2}]$, $\fhat_b(z) = \bigO(1)$ uniformly for $z \in \partial D(\bend(t), \epsilon)$, and then by the asymptotic formula of the Airy parametrix \eqref{eq:defn_Psi^Ai_infty},
\begin{align}
  \Phat^{(b)}(z) = {}& I + \bigO(n^{-1}), & \Phattilde^{(b)}(z) = {}& I + \bigO(n^{-1}), & z \in {}& \partial D(\bend(t), \epsilon).
\end{align}

\paragraph{Local parametrices at $\aend(t)$}

Let $\epsilon' > 0$ be a small enough constant, let ($u = t - 1$ as in Sections \ref{subsec:deformation_along_C2_small_t} and \ref{subsec:deformation_along_C2_large_t})
\begin{equation} \label{eq:defn_rhat_n}
  \rhat_n = \rhat_n(t) = \epsilon' (-u)^{1 + \frac{1}{\theta}},
\end{equation}
and let $\fhat_a(z)$ be a function analytic in $D(\aend(t), \rhat_n)$, such that
\begin{equation}
  \frac{2}{3} \fhat_a(z)^{\frac{3}{2}} = -\frac{1}{2} \phihat(z) \pm \pi i.
\end{equation}
We have $\fhat_a(\aend(t)) = 0$ and $\fhat'_a(\aend(t)) < 0$. Then we specify the shape of $\Sigmahat_1$ in $D(\aend(t), \rhat_n)$ as $\fhat^{-1}_a(\{ e^{-\frac{2\pi i}{3}} [0, +\infty) \}) \cap D(\aend(t), \rhat_n)$, and define
\begin{align} \label{def:Pa}
  \Phat^{(a)}(z) = {}& \Ehat^{(a)}(z) \Psi^{(\Ai)}(n^{\frac{2}{3}} \fhat_a(z))
  \begin{pmatrix}
    e^{-\frac{n}{2} \phihat(z)} \ghat^{(a, b)}_1(z) & 0 \\
    0 & -e^{\frac{n}{2} \phihat(z)} \ghat^{(a, b)}_1(z)
  \end{pmatrix}, \\
  \Phattilde^{(a)}(z) = {}& \Ehattilde^{(a)}(z) \Psi^{(\Ai)}(n^{\frac{2}{3}} \fhat_a(z))
  \begin{pmatrix}
    e^{-\frac{n}{2} \phihat(z)} \ghattilde^{(a, b)}_1(z) & 0 \\
    0 & -e^{\frac{n}{2} \phihat(z)} \ghattilde^{(a, b)}_1(z)
  \end{pmatrix}, \\
\end{align}
where, with $\star = \hat{} \text{ or } \widetilde{\hat{}}$,
\begin{equation}
  \overset{\star}{E}^{(a)}(z) = \frac{1}{\sqrt{2}}
  \begin{pmatrix}
    \overset{\star}{g}^{(a, b)}_1(z) & 0 \\
    0 & -\overset{\star}{g}^{(a, b)}_2(z)
  \end{pmatrix}^{-1}
  e^{\frac{\pi i}{4} \sigma_3}
  \begin{pmatrix}
    1 & -1 \\
    1 & 1
  \end{pmatrix}
  \begin{pmatrix}
    n^{\frac{1}{6}} \fhat_a(z)^{\frac{1}{4}} & 0 \\
    0 & n^{-\frac{1}{6}} \fhat_a(z)^{-\frac{1}{4}}
  \end{pmatrix}.
\end{equation}
By \eqref{eq:phihat_est_a} and its generalization stated in Proposition \ref{prop:extension_t}, for all $t \in [C^{-1}, 1 - C n^{-1/2}]$, $\fhat_a(z) = \bigO((-u)^{4/3})$ uniformly for $z \in \partial D(\aend(t), \rhat_n)$, and then by the asymptotic formula of the Airy parametrix, we have a constant $C' > 0$ depending on $C$, such that
\begin{align} \label{eq:Phat^a_est}
  \lVert \Phat^{(a)}(z) - I \rVert \leq {}&  C' n^{-1} (-u)^{-2}, & \lVert \Phattilde^{(a)}(z) - I \rVert \leq {}& C' n^{-1} (-u)^{-2}, & z \in {}& \partial D(\aend(t), \rhat_n).
\end{align}
where $\lVert \cdot \rVert$ is the max matrix norm.

\paragraph{Local parametrices at $0$}

Our construction here is similar to \cite[Section 4.3.2]{Bogatskiy-Claeys-Its16}. With $\rhat_n$ defined in \eqref{eq:defn_rhat_n}, we define the $(\theta + 1) \times (\theta + 1)$ matrix-valued function $\Phat^{(0, \pre)}(z) = (\Phat^{(0, \pre)}_{j, k}(z))^{\theta}_{j, k = 0}$ by
\begin{equation}
  \Phat^{(0, \pre)}(z) = I_{1} \oplus (z^{\frac{j - 1}{\theta}} e^{\frac{(j - 1)(k - 1) 2\pi i}{\theta}})^{\theta}_{j, k = 1}.
\end{equation}
We denote 
\begin{align}
  f_0(z) = {}& \frac{1}{2\pi i} \int^{2\rhat_n}_0 \frac{\Phattilde_1 ^{(\infty)}(x)}{\Phattilde^{(\infty)}_2(x)} x^{\alpha}e^{n \phihat(x)} \frac{dx}{x - z}, & & \\
  f_j(z) = {}& \frac{1}{2\pi i} \int^{(2\rhat_n)^{\theta}}_0 \frac{x^{\frac{\alpha + j - \theta}{\theta}} \Phat^{(\infty)}_1(x^{\frac{1}{\theta}})}{\theta \Phat^{(\infty)}_2(x^{\frac{1}{\theta}})} e^{n \phihat(x^{\frac{1}{\theta}})} \frac{dx}{x - z}, & j = {}& 1, \dotsc, \theta.
\end{align}
In $D(0, \rhat^{\theta}_n) \setminus [0, +\infty)$ define the $(\theta + 1) \times (\theta + 1)$ matrix-valued functions $\Pmodelhat^{(0)}(z) = (\Pmodelhat^{(0)}_{j, k}(z))^{\theta}_{j, k = 0}$ and $\Pmodelhattilde^{(0)}(z) = (\Pmodelhattilde^{(0)}_{j, k}(z))^{\theta}_{j, k = 0}$ by
\begin{align} \label{eq:defn_P0}
  \Pmodelhat^{(0)}_{j,k}(z) = {}&
  \begin{cases}
    f_j(z), & k = 0 \text{ and } j = 1, \dotsc, \theta, \\
    \Pmodelhat^{(0, \pre)}_{j,k}(z), & \text{otherwise},
  \end{cases} \\
  \Pmodelhattilde^{(0)}_{j,k}(z) = {}&
  \begin{cases}
    f_0(e^{\frac{2\pi i}{\theta} (k - 1)}z^{\frac{1}{\theta}}), & j = 0 \text{ and } k = 1, \dotsc, \theta, \\
    \Pmodelhat^{(0, \pre)}_{j,k}(z), & \text{otherwise}.
  \end{cases}
\end{align}
It is straightforward to see that for $j = 1, \dotsc, \theta$
\begin{align}
  \Pmodelhat^{(0)}_{j, 0}(z) = {}&
  \begin{cases}
    \bigO(1), & \alpha + j - \theta > 0, \\
    \bigO(\log z), & \alpha + j - \theta = 0, \\
    \bigO(z^{\frac{\alpha + j - \theta}{\theta}}), & \alpha + j - \theta < 0,
  \end{cases}
  &
  \Pmodelhat^{(0)}_{0, k}(z) = {}&
  \begin{cases}
    \bigO(1), & \alpha > 0, \\
    \bigO(\log z), & \alpha = 0, \\
    \bigO(z^{\frac{\alpha}{\theta}}), & \alpha < 0,
  \end{cases}
\end{align}
Then in $D(0, \rhat_n) \setminus [0, +\infty)$ we define the $(\theta + 1) \times (\theta + 1)$ matrix-valued function $\Phat^{(0)}(z) = (\Phat^{(0)}_{j, k}(z))^{\theta}_{j, k = 0}$ by
\begin{align}
  \Phat^{(0)}(z) = {}& (\Phat^{(0, \pre)}(z))^{-1} \Pmodelhat^{(0)}(z) & \Phattilde^{(0)}(z) = {}& (\Phat^{(0, \pre)}(z))^{-1} \Pmodelhattilde^{(0)}(z).
\end{align}
For all $t \in [C^{-1}, 1 - C n^{-1/2}]$, we have that uniformly
\begin{align}
  \frac{\Phat^{(\infty)}_1(x^{1/\theta})}{\Phat^{(\infty)}_2(x^{1/\theta})} = {}& \bigO(u^{-\frac{\theta + 1}{\theta}(\alpha - 1 - \theta)}), & \frac{\Phattilde^{(\infty)}_1(x^{1/\theta})}{\Phattilde^{(\infty)}_2(x^{1/\theta})} = {}& \bigO(u^{-\frac{\theta + 1}{\theta}\alpha}), & x \in {}& [0, 2\rhat_n].
\end{align}
Together with the estimate of $\phihat(x^{1/\theta})$ by Proposition \ref{prop:phihat} and its generalization in Proposition \ref{prop:extension_t}, we derive that there exists $\epsilon'' > 0$ such that for all $t \in [C^{-1}, 1 - C n^{-1/2}]$,
\begin{align} \label{eq:Phat^0_est}
  \lVert \Phat^{(0)}(z) - I \rVert \leq {}& e^{-\epsilon'' n (-u)^2}, & \lVert \Phat^{(0)}(z) - I \rVert \leq {}& e^{-\epsilon'' n (-u)^2}, & z \in {}& \partial D(0, \rhat_n).
\end{align}

\paragraph{Shape of $\Sigmahat_1$}

 The shape of $\Sigmahat_1$ has been determined in $D(\aend(t), \rhat_n)$ and $D(\bend(t), \epsilon)$ above in this subsection. Let $\sigma_a$ and $\sigma_b$ be the intersections of $\Sigmahat_1$ with $\partial D(\aend(t), \rhat_n)$ and $\partial D(\bend(t), \epsilon)$, respectively, and we let the undetermined part of $\Sigmahat_1$ be a contour $\Sigmahat^R_1$ connecting $\sigma_a$ and $\sigma_b$ such that $\Re \phihat(z) > \epsilon''' \lvert z \rvert^{2\theta/(\theta + 1)}$ for $z \in \Sigmahat^R_1$, where $\epsilon'''$ is a small positive constant that is independent of $t$. 

To justify the existence of contour $\Sigmahat^R_1$, we consider two cases separately. Let $c > 0$ be a small constant, and for $t \in [C^{-1}, 1 - c]$, the property that defines $\Sigmahat^R_1$ is equivalent to that $\Re \phihat(z) > 0$ for all $z \in \Sigmahat^R_1$, and this follows directly from Regularity Conditions \ref{reg:primary} and \ref{reg:soft-edge}, like in the definition of $\Sigmahat^R_1$ at the beginning of Section \ref{subsubsec:final_trans_p}, we may choose $\Sigmahat^R_1$ to be close to the real axis. For $t \in [1 - c, 1 - C n^{-1/2}]$, by the property of $\phihat(z)$ in the neighbourhood of $0$ as stated in Section \ref{subsec:deformation_along_C2_small_t}, and the regularity stated in Regularity Condition \ref{reg:primary}, we also justify the existence of $\Sigmahat^R_1$.

Since we need to define $\Sigmahat$ for all $t \in [C^{-1}, 1 - C n^{-1/2}]$, we define $\Sigmahat^R_1$ in a uniform way as follows. We assume that $\epsilon'$ is small relative to $\epsilon$, so that $\Im \sigma_a < \Im \sigma_b$ for all $t \in [C^{-1}, 1 - C n^{-1/2}]$. Then let $\Sigmahat^R_1$ be a polygonal line that connects $\sigma_a$ to vertex $\sigma'$ and then to $\sigma_b$, such that $\Im \sigma' = \Im \sigma_b$ and $\arg (\sigma' - \sigma_a) = \pi/3$. If $\epsilon$ is small enough, by arguments above, we find that this definition works.

\subsection{Scalar shifted RH problem for $\tRhat$}

We define the vector-valued functions $\Uhat^{(0)}(z) = (\Uhat^{(0)}_0(z), \dotsc, \Uhat^{(0)}_{\theta}(z))$ and $\Uhattilde^{(0)}(z) = (\Uhattilde^{(0)}_0(z), \dotsc, \Uhattilde^{(0)}_{\theta}(z))$ as
\begin{align}
  \Uhat_0(z) = {}& \Qhat_2(z^{\frac{1}{\theta}}), & \Uhattilde_0(z) = {}& \Qhattilde_1(z^{\frac{1}{\theta}}), & z \in {}& D(0, \rhat^{\theta}_n) \setminus (-\rhat^{\theta}_n, \rhat^{\theta}_n), \\
  \Uhat_k(z) = {}& \Qhat_1(z^{\frac{1}{\theta}} e^{\frac{2(k - 1)}{\theta} \pi i}), & \Uhattilde_k(z) = {}& \Qhattilde_2(z^{\frac{1}{\theta}} e^{\frac{2(k - 1)}{\theta} \pi i}), &z \in {}& D(0, \rhat^{\theta}_n) \setminus (-\rhat^{\theta}_n, \rhat^{\theta}_n), \quad k = 1, 2, \dotsc, \theta.
\end{align}
It is straightforward to check that the vector-valued functions
\begin{align}
  \Vhat^{(0, \pre)}(z) = {}& \Uhat(z) \Pmodelhat^{(0)}(z)^{-1}, & \Vhattilde^{(0, \pre)}(z) = {}& \Uhattilde(z) \Pmodelhattilde^{(0)}(z)^{-1}
\end{align}
are analytic in $D(0, \rhat^{\theta}_n) \setminus \{ 0 \}$, and all their components are $o(z^{-1})$ as $z \to 0$. Hence all components of $\Vhat^{(0, \pre)}(z)$ and $\Vhattilde^{(0, \pre)}(z)$ are analytic in $D(0, \rhat^{\theta}_n)$. Then let the $(\theta + 1)$-dimensional vector-valued functions $\Vhat^{(0)}(z) = (\Vhat^{(0)}_0(z), \dotsc, \Vhat^{(0)}_{\theta}(z))$ and $\Vhattilde^{(0)}(z) = (\Vhattilde^{(0)}_0(z), \dotsc, \Vhattilde^{(0)}_{\theta}(z))$ be
\begin{align}
  \Vhat^{(0)}(z) = {}& \Uhat(z) \Phat^{(0)}(z)^{-1} = \Vhat^{(0, \pre)}(z) \Phat^{(0, \pre)}(z), \\
  \Vhattilde^{(0)}(z) = {}& \Uhattilde(z) \Phattilde^{(0)}(z)^{-1} = \Vhattilde^{(0, \pre)}(z) \Phattilde^{(0, \pre)}(z).
\end{align}

We define the $2$-dimensional vector-valued functions
\begin{align}
  \Vhat^{(a)}(z) = {}& (\Vhat^{(a)}_1(z), \Vhat^{(a)}_2(z)) = \Qhat(z) \Phat^{(a)}(z)^{-1}, && \notag \\
  \Vhattilde^{(a)}(z) = {}& (\Vhattilde^{(a)}_1(z), \Vhattilde^{(a)}_2(z)) = \Qhattilde(z) \Phattilde^{(a)}(z)^{-1}, & z \in {}& D(\aend(t), \rhat_n), \\
  \Vhat^{(b)}(z) = {}& (\Vhat^{(b)}_1(z), \Vhat^{(b)}_2(z)) = \Qhat(z) \Phat^{(b)}(z)^{-1}, && \notag \\
  \Vhattilde^{(b)}(z) = {}& (\Vhattilde^{(b)}_1(z), \Vhattilde^{(b)}_2(z)) = \Qhattilde(z) \Phattilde^{(b)}(z)^{-1}, & z \in {}& D(\bend(t), \epsilon).
\end{align}

Next we define $\Rhat(z)$ and $\Rhattilde(z)$. We first introduce some contours and domains. We set
\begin{equation} 
 \Sigmahat^{R} := [\rhat_n, \aend(t) - \rhat_n] \cup [\aend(t), \bend(t)] \cup [\bend(t)+\epsilon, \infty) \cup \partial D(0, \rhat_n) \cup \partial D(\aend(t), \rhat_n) \cup \partial D(\bend(t),\epsilon) \cup \Sigmahat_1^R \cup \Sigmahat_2^R,
\end{equation}
where $\Sigmahat_1^R$ is defined at the end of Section \ref{subsubsec:three_parametrices}, and $\Sigmahat_2^R = \overline{\Sigmahat_1^R}$.
see Figure \ref{fig:Sigmahat_R} for an illustration.
\begin{figure}[htb]
 \begin{minipage}{0.35\linewidth}
   \centering
   \includegraphics{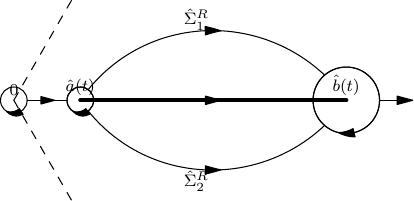}
   \caption{Jump contour $\Sigmahat^R$.}
   \label{fig:Sigmahat_R}
 \end{minipage}
 \hspace{\stretch{1}}
 \begin{minipage}{0.6\linewidth}
   \centering
   \includegraphics{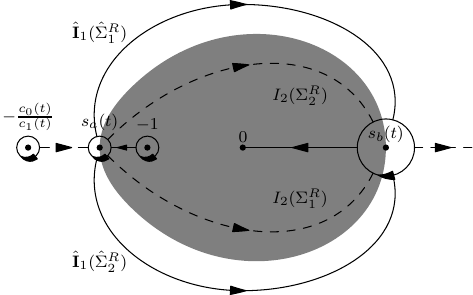}
   \caption{Jump contour $\SigmahatR$. The jump is trivial for $\tRhat(s)$ on the dashed part.}
   \label{fig:jump_R_scalar}
 \end{minipage}
\end{figure}
Like \eqref{def:Wk}, we divide the open disk $D(0, \rhat_n)$ into $\theta$ parts
\begin{equation}
  \What_k:=\{z\in D(0, \rhat_n) \setminus \{0\} : \arg z \in (\frac{(2k - 3)\pi}{\theta}, \frac{(2k - 1)\pi}{\theta}) \},\quad k=1, \dotsc, \theta,
\end{equation}
and denote by $\Gamma_k$ the arc boundary of $ \What_k$, $k=1,\ldots, \theta$. We then define a $1\times 2$ vector-valued functions $\Rhat(z) = (\Rhat_1(z), \Rhat_2(z))$ and $\Rhattilde(z) = (\Rhattilde_1(z), \Rhattilde_2(z))$  such that $\Rhat_1(z)$ and $\Rhattilde_2(z)$ are analytic in $\compC \setminus \Sigmahat^R$ and $\Rhat_2(z)$ and $\Rhattilde_1(z)$ are analytic in $\halfH \setminus \Sigmahat^R$, as follows:
\begin{align}
  \Rhat_1(z) = {}&
  \begin{cases}
    \Qhat_1(z), & \hbox{$z\in \mathbb{C} \setminus \{D(\bend(t),\epsilon) \cup D(\aend(t), \rhat_n) \cup D(0, \rhat_n) \cup \Sigmahat^R \}$,} \\
    \Vhat_1^{(b)}(z), & \hbox{$z\in D(\bend(t),\epsilon) \setminus [\bend(t) - \epsilon, \bend(t)]$,} \\
    \Vhat_1^{(a)}(z), & \hbox{$z\in D(\aend(t), \rhat_n) \setminus [\aend(t), \aend(t) + \rhat_n]$,} \\
    \Vhat^{(0)}_k(z^{\theta}), & z \in \What_k, \ k = 1, \dotsc, \theta,
  \end{cases} \label{eq:Rhat_1} \\
  \Rhattilde_2(z) = {}&
  \begin{cases}
    \Qhattilde_2(z), & \hbox{$z\in \mathbb{C} \setminus \{D(\bend(t),\epsilon) \cup D(\aend(t), \rhat_n) \cup D(0, \rhat_n) \cup \Sigmahat^R \}$,} \\
    \Vhattilde_2^{(b)}(z), & \hbox{$z\in D(\bend(t),\epsilon) \setminus [\bend(t) - \epsilon, \bend(t)]$,} \\
    \Vhattilde_2^{(a)}(z), & \hbox{$z\in D(\aend(t), \rhat_n) \setminus [\aend(t), \aend(t) + \rhat_n]$,} \\
    \Vhattilde^{(0)}_k(z^{\theta}), & z \in \What_k, \ k = 1, \dotsc, \theta,
  \end{cases} \\
  \Rhat_2(z) = {}&
  \begin{cases}
    \Qhat_2(z), & \hbox{$z\in \halfH \setminus \{D(\bend(t),\epsilon) \cup D(\aend(t), \rhat_n) \cup D(0, \rhat_n) \cup \Sigmahat^R \}$,} \\
    \Vhat_2^{(b)}(z), & \hbox{$z\in D(\bend(t),\epsilon) \setminus [\bend(t) - \epsilon, \bend(t)]$,} \\
    \Vhat_2^{(a)}(z), & \hbox{$z\in D(\aend(t), \rhat_n) \setminus [\aend(t), \aend(t) + \rhat_n]$,} \\
    \Vhat^{(0)}_0(z^{\theta}), & z \in \What_1 \setminus [0, \rhat_n], \label{eq:Rhat_2}
  \end{cases} \\
  \Rhattilde_1(z) = {}&
  \begin{cases}
    \Qhattilde_1(z), & \hbox{$z\in \halfH \setminus \{D(\bend(t),\epsilon) \cup D(\aend(t), \rhat_n) \cup D(0, \rhat_n) \cup \Sigmahat^R \}$,} \\
    \Vhattilde_1^{(b)}(z), & \hbox{$z\in D(\bend(t),\epsilon) \setminus [\bend(t) - \epsilon, \bend(t)]$,} \\
    \Vhattilde_1^{(a)}(z), & \hbox{$z\in D(\aend(t), \rhat_n) \setminus [\aend(t), \aend(t) + \rhat_n]$,} \\
    \Vhattilde^{(0)}_0(z^{\theta}), & z \in \What_1 \setminus [0, \rhat_n]. 
  \end{cases}
\end{align}
We note that both $\Rhat_1(z)$ and $\Rhattilde_2(z)$ are bounded as $z \to 0$ in $\compC$ and both $\Rhat_2(z)$ and $\Rhattilde_1(z)$ are bounded as $z \to 0$ in $\halfH$. Recall the function $J_{c_1, c_2}$ ($J_{c_1(t), c_2(t)}$ before the $t$-dependence is suppressed) defined in \eqref{eq:J_function} with $c_1(t)$ and $c_0(t)$ specified in Section \ref{subsec:deformation_along_C2_small_t} for $1 - t$ small, and Section \ref{subsec:deformation_along_C2_large_t} for $t \in (0, 1)$. Then we define
\begin{align} \label{eq:def_tRhat}
  \tRhat(s) = {}&
  \begin{cases}
    \Rhat_1(J_{c_1, c_0}(s)), & \text{$s \in \compC \setminus \overline{D}$ and $s \notin \Iinvhat_1(\Sigmahat^R )$,} \\
    \Rhat_2(J_{c_1, c_0}(s)), & \text{$s \in D \setminus[-1, 0]$ and $s \notin \Iinvhat_2(\Sigmahat^R)$,}
  \end{cases} \\
  \tRhattilde(s) = {}&
  \begin{cases}
    \Rhattilde_2(J_{c_1, c_0}(s)), & \text{$s \in \compC \setminus \overline{D}$ and $s \notin \Iinvhat_1(\Sigmahat^R )$,} \\
    \Rhattilde_1(J_{c_1, c_0}(s)), & \text{$s \in D \setminus[-1, 0]$ and $s \notin \Iinvhat_2(\Sigmahat^R)$,}
  \end{cases}
\end{align}
Here we remark that $\tRhat(s)$ and $\tRhattilde(s)$ are well defined and analytic at $s_a(t)$, $s_b(t)$, $-1 = \Iinvhat_2(0)$ and $-c_0(t)/c_1(t) = \Iinvhat_1(0)$. We have that $\tRhat(s)$ satisfies a shifted RH problem that is similar to \cite[RH problem 3.19]{Wang-Zhang21}, and $\tRhattilde(s)$ satisfies a shifted RH problem that is similar to \cite[RH problem 4.18]{Wang-Zhang21}. Below we state the RH problem for $\tRhat(s)$, and omit that for $\tRhattilde(s)$ for brevity. To state it, we define
\begin{equation}\label{def:SigmaRi}
  \begin{aligned}
    \SigmahatR^{(1)} := {}& \Iinvhat_1(\Sigmahat^R_1 \cup \Sigmahat^R_2) \subseteq \compC \setminus \overline{D}, & \SigmahatR^{(1')} := {}& \Iinvhat_2(\Sigmahat^R_1 \cup \Sigmahat^R_2) \subseteq D, \\
    \SigmahatR^{(2, L)} := {}& \Iinvhat_1((\rhat^n, \aend(t) - \rhat_n)) \subseteq \compC \setminus \overline{D}, & \SigmahatR^{(2', L)} := {}& \Iinvhat_2((\rhat^n, \aend(t) - \rhat_n)) \subseteq D,  \\
    \SigmahatR^{(2, R)} := {}& \Iinvhat_1((\bend(t) + \epsilon, +\infty)) \subseteq \compC \setminus \overline{D}, & \SigmahatR^{(2', R)} := {}& \Iinvhat_2((\bend(t) + \epsilon, +\infty)) \subseteq D, \\
    \SigmahatR^{(3)} := {}& \Iinvhat_1(\partial D(\bend(t),\epsilon)) \subseteq \compC \setminus \overline{D}, & \SigmahatR^{(3')} := {}& \Iinvhat_2(\partial D(\bend(t),\epsilon)) \subseteq D, \\
    \SigmahatR^{(4)} := {}& \Iinvhat_1(\partial D(\aend(t),\rhat_n)) \subseteq \compC \setminus \overline{D}, & \SigmahatR^{(4')} := {}& \Iinvhat_2(\partial D(\aend(t),\rhat_n)) \subseteq D, \\
    \SigmahatR^{(5)} := {}& \Iinvhat_2(\Gamma_1) \subseteq D, & \SigmahatR^{(5')}_k := {}& \Iinvhat_1(\Gamma_k) \subseteq \compC \setminus \overline{D}, \qquad k = 1, \dotsc, \theta,
  \end{aligned}
\end{equation}
and set
\begin{equation}\label{def:SigmaR}
  \SigmahatR = \SigmahatR^{(1)} \cup \SigmahatR^{(1')} \cup \SigmahatR^{(2, L)} \cup \SigmahatR^{(2, R)} \cup \SigmahatR^{(2', L)} \cup \SigmahatR^{(2', R)} \cup \SigmahatR^{(3)} \cup \SigmahatR^{(3')}  \cup \SigmahatR^{(4)} \cup \SigmahatR^{(4')} \cup \SigmahatR^{(5)} \cup \bigcup_{k=1}^{\theta} \SigmahatR^{(5')}_k,
\end{equation}
as shown in Figure \ref{fig:jump_R_scalar}.
\begin{rmk} \label{rmk:shapes_of_contour}
  By \eqref{eq:limit_Ihat_1} and \eqref{eq:limit_Ihat_2}, it is readily seen that the radius of the circular contour $\SigmahatR^{(5)} \cup \bigcup_{k=1}^{\theta} \SigmahatR^{(5')}_k$ around $s_a(t)$ and the radius of the circular contour $\SigmahatR^{(2', L)}$ around $-1$ are both of the order $\bigO(-u)$. 
\end{rmk}
We also define the following functions on each curve constituting $\SigmahatR$:
\begin{align}
  J_{\SigmahatR^{(1)}}(s) = {}& \frac{\theta \Phat^{(\infty)}_2(z)}{z^{\alpha + 1 - \theta} \Phat^{(\infty)}_1(z)} e^{-n\phihat(z)}, & s \in {}& \SigmahatR^{(1)}, \label{def:Jsigma1} \\
  \intertext{where $z = J_{c_1, c_2}(s) \in \Sigmahat^R_1 \cup \Sigmahat^R_2$,}
  J_{\SigmahatR^{(2')}}(s) = {}& \frac{z^{\alpha + 1 - \theta} \Phat^{(\infty)}_1(z)}{\theta \Phat^{(\infty)}_2(z)} e^{n\phi(z)}, & s \in {}& \SigmahatR^{(2', L)} \cup \SigmahatR^{(2', R)}, \label{def:Jsigma2'} \\
  \intertext{where $z = J_{c_1, c_2}(s) \in (\rhat_n, \aend(t) - \rhat_n) \cup (\bend(t) + \epsilon, +\infty)$,}
  J^{1}_{\SigmahatR^{(3)}}(s) = {}& \Phat^{(b)}_{11}(z) - 1, \quad J^{2}_{\SigmahatR^{(3)}}(s) = \Phat^{(b)}_{21}(z), & s \in {}& \SigmahatR^{(3)}, \label{def:Jsigma3} \\
  \intertext{where $z = J_{c_1, c_2}(s) \in \partial D(\bend(t), \epsilon)$,}
  J^{1}_{\SigmahatR^{(3')}}(s) = {}& \Phat^{(b)}_{22}(z) - 1, \quad J^{2}_{\SigmahatR^{(3')}}(s) = \Phat^{(b)}_{12}(z), & s \in {}& \SigmahatR^{(3')},
  \label{def:Jsigma32}\\
  \intertext{where $z = J_{c_1, c_2}(s) \in \partial D(\bend(t), \epsilon)$,}
  J^{1}_{\SigmahatR^{(4)}}(s) = {}& \Phat^{(a)}_{11}(z) - 1, \quad J^{2}_{\SigmahatR^{(4)}}(s) = \Phat^{(a)}_{21}(z), & s \in {}& \SigmahatR^{(4)}, \label{def:Jsigma41} \\
  \intertext{where $z = J_{c_1, c_2}(s) \in \partial D(\aend(t), \rhat_n)$,}
  J^{1}_{\SigmahatR^{(4')}}(s) = {}& \Phat^{(a)}_{22}(z) - 1, \quad J^{2}_{\SigmahatR^{(4')}}(s) = \Phat^{(a)}_{12}(z), & s \in {}& \SigmahatR^{(4')},
  \label{def:Jsigma42} \\
  \intertext{where $z = J_{c_1, c_2}(s) \in \partial D(\aend(t), \rhat_n)$,}
  J^{0}_{\SigmahatR^{(5)}}(s) = {}&  \Phat^{(0)}_{00}(z^{\theta}) - 1, \quad J^{j}_{\SigmahatR^{(5)}}(s) = \Phat^{(0)}_{j0}(z^{\theta}), & s \in {}& \SigmahatR^{(5)}, \label{def:Jsigma5} \\
  \intertext{where $z = J_{c_1, c_2}(s) \in \Gamma_1$ and $j=1,\ldots,\theta$,}
  J^k_{\SigmahatR^{(5')}_k}(s) = {}& \Phat^{(0)}_{kk}(z^{\theta}) - 1, \quad J^j_{\SigmahatR^{(5')}_k}(s) = \Phat^{(0)}_{jk}(z^{\theta}), \quad j\neq k, & s \in {}& \SigmahatR^{(5')}_k, \label{def:Jsigma5'}
\end{align}
where $z = J_{c_1, c_0}(s) \in \Gamma_k$, $k=1,\ldots,\theta$, and $j=0,1,\ldots,\theta$. In \eqref{def:Jsigma3} and \eqref{def:Jsigma32}, $\Phat^{(b)}(z)=(\Phat^{(b)}_{jk}(z))_{j,k=1}^2$ is defined in \eqref{def:Pb}. In \eqref{def:Jsigma41} and \eqref{def:Jsigma42}, $\Phat^{(a)}(z)=(\Phat^{(a)}_{jk}(z))_{j,k=1}^2$ is defined in \eqref{def:Pa}. In \eqref{def:Jsigma5} and \eqref{def:Jsigma5'}, $\Phat^{(0)}(z)=(\Phat^{(0)}_{jk}(z))^{\theta}_{j,k=0}$ is defined in \eqref{eq:defn_P0}. With the aid of these functions, we further define an operator
$\DeltaR$ that acts on functions defined on $\SigmahatR$ by
\begin{equation}\label{def:DeltaR}
  \DeltaR f(s) =
  \begin{cases}
    J_{\SigmahatR^{(1)}}(s) f(\s), & \text{$s \in \SigmahatR^{(1)}$ and $\s = \Iinvhat_2(J_c(s))$}, \\
    J_{\SigmahatR^{(2')}}(s) f(\s), & \text{$s \in \SigmahatR^{(2', L)} \cup \SigmahatR^{(2', R)}$ and $\s = \Iinvhat_1(J_c(s))$}, \\
    J^{1}_{\SigmahatR^{(3)}}(s) f(s) + J^{2}_{\SigmahatR^{(3)}}(s) f(\s), & \text{$s \in \SigmahatR^{(3)}$ and $\s = \Iinvhat_2(J_c(s))$}, \\
    J^{1}_{\SigmahatR^{(3')}}(s) f(s) + J^{2}_{\SigmahatR^{(3')}}(s) f(\s), & \text{$s \in \SigmahatR^{(3')}$ and $\s = \Iinvhat_1(J_c(s))$}, \\
    J^{1}_{\SigmahatR^{(4)}}(s) f(s) + J^{2}_{\SigmahatR^{(4)}}(s) f(\s), & \text{$s \in \SigmahatR^{(4)}$ and $\s = \Iinvhat_2(J_c(s))$}, \\
    J^{1}_{\SigmahatR^{(4')}}(s) f(s) + J^{2}_{\SigmahatR^{(4')}}(s) f(\s), & \text{$s \in \SigmahatR^{(4')}$ and $\s = \Iinvhat_1(J_c(s))$}, \\
    J^{0}_{\SigmahatR^{(5)}}(s) f(s) + \sum\limits ^{\theta}_{j = 1} J^{j}_{\SigmahatR^{(5)}}(s) f(\s_j), & \text{$s \in \SigmahatR^{(5)}$ and $\s_j = \Iinvhat_1(J_c(s) e^{\frac{2(j - 1)\pi i}{\theta}}) \in \Iinvhat_1(\Gamma_j)$}, \\
    \sum\limits^{\theta}_{j = 0} J^j_{\SigmahatR^{(5')}_k}(s) f(\s_j), & \text{$s \in \SigmahatR^{(5')}_k$ and  $\s_0 = \Iinvhat_2(J_c(s) e^{\frac{2(1 - k)\pi i}{\theta}}) \in D$,} \\
    & \text{$\s_j = \Iinvhat_1(J_c(s) e^{\frac{2(j - k)\pi i}{\theta}}) \in \Iinvhat_1(\Gamma_j) \subseteq \compC \setminus \overline{D}$} \\
    & \text{for $j = 1, \dotsc, \theta$,  such that $\s_k = s$}, \\
    0, & s \in \SigmahatR^{(1')} \cup \SigmahatR^{(2, L)} \cup \SigmahatR^{(2, R)},
  \end{cases}
\end{equation}
where $f$ is a complex-valued function defined on $\SigmahatR$. Hence, we can define a scalar shifted RH problem as follows.

\begin{RHP} \label{rhp:tR} \hfill
  \begin{enumerate} 
  \item
    $\tRhat(s)$ is analytic in $\compC \setminus \SigmahatR$ and continuous up to boundary, where the contour $\SigmahatR$ is defined in \eqref{def:SigmaR}.
  \item \label{enu:rhp:tR:2}
    For $s\in \SigmahatR$, we have
    \begin{equation}
      \tRhat_+(s) - \tRhat_-(s) = \DeltaR \tRhat_-(s),
    \end{equation}
    where $\DeltaR$ is the operator defined in \eqref{def:DeltaR}.
  \item
    As $s \to \infty$, we have
    \begin{equation}
      \tRhat(s)=1+\bigO(s^{-1}).
    \end{equation}
  \end{enumerate}
\end{RHP}

\begin{prop} \label{prop:uniqueness_tR}
  RH problem \ref{rhp:tR} has a unique solution.
\end{prop}
The proof of the proposition is analogous to that of \cite[Proposition 3.20]{Wang-Zhang21}, and we omit it. The idea of the proof is that RH problem \ref{rhp:tR} is equivalent to RH problem \ref{RHP:OPs} for $Y = (Y_1, Y_2) := (p_n, C p_n)$, whose unique solvability is guaranteed by the biorthogonality.

\subsection{Small norm argument} \label{subsec:small_norm_soft_edge}

In this subsubsection, we obtain the following estimate:
\begin{lemma} \label{lem:tRest_soft}
  As $n\to \infty$, we have that there exists $C'' > 0$ that depends on $C$, such that for all $t \in [C^{-1}, 1 - C n^{-1/2}]$, and all $s \in \compC \setminus \SigmahatR$, 
  \begin{align} 
    \lvert \tRhat(s) - 1 \rvert \leq {}& C'' n^{-1} (-u)^2, & \lvert \tRhattilde(s) - 1 \rvert \leq {}& C'' n^{-1} (-u)^2.
  \end{align}
\end{lemma}
We are only going to prove the result for $\tRhat(s)$, because the proof depends on the shifted RH problems satisfied by $\tRhat(s)$ and $\tRhattilde(s)$, and we have only explicitly stated it for $\tRhat(s)$. The omission of the proof for $\tRhattilde(s)$ is only for brevity, since it is analogous to the proof for $\tRhat(s)$.

We note that there exist $C' > 0$ and $\epsilon'' > 0$ such that ($i = 1, 2$ and $k = 1, \dotsc, \theta$)
\begin{align}
  \lvert J_{\SigmahatR^{(1)}}(s) \rvert < {}& e^{-\epsilon'' n \lvert s - s_a(t) \rvert^2}, & s \in {}& \SigmahatR^{(1)}; & & \label{eq:est_J_SigmahatR:1} \\
  \lvert J_{\SigmahatR^{(2')}}(s) \rvert < {}& e^{-\epsilon'' n (-u)^2}, & s \in {}& \SigmahatR^{(2', L)}; & \lvert J_{\SigmahatR^{(2')}}(s) \rvert < {}& e^{-\epsilon'' n}, & s \in {}& \SigmahatR^{(2', R)}; \\
  \lvert J^{i}_{\SigmahatR^{(3)}}(s) \rvert < {}& C' n^{-1}, & s \in {}& \SigmahatR^{(3)}; & \lvert J^{i}_{\SigmahatR^{(3')}}(s) \rvert < {}& C' n^{-1}, & s \in {}& \SigmahatR^{(3')}; \\
  \lvert J^{i}_{\SigmahatR^{(4)}}(s) \rvert < {}& C' n^{-1} (-u)^{-2}, & s \in {}& \SigmahatR^{(4)}; & \lvert J^{i}_{\SigmahatR^{(4')}}(s) \rvert < {}& C' n^{-1} (-u)^{-2}, & s \in {}& \SigmahatR^{(4')}; \\
  \lvert J^{i}_{\SigmahatR^{(5)}}(s) \rvert < {}& e^{-\epsilon'' n(-u)^2}, & s \in {}& \SigmahatR^{(5)}; & \lvert J^{k}_{\SigmahatR^{(5')}}(s) \rvert < {}& e^{-\epsilon'' n(-u)^2}, & s \in {}& \SigmahatR^{(5')}. \label{eq:est_J_SigmahatR:5}
\end{align}
We note that $C'$ and $\epsilon''$ here may not be the same as $C'$ in \eqref{eq:Phat^a_est} and $\epsilon''$ in \eqref{eq:Phat^0_est}.

To show that $\tRhat(s) \to 1$ as $n\to\infty$, we use the strategy proposed in \cite{Claeys-Wang11} and also used in \cite{Wang-Zhang21}. We start with the claim that $\tRhat$ satisfies the integral equation
\begin{equation} \label{eq:constr_R}
  \tRhat(s) = 1 + \mathcal{C}(\DeltaR \tRhat_-)(s),
\end{equation}
where
\begin{equation}
  \mathcal{C}g(s) = \frac{1}{2\pi i} \int_{\SigmahatR} \frac{g(\xi)}{\xi - s} d\xi,  \qquad s\in\compC \setminus \SigmahatR,
\end{equation}
is the Cauchy transform of a function $g$. Indeed, due to the uniqueness of RH problem \ref{rhp:tR} stated in Proposition \ref{prop:uniqueness_tR}, it suffices to show that the right-hand side of \eqref{eq:constr_R} satisfies the RH problem for $\tRhat$ and the verification is straightforward. As a consequence of \eqref{eq:constr_R}, it is readily seen that
\begin{equation} \label{eq:tRs-1split}
  \tRhat(s) - 1 = \frac{1}{2\pi i} \int_{\SigmahatR} \frac{\DeltaR(\tRhat_- - 1)(\xi)}{\xi - s} d\xi + \frac{1}{2\pi i} \int_{\SigmahatR} \frac{\DeltaR(1)(\xi)}{\xi - s} d\xi,\qquad s\in\compC \setminus \SigmahatR.
\end{equation}
To estimate the two terms on the right-hand side of the above formula, we need the following estimate of the operator $\DeltaR$ which is a direct consequence of the estimates \eqref{eq:est_J_SigmahatR:1}--\eqref{eq:est_J_SigmahatR:5}.

\begin{prop}\label{prop:estdletaR}
Let $\DeltaR(1)(\xi)$ be the function on $\SigmahatR$ in \eqref{eq:tRs-1split} and $\DeltaR: L^2(\SigmahatR) \to L^2(\SigmahatR)$ be the operator defined in \eqref{def:DeltaR}. There exists a constant $C'''>0$ depending on $C$ such that for all $t \in [C^{-1}, 1 - C n^{-1/2}]$
\begin{align}
  \lVert \DeltaR(1) \rVert_{L^2(\SigmahatR)} \leq {}& C''' n^{-1} (-u)^{-3/2}, \label{eq:DeltaR1norm} \\
  \lVert \DeltaR \rVert_{L^2(\SigmahatR)} \leq {}& C''' n^{-1} (-u)^{-2}, \label{eq:estOperator}
\end{align}
if $n$ is large enough.
\end{prop}

By taking the limit where $s$ approaches the minus side of $\SigmahatR$, we obtain from \eqref{eq:tRs-1split} that
\begin{equation}\label{eq:tR-}
  \tRhat_-(s) - 1 = \mathcal{C}_{\DeltaR}(\tRhat_- - 1)(s) + \mathcal{C}_-(\DeltaR(1))(s),
\end{equation}
where
\begin{equation}
  \mathcal{C}_{\DeltaR}f(s) = \mathcal{C}_-(\DeltaR(f))(s), \qquad   \mathcal{C}_-g(s) = \frac{1}{2\pi i} \lim_{s' \to s_-} \int_{\SigmahatR} \frac{g(\xi)}{\xi - s'} d\xi,
\end{equation}
and the limit $s' \to s_-$ is taken when approaching the contour from the minus side. Since the Cauchy operator $\mathcal{C}_-$ is bounded, we see from Proposition \ref{prop:estdletaR} that the operator norm of $\mathcal{C}_{\DeltaR}$ is also uniformly $\bigO(n^{-1} (-u)^{-2})$ as $n\to \infty$ and $t \in [C^{-1}, 1 - C n^{-1/2}]$. Hence, by enlarging $C$ if necessary, we have that the operator $1-\mathcal{C}_{\DeltaR}$ is invertible for large enough $n$ and all $t \in [C^{-1}, 1 - C n^{-1/2}]$. Then we rewrite \eqref{eq:tR-} as
\begin{equation}
\tRhat_-(s) - 1 = (1 - \mathcal{C}_{\DeltaR})^{-1}( \mathcal{C}_-(\DeltaR(1)))(s).
\end{equation}
Combining formulas \eqref{eq:DeltaR1norm} and \eqref{eq:estOperator} gives us
\begin{equation}\label{eq:tR-1norm}
  \lVert \tRhat_- - 1 \rVert_{L^2(\SigmahatR)} = \bigO(n^{-1} (-u)^{-3/2}).
\end{equation}

At last, we prove Lemma \ref{lem:tRest_soft} under the condition that the distance between $s$ and $\SigmahatR$ is bounded below by $\epsilon''' (-u)$ where $\epsilon''' > 0$ is a small constant. By Remark \ref{rmk:shapes_of_contour}, this covers the interior of the circular contours $\SigmahatR^{(5)} \cup \bigcup_{k=1}^{\theta} \SigmahatR^{(5')}_k$ and $\SigmahatR^{(2', L)}$. To prove Lemma \ref{lem:tRest_soft} for all $s \in \compC \setminus \SigmahatR$, we may use the argument of deformation of the contour $\SigmahatR$, like in the proof of \cite[Lemma 3.16]{Wang23a}, and we omit the detail. From \eqref{eq:tRs-1split}, \eqref{eq:DeltaR1norm}, \eqref{eq:tR-1norm} and the Cauchy-Schwarz inequality, we have
  \begin{equation}
    \begin{split}
      \lvert \tRhat(s) - 1 \rvert & \leq  \frac{1}{2\pi} \left( \lVert \DeltaR(\tRhat_- - 1) \rVert_{L^2(\SigmahatR)} + \lVert \DeltaR(1) \rVert_{L^2(\SigmahatR)} \right) \cdot \lVert \frac{1}{\xi - s} \rVert_{L^2(\SigmahatR)} \\
      &\leq \frac{1}{2\pi} \left( \bigO(n^{-2} (-u)^{-7/2}) + \bigO(n^{-1} (-u)^{-3/2}) \right) \cdot \bigO(u^{-1/2})= \bigO(n^{-1} (-u)^2).
    \end{split}
  \end{equation}
This finishes the proof of Lemma \ref{lem:tRest_soft}.

\subsection{Proof of  a technical lemma}

In this subsection, we  prove the following result that will be used in the proof of Theorem \ref{thm:universality}.
\begin{lemma} \label{lem:tail}
  Let $C$ be a large enough positive number and $t \in [C^{-1}, 1 - C n^{-1/2}]$. Let $p_j(x) = p^{(V_t)}_{n, k}(x)$, $q_k(x) = q^{(V_t)}_{n, k}(x)$ and $\kappa_j = \kappa^{(V_t)}_{n,j}$ be defined by \eqref{eq:biorthogonality}. For $j = k = n$, if $n$ is large enough, then there exists some constant $C'>0$ such that for all  $x \in [0, \rhat_n]$, where $\rhat_n = \rhat_n(t)$ is defined in \eqref{eq:defn_rhat_n},
  \begin{align}
    \lvert p_n(x) \rvert < {}& C' (-u)^{\frac{1}{2} - \frac{\alpha + 1}{\theta}} e^{n \Re \gfnhat_t(x)}, \label{eq:est_p_n_soft} \\
    \lvert q_n(x^{\theta}) \rvert < {}& C' (-u)^{-\frac{1}{2} - \alpha} e^{n \Re \gfnhattilde_t(x)}, \label{eq:est_q_n_soft} \\
    \lvert \kappa_n(t) \rvert > {}& C'^{-1} e^{n\ellhat_t}, \label{eq:est_kappa_n_soft}
  \end{align}
  where $u=t-1$.
\end{lemma}
\begin{proof}
When $t \in [C^{-1}, 1 - C n^{-1/2}]$ and $x \in [0, \rhat_n]$, we note that $\Rhat_1(x) = \Qhat_1(x)$. Then by tracing back the transformations $Y \to \That \to \Shat \to \Qhat$, we find that
\begin{equation}
  p_n(x) = \Rhat_1(x) \Phat^{(\infty)}_1(x) e^{n\gfnhat_t(x)}. 
\end{equation}
By the estimate of $\tRhat$ in Lemma \ref{lem:tRest_soft} and the estimate that $\Phat^{(\infty)}_1(x) = \bigO((-u)^{\frac{1}{2} - \frac{\alpha + 1}{\theta}})$, we prove \eqref{eq:est_p_n_soft}.

Similarly,  by tracing back the transformations $\Y \to \Thattilde \to \Shattilde \to \Qhattilde$ and using $\Rhattilde_1(x) = \Qhattilde_1(x)$, 
 we find that
\begin{equation}
  q_n(x^{\theta}) = \Rhattilde_1(x) \Phattilde^{(\infty)}_1(x) e^{n\gfnhattilde_t(x)}. 
\end{equation}
By the estimate of $\tRhattilde$ in Lemma \ref{lem:tRest_soft} and the estimate that $\Phattilde^{(\infty)}_1(x) = \bigO((-u)^{-\frac{1}{2} - \alpha})$, we prove \eqref{eq:est_q_n_soft}.

 By \eqref{eq:defn_Cp_n} and the definition of $Y$, we have
\begin{equation}\label{eq:kapparep}
  Y_2(z) = Cp_n(z) = \frac{i}{2\pi} \kappa_n z^{-(n + 1)\theta} + \bigO(z^{-(n + 2)\theta}), \quad z \to +\infty.
\end{equation}
On the other hand, we have
\begin{equation} 
  Y_2(z) = \Rhat_2(z) \Phat^{(\infty)}_2(z) e^{n(\ellhat_t - \gfnhattilde_t(z))}.
\end{equation}
Letting $z \to \infty$ in the above equation and using the estimate of $\tRhat$ in Lemma \ref{lem:tRest_soft}, we derive \eqref{eq:est_kappa_n_soft}.

\end{proof}
\section{Proof of Theorems \ref{thm:main} and \ref{thm:universality}} \label{sec:proof_main}

\subsection{Proof of Theorem \ref{thm:main}} 


The strategy used in this subsection is the same as that in \cite[Section 5]{Wang-Zhang21}. 

\paragraph{Proof of \eqref{eq:asy_p_n}}

Let $\epsilon > 0$ be a small positive number. Analogous to \cite[Equation (5.4)]{Wang-Zhang21}, we define
\begin{equation}\label{def:v}
  v(s) =
  \begin{cases}
    R_2(e^{-\frac{\pi i}{\theta}} s^{1 + \frac{1}{\theta}}), & \arg s \in (0, \frac{\pi}{\theta + 1}), \\
    R_2(e^{\frac{\pi i}{\theta}} s^{1 + \frac{1}{\theta}}), & \arg s \in (-\frac{\pi}{\theta + 1}, 0), \\
    R_1(-(-s)^{1 + \frac{1}{\theta}}), & \arg s \in (\frac{\pi}{\theta + 1}, \pi] \cup (-\pi, -\frac{\pi}{\theta + 1}),
  \end{cases}
\end{equation}
for $0 < \lvert s \rvert < \epsilon r^{\theta/(\theta + 1)}_n = \epsilon n^{-\theta/(2\theta + 1)}$ with $r_n$  defined in \eqref{def:rn}. Since $R_1$ and $R_2$ are defined by $V^{(0)}$  near the origin through \eqref{def:R1} and \eqref{def:R2} respectively, and $V^{(0)}$ satisfies a RH problem given in Proposition \ref{rhp:V0}, we find that $v(s)$ can be extended analytically in the disc $\lvert s \rvert \leq \epsilon n^{-\theta/(2\theta + 1)}$, and thus admits the Taylor expansion $v(s) = \sum^{\infty}_{k = 0} c_k s^k$ there. Moreover, as $\epsilon$ is small enough, we see from \eqref{def:v}, \eqref{eq:scalar_R_defn} and Lemma \ref{lem:tRest} that uniformly for $s$ in this disk
\begin{equation}\label{eq:estv}
  v(s)=1+\bigO(n^{-\frac{1}{2(2\theta+1)}}).
\end{equation}
Then analogous to \cite[Equation (5.7)]{Wang-Zhang21}, we define for $\lvert s \rvert \leq \epsilon^{\theta + 1} n^{-\theta(\theta + 1)/(2\theta + 1)}$,
\begin{equation} \label{eq:defn_v_k(s)}
  v_k(s) = \frac{1}{2\pi i} \oint_{\lvert \zeta \rvert = \epsilon n^{-\theta/(2\theta + 1)}} \frac{v(\zeta)}{\zeta^{k + 1}(1 - s/\zeta^{\theta + 1})} d\zeta = \sum^{\infty}_{l = 0} c_{k + l(\theta + 1)} s^{l}, \quad k = 0, 1, \dotsc, \theta,
\end{equation}
which are analytic near the origin and $v(s)=\sum_{k=0}^\theta s^k v_k(s^{\theta+1})$ for $\lvert s \rvert < \epsilon n^{-\theta/(2\theta + 1)}$. Hence, analogous to \cite[Equation (5.9)]{Wang-Zhang21}, for $k=0,1,\ldots,\theta$, we have
\begin{equation} \label{eq:U_j_asy_simplified}
  \begin{split}
    U_k(z) = {}& \sum^{\theta}_{i = 0} V^{(0)}_i(z) P^{(0)}_{ik}(z) = \sum^{\theta}_{i = 0} v_i(-z) \mathsf{P}^{(0)}_{ik}(z) \\
    = {}& \sum^{\theta}_{l = 0} \frac{v_l(-z)}{(\rho n)^{l/2}} \Phi_{lk}((\rho n)^{\frac{\theta + 1}{2}} z) n^{(\pre)}_k(z) n_k(z)^{-1} e^{n m_k(z)},
  \end{split}
\end{equation}
    where $M$, $N$, $N^{(\pre)}$, $\mathsf{P}^{(0)}$ and $P_0$ are defined in \eqref{def:M}, \eqref{def:N}, \eqref{eq:Npre} \eqref{def:sfP0} and \eqref{def:P0},  respectively, and $\Phi$ is the solution of the RH problem \ref{RHP:general_model}. 
We note that
\begin{align} \label{eq:est_v_k}
  v_0(s) = {}& 1 + \bigO(n^{-\frac{1}{2(2\theta + 1)}}), & v_k(s) = {}& \bigO(n^{\frac{2k\theta - 1}{2(2\theta + 1)}}), \quad \text{for } k = 1, \dotsc, \theta.
\end{align}
By tracing back the transformations $Y \to T \to S \to Q$ and \eqref{eq:U_0_def} and \eqref{eq:U_0_def_2}, we prove \eqref{eq:asy_p_n}.


\paragraph{Proof of \eqref{eq:asy_q_n}}

Similar to \eqref{def:v}, we define
\begin{equation}\label{def:tildev}
  \v(s) =
  \begin{cases}
    \R_1(e^{-\frac{\pi i}{\theta}} s^{1 + \frac{1}{\theta}}), & \arg s \in (0, \frac{\pi}{\theta + 1}), \\
    \R_1(e^{\frac{\pi i}{\theta}} s^{1 + \frac{1}{\theta}}), & \arg s \in (-\frac{\pi}{\theta + 1}, 0), \\
    \R_2(-(-s)^{1 + \frac{1}{\theta}}), & \arg s \in (\frac{\pi}{\theta + 1}, \pi] \cup (-\pi, -\frac{\pi}{\theta + 1}),
  \end{cases}
\end{equation}
for $0 < \lvert s \rvert < \epsilon r^{\theta/(\theta + 1)}_n = \epsilon n^{-\theta/(2\theta + 1)}$ with $r_n$  defined in \eqref{def:rn}.
Since $\R_1$ and $\R_2$ are defined by $\V^{(0)}$ by \eqref{def:tildeR1} and \eqref{def:tildeR2} respectively, and $\V^{(0)}$ satisfies RH problem \ref{rhp:tildeV0}, we find that $\v(s)$ can be extended analytically in the disc $\lvert s \rvert \leq \epsilon n^{-\theta/(2\theta + 1)}$, and thus admits the Taylor expansion $\v(s) = \sum^{\infty}_{k = 0} \c_k s^k$ there. Moreover, as $\epsilon$ is small enough, we see from \eqref{def:tildev}, \eqref{eq:scalar_tildeR_defn} and Lemma \ref{lem:tildetRest} that uniformly for $s$ in this disk
\begin{equation}\label{eq:esttildev}
  \v(s)=1+\bigO(n^{-\frac{1}{2(2\theta+1)}}).
\end{equation}
Then analogous to \cite[Equation (5.24)]{Wang-Zhang21}, we define for $\lvert s \rvert \leq \epsilon^{\theta + 1} n^{-\theta(\theta + 1)/(2\theta + 1)}$,
\begin{equation}\label{eq:defn_tildev_k(s)}
  \v_k(s) = \frac{1}{2\pi i} \oint_{\lvert \zeta \rvert = \epsilon n^{-\theta/(2\theta + 1)}} \frac{\v(\zeta)}{\zeta^{k + 1}(1 - s/\zeta^{\theta + 1})} d\zeta  = \sum^{\infty}_{l = 0} \c_{k + l(\theta + 1)} s^{l}, \quad k = 0, 1, \dotsc, \theta,
\end{equation}
which are analytic near the origin and $\v(s)=\sum_{k=0}^\theta s^k \v_k(s^{\theta+1})$ for $\lvert s \rvert < \epsilon n^{-\theta/(2\theta + 1)}$. Hence, analogous to \cite[Equation (5.9)]{Wang-Zhang21}, for $k=0,1,\ldots,\theta$, we have

\begin{equation}
  \v_0(s) = 1 + \bigO(n^{-\frac{1}{2(2\theta + 1)}}), \quad \v_k(s) = \bigO(n^{\frac{2k\theta - 1}{2(2\theta + 1)}}), \quad k = 1, \dotsc, \theta,
\end{equation}

\begin{equation} \label{eq:tildeQ}
  \begin{split}
    \U_k(z^{\theta}) = {}& \sum^{\theta}_{i = 0} \V^{(0)}_i(z^{\theta}) \widetilde{P}^{(0)}_{i, k}(z^{\theta}) = \sum^{\theta}_{i = 0} \v_i(-z^{\theta}) \Pmodeltilde^{(0)}_{i, k}(z^{\theta}) \\
    = {}& \sum^{\theta}_{l = 0} \frac{\v_l(-z)}{(\rho n)^{l/2}} \Phitilde_{lk}((\rho n)^{\frac{\theta + 1}{2}} z) \n^{(\pre)}_k(z) \n_k(z)^{-1} e^{-n m_k(z)},
  \end{split}
\end{equation}
where $M$, $\N$, $\N^{(\pre)}$, $\widetilde{\mathsf{P}}^{(0)}$ and $\widetilde{P}^{(0)}$ are defined in \eqref{def:M}, \eqref{def:tildeN}, \eqref{eq:Ntildepre}, \eqref{def:tildesfP0} and \eqref{def:tildeP0}, respectively, and $\Phitilde$ is the solution of the RH problem \ref{RHP:general_model_tilde}. 
By tracing back the transformations $\Y \to \T \to \St \to \Q$ and \eqref{eq:Utilde_0_def} and \eqref{eq:Utilde_0_def_2}, we prove \eqref{eq:asy_q_n}.

%

\paragraph{Proof of \eqref{eq:asy_kappa}}

We see from the transformations $Y \to T \to S \to Q$ and \eqref{def:R1} and \eqref{def:R2} that for $z \in \halfH \setminus [0,\infty)$ large enough,
\begin{equation}\label{eq:Y2rep}
  Y_2(z) = R_2(z) P^{(\infty)}_2(z) e^{n(\ell_t - \gfntilde_t(z))},
\end{equation}
where the functions $R_2(z)$, $P^{(\infty)}_2(z)$ and $\gfntilde_t(z)$ are given in \eqref{def:R2}, \eqref{eq:Pinfty_2} and \eqref{eq:gfn_gfntilde_t}, respectively.
As $z \to \infty$, it is readily seen from the definitions of $\gfntilde(z)$ and $P^{(\infty)}_2(z)$, and \eqref{eq:Pinfty_in_one} that
\begin{equation} \label{eq:entglim}
  \lim_{z \to \infty} \frac{e^{n\gfntilde_t(z)}}{z^{n\theta}} = 1, \qquad  \lim_{z \to \infty} z^{\theta} P^{(\infty)}_2(z) = \lim_{s \to 0}(J_{c(t)}(s))^{\theta} \tP(s) = \frac{c(t)^{\alpha + 1} i}{\sqrt{\theta}},
\end{equation}
where $J_{c(t)}$ is defined in \eqref{eq:J_function} with $c(t)$ given in Section \ref{subsec:deformation_algebraic}, and $\P$ is defined in \eqref{eq:Pinfty_in_one}. From \eqref{eq:scalar_R_defn} and Lemma \ref{lem:tRest}, we have
\begin{equation}\label{eq:R2lim}
  \lim_{z \to \infty} R_2(z) = \lim_{z \to \infty} \tR(\Iinv_{2, t}(z)) =\lim_{s \to 0} \tR(s) = 1 + \bigO ( n^{-\frac{1}{2(2\theta + 1)}} ), \quad n\to \infty,
\end{equation}
where $\Iinv_{2, t}$ is defined in Section \ref{subsec:deformation_C1} and $\tR$ is defined in \eqref{eq:scalar_R_defn}. Inserting \eqref{eq:entglim} and \eqref{eq:R2lim} into \eqref{eq:Y2rep} and using \eqref{eq:kapparep}, we  have  \eqref{eq:asy_kappa}.


\subsection{Proof of Theorem \ref{thm:universality}}

In this subsection, we prove the following proposition, which then implies Theorem \ref{thm:universality} directly. 
\begin{prop} \label{prop:universality_in_three}
  Let $V$ satisfy the regularity condition as in Theorem \ref{thm:universality}, and $t = 1 - \sqrt{A_1} \tau/\sqrt{n}$, where $\tau$ is in a compact subset of $\realR$ and $A_1$ is defined in \eqref{eq:defn_rho}. Let $x, y$ be in a compact subset of $[0, +\infty)$. Suppose $\epsilon > 0$ is any constant.
  \begin{enumerate}
  \item
    Let $M$ be a positive constant. Then, we have 
    \begin{multline} \label{eq:universality_left_part}
      \lim_{n \to \infty} \frac{1}{(\rho n)^{\frac{\theta + 1}{2\theta}}} x^{\alpha} e^{-nV_t(x)}\sum^{n - 1}_{j = \lfloor n - M n^{1/2} \rfloor} \frac{1}{\kappa^{(V_t)}_{n, j}}p^{(V_t)}_{n, j} \left( \frac{x}{(\rho n)^{\frac{\theta + 1}{2\theta}}} \right) q^{(V_t)}_{n, j} \left( \frac{y^{\theta}}{(\rho n)^{\frac{\theta + 1}{2}}} \right) = \\
      \frac{\theta}{2\pi} x^{\alpha} \int^{\tau +M/\sqrt{A_1}}_{\sigma = \tau} \phi^{(\sigma)}(x) \phitilde^{(\sigma)}(y) d\sigma.
    \end{multline}
  \item
    Let $\delta = \delta(\epsilon) > 0$ be a small enough constant. Then for all large enough $n$, we have 
    \begin{equation} \label{eq:universality_middle_part}
      \left\lvert \frac{1}{(\rho n)^{\frac{\theta + 1}{2\theta}}} x^{\alpha} e^{-nV_t(x)}\sum^{\lfloor \delta n \rfloor}_{j = 0}  \frac{1}{\kappa^{(V_t)}_{n, j}}p^{(V_t)}_{n, j} \left( \frac{x}{(\rho n)^{\frac{\theta + 1}{2\theta}}} \right) q^{(V_t)}_{n, j} \left( \frac{y^{\theta}}{(\rho n)^{\frac{\theta + 1}{2}}} \right) \right\rvert < \epsilon.
    \end{equation}
  \item
    For any $\delta > 0$, if $M = M(\epsilon)$ is large enough, then for all large enough $n$, we have
    \begin{equation} \label{eq:universality_right_part}
      \left\lvert \frac{1}{(\rho n)^{\frac{\theta + 1}{2\theta}}} x^{\alpha} e^{-nV_t(x)}\sum^{\lfloor n - M n^{1/2} \rfloor - 1}_{j = \lfloor \delta n \rfloor + 1} \frac{1}{\kappa^{(V_t)}_{n, j}}p^{(V_t)}_{n, j} \left( \frac{x}{(\rho n)^{\frac{\theta + 1}{2\theta}}} \right) q^{(V_t)}_{n, j} \left( \frac{y^{\theta}}{(\rho n)^{\frac{\theta + 1}{2}}} \right) \right\rvert < \epsilon.
    \end{equation}
    
  \end{enumerate}
\end{prop}

\subsubsection{Proof of \eqref{eq:universality_left_part}}

We note that for $j = 1, \dotsc, n - 1$,
\begin{align} \label{eq:scaling_argument}
  p^{(V_t)}_{n, j}(x) = {}& p^{(V_{\frac{j}{n} t})}_{j, j}(x), & q^{(V_t)}_{n, j}(x) = {}& q^{(V_{\frac{j}{n} t})}_{j, j}(x), & \kappa^{(V_t)}_{n, j} = {}& \kappa^{(V_{\frac{j}{n} t})}_{j, j}.
\end{align}
Let $M$ be a positive constant and $m  \in [ \lfloor n - M n^{1/2} \rfloor, n - 1]$. 
Recalling $t=1 - \frac{\sqrt{A_1} \tau}{\sqrt{n}}$, we have for large $n$
\begin{equation}\label{eq:tau_m}\frac{m}{n}t= \frac{m}{n}\left( 1 - \frac{\sqrt{A_1} \tau}{\sqrt{n}}\right)=1-\frac{\sqrt{A_1} }{\sqrt{m}}\left(\tau+\frac{n-m}{\sqrt{A_1n}}\right)+\bigO(n^{-1}),\end{equation}
and 
\begin{equation}\label{eq:x_m}
\frac{x}{(\rho n)^{\frac{\theta + 1}{2\theta}}}=\frac{x}{(\rho m)^{\frac{\theta + 1}{2\theta}}} (1+\bigO(n^{-1/2})), \quad 
 \frac{y^{\theta}}{(\rho n)^{\frac{\theta + 1}{2}}}= \frac{y^{\theta}}{(\rho m)^{\frac{\theta + 1}{2}}}(1+\bigO(n^{-1/2})).\end{equation}
From Theorem \ref{thm:main} and \eqref{eq:scaling_argument}-\eqref{eq:x_m}, we have as $n\to\infty$
\begin{multline} \label{eq:individual_pq_product}
 \frac{1}{(\rho n)^{\frac{\theta + 1}{2\theta}}} \frac{1}{\kappa^{(V_t)}_{n, m}} p^{(V_t)}_{n, m} \left( \frac{x}{(\rho n)^{\frac{\theta + 1}{2\theta}}} \right) q^{(V_t)}_{n, m} \left( \frac{y^{\theta}}{(\rho n)^{\frac{\theta + 1}{2}}} \right) \left( \frac{x}{(\rho n)^{\frac{\theta + 1}{2\theta}}} \right)^{\alpha} \exp \left[ -nV \left( \frac{x}{(\rho n)^{\frac{\theta + 1}{2\theta}}} \right) \right] = \\
  \frac{\theta}{2\pi} x^{\alpha} \frac{1}{\sqrt{A_1 n}} \left( \phi^{(\tau+\frac{n-m}{\sqrt{A_1n}})}(x) \phitilde^{(\tau+\frac{n-m}{\sqrt{A_1n}})}(y) + \bigO(n^{-\frac{1}{2\theta + 1}}) \right),
\end{multline}
and the error terms are uniform for  $m  \in [ \lfloor n - M n^{1/2} \rfloor, n - 1]$. Since $\phi^{(\sigma)}(x)$ and $\phitilde^{(\sigma)}(y)$ are continuous for $x,y\in [0,\infty)$ and $\sigma\in\mathbb{R}$, the summation of the leading term on the right-hand side of \eqref{eq:individual_pq_product}  is a Riemann sum of the definite integral on the right-hand side of \eqref{eq:universality_left_part}. Therefore, noting that the  summation  of the error terms with the index $m$ running from $\lfloor n - M n^{1/2} \rfloor$ to $n - 1$ is of order $\bigO(n^{-1/(2\theta + 1)})$, we have \eqref{eq:universality_left_part}. 

\subsubsection{Proof of \eqref{eq:universality_middle_part}}

Since $V_t$ satisfies Regularity Conditions \ref{reg:primary} and \ref{reg:soft-edge}, for all $t \in (0, 1)$, the solution of the Euler-Lagrange equations \eqref{eq:E-L_1} and \eqref{eq:E-L_2}, with $V$ replaced by $V_t$, is a measure supported on an interval $[\aend(t), \bend(t)]$.

We denote $v_{\min} = \min \{ V(x) : x \in [0, +\infty) \}$. The regularity condition implies that $V(0) > v_{\min}$. To see this, by Regularity Condition \ref{reg:soft-edge}, we have that for any $t \in (0, 1)$, $V_t(\aend(t)) < V_t(0)$,  which is equivalent to $V(\aend(t)) < V(0)$.

Next, we show that for any $\epsilon > 0$, there is $T_{\epsilon}$ such that if $t \in (0, T_{\epsilon})$, then $\max \{ V(x) : x \in [\aend(t), \bend(t)] \} < v_{\min} + \epsilon$. This implies that as $t \to 0_+$, $\bend(t) - \aend(t)$ converges to $0$. To see this, we suppose the opposite is true. Then  there exists a sequence $t_k \to 0_+$, such that  $v_k := \max \{ V(x) : x \in [\aend(t_k), \bend(t_k)] \} > v_{\min} + \epsilon$. We denote
\begin{equation} \label{eq:eq_measure_t}
  F_t(x; \nu) = \int \log \lvert x - y \rvert^{-1} d\nu(y) + \int \log \lvert x^{\theta} - y^{\theta} \rvert^{-1} d\nu(y) + V_t(x).
\end{equation}
Then the equilibrium measure $\mu^{(V_{t_k})}$ is supported on $[\aend(t_k), \bend(t_k)]$, and satisfies
\begin{equation}
  F_{t_k}(x; \mu^{(V_{t_k})}) = -\ell_{t_k}, \quad x \in [\aend(t_k), \bend(t_k)].
\end{equation}
Moreover, we have that 
\begin{equation}
  I^{(V_{t_k})}(\mu^{(V_{t_k})}) = \int F_{t_k}(x; \mu^{(V_{t_k})}) d\mu^{(V_{t_k})}(x) = -\ell_{t_k}
\end{equation}
is the minimum of $\{ I^{(V_t)}(\nu) : \nu \text{ is a probability measure on } [0, +\infty) \}$. However, we have
\begin{equation}
  \begin{split}
    F_{t_k}(x; \mu^{(V_{t_k})}) \geq {}& \max_{x \in [\aend(t_k), \bend(t_k)]} V_{t_k}(x) + \log (\bend(t_k) - \aend(t_k))^{-1} + \log (\bend(t_k)^{\theta} - \aend(t_k)^{\theta})^{-1} \\
    = {}& \frac{v_k}{t_k} - \log [(\bend(t_k) - \aend(t_k)) (\bend(t_k)^{\theta} - \aend(t_k)^{\theta})].
  \end{split}
\end{equation}
When $t_k$ is small enough, we have $F_{t_k}(x; \mu^{(V_{t_k})}) > v_{\min}/t_k + \epsilon/t_k$. On the other hand, it is easy to see that we can construct a probability measure $\nu$ with $\max \{ V(x) : x \in \supp(\nu) \} < v_{\min} + \epsilon/2$, and then when $t_k$ is small enough, $F_{t_k}(x; \mu^{(V_{t_k})}) < v_{\min}/t_k + \epsilon/t_k$ by a rough estimate of \eqref{eq:eq_measure_t}. This is contradictory to the minimal property of $\mu^{(V_{t_k})}$, and the claim is proved.

Hence, we have a $t^* \in (0, 1)$ such that $\max \{ V(x) : x \in [\aend(t^*), \bend(t^*)] \} < v_{\min} + \epsilon$. By choosing a possibly smaller $\epsilon$ and $t^*$, we can assume that for some $\delta' > 0$, $\max \{ V(x) : x \in [\aend(t^*) - 2\delta', \bend(t^*) + 2\delta'] \} < v_{\min} + \epsilon$, and $\min \{ V(x) : x \in [0, \delta'] \} > v_{\min} + 2\epsilon$.

Because functions $\{ x^{k\theta} \}^{n - 1}_{k = 1}$ form a Chebyshev system in the sense of \cite[First definition in Section 4.4]{Nikishin-Sorokin91}, all the zeros of $p_k(x)$ in $\compC$ are on $\realR_+$, and all the zeros of $q_k(x^{\theta})$ in $\halfH$ are on $\realR_+$.

By taking the transformations $Y \to \That \to \Shat \to \Qhat$, and using \eqref{eq:Rhat_1} and \eqref{eq:Rhat_2}, we find that for all $x \in [0, \aend(t) - \delta'] \cup [\bend(t) + \delta', +\infty)$, $p^{(V_{t^*})}_{n, n}(x) \neq 0$ if $n$ is large enough. To see this, we note that for $x \in (\rhat_n, \aend(t) - \delta'] \cup [\bend(t) + \delta', +\infty)$,
\begin{equation}
  p^{(V_{t^*})}_{n, n}(x) = e^{n \gfnhat_{t^*}(x)} \Phat^{(\infty)}_1(x) \Rhat_1(x),
\end{equation}
where $\rhat_n$, $\Phat^{(\infty)}_1(x)$ and $\Rhat_1(x)$ are defined in \eqref{eq:defn_rhat_n}, \eqref{eq:Phat^infty_1} and \eqref{eq:Rhat_1} with $t$ replaced by $t^*$, and it has a slightly different formula on $[0, \rhat_n]$. By the formulas of $\Phat^{(\infty)}_1(x)$ and $\Rhat_1(x)$, we find that they have no zeros there, and prove this claim. By a simple scaling argument, we find that if $n$ is large enough, $p^{(V_t)}_{n, m}(x)$ is nonzero on $[0, \aend(t) - \delta'] \cup [\bend(t) + \delta', +\infty)$ if $m/n = t^*$. Then by the interlacing property of $p^{(V_t)}_{n, m}$, we have that for all $m < n t^*$, the zeros of $p^{(V_t)}_{n, m}(x)$ lie in $[\aend(t) - \delta', \bend(t) + \delta']$. This argument also applies to $q^{(V_t)}_{n, m}(x^{\theta})$.

For $x \in [0, \delta']$, if $m/n$ is small enough, we have 
\begin{align}
  \lvert p^{(V_t)}_{n, m}(x) \rvert \leq {}& (\bend(t) + 2\delta')^m < e^{\frac{1}{4} \epsilon n}, & \lvert q^{(V_t)}_{n, m}(x) \rvert \leq {}& (\bend(t) + 2\delta')^{(\theta + 1)m} < e^{\frac{1}{4} \epsilon n},
\end{align}
and, if $p^{(V_t)}_{n, m}(x) = \prod^m_{k = 1} (x - c_k)$, we define $\check{p}^{(V_t)}_{n, m}(x) = \prod^m_{k = 1} (x^{\theta} - c^{\theta}_k)$, and have
\begin{equation}
  \begin{split}
    \kappa^{(V_t)}_{n, m} = {}& \int^{\infty}_0 p^{(V_t)}_{n, m}(x) q^{(V_t)}_{n, m}(x) x^{\alpha} e^{-nV_t(x)} dx \\
    = {}& \int^{\infty}_0 p^{(V_t)}_{n, m}(x) \check{p}^{(V_t)}_{n, m}(x) x^{\alpha} e^{-nV_t(x)} dx \\
    \geq {}& \int^{\bend(t) + 2\delta'}_{\bend(t) + \delta'} p^{(V_t)}_{n, m}(x) \check{p}^{(V_t)}_{n, m}(x) x^{\alpha} e^{-nV_t(x)} dx \\
    \geq {}& \int^{\bend(t) + 2\delta'}_{\bend(t) + \delta'} (x - \bend(t) - \delta')^m (x^{\theta} - (\bend(t) + \delta')^{\theta})^m x^{\alpha} e^{-nV_t(x)} dx,
  \end{split}
\end{equation}
and if $m/n$ is small enough,
\begin{equation}
  \kappa^{(V_t)}_{n, m} > e^{-n(v_{\min} + \epsilon)}.
\end{equation}
In conclusion, we have that for $x, y \in [0, \delta']$,
\begin{equation}
  \frac{1}{\kappa^{(V_t)}_{n, m}} \lvert p^{(V_t)}_{n, m}(x) q^{(V_t)}_{n, m}(y) \rvert x^{\alpha} e^{-nV_t(x)} < e^{-\frac{1}{2} \epsilon n}.
\end{equation}
As a consequence, we have \eqref{eq:universality_middle_part}. 

\subsubsection{Proof of \eqref{eq:universality_right_part}}

Let $x = \xi (\rho n)^{-\frac{\theta + 1}{2\theta}}$ and $y = \eta (\rho n)^{-\frac{\theta + 1}{2\theta}}$. 
Let $C$ be a large constant with $C > \delta^{-1}$. Then,  for any $j \in [C^{-1} n, n - C n^{1/2}]$, we have 
$\frac{j}{n}t\in [C^{-1}, 1 - C n^{-1/2}]$. Therefore, from Lemma \ref{lem:tail} and \eqref{eq:scaling_argument}, we have that  as $n \to \infty$, for some $C' > 0$ depending on $C$,
\begin{multline}\label{eq:tail_expand}
  \left\lvert \frac{1}{(\rho n)^{\frac{\theta + 1}{2\theta}}} \frac{1}{\kappa^{(V_t)}_{n, j}} x^{\alpha} e^{-nV_t(x)} p^{(V_t)}_{n, j}(x) q^{(V_t)}_{n, j}(y^{\theta}) \right\rvert \\
  \leq C' \xi^{\alpha} (\rho n)^{-\frac{\theta + 1}{2\theta} (\alpha + 1)} \left( 1 - \frac{j}{n} t \right)^{-(\alpha + \frac{\alpha + 1}{\theta})} e^{j(\gfnhat_{\frac{j}{n} t}(x) + \gfnhat_{\frac{j}{n} t}(y) - V_{\frac{j}{n} t}(x) - \ellhat_{\frac{j}{n} t})}.
\end{multline}
For $j \in [C^{-1} n, n - C n^{1/2}]$, we have, for some $C'' > 0$ depending on $C$,
\begin{equation}\label{eq:tail_est1}
  (\rho n)^{-\frac{\theta + 1}{2\theta} (\alpha + 1)} \left( 1 - \frac{j}{n} t \right)^{-(\alpha + \frac{\alpha + 1}{\theta})} \leq C'' n^{-\frac{1}{2}}.
\end{equation}
By Proposition \ref{prop:phihat} and its generalization in Proposition \ref{prop:extension_t}, we find that
\begin{equation}\label{eq:tail_est2}
  \gfnhat_{\frac{j}{n} t}(x) + \gfnhat_{\frac{j}{n} t}(y) - V_{\frac{j}{n} t}(x) - \ellhat_{\frac{j}{n} t}\leq - \epsilon'' \left( 1 - \frac{j}{n} t \right)^2.
\end{equation}
Substituting \eqref{eq:tail_est1} and \eqref{eq:tail_est2} into \eqref{eq:tail_expand}, we obtain \eqref{eq:universality_right_part}, if $M > C$ is large enough.

\section{Solvability of RH problem \ref{RHP:general_model} via vanishing lemma} \label{sec:vanishing_lemma}

In this section, we show that RH problem \ref{RHP:general_model} has a unique solution. More precisely, we prove the following result.

\begin{prop} \label{lem:existence_H_to_S}
  For every $\tau \in \realR$, the following holds:
  \begin{enumerate}
  \item
    RH problem \ref{RHP:general_model} is uniquely solvable.
  \item
    The solution $\Phi(\xi) = \Phi^{(\tau)}(\xi)$ has a full asymptotic expansion in powers of $\xi^{-1}$ as follows:
    \begin{equation} \label{eq:asy_expansion_H_to_S}
      \Phi^{(\tau)}(\xi) \sim \left( I + \sum^{\infty}_{k = 1} M_k \xi^{-k} \right) \diag \left( e^{-\frac{k}{\theta + 1} \pi i} \xi^{\frac{k}{\theta + 1}} \right)^{\theta}_{k = 0} \Omega_{\pm} e^{-\Theta(\xi)},
    \end{equation}
    as $\xi \to \infty$, uniformly in $\compC \setminus \{ \realR \cup i\realR \}$. Here $M_k = M_k(\tau)$ are real valued and depending on $\tau$ analytically.
  \end{enumerate}
\end{prop}

We remark that Proposition \ref{lem:existence_H_to_S} is only proved for $\tau \in \realR$.

To prove this proposition, we follow the strategy in \cite[Section 5.3]{Deift-Kriecherbauer-McLaughlin-Venakides-Zhou99} and \cite[Section 2.2]{Claeys-Vanlessen07a}. We transform RH problem \ref{RHP:general_model} into an equivalent RH problem for $\Phihat$ such that the jump matrix for $\Phihat$ is continuous on the jump contour $\hat{\Sigma}$ and converges exponentially to the identity matrix as $\xi \to \infty$ on $\hat{\Sigma}$ and such that the RH problem for $\Phihat$ is normalized at infinity. To this end, we define
\begin{equation}\label{eq:Transform}
  \Phihat(\xi) =
  \begin{cases}
    \Phi(\xi) e^{\Theta(\xi)} e^{-\Lambda(\xi)} \Phi^{(\Mei)}(\xi)^{-1}, & \lvert \xi \rvert > 1, \\
    \Phi(\xi) \Phi^{(\Mei)}(\xi)^{-1}, & \lvert \xi \rvert < 1,
  \end{cases}
\end{equation}
for $\xi \in \compC \setminus \{ \realR \cup i\realR \}$, where $\Phi^{(\Mei)}$ is defined in Appendix \ref{sec:Phi^Mei}. Then $\Phihat(\xi)$ satisfies the following RH problem:
\begin{RHP} \label{RHP:Phihat}
  $\Phihat(\xi)$ is a $(\theta + 1) \times (\theta + 1)$ matrix-valued function on $\compC \setminus \hat{\Sigma}$, where
  \begin{align}
      \hat{\Sigma} = {}& \sum^2_{j = 0} \hat{\Sigma}_j, & \hat{\Sigma}_0 = {}& \{ e^{2\pi i t} : t \in (0, 1] \}, & \hat{\Sigma}_1 = {}& \{ i t: t \in (1, +\infty) \}, & \hat{\Sigma}_2 = {}& \{ i t: t \in (-\infty, -1) \},
  \end{align}
  with the orientation of $\hat{\Sigma}$ shown in Figure \ref{fig:hatSigma}. It satisfies
  \begin{enumerate}
  \item \label{enu:RHP:Phihat:1}
    For $\xi \in \hat{\Sigma}$,
    \begin{multline}
      \Phihat_+(\xi) = \Phihat_-(\xi) J_{\Phihat}(\xi), \quad \text{where} \\
      J_{\Phihat}(\xi) =
      \begin{cases}
        \Phi^{(\Mei)}(\xi) e^{\Lambda(\xi) - \Theta(\xi)} \Phi^{(\Mei)}(\xi)^{-1}, & \xi \in \hat{\Sigma}_0, \\
        (\Phi^{(\Mei)}(\xi) e^{\Lambda(\xi)})_-
        \left[
        \begin{pmatrix}
          1 & e^{(*)} \\
          0 & 1
        \end{pmatrix}
        \oplus I_{\theta - 1} \right] (\Phi^{(\Mei)}(\xi) e^{\Lambda(\xi)})^{-1}_{+}, & \xi \in \hat{\Sigma}_1, \\
        (\Phi^{(\Mei)}(\xi) e^{\Lambda(\xi)})_-
        \left[
        \begin{pmatrix}
          1 & -e^{-(*)} \\
          0 & 1
        \end{pmatrix}
        \oplus I_{\theta - 1} \right] (\Phi^{(\Mei)}(\xi) e^{\Lambda(\xi)})^{-1}_+, & \xi \in \hat{\Sigma}_2,
      \end{cases}
    \end{multline}
    with
      \begin{equation}
      (*) = \beta \pi i + \frac{\theta + 1}{\theta} \sin \left( \frac{2\pi}{\theta + 1} \right) i \xi^{\frac{2}{\theta + 1}} -2\tau \sin \left( \frac{\pi}{\theta + 1} \right) i \xi^{\frac{1}{\theta + 1}}.
    \end{equation} 
  \item \label{enu:RHP:Phihat:3}
    $\Phihat(\xi)$ is continuous up to boundary on $\hat{\Sigma}$, and
    \begin{equation}\label{eq:Phihat_infty}
      \Phihat(\xi) = I + \bigO(\xi^{-1}), \quad \xi \to \infty.
    \end{equation}
  \end{enumerate}
\end{RHP}

\begin{figure}[htb]
  \begin{minipage}[t]{0.45\linewidth}
  \centering
  \includegraphics{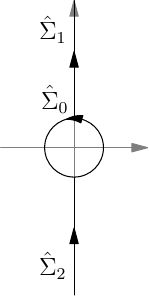}
  \caption{Contour $\hat{\Sigma}$ in RH problem \ref{RHP:Phihat}.}
  \label{fig:hatSigma}
  \end{minipage}
  \hspace{\stretch{1}}
  \begin{minipage}[t]{0.45\linewidth}
    \centering
    \includegraphics{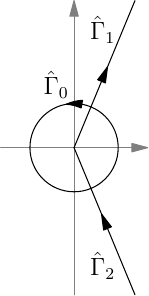}
    \caption{Contour $\hat{\Gamma}$ ($\theta = 3$) in the proof of Lemma \ref{lem:finite_zeros}.}
    \label{fig:Gamma}
  \end{minipage}
\end{figure}
  


 From the relation  \eqref{eq:Transform}, the solvability of RH problem \ref{RHP:general_model} is equivalent to the solvability of RH problem \ref{RHP:Phihat}. By general theory of the construction of solutions of RH problems, this is reduced to the study of the singular integral operator on the $L^2$ space of $(\theta + 1) \times (\theta + 1)$ matrix-valued functions on $\hat{\Sigma}$
  \begin{equation}
    C_{\Phihat}: L^2(\hat{\Sigma}) \to L^2(\hat{\Sigma}): \quad f \mapsto C_+[f(1 - J^{-1}_{\Phihat})],
  \end{equation}
  where $C_+$ is the $+$ boundary value of the Cauchy operator
  \begin{equation}
    Cf(z) = \frac{1}{2\pi i} \int_{\hat{\Sigma}} \frac{f(s)}{s - z} ds, \quad z \in \compC \setminus \hat{\Sigma}.
  \end{equation}
  Suppose that $I - C_{\Phihat}$ is invertible in $L^2(\hat{\Sigma})$, then there exists $\mu \in L^2(\hat{\Sigma})$ such that $(I - C_{\Phihat}) \mu = C_+(I - J^{-1}_{\Phihat})$, and
  \begin{equation}\label{eq:Integral_Rep_Phi}
    \Phihat(\xi) = I + \frac{1}{2\pi i} \int_{\hat{\Sigma}} \frac{(I + \mu(s))(I - J_{\Phihat}(s)^{-1})}{s - \xi} ds, \quad \xi \in \compC \setminus\hat{\Sigma}
  \end{equation}
  satisfies part \ref{enu:RHP:Phihat:1} of RH problem \ref{RHP:Phihat} in the $L^2$ sense. Furthermore, one can use the analyticity of $J_{\Phihat}$ to show that $\Phihat$ satisfies the jump condition in the sense of continuous boundary values as well, see \cite[Step 3 of Sections 5.2 and 5.3]{Deift-Kriecherbauer-McLaughlin-Venakides-Zhou99}. Also it follows from the exponential decaying of $I - J^{-1}_{\Phihat}$ as $\xi \to \infty$ on $\hat{\Sigma}$ that part \ref{enu:RHP:Phihat:3} of RH problem \ref{RHP:Phihat} is also satisfied, see \cite[Proposition 5.4]{Deift-Kriecherbauer-McLaughlin-Venakides-Zhou99}. Hence, RH problem \ref{RHP:Phihat} is solvable if the singular integral operator $I - C_{\Phihat}$ is invertible in $L^2(\hat{\Sigma})$.

  To show that $I - C_{\Phihat}$ is invertible, we need to check that it is a Fredholm operator with zero index and kernel $\{ 0 \}$. As in \cite[Steps 1 and 2 of Section 5.3]{Deift-Kriecherbauer-McLaughlin-Venakides-Zhou99}, one can show that $I - C_{\Phihat}$ is a Fredholm operator with zero index. 
  To show that its kernel is trivial, we  suppose $\mu_0 \in L^2(\hat{\Sigma})$ is such that $(I - C_{\Phihat}) \mu_0 = 0$. One can then show that the matrix-valued function $\Phihat_0$ defined by
  \begin{equation}
    \Phihat_0(\xi) = \frac{1}{2\pi i} \int_{\hat{\Sigma}}\frac{\mu_0(s)(I - J_{\Phihat}(s)^{-1})}{s - \xi} ds, \quad \xi \in \compC \setminus \hat{\Sigma}
  \end{equation}
  is a solution to the homogeneous version of RH problem \ref{RHP:Phihat}, that is, with the condition \eqref{eq:Phihat_infty} replaced by
  \begin{equation}
    \Phihat(\xi) = \bigO(\xi^{-1}), \quad \xi \to \infty.
  \end{equation}
 Hence, to show the solvability of RH problem \ref{RHP:Phihat},  we only need to show that the associated homogeneous version of RH problem has only a trivial solution. 
  Equivalently, to show the solvability of the RH problem \ref{RHP:general_model} ,  we only need to verify the following vanishing lemma.

\begin{lemma}[vanishing lemma] \label{lem:vanishing}
  Let $\Phi_0$ be analytic on $\compC \setminus \{ \realR \cup i\realR \}$ that satisfies the homogeneous version of RH problem \ref{RHP:general_model}, such that parts \ref{enu:RHP:general_model:1} and \ref{enu:RHP:general_model:3} are the same while part \ref{enu:RHP:general_model:2} is changed to
  \begin{equation} \label{eq:Psi_hardtosoft_at_infty_homo}
    \Phi_0(\xi) = \bigO(\xi^{-1}) \Upsilon(\xi) \Omega_{\pm} e^{-\Theta(\xi)}.
  \end{equation}
  Then $\Phi_0(\xi) = 0$.
\end{lemma}
 
\begin{proof}[Proof of Proposition \ref{lem:existence_H_to_S} assuming  Lemma \ref{lem:vanishing}]
  From the discussion above, we have the solvability of RH problem \ref{RHP:general_model} for $\tau\in\mathbb{R}$. The uniqueness is easier. Suppose $\Phi(\xi)$ and $\Phicheck(\xi)$ both satisfy RH problem \ref{RHP:general_model}, then $\Phicheck(\xi) \Phi^{-1}(\xi)$ is analytic on $\compC$ and it is $I + \bigO(\xi^{-1/(\theta + 1)})$ at $\infty$. Hence, by Liouville's theorem, $\Phicheck(\xi) \Phi^{-1}(\xi) = I$ and the uniqueness is proved.

  Moreover, from the analytic Fredholm alternative theorem \cite[Proposition 4.3]{Zhou89},  it follows that the solution $\Phi(\xi) = \Phi^{(\tau)}(\xi)$ is meromorphic for $\tau\in\mathbb{C}$ and pole free on the real axis. From \eqref{eq:Transform}, \eqref{eq:Integral_Rep_Phi} and the exponential decaying of $I - J^{-1}_{\Phihat}$ as $\xi \to \infty$ on $\realR \cup i\realR$, we have the asymptotic expansion \eqref{eq:asy_expansion_H_to_S} uniformly for $\xi$ bounded away from the integral contour $\realR \cup i\realR$. The coefficients $M_k(\tau)$ in  \eqref{eq:asy_expansion_H_to_S} are analytic for  $\tau\in\mathbb{R}$, since $\Phi^{(\tau)}(\xi)$ is analytic  for $\tau\in\mathbb{R}$. 
  
  For $\tau\in\mathbb{R}$, it is straightforward to check that the function 
  \begin{equation} \label{eq:Phi_conjugate}
    \overline{ \Phi^{(\tau)}(\bar{\xi})} (I_2\oplus J_{\theta - 1})
  \end{equation}
  also solves RH problem \ref{RHP:general_model}. Here $J_{\theta - 1}$ is like $J_{\theta + 1}$ in \eqref{eq:defn_Xicheck}, the row-reversed identity matrix defined in the end of Section \ref{sec:introduction}. From the uniqueness of the solution to  RH problem \ref{RHP:general_model}, we have
  \begin{equation} \label{eq:Phi_symmetry}
    \Phi^{(\tau)}(\xi)=\overline{ \Phi^{(\tau)}(\bar{\xi})} (I_2\oplus J_{\theta - 1}).
  \end{equation}
  Substituting the large-$\xi$ expansion \eqref{eq:asy_expansion_H_to_S}  into the above equation and comparing the coefficients, we find 
  \begin{equation} 
    M_k(\tau)=\overline{ M_k(\tau)},
  \end{equation}
  for real $\tau$. Thus, the coefficients $M_k(\tau)$ are real valued for real $\tau$.
  We complete the proof of  Proposition \ref{lem:existence_H_to_S}.
\end{proof}

\subsection{Proof of  vanishing lemma }

Let $\Phi_0$ be a nontrivial solution of the homogeneous version of RH problem \ref{RHP:general_model}. We assume that $\Phi_0$ satisfies
\begin{equation} \label{eq:realty_Phi_0}
  \Phi_0(\xi) = \overline{\Phi_0(\bar{\xi})} (I_2\oplus J_{\theta - 1}).
\end{equation}
Otherwise, if $\Phi_0(\xi)$ is a nontrivial but does not satisfy \eqref{eq:realty_Phi_0}, we consider $\Phi_0(\xi) + \overline{\Phi_0(\bar{\xi})} (I_2\oplus J_{\theta - 1})$ and $i(\Phi_0(\xi) - \overline{\Phi_0(\bar{\xi})} (I_2\oplus J_{\theta - 1}))$. We have that both of them are solutions of the homogeneous version of RH problem \ref{RHP:general_model}, and at least one between them is nontrivial. We can take a nontrivial one as our $\Phi_0$.

We note that to prove the vanishing result in Lemma \ref{lem:vanishing}, we can consider the rows of $\Phi_0$ separately, and need only to show that any of its rows are zero vectors. Before giving the proof, we fix some notation and state a technical lemma that will be used in the proof.

Let $(\phi_0, \phi_1, \dotsc, \phi_{\theta})$ be one row of $\Phi_0$. Let $f^{(\pre)}(z)$ be the analytic function on $\compC \setminus [0, +\infty)$ such that
\begin{equation}\label{eq: f_Pre_definition}
  f^{(\pre)}(z) =
  \begin{cases}
    \phi_1(z^{\theta}), & \arg z \in (-\frac{\pi}{\theta}, -\frac{\pi}{2\theta}) \cup (\frac{\pi}{2\theta}, \frac{\pi}{\theta}), \\
    \phi_j(z^{\theta}), & \arg z \in (\frac{(2j - 3)\pi}{\theta}, \frac{(2j - 1)\pi}{\theta}), \quad j = 2, \dotsc, \theta, \\
    \phi_1(z^{\theta}) + e^{\beta \pi i} \phi_0(z^{\theta}), & \arg z \in (0, \frac{\pi}{2\theta}), \\
    \phi_1(z^{\theta}) + e^{-\beta \pi i} \phi_0(z^{\theta}), & \arg z \in (-\frac{\pi}{2\theta}, 0),
  \end{cases}
\end{equation}
and then let $f(z)$ be the analytic function on $\compC \setminus \{ 0 \}$ such that
\begin{equation}\label{eq: f_definition}
  f(z) = f^{(\pre)}(z) \times
  \begin{cases}
    e^{-\frac{\beta \pi i}{2}} z^{\frac{\beta}{2}}, & z \in \compC_+, \\
    e^{\frac{\beta \pi i}{2}} z^{\frac{\beta}{2}}, & z \in \compC_-.
  \end{cases}
\end{equation}
Similarly, let $g(z)$ be the analytic function on the sectors $\arg z \in (0, \pi/\theta) \cup (-\pi/\theta, 0)$ such that
\begin{equation}\label{eq:g_defnition}
  g(z) = 
  \begin{cases}
    e^{\frac{\beta \pi i}{2}} z^{\frac{\beta}{2}} \phi_0(z^{\theta}), & \arg z \in (0, \frac{\pi}{\theta}) \\
    -e^{-\frac{\beta \pi i}{2}} z^{\frac{\beta}{2}} \phi_0(z^{\theta}), & \arg z \in (-\frac{\pi}{\theta}, 0).
  \end{cases}
\end{equation}
We have that for all $x \in (0, +\infty)$,
\begin{align} \label{eq:g_and_f_relation}
  g_+(x) - g_-(x) = {}& f(x), & g(e^{\frac{\pi i}{\theta}} x) = {}& e^{-\frac{\alpha + 1}{\theta} 2\pi i} g(e^{-\frac{\pi i}{\theta}}x). 
\end{align}
We have as $z \to \infty$, 
\begin{align}
  f(z) = {}& \bigO(z^{\frac{\beta}{2} - \frac{\theta}{\theta + 1}}) e^{-\frac{\theta + 1}{2\theta} (-z)^{\frac{2\theta}{\theta + 1}} - \tau (-z)^{\frac{\theta}{\theta + 1}}}, \\
  g(z) = {}& \bigO(z^{\frac{\beta}{2} - \frac{\theta}{\theta + 1}}) \times
             \begin{cases}
               e^{-\frac{\theta + 1}{2\theta} e^{\frac{-2\pi i}{\theta + 1}} z^{\frac{2\theta}{\theta + 1}} + \tau e^{\frac{-\pi i}{\theta + 1}} z^{\frac{\theta}{\theta + 1}}}, & \arg z \in (0, \frac{\pi}{\theta}), \\
               e^{-\frac{\theta + 1}{2\theta} e^{\frac{2\pi i}{\theta + 1}} z^{\frac{2\theta}{\theta + 1}} + \tau e^{\frac{\pi i}{\theta + 1}} z^{\frac{\theta}{\theta + 1}}}, & \arg z \in (-\frac{\pi}{\theta}, 0),
             \end{cases}
\end{align}
and as $z \to 0$, 
\begin{align}
  f(z) = {}& \bigO(1), & g(z) = {}&
                                    \begin{cases}
                                      \bigO(z^{\theta - \alpha - 1}), & \alpha > \theta - 1, \\
                                      \bigO(1), & \alpha < \theta - 1, \\
                                      \bigO(\log z), & \alpha = \theta - 1.
                                    \end{cases}
\end{align}

Hence $f(z)$ is an entire function. We have the following lemma about $f(z)$, whose proof will be given in the end of this section.
\begin{lemma} \label{lem:finite_zeros}
  If $f(x)$ is not identically $0$, then it has only finitely many zeros on the $(\theta - 1)$ rays
  \begin{equation}
    \Gamma_j := \{ e^{\frac{2j\pi i}{\theta}} t : t \in (0, +\infty) \}, \quad j = 1, 2, \dotsc, \theta - 1.
  \end{equation}
\end{lemma}

\begin{proof}[Proof of Lemma \ref{lem:vanishing}]
  We prove the vanishing lemma by showing that $f(z)=0$. Below we assume that $f(z)$ is a non-trivial entire function, and derive a contradiction.

  For this purpose, we denote, for $r \in (0, +\infty)$, $\Sigma_r$ the wedge contour 
  \begin{align}
    \Sigma_r = {}& \bigcup^2_{i = 0} \Sigma^{(i)}_r, & \Sigma^{(0)}_r = {}& \{ r e^{it} : t \in (-\frac{\pi}{\theta}, \frac{\pi}{\theta}) \}, &
    \begin{aligned}
      \Sigma^{(1)}_r = {}& \{ e^{-\frac{\pi i}{\theta}} t : t \in (0, r) \}, \\
      \Sigma^{(2)}_r = {}& \{ -e^{\frac{\pi i}{\theta}} t : t \in (-r, 0) \},
    \end{aligned}
  \end{align}
  with counterclockwise orientation. We have, by the asymptotics of $f(z)$ and $g(z)$ as $z \to \infty$, that
  \begin{equation}\label{eq:Integral_zero}
    \lim_{r \to \infty} \int_{\Sigma_r} z^{\alpha} g(z) \prod^{\theta - 1}_{j = 0} f(z e^{\frac{2j\pi i}{\theta}}) dz = 0.
  \end{equation}
  To see it, we note that the integral over the $\Sigma^{(1)}_1$ and $\Sigma^{(2)}_r$ cancel each other, and the integral over $\Sigma^{(0)}_r$ tends to $0$, since the integrand is $\bigO(z^{-1 - \theta/2})$, as $r \to \infty$.  
  From Lemma \ref{lem:finite_zeros}, we denote the zeros of $f(x)$ on rays $\Gamma_1, \dotsc, \Gamma_{\theta - 1}$ as $a_1, \dotsc, a_m$ where we allow repetitions for multiplicity. 
  Then, similar to \eqref{eq:Integral_zero}, we have
  \begin{equation}
    \lim_{r \to \infty} \int_{\Sigma_r} z^{\alpha} g(z) \prod^{\theta - 1}_{j = 0} f(z e^{\frac{2j\pi i}{\theta}}) \prod^m_{k = 1} \frac{1}{z^{\theta} - a^{\theta}_k} dz = 0.
  \end{equation}
  Now we deform $\Sigma_r$ into flat contours $\{ t + i\epsilon : 0 \leq t \leq r \} \cup \{ t - i\epsilon : 0 \leq t \leq r \}$, where $\epsilon \to 0_+$, the upper contour is oriented leftward and the lower one rightward. Then using the relation \eqref{eq:g_and_f_relation} and letting $r \to \infty$, we derive that
  \begin{equation} \label{eq:integral_of_nonnegative}
    \int^{\infty}_0 x^{\alpha} f^2(x) W(x)\prod^m_{k = 1} \frac{1}{x^{\theta} - a^{\theta}_k} dz = 0, \quad W(x) = \prod^{\theta - 1}_{j = 1} f(x e^{\frac{2j\pi i}{\theta}}).
  \end{equation}

  From \eqref{eq:realty_Phi_0}, we find that both $f(x)$ and $W(x)$ are real on $(0, +\infty)$. We also have that the function $W(x) \prod^m_{k = 1} \frac{1}{z^{\theta} - a^{\theta}_k}$ has no zeros on $(0, +\infty)$, that is, it does not change sign on $(0, +\infty)$, and it is either strictly positive or strictly negative on $(0, +\infty)$.
  
 We note that $x^{\alpha} f^2(x)$ is non-negative on $(0, +\infty)$. Hence \eqref{eq:integral_of_nonnegative} implies that $f(x) = 0$ on $(0, +\infty)$, and then $f(z) = 0$ on $\compC$, contradictory to the assumption that $f(z)$ is non-trivial.
  
Next we show that $g(z)=0$. From \eqref{eq:g_defnition}, \eqref{eq:g_and_f_relation} and $f(z) = 0$, we see that $\hat{g}(z)=z^{\frac{1}{2}-\frac{(\theta + 1)\beta}{2\theta}}g(z^{\frac{1}{\theta}})$ can be analytically extended to the whole complex plane. If $\theta \geq 4$, from the asymptotic behaviour of $g(z)$ near infinity, it follows that $\hat{g}(z)$ vanishes exponentially at infinity from all directions. Hence, $\hat{g}(z)$ is identically zero. If $\theta = 2, 3$, although $\hat{g}(z)$ does not vanish exponentially if $z \to \infty$ in some directions, like the positive real axis, $\hat{g}(z) \hat{g}(-z)$ vanishes exponentially at infinity from all directions, and it is identically zero. We conclude that for all $\theta \geq 2$, we have $g(z) = 0$.

By \eqref{eq: f_Pre_definition}, \eqref{eq: f_definition} and  \eqref{eq:g_defnition},  the vanishing of $f(z)$ and $g(z)$ implies that the row vector $(\phi_0, \phi_1, \dotsc, \phi_{\theta})$  vanishes in the whole complex plane.  This completes the proof of the vanishing lemma. 
  
\end{proof}

\begin{proof}[Proof of Lemma \ref{lem:finite_zeros}]
  We intend to prove the technical lemma concerning  the zeros of $f(z)$. We note that from the definition of $f(z)$, its zero points have no finite accumulation points. So we only need to show that on the rays $\Gamma_j$, it has no zeros outside a finite region.

  For this purpose, we define
  \begin{equation}\label{eq:F_definition}
    F(z) = e^{ \frac{\theta + 1}{2\theta} z^2 - \tau z} \times
    \begin{cases}
      f^{(\pre)}(-(-z)^{\frac{\theta + 1}{\theta}}), & \arg (-z) \in (-(\theta - \frac{1}{2}) \frac{\pi}{\theta + 1}, (\theta - \frac{1}{2}) \frac{\pi}{\theta + 1}), \\
      \phi_0(-z^{\theta + 1}), & \arg z \in (-\frac{\pi}{\theta + 1}, \frac{\pi}{\theta + 1}), \\
      \phi_1(-z^{\theta + 1}), & \arg z \in (\frac{\pi}{\theta + 1}, \frac{3\pi}{2(\theta + 1)}) \cup (-\frac{3\pi}{2(\theta + 1)}, -\frac{\pi}{\theta}).
    \end{cases}
  \end{equation}
  We thus need to show that $F(z)$ has no zeros on
  \begin{equation}
    \tilde{\Gamma}_j := \{ e^{\frac{(2j + 1)\pi i}{\theta+1}} t : t \in (0, +\infty) \}, \quad j = 1, 2, \dotsc, \theta - 1
  \end{equation}
  outside a finite region. It is clear that $F(z)$ is analytic in the sectors $\arg (-z) \in (-(\theta - 1/2)\pi/(\theta + 1), (\theta - 1/2)\pi/(\theta + 1))$ and $\arg z \in (-3\pi/(2(\theta + 1)), 3\pi/(2(\theta + 1)))$ separately, and is continuous up to the boundary in either sector.  On the rays $\{ \arg z = 0, \pm \pi/(\theta + 1) \}$, $F(z)$ is defined by straightforward continuation.  Also we have that
  \begin{align}
    F(z) = {}& \bigO(z^{-1}), && \hspace{-3cm} z \to \infty, \\
    F(z) = {}& \bigO(z^{\frac{\alpha + 1}{\theta} - \frac{1}{2}}), && \hspace{-3cm} z \to 0 \text{ in the sector } \arg (-z) \in (-(\theta - \frac{1}{2}) \frac{\pi}{\theta + 1}, (\theta - \frac{1}{2}) \frac{\pi}{\theta + 1}), \\
    F(z) = {}&
               \begin{cases}
                 \bigO(z^{\theta - \alpha - \frac{1}{2}}), & \alpha > \theta - 1, \\
                 \bigO(z^{\frac{1}{2}} \log z), & \alpha = \theta - 1, \\
                 \bigO(z^{\frac{\alpha + 1}{\theta} - \frac{1}{2}}), & -1 < \alpha < \theta - 1,
               \end{cases}
                              && z \to 0 \text{ in the sector } \arg z \in (-\frac{3\pi}{2(\theta + 1)}, \frac{3\pi}{2(\theta + 1)}).
  \end{align}
  
  Let $M_{\alpha}$ be a nonnegative integer such that
  \begin{equation}
    M_{\alpha} =
    \begin{cases}
      0, & -1 < \alpha \leq \theta, \\
      \lceil \alpha - \theta \rceil, & \alpha > \theta.
    \end{cases}
  \end{equation}
  Then we define
  \begin{equation}
    \hat{F}(z) =
    \begin{cases}
      F(z), & \lvert z \rvert > 1, \\
      z^{M_{\alpha}}F(z), & \lvert z \rvert < 1.
    \end{cases}
  \end{equation}
  We have that $\hat{F}(z)$ is analytic on $\compC \setminus \hat{\Gamma}$, where
  \begin{align}
    \hat{ \Gamma} = {}& \bigcup^2_{j = 0}   \hat{ \Gamma} _j, &  \hat{ \Gamma} _0 = {}& \{ e^{2\pi it} : t \in (0, 1] \}, &
                                                                                                                         \begin{aligned}
                                                                                                                           \hat{ \Gamma} _1 = {}& \{ e^{\frac{3 \pi i}{2(\theta + 1)}} t: t \in (0, +\infty) \}, \\
                                                                                                                           \hat{ \Gamma} _2 = {}& \{ -e^{-\frac{3\pi i}{2(\theta + 1)}} t: t \in (-\infty, 0) \},
                                                                                                                         \end{aligned}
  \end{align}
  see Figure \ref{fig:Gamma} for the orientation of the contour.

  Denote $\rho(z)$ a function on $  \hat{ \Gamma} $ such that
  \begin{equation}
    \rho(z) = 
    \begin{cases}
      e^{\beta \pi i} F(z e^{-\frac{2\pi}{\theta + 1} i}) e^{\frac{\theta + 1}{2\theta} (1 - e^{-\frac{4\pi i}{\theta + 1}}) z^2 - \tau (1 - e^{-\frac{2\pi i}{\theta + 1}}) z}, & z \in  \hat{ \Gamma} _1 \text{ and } \lvert z \rvert > 1, \\
      e^{-\beta \pi i} F(z e^{\frac{2\pi}{\theta + 1} i}) e^{\frac{\theta + 1}{2\theta} (1 - e^{\frac{4\pi i}{\theta + 1}}) z^2 - \tau (1 - e^{\frac{2\pi i}{\theta + 1}}) z}, & z \in  \hat{ \Gamma} _2 \text{ and } \lvert z \rvert > 1, \\
      e^{\beta \pi i}z^{M_{\alpha}} F(z e^{-\frac{2\pi}{\theta + 1} i}) e^{\frac{\theta + 1}{2\theta} (1 - e^{-\frac{4\pi i}{\theta + 1}}) z^2 - \tau (1 - e^{-\frac{2\pi i}{\theta + 1}}) z}, & z \in  \hat{ \Gamma} _1 \text{ and } \lvert z \rvert < 1, \\
      e^{-\beta \pi i} z^{M_{\alpha}}F(z e^{\frac{2\pi}{\theta + 1} i}) e^{\frac{\theta + 1}{2\theta} (1 - e^{\frac{4\pi i}{\theta + 1}}) z^2 - \tau (1 - e^{\frac{2\pi i}{\theta + 1}}) z}, & z \in  \hat{ \Gamma} _2 \text{ and } \lvert z \rvert <1 , \\
      (z^{M_{\alpha}}-1)F(z)   & z \in  \hat{ \Gamma} _0.
    \end{cases}
  \end{equation}
  Then we have
  \begin{equation}
    \hat{F}_+(z) - \hat{F}_-(z) = \rho(z), \quad z \in   \hat{ \Gamma}.
  \end{equation}
  It is seen that $\rho(z)$ is continuous and bounded on $  \hat{ \Gamma} $ and it vanishes exponentially as $z \to \infty$. Therefore, we have that
  \begin{equation}
    \hat{F}(z) = \frac{1}{2\pi i} \int_{  \hat{ \Gamma} } \frac{\rho(w)}{w - z} dw.
  \end{equation}
  From this representation of $\hat{F}(z)$, we have that if there is a minimum $k$ such that
  \begin{equation} \label{eq:k-moment_minimum}
    \frac{1}{2\pi i} \int_{  \hat{ \Gamma} } w^k \rho(w) dw = c_k \neq 0,
  \end{equation}
  then as $z \to \infty$,
  \begin{equation}
    F(z) = \hat{F}(z) = c_k z^{-k - 1} + \bigO(z^{-k - 2}),
  \end{equation}
  and then $F(z)$ do not have zeros along the rays $ \tilde{ \Gamma}_k$ when $\lvert z \rvert$ is large enough, and we conclude the proof.

  Otherwise, \eqref{eq:k-moment_minimum} does not hold for any $k \geq 0$, and for all $m \geq 0$,
  \begin{equation}
    \frac{1}{2\pi i} \int_{  \hat{ \Gamma} } w^m \rho(w) dw = 0.
  \end{equation}
  We consider the value of $F(iy)$ where $y \in \realR$ and $\lvert y \rvert > 1$. We have that for any $m$,
  \begin{equation}
    \begin{split}
      F(z) = {}& \frac{1}{2\pi i} \int_{  \hat{ \Gamma} } \frac{\rho(w)}{w - z} dw \\
      = {}& \frac{1}{2\pi i} \int_{  \hat{ \Gamma} } \rho(w) \left( \frac{1}{w - z} - \frac{1 - (w/z)^m}{w - z} \right) dw \\
      = {}& \frac{1}{z^m} \frac{1}{2\pi i} \int_{  \hat{ \Gamma} } \rho(w) \frac{w^m}{w - z} dw.
    \end{split}
  \end{equation}

  Then we show that $F(iy)$ vanishes faster than exponential function as $y \to \pm \infty$. To see it, we note that
  \begin{equation}
    \lvert \rho(w) \rvert \leq M e^{-\epsilon \lvert w \rvert^2}
  \end{equation}
  for some $M, \epsilon > 0$. Then for large enough $y$
  \begin{equation} \label{eq:F(z)_est}
    \lvert F(iy) \rvert \leq \int^{\infty}_0 e^{-\epsilon x^2} x^m dx \lvert y \rvert^{-m}. 
  \end{equation}
  We take $m $ the integer such that $ m\leq 2 \epsilon \lvert y \rvert^2<m+1$. 
  Then from \eqref{eq:F(z)_est}, we derive that for large enough $y$
  \begin{equation}\label{eq:Expnential_Decay}
    \lvert F(iy) \rvert \leq e^{-\epsilon  \lvert y \rvert^2}.
  \end{equation}
  Thus, $z^{M_{\alpha}}F(z)$ is analytic and bounded in the sector $\arg (-z) \in (-(\theta - 1/4)\pi/(\theta + 1), (\theta - 1/4)\pi/(\theta + 1))$ and decays faster than exponential function at infinity along the imaginary axis.  Applying  Carlson's theorem, we have $F(z)$ is the zero function, which implies that $f(z)$ is also identically zero. We complete the proof of Lemma \ref{lem:finite_zeros}.
\end{proof}

\section{Proof of Theorem \ref{Pro:LaxPair}: Lax pair and nonlinear differential equations} \label{sec:Lax_pair}

From \eqref{eq:asyNew} and Proposition \ref{lem:existence_H_to_S}, we have the following asymptotic behaviours as $\xi \to \infty$ in $\compC_+$. To simplify the notation, we let $D = \diag(1, \sqrt{2}, 2)$ as in \eqref{eq:ABExpand}, and write \eqref{eq:asy_expansion_H_to_S} for $\theta = 2$ as
\begin{multline} \label{eq:PhiInfty}
  \Phi(\xi) = D (I + 2^{\frac{3}{2}} \frac{M(\tau)}{\xi} + 8 \frac{\M(\tau)}{\xi^2}+\bigO(\xi^{-3})) D^{-1} \diag(1, -\omega \xi^{\frac{1}{3}}, \omega^2 \xi^{\frac{2}{3}}) \\
  \times \begin{pmatrix}
      1 & 1 & 1 \\
      1 & \omega & \omega^2 \\
      1 & \omega^2 & \omega
    \end{pmatrix}
    \diag \left( e^{-\frac{3}{4} \omega^2 \xi^{\frac{2}{3}} - \tau \omega \xi^{\frac{1}{3}}}, e^{-\frac{3}{4} \omega \xi^{\frac{2}{3}} - \tau \omega^2 \xi^{\frac{1}{3}}}, e^{-\frac{3}{4} \xi^{\frac{2}{3}} - \tau \xi^{\frac{1}{3}}} \right), 
\end{multline}
for $\Im \xi>0$, where $M(\tau) = (m_{ij}(\tau))^3_{i, j = 1}$, $\M(\tau) = (\m_{ij}(\tau))^3_{i, j = 1}$ are analytic functions for $\tau\in\mathbb{R}$. 

Therefore, we have that the coefficients $A$ and $B$ in \eqref{eq:LaxPair},
as matrix-valued functions in $\xi$, are analytic in $\xi \in \compC \setminus \{ 0 \}$.  
As $\xi \to \infty$, we have the asymptotic expansion of $A$ and $B$
\begin{align}
  A = {}& 2^{-\frac{3}{2}} D \left( \sum^{\infty}_{k = 0} A_{-k} (2^{-\frac{3}{2}} \xi)^{-k} \right) D^{-1}, & B = {}& \sqrt{2} D \left( \sum^{\infty}_{k = -1} B_{-k} (2^{-\frac{3}{2}} \xi)^{-k} \right) D^{-1},
\end{align}
where $A_k$ and $B_k$ are functions in $\tau$ and  analytic for real $\tau$. Then by taking $\xi \to 0$, we find that $A = \bigO(\xi^{-1})$ and $B = \bigO(1)$. Hence we have the expressions of $A$ and $B$ as given in \eqref{eq:ABExpand}.
From the asymptotic expansion \eqref{eq:PhiInfty}, we have \eqref{eq:Coefficients} with
\begin{equation}
  \begin{aligned}
    a = {}& m_{22} - m_{11} - \frac{\sqrt{2}}{3} \tau m_{23}, & b = {}& m_{12} - \frac{\sqrt{2}}{3} \tau m_{13}, & c = {}& m_{13}, \\
    d = {}& m_{32} - m_{21} + \frac{\sqrt{2}}{3} \tau (m_{11} - m_{33}), & f = {}& -m_{23} + \frac{\sqrt{2}}{3} \tau m_{13}, & k = {}& m_{33} - m_{22} + \frac{\sqrt{2}}{3} \tau m_{12}.
  \end{aligned}
\end{equation}
The above functions are analytic for $\tau\in\mathbb{R}$.

Furthermore, we have
\begin{equation}
  \begin{split}
    & A_{-2} =
                 \left[\M, \begin{pmatrix}
                             0 & 0 & 0 \\
                             1 & 0 & 0 \\
                             -\frac{\sqrt{2}}{3} \tau & 1 & 0
                           \end{pmatrix} \right]-\left[M,\begin{pmatrix}
                                                           0 & 0 & 0 \\
                                                           1 & 0 & 0 \\
                                                           -\frac{\sqrt{2}}{3} \tau & 1 & 0
                                                         \end{pmatrix} \right]M+
                 \left[M, \begin{pmatrix}
                            0 & \frac{\sqrt{2}}{3} \tau & -1 \\
                            0 & \frac{1}{3} & \frac{\sqrt{2}}{3} \tau \\
                            0 & 0 & \frac{2}{3}
                          \end{pmatrix} \right]-
                 M\\
    & = {} \textstyle
          \begin{pmatrix}
            \m_{12} & \m_{13} & 0 \\
            \m_{22} - \m_{11} & \m_{23} - \m_{12} & -\m_{13} \\
            \m_{32} - \m_{21} & \m_{33} - \m_{22} & -\m_{23}
          \end{pmatrix}-\frac{\sqrt{2}}{3} \tau \begin{pmatrix}
                                                  \m_{13} & 0 & 0 \\
                                                  \m_{23} & 0 & 0 \\
                                                  \m_{33} - \m_{11} & -\m_{12} & -\m_{13}
                                                \end{pmatrix}\\
               & \textstyle
                 +\left(- \begin{pmatrix}
                            m_{12}+1 & m_{13} & 0 \\
                            m_{22} - m_{11} & m_{23} - m_{12}+1 & -m_{13} \\
                            m_{32} - m_{21} & m_{33} - m_{22} & -m_{23}+1
                          \end{pmatrix}+\frac{\sqrt{2}}{3} \tau  \begin{pmatrix}
                                                                   m_{13} & 0 & 0 \\
                                                                   m_{23} & 0 & 0 \\
                                                                   m_{33} - m_{11} & -m_{12} & -m_{13}
                                                                 \end{pmatrix}\right)M\\
               &\textstyle
                 +\frac{1}{3} \begin{pmatrix}
                               3m_{31} & m_{12}+3m_{32} & -3m_{11}+2m_{13}+3m_{33} \\
                               - m_{21} &0 &-3m_{21}+m_{23} \\
                               -2m_{31}  & - m_{32} & -3m_{31}
                             \end{pmatrix}
                 +\frac{\sqrt{2}}{3} \tau   \begin{pmatrix}
                                              -m_{21} & m_{11} - m_{22} & m_{12} - m_{23} \\
                                              -m_{31} & m_{21} - m_{32} & m_{22} - m_{33} \\
                                              0 & m_{31} & m_{32}
                                            \end{pmatrix}.
  \end{split}
\end{equation}
Since $A_{-2} = 0$, we have 
\begin{equation}\label{eq:Relation_1}
  (A_{-2})_{13}= a+k+c(f-b- \frac{\sqrt{2}}{3} c\tau-\frac{1}{3})=0. 
\end{equation}
Similarly, we have
\begin{align}\label{eq:Relation_2}
  (A_{-2})_{12}+ (A_{-2})_{23}&\textstyle =m_{13}(m_{11}-2m_{22}+m_{33}) +m_{23}(m_{12}-m_{23}-\frac{2}{3})-m_{12}(m_{12}+\frac{2}{3})\nonumber\\
   &~~~~~~~~~~~~~~~~~~~~~~~\textstyle +(m_{32}-m_{21})+\frac{\sqrt{2}}{3} \tau((m_{11}-m_{33})+m_{13}(m_{12}+m_{23})) \nonumber \\
 &\textstyle =d+c(k-a)-(b^2+bf+f^2)+\frac{2}{3}(f-b)+\frac{\sqrt{2}}{3} c\tau(f-b-\frac{\sqrt{2}}{3} c\tau-\frac{4}{3}) \nonumber \\
  \textstyle &=0.
\end{align}
From  \eqref{eq:defn_N(xi)} with $\theta = 2$, we have that the eigenvalues of $A_{-1}$ are $1/2 - \alpha/3$, $\alpha/6$ and $1/2 + \alpha/6$. Hence, by computing the characteristic function of $A_{-1}$ and comparing its linear term with that of $(z - (1/2 - \alpha/3))(z - \alpha/6)(z - (1/2 + \alpha/6))$, we obtain the relation \eqref{eq:Consatant_Equation_3}.
From  \eqref{eq:Consatant_Equation_3},  \eqref{eq:Relation_1} and  \eqref{eq:Relation_2}, we have \eqref{eq:Constant_Equation_1}.
Substituting \eqref{eq:Constant_Equation_1} into \eqref{eq:Relation_1}, we have \eqref{eq:Constant_Equation_2}.

The compatibility of the Lax pair yields the zero-curvature equation
\begin{equation}\label{eq:ZeroC}
  \frac{dA}{d\tau} - \frac{dB}{d\xi} + AB - BA = 0,
\end{equation}
from which we derive the system of nonlinear differential equations \eqref{eq:NonEqs_c}-\eqref{eq:NonEqs_d}.
Using \eqref{eq:Constant_Equation_1}-\eqref{eq:NonEqs_f}, we may express $a$, $b$,  $f$ and $k$ in terms of $c$, $c'$ and $c''$
as follows
\begin{align}
  b = {}& \textstyle -\frac{1}{2}\left(\frac{c'}{\sqrt{2}}+c^2+ \frac{\sqrt{2}}{3} \tau c\right)+\frac{1}{2}\gamma, \label{eq:Exp_b} \\
  f= {}& \textstyle  -\frac{1}{2}\left(\frac{c'}{\sqrt{2}}+c^2-\frac{\sqrt{2}}{3} \tau c\right)-\frac{1}{2}\gamma, \label{eq:Exp_f} \\
  a= {}& \textstyle -\frac{1}{4}c''-\frac{3}{2\sqrt{2}}cc'-\frac{1}{2}c^3+\frac{c}{2}\left(-\frac{2}{9} \tau^2+\gamma+\frac{1}{3}\right)-\frac{1}{3\sqrt{2}}\gamma\tau, \label{eq:Exp_a} \\
  k = {}& \textstyle \frac{1}{4}c''+\frac{3}{2\sqrt{2}}cc'+\frac{1}{2}c^3+\frac{c}{2}\left(\frac{2}{9} \tau^2+\gamma+\frac{1}{3}\right)+\frac{1}{3\sqrt{2}}\gamma\tau,\label{eq:Exp_k} 
 \end{align}
where $\gamma= \frac{1}{36} + \frac{\alpha}{12} - \frac{\alpha^2}{12}$.
From  \eqref{eq:Consatant_Equation_3}, \eqref{eq:NonEqs_a} and  \eqref{eq:NonEqs_k}, we have
\begin{equation}\label{eq:NonEqs_ak}
  \textstyle  \frac{(a-k)'}{\sqrt{2}}=3(f+b)^2+3c(a-k)-\frac{4}{9} \tau^2(b+f)-\frac{1}{3}(b-f)+ \frac{4}{9} \tau^2c^2+ \frac{\sqrt{2}}{3} \tau c(b-f)+\frac{2\sqrt{2}}{3} \tau(a+k)+\frac{4\sqrt{2}}{9}\tau c+2\gamma.
\end{equation}
Substituting \eqref{eq:Exp_b}-\eqref{eq:Exp_k} into \eqref{eq:NonEqs_ak}, 
 we obtain \eqref{eq:Third_order_Eqs_c}. 
The  Chazy-I equation \eqref{eq:Chazy_Eqs_y} follows directly from  \eqref{eq:Third_order_Eqs_c}.  And the equation \eqref{eq:Chazy_Eqs_u}
can be derived from \eqref{eq:Consatant_Equation_3},  \eqref{eq:Exp_b}-\eqref{eq:Exp_k} and the fact that $\det(A_{-1})=-(\frac{\alpha^{3}}{108} + \frac{\alpha^{2}}{72} - \frac{\alpha}{24})$.
We complete the proof of Theorem \ref{Pro:LaxPair}.

\section{Asymptotics of $\Phi^{(\tau)}(\xi)$ and $K^{(\tau)}(x, y)$ as $\tau\to-\infty$} \label{sec:asy_tau_to_negative}

Let
\begin{equation}
  P^{(\infty)}(\xi) = \Upsilon(\xi) \Omega_{\pm} e^{-\Theta(\xi)},
\end{equation}
where $ \Upsilon(\xi)$, $\Omega_{\pm}$ and $\Theta(\xi)$ are the same as in \eqref{eq:asyNew}. 
Then we denote the contour $\Sigma^{(\tau)}_R$ as $\partial D(0, (-\tau)^{-1}) \cup \Sigma^{(\tau)}_{R, 1} \cup \Sigma^{(\tau)}_{R, 2}$ where $\Sigma^{(\tau)}_{R, 1} = \{ iy : y \in (\tau^{-1}, +\infty) \}$ and $\Sigma^{(\tau)}_{R, 2} = \{ iy : y \in (-\infty, -\tau^{-1}) \}$, such that $\partial D(0, (-\tau)^{-1})$ is oriented clockwise, $\Sigma^{(\tau)}_{R, 1}$ upward, and $\Sigma^{(\tau)}_{R, 2}$ downward.
Let
\begin{equation}
  R(\xi) = \diag((-\tau)^{\frac{k}{\theta + 1}})^{\theta}_{k = 0} \Phi^{(\tau)}(\xi) \times
  \begin{cases}
    P^{(\infty)}(\xi)^{-1} \diag((-\tau)^{-\frac{k}{\theta + 1}})^{\theta}_{k = 0}, & \lvert \xi \rvert > (-\tau)^{-1}, \\
    \Phi^{(\Mei)}((-\tau)^{\theta + 1 }\xi)^{-1} \diag((-\tau)^{\frac{\theta k}{\theta + 1}})^{\theta}_{k = 0}, & \lvert \xi \rvert < (-\tau)^{-1},
  \end{cases}
\end{equation}
where $\Phi^{(\Mei)}$ is defined in \eqref{eq:defn_Phi^Mei}. Although $\Phi^{(\tau)}(\xi)$ is not continuous on $\realR$ and $\{ iy : \lvert y \rvert < (-\tau)^{-1} \}$, it is clear that $R(\xi)$ is analytic there by simple analytic continuation. Furthermore, $R(\xi)$ satisfies the following RH problem:
\begin{RHP}
  $R(\xi)$ is analytic in $\compC \setminus \Sigma^{(\tau)}_R$, and is continuous up to the boundary. It satisfies
  \begin{enumerate}
  \item
    For $\xi \in \Sigma^{(\tau)}_R$,
    \begin{equation}
      R_+(\xi) = R_-(\xi) J_R(\xi), 
    \end{equation}
    where
    \begin{multline}
      J_R(\xi) = \\
      \begin{cases}
        \diag((-\tau)^{\frac{k}{\theta + 1}})^{\theta}_{k = 0}  \Upsilon(\xi) \Omega_+ \left(
        \begin{pmatrix}
          1 & e^{\beta \pi i} \exp(i g^{(\tau)}(\xi)) \\
          0 & 1
        \end{pmatrix}
        \oplus I_{\theta - 1}
        \right) & \\
        \times \Omega^{-1}_+ \Upsilon(\xi)^{-1} \diag((-\tau)^{-\frac{k}{\theta + 1}})^{\theta}_{k = 0}, & \xi \in \Sigma^{(\tau)}_{R, 1}, \\
        \diag((-\tau)^{\frac{k}{\theta + 1}})^{\theta}_{k = 0}  \Upsilon(\xi) \Omega_- \left(
        \begin{pmatrix}
          1 & -e^{-\beta \pi i} \exp(-i g^{(\tau)}(\xi)) \\
          0 & 1
        \end{pmatrix}
        \oplus I_{\theta - 1}
        \right) & \\
        \times \Omega^{-1}_- \Upsilon(\xi)^{-1} \diag((-\tau)^{-\frac{k}{\theta + 1}})^{\theta}_{k = 0}, & \xi \in \Sigma^{(\tau)}_{R, 2}, \\
        \diag((-\tau)^{\frac{k}{\theta + 1}})^{\theta}_{k = 0} P^{(\infty)}(\xi) \Phi^{(\Mei)}((-\tau)^{\theta + 1 }\xi)^{-1} \diag((-\tau)^{\frac{\theta k}{\theta + 1}})^{\theta}_{k = 0}, & \lvert \xi \rvert = (-\tau)^{-1},
      \end{cases}
    \end{multline}
    where
    \begin{equation}
      g^{(\tau)}(\xi)= \frac{\theta + 1}{\theta} \sin\left( \frac{2\pi}{\theta+1} \right) \xi^{\frac{2}{\theta + 1}} - 2\tau \sin\left( \frac{\pi}{\theta + 1} \right) \xi^{\frac{1}{\theta + 1}}.
    \end{equation}
  \item
    As $\xi \to \infty$, 
    \begin{equation}
      R(\xi) = I + \bigO(\xi^{-1}).
    \end{equation}
  \end{enumerate}
\end{RHP}

It is straightforward to check that $\lvert J_R(\xi) - I \rvert = \bigO(\exp(-\epsilon (-\tau)^{\frac{\theta}{\theta + 1}}) \lvert \xi \rvert^{\frac{1}{\theta + 1}})$ for some $\epsilon > 0$ on $\Sigma^{(\tau)}_{R, 1} \cup \Sigma^{(\tau)}_{R, 2}$, as $\tau \to -\infty$ and uniformly in $\xi$. We also have $\lvert J_R(\xi) - I \rvert = \bigO((-\tau)^{-2/(\theta + 1)})$  as $\tau\to-\infty$ on $\partial D(0, (-\tau)^{-1})$. To see this, we note
\begin{align}
  \diag((-\tau)^{\frac{k}{\theta + 1}})^{\theta}_{k = 0} P^{(\infty)}(\xi) = {}& \Upsilon((-\tau) \xi) \Omega_{\pm} e^{-\Theta(\xi)}, \\
  \Phi^{(\Mei)}((-\tau)^{\theta + 1 }\xi)^{-1} \diag((-\tau)^{\frac{\theta k}{\theta + 1}})^{\theta}_{k = 0} = {}& e^{-\tau \Lambda(\xi)} \Omega^{-1}_{\pm} \Upsilon((-\tau) \xi)^{-1} (I + \bigO(\tau^{-\frac{\theta}{\theta + 1}})),
\end{align}
where $\Lambda(\xi)$ is defined in \eqref{eq:defn_Lambda}, and $e^{-\Theta(\xi)} e^{-\tau \Lambda(\xi)} = I + \bigO(\tau^{-2/(\theta + 1)})$. Hence, by the standard small norm argument, we have that $R(\xi) = I + \bigO((-\tau)^{-2/(\theta + 1)})$ uniformly in $\xi$, if $(-\tau)$ is large enough. We then conclude that
\begin{equation}
  \Phi^{(\tau)}(\xi) = \diag((-\tau)^{-\frac{k}{\theta + 1}})^{\theta}_{k=0} (I + \bigO((-\tau)^{-\frac{2}{\theta+1}})) \diag((-\tau)^{-\frac{\theta k}{\theta + 1}})^{\theta}_{k=0} \Phi^{(\Mei)}((-\tau)^{\theta + 1 }\xi)
\end{equation}
for all $\lvert \xi \rvert < (-\tau)^{-1}$.

As a consequence, we have that for all $k = 0, 1, \dotsc, \theta$, and $\xi$ in a compact subset of $\compC$,
\begin{align}
  \Phi^{(\tau)}_{0, k}((-\tau)^{-\theta - 1} \xi) = {}& \Phi^{(\Mei)}_{0, k}(\xi) (1 + \bigO((-\tau)^{-\frac{2}{\theta + 1}})), \\
  (\Phi^{(\tau)}((-\tau)^{-\theta - 1} \xi))^{-1}_{k, \theta} = {}& (-\tau)^{\theta} (\Phi^{(\Mei)}(\xi))^{-1}_{k, \theta} (1 + \bigO((-\tau)^{-\frac{2}{\theta + 1}})).
\end{align}
Recall that $\phi^{(\tau)}$ is defined by the first row of $\Phi^{(\tau)}$, and $\phitilde^{(\tau)}$ is defined by the last column of $(\Phi^{(\tau)})^{-1}$. Also recall the relation between $\Phi^{(\Mei)}$ and $\Psi^{(\Mei)}$ given in \eqref{eq:Phi^Mei_pre_and_Psi_Mei}, and the relation between $(\Phi^{(\Mei)})^{-1}$ and $\Psitilde^{(\Mei)}$ given in \eqref{eq:Phi^Mei_inverse_pre_and_Psitilde_Mei}, \eqref{eq:defn_Phitilde^Mei} and \eqref{eq:Phitilde^Mei_and_Phi^Mei_inverse}. We have that as $\tau \to -\infty$, the limits of $\phi^{(\tau)}$ and $\phitilde^{(\tau)}$ can be expressed by the first row of $\Psi^{(\Mei)}$ and the first row of $\Psitilde^{(\Mei)}$, respectively. Hence, we obtain \eqref{eq:expr_phi^Mei} and \eqref{eq:expr_phitilde^Mei} and prove part \ref{enu:lem:phi_phitilde_sy:1} of Lemma \ref{lem:phi_phitilde_sy}.

Now we prove part \ref{enu:thm:kernel_asy:1} of Theorem \ref{thm:kernel_asy}. For $\xi, \eta \in (0, \infty)$ and a large $C > 0$, we divide the integral domain on the right-hand side of \eqref{eq:kernel_limit_formula} that defines $K^{(\tau)}(x, y)$ into three parts: $(-\tau, -C) \cup [-C, C] \cup (C, +\infty)$. On $(-\tau, -C)$, as $\tau \to -\infty$, by part \ref{enu:lem:phi_phitilde_sy:1} of Lemma \ref{lem:phi_phitilde_sy},
\begin{equation}
  \begin{split}
    & ((-\tau)^{-\frac{\theta + 1}{\theta}} \xi)^{\alpha} \int^{-C}_{\tau} \phi^{(\sigma)}((-\tau)^{-\frac{\theta + 1}{\theta}} \xi) \phitilde^{(\sigma)}((-\tau)^{-\frac{\theta + 1}{\theta}} \eta) d\sigma \\
    = {}& 2\pi (\theta + 1) (-\tau)^{\frac{1}{\theta}} \int^{-C}_{\tau} \left( \frac{\sigma}{\tau} \right)^{\frac{1}{\theta}} \left( \left( \frac{\sigma}{\tau} \right)^{\frac{\theta + 1}{\theta}} \xi \right)^{\alpha} \\
    & \phantom{2\pi (\theta + 1) (-\tau)^{\frac{1}{\theta}}}
      \times \phi^{(\Mei)} \left( \left( \frac{\sigma}{\tau} \right)^{\frac{\theta + 1}{\theta}} \xi \right) \phitilde^{(\Mei)} \left( \left( \frac{\sigma}{\tau} \right)^{\frac{\theta + 1}{\theta}} \eta \right) (1 + \bigO((-\sigma)^{-\frac{2}{\theta + 1}})) d\sigma \\
    = {}& 2\pi \theta (-\tau)^{\frac{\theta + 1}{\theta}} \left[ \int^1_{(-C/\tau)^{\frac{\theta + 1}{\theta}}} (u\xi)^{\alpha} \phi^{(\Mei)}(u \xi) \phitilde^{(\Mei)}(u \eta) du + \epsilon(C) \right],
  \end{split}
\end{equation} 
where the error term $\epsilon(C) \to 0$ as $C \to +\infty$, uniformly in $\tau$. Next, we consider the integral
\begin{equation}
  \int_I \phi^{(\sigma)}((-\tau)^{-\frac{\theta + 1}{\theta}} \xi) \phitilde^{(\sigma)}((-\tau)^{-\frac{\theta + 1}{\theta}} \eta) d\sigma
\end{equation}
for $I = (C, +\infty)$ or $I = [-C, C]$. By estimate \eqref{eq:estimate_phi_phitilde_prod} at the end of Section \ref{sec:asy_tao_to_+infty}, the integral over $(C, +\infty)$ is bounded by $e^{-\epsilon C^2}$ for some $\epsilon > 0$. Finally, the integral over $[-C, C]$ is bounded. Summing over the three estimates, we thus finish the proof.

\section{Asymptotics of $\Phi^{(\tau)}(\xi)$ and $K^{(\tau)}(x, y)$ as $\tau\to\infty$} \label{sec:asy_tao_to_+infty}

\subsection{Properties of $g$-functions}

We use notation set up in Appendix \ref{sec:limiting_J}. Let
\begin{equation} \label{eq:defn_M(s)_Cg}
  M(s) = - C_g \left[ (\theta^{-\frac{1}{\theta + 1}} s)^2 + \frac{2}{\theta} (\theta^{-\frac{1}{\theta + 1}} s) + \frac{1}{\theta} - 1 \right], \quad C_g = \frac{1}{2} \frac{\theta^3}{(\theta + 1)(\theta - 1)^2}.
\end{equation}
We introduce the $g$-functions $g_0(z), g_1(z), \dotsc, g_{\theta}(z)$ on $\{ \arg z \in (0, \pi) \cup (-\pi, 0) \}$
\begin{align} \label{eq: g_another}
  g_0(z) = {}& M(\Iinvhat^{(\pre)}_2(z^{\frac{1}{\theta}})), & g_k(z) = M(\Iinvhat^{(\pre)}_1(e^{\frac{2(k - 1)\pi i}{\theta}} z^{\frac{1}{\theta}})), \quad k = 1, \dotsc, \theta.
\end{align}
Then,  the $g$-functions satisfy the following properties. 
 
\begin{lemma} \label{lem:g_limit}
  The functions $g_j(z),0\leq k\leq \theta$ defined in \eqref{eq: g_another} satisfy the properties below.
  \begin{enumerate} 
  \item
    $g_0(z)$ is analytic on $\compC \setminus [1, +\infty)$, $g_1(z)$ is analytic on $\compC \setminus ( [1, +\infty) \cup (-\infty, 0] )$, and $g_k(z)$ are analytic on $\compC \setminus (-\infty, 0]$ for $k = 2, \dotsc, \theta$, and they satisfy the relations (with $g_{\theta + 1}(z) = g_1(z)$ in \eqref{eq: gR2})
    \begin{align}
      g_{0,\pm}(x) = {}& g_{1,\mp}(x), & x \in {}& (1,\infty), \label{eq: gR1} \\
      g_{k,+}(x) = {}& g_{k+1,-}(x), & x \in {}& (-\infty,0), \quad 1\leq k\leq  \theta, \label{eq: gR2}
    \end{align}
  \item
    As $z \to \infty$, we have (with $C_g$ defined in \eqref{eq:defn_M(s)_Cg})
    \begin{align}
      g_0(z) = {}&
                   \begin{cases}
                     -C_g \left(  e^{-\frac{2\pi i}{\theta+1}}\left(\frac{z}{\theta}\right)^{\frac{2}{\theta+1}}-\frac{2}{\theta}(\theta-1) e^{-\frac{\pi i}{\theta+1}}\left(\frac{z}{\theta}\right)^{\frac{1}{\theta+1}}+\bigO(z^{-\frac{1}{\theta+1}})\right), & \arg z\in(0,\pi),
                     \\
                     -C_g \left(  e^{\frac{2\pi i}{\theta+1}}\left(\frac{z}{\theta}\right)^{\frac{2}{\theta+1}}-\frac{2}{\theta}(\theta-1) e^{\frac{\pi i}{\theta+1}}\left(\frac{z}{\theta}\right)^{\frac{1}{\theta+1}}+\bigO(z^{-\frac{1}{\theta+1}})\right), & \arg z\in(-\pi,0),
                   \end{cases} \label{eq: g0Expand} \\
      g_1(z) = {}&
                   \begin{cases}
                     -C_g \left(  e^{\frac{2\pi i}{\theta+1}}\left(\frac{z}{\theta}\right)^{\frac{2}{\theta+1}}-\frac{2}{\theta}(\theta-1) e^{\frac{\pi i}{\theta+1}}\left(\frac{z}{\theta}\right)^{\frac{1}{\theta+1}}+\bigO(z^{-\frac{1}{\theta+1}})\right), & \arg z\in(0,\pi),
                     \\
                     - C_g \left(  e^{-\frac{2\pi i}{\theta+1}}\left(\frac{z}{\theta}\right)^{\frac{2}{\theta+1}}-\frac{2}{\theta}(\theta-1) e^{-\frac{\pi i}{\theta+1}}\left(\frac{z}{\theta}\right)^{\frac{1}{\theta+1}}+\bigO(z^{-\frac{1}{\theta+1}})\right), & \arg z\in(-\pi,0),
                   \end{cases} \label{eq: g1Expand} \\
      \intertext{and for $2\leq k\leq \theta$}
      g_k(z) = {}&  \textstyle{- C_g \left(  e^{\frac{2(2k-1)\pi i}{\theta+1}}\left(\frac{z}{\theta}\right)^{\frac{2}{\theta+1}}-\frac{2}{\theta}(\theta-1) e^{\frac{(2k-1)\pi i}{\theta+1}}\left(\frac{z}{\theta}\right)^{\frac{1}{\theta+1}}+\bigO(z^{-\frac{1}{\theta+1}})\right)}, \quad \arg z\in(-\pi,\pi), \label{eq: gkExpand}
    \end{align}
    Moreover, all the $\bigO(z^{-1/(\theta + 1)})$ terms have asymptotic expansions as power series of $z^{-1/(\theta + 1)}$.
  \item
    In a small neighbourhood of $1$, we have an analytic function $f(z)$, such that
    \begin{align} \label{eq:defn_f_Airy}
      g_0(z) - g_1(z) = {}& \frac{4}{3} f(z)^{\frac{3}{2}}, & f(1) = {}& 0, & f'(1) = {}& -\frac{2^{\frac{1}{3}} \theta}{(\theta - 1)^{\frac{2}{3}} (\theta + 1)^{\frac{5}{3}}}.
    \end{align}
  \item
    We have
    \begin{equation}\label{eq:Inequlity_g0_g1}
      g_0(z)-g_1(z)>0 ,\quad z\in(0,1),
    \end{equation}
    and there exists a contour $\Gamma_1 \subseteq \compC_+$ that connects $1$ and $i \cdot \infty$, such that
    \begin{equation}\label{eq:Inequlity_Re_g0_g1}
      \Re (g_0(z)-g_1(z)) = \Re (g_0(\bar{z}) - g_1(\bar{z})) < 0, \quad z \in \Gamma_1.
    \end{equation}
  \end{enumerate}
\end{lemma}
\begin{proof} 
  Equation \eqref{eq:PreImage} together with \eqref{eq: g_another}, implies the jump relation \eqref{eq: gR1}. The other relation \eqref{eq: gR2} follows directly from the definition of  $\Iinvhat^{(\pre)}_1$ and \eqref{eq: g_another}. 
  
  From \eqref{eq:J_mapping} and \eqref{eq:PreImage}, we can derive the large $z$ expansion of   $\Iinvhat^{(\pre)}_1$ and $\Iinvhat^{(\pre)}_2$. Using this and \eqref{eq: g_another}, we obtain the large $z$ expansion of  $g_k(z)$ as given in \eqref{eq: g0Expand}-\eqref{eq: gkExpand}. Similarly, from \eqref{eq:J_mapping},  \eqref{eq: g_another} and \eqref{eq:PreImage}, we have \eqref{eq:defn_f_Airy}.
  
  It follows from the  definition of the mappings $\Iinvhat^{(\pre)}_1$ and $\Iinvhat^{(\pre)}_2$ that  
  \begin{equation}
    \Iinvhat^{(\pre)}_{1}((0,1))=(-(\theta + 1)\theta^{-\frac{\theta}{\theta + 1}},-\theta^{-\frac{\theta}{\theta + 1}}), ~~\Iinvhat^{(\pre)}_{2 }((0,1))=(-\theta^{-\frac{\theta}{\theta + 1}},0).\end{equation}
  This, together with \eqref{eq: g_another}, implies \eqref{eq:Inequlity_g0_g1}. 

  The existence of a contour $\Gamma_1 \subseteq \compC_+$ such that \eqref{eq:Inequlity_Re_g0_g1} holds follows from the analyticity of $g_0$ and $g_1$ in $\mathbb{C}\setminus\mathbb{R}$ and the behaviours of $g_0$ and $g_1$ given in \eqref{eq: g0Expand}, \eqref{eq: g1Expand} and  \eqref{eq:Inequlity_g0_g1}.
\end{proof} 

\subsection{Normalization: $\Phi\to T$}

Recall the contour $\Gamma_1$ in Lemma \ref{lem:g_limit}. Let $\Gamma_2 = \{ z \in \compC_- : \bar{z} \in \Gamma_1 \}$ be the conjugate of $\Gamma_2$, and orient both of them from $1$ to $\infty$. We denote $\Gamma = \Gamma_1 \cup \Gamma_2 \cup (-\infty, 0) \cup (0, 1) \cup (1, +\infty)$, and orient the real part of $\Gamma$ from left to right. We divide $\compC$ into four infinite regions: $\Omega_1$ between $[1, +\infty)$ and $\Gamma_1$, $\Omega_2$ between $\Gamma_1$ and $(-\infty, 1]$, $\Omega_3$ between $(-\infty, 1]$ and $\Gamma_2$, and $\Gamma_4$ between $\Gamma_2$ and $[1, +\infty)$. Each $\Omega_i$ overlaps with the $i$-th quadrant. 

Let the $(\theta + 1) \times (\theta + 1)$ matrix-valued function $T(\xi)$ be defined on $\compC \setminus \Gamma$, such that, with $c_1$ defined in \eqref{eq:defn_c_1_c_2},
\begin{equation} \label{eq:Phi_to_T}
  T^{(\pre)}(\xi) = \diag((c_1\tau)^{-k})_{k=0}^{\theta}\Phi((c_1\tau)^{\theta+1}\xi) \diag \left( e^{-\tau^2g_k(\xi)} \right)_{k=0}^{\theta}, 
\end{equation}     
for $z$ in the intersection between $\Omega_i$ and the $i$-th quadrant for $i = 1, 2, 3, 4$, and by analytic continuation in other parts of its domain.

Then, $T^{(\pre)}(\xi)$ satisfies the following RH problem:
\begin{RHP} \label{RHP:T^pre}
  \begin{enumerate}
  \item \label{enu:RHP:T^pre_1}
    On $\Gamma$, $T^{(\pre)}(\xi)$ satisfies, with $\sigma_1 =(
    \begin{smallmatrix}
      0 & 1 \\
      1 & 0
    \end{smallmatrix})
$,
    \begin{equation} \label{eq:Tjumps}
      T^{(\pre)}_+(\xi) = T^{(\pre)}_-(\xi) J_{T}(\xi), \quad J_{T}(\xi) =
      \begin{cases}
        \begin{pmatrix}
          1 & e^{\pi i\beta}  e^{\tau^2(g_0(\xi)-g_1(\xi))}  \\
          0 & 1  \\
        \end{pmatrix}\oplus I_{\theta-1},
        & \xi \in \Gamma_1, \\
        \begin{pmatrix}
          1 & -e^{-\pi i\beta}  e^{\tau^2(g_0(\xi)-g_1(\xi))}  \\
          0 & 1  \\
        \end{pmatrix}\oplus I_{\theta-1},
        & \xi \in \Gamma_2, \\
        \begin{pmatrix}
          -e^{-\pi i\beta}  &0 \\
          e^{\tau^2(g_1(\xi)-g_0(\xi))} & e^{\pi i\beta}   \\
        \end{pmatrix}\oplus I_{\theta-1},
        & \xi \in (0,1), \\
        \sigma_1\oplus I_{\theta-1},
            & \xi \in (1,\infty), \\
        \Mcyclic,
            & \xi \in \realR_-.
      \end{cases}
    \end{equation}
  \item \label{enu:RHP:T^pre_3}
    $T^{(\pre)}(\xi)$ has the following local growth condition as $\xi \to 0$
    \begin{equation}
      T^{(\pre)} (\xi) =  \xi^{-\frac{\beta}{2\theta}} T^{(\pre)}_0(\xi) \diag \left( \xi^{-\frac{\alpha + 1 - \theta}{\theta}}, 1, \xi^{\frac{1}{\theta}}, \xi^{\frac{2}{\theta}}, \dotsc, \xi^{\frac{\theta - 1}{\theta}} \right)E\diag \left( e^{-\tau^2g_k(\xi)} \right)_{k=0}^{\theta},
    \end{equation}
    where $T^{(\pre)}_0$ is analytic near $\xi=0$  and $E$ the matrix defined in RH problem \ref{RHP:general_model} in the way that $E$ in $\Omega_2$ (resp.~ $\Omega_3$) here is equal to $E$ in the second (resp.~the third) quadrant there.
  \item \label{enu:RHP:T^pre_2}
    $T^{(\pre)}(\xi)$ has the following asymptotic condition at infinity as $\xi \to \infty$ 
      
     \begin{equation}\label{eq:Tinfty_pre}
      T^{(\pre)}(\xi) = \left( I + \bigO(\xi^{-1}) \right) \Upsilon(\xi) \Omega_{\pm} \left( I + \bigO(\xi^{-\frac{1}{\theta + 1}}) \right),
    \end{equation}
    where the $\bigO(\xi^{-1})$ has an asymptotic expansion as power series of $\xi^{-1}$, and the $\bigO(\xi^{-1/(\theta + 1)})$ has an asymptotic expansion as power series of $\xi^{-1/(\theta + 1)}$.
  \end{enumerate}
\end{RHP}

From the asymptotic expansion \eqref{eq:Tinfty_pre}, we find that $T(\xi)$ has an asymptotic expansion
\begin{equation}
  T(\xi) = \Upsilon(\xi) \Omega_{\pm} \left( I+\sum^{\infty}_{n = 1} B_n \xi^{-\frac{n}{\theta + 1}} \right),
\end{equation}
and the right-hand side, as a formal Puiseux series, satisfies the condition in Lemma \ref{lem:RH_for_Phi_0}. Hence, we can use Lemma \ref{lem:RH_for_Phi_0} and get a constant lower triangular matrix $C_0(\tau)$ that is independent of $\xi$ whose diagonal entries are all $1$, such that if we define
\begin{equation} \label{eq:T_from_T^pre}
  T(\xi) = C_0(\tau) T^{(\pre)}(\xi),
\end{equation}
then $T(x)$ satisfies the following RH problem:
\begin{RHP} \label{enu:RHP:T_2}
  \begin{enumerate}
  \item 
    $T(\xi)$ satisfies part \ref{enu:RHP:T^pre_1} of RH problem \ref{RHP:T^pre}, with $T^{(\pre)}(\xi)$ replaced by $T(\xi)$ in \eqref{eq:Tjumps}.
  \item
    $T(\xi)$ satisfies part \ref{enu:RHP:T^pre_3} of RH problem \ref{RHP:T^pre}, with $T^{(\pre)}(\xi)$ replaced by $T(\xi)$ and $T_0(\xi)$ replaced by $T_0(\xi) := C_0(\tau) T^{(\pre)}_0(\xi)$ in \eqref{eq:Tjumps}.
  \item
    $T(\xi)$ satisfies a stronger version of part \ref{enu:RHP:T^pre_2} of RH problem \ref{RHP:T^pre}, such that
    \begin{equation}\label{eq:Tinfty}
      T(\xi) = \left( I + \bigO(\xi^{-1}) \right) \Upsilon(\xi) \Omega_{\pm}.
    \end{equation}
  \end{enumerate}
\end{RHP}

\subsection{Global parametrix}

Let, for $j = 0, \dotsc, \theta$,
\begin{align}
  N_j(s) = {}& (s + (\theta + 1)\theta^{-\frac{\theta}{\theta + 1}})^j \left( \frac{s}{s + (\theta + 1)\theta^{-\frac{\theta}{\theta + 1}}} \right)^{\frac{\beta}{2}} \left( \frac{s}{s + \theta^{-\frac{\theta}{\theta + 1}}} \right)^{\frac{1}{2}}, \\
  \N_j(s) = {}& (s + (\theta + 1)\theta^{-\frac{\theta}{\theta + 1}})^j \left( \frac{s}{s + (\theta + 1)\theta^{-\frac{\theta}{\theta + 1}}} \right)^{-\frac{\betatilde + 1}{2}} \left( \frac{s}{s + \theta^{-\frac{\theta}{\theta + 1}}} \right)^{\frac{1}{2}},
\end{align}
be defined on $\compC \setminus [-(\theta + 1) \theta^{-\frac{\theta}{\theta + 1}}, 0]$ such that both $N_j(s)$ and $\N_j(s)$ are $s^j + \bigO(s^{j - 1})$ as $s \to \infty$. Then define the $(\theta + 1) \times (\theta + 1)$ matrix-valued functions $P^{(\infty, \pre)}(\xi) = (P^{(\infty, \pre)}(\xi))^{\theta}_{j, k = 0}$ and $\widetilde{P}^{(\infty, \pre)}(\xi) = (\widetilde{P}^{(\infty, \pre)}(\xi))^{\theta}_{j, k = 0}$ on $\compC \setminus \realR$ by
\begin{align}
  P^{(\infty, \pre)}_{j, 0}(\xi) = {}& N_j(\Iinvhat^{(\pre)}_2(\xi^{\frac{1}{\theta}})), & P^{(\infty, \pre)}_{j, k}(\xi) = N_j(\Iinvhat^{(\pre)}_1(e^{\frac{2(k - 1)\pi i}{\theta}} \xi^{\frac{1}{\theta}}), \quad k = 1, \dotsc, \theta, \\
  \widetilde{P}^{(\infty, \pre)}_{j, 0}(\xi) = {}& \N_j(\Iinvhat^{(\pre)}_2(\xi^{\frac{1}{\theta}})), & \widetilde{P}^{(\infty, \pre)}_{j, k}(\xi) = \N_j(\Iinvhat^{(\pre)}_1(e^{\frac{2(k - 1)\pi i}{\theta}} \xi^{\frac{1}{\theta}}), \quad k = 1, \dotsc, \theta.
\end{align}
Then we find that $P^{(\infty, \pre)}(\xi)$ satisfies the following RH problem:
\begin{RHP} \label{RHP:Pinfty_limit}
  \begin{enumerate} 
  \item \label{enu:RHP:Pinfty_limit_1}
    On the real axis oriented from left to right, we have
    \begin{align}
      P^{(\infty, \pre)}_+(\xi) = {}& P^{(\infty, \pre)}_{-}(\xi)
      \begin{cases}
        \sigma_1\oplus I_{\theta-1},
        & \xi \in (1,\infty), \\
        \diag(
        -e^{-\pi i\beta} , e^{\pi i\beta})\oplus I_{\theta-1},
        & \xi \in (0,1), \\
        \Mcyclic,
        & \xi \in (-\infty,0).
      \end{cases} \label{eq:PinftyJump} \\
      \widetilde{P}^{(\infty, \pre)}_+(\xi) = {}& \widetilde{P}^{(\infty, \pre)}_{-}(\xi)
      \begin{cases}
        \sigma_1\oplus I_{\theta-1},
        & \xi \in (1,\infty), \\
        \diag(
        e^{\pi i\betatilde} , -e^{-\pi i\betatilde})\oplus I_{\theta-1},
        & \xi \in (0,1), \\
        \Mcyclic,
        & \xi \in (-\infty,0).
      \end{cases} \label{eq:PtildeinftyJump}
    \end{align}
    
  \item \label{enu:RHP:Pinfty_limit_3}
    $P^{(\infty, \pre)}(\xi)$ and $\widetilde{P}^{(\infty, \pre)}(\xi)$ have the following local growth condition conditions as $\xi \to 0$
    \begin{align}
      P^{(\infty, \pre)}(\xi)= {}& \xi^{-\frac{\beta}{2\theta}} P^{(\infty, \pre)}_0(\xi) \diag \left( \xi^{-\frac{\alpha + 1 - \theta}{\theta}}, 1, \xi^{\frac{1}{\theta}}, \xi^{\frac{2}{\theta}}, \dotsc, \xi^{\frac{\theta - 1}{\theta}} \right) E^{(\infty)}, \label{eq:Pzero} \\
      \widetilde{P}^{(\infty, \pre)}(\xi)= {}& \widetilde{P}^{(\infty, \pre)}_0(\xi) \diag \left( \xi^{-\frac{\betatilde}{2}}, \xi^{\frac{\betatilde + 2\theta - 1}{2\theta}}, \dotsc, \xi^{\frac{\betatilde + 5}{2\theta}} \xi^{\frac{\betatilde + 3}{2\theta}}, \xi^{\frac{\betatilde + 1}{2\theta}} \right) ((E^{(\infty)})^{-1})^T \Xicheck^{-1}_{\pm}, \label{eq:Ptildezero}
    \end{align}
    where $P^{(\infty, \pre)}_0(\xi)$ and $\widetilde{P}^{(\infty, \pre)}_0(\xi)$ are analytic near the origin, $\arg(\xi)\in(-\pi,\pi)$, $\Xicheck^{-1}_{\pm}$ is defined in \eqref{eq:defn_Xicheck_real}, and $E^{(\infty)}$ is a constant matrix in $\compC_+$ and $\compC_-$, such that $E^{(\infty)} = I_{1 \times 1} \oplus C_{\theta \times \theta}$ in $\compC_+$, where $C_{\theta \times \theta}$ is defined in \eqref{eq:defn_E_region_II}. Moreover, $P^{(\infty, \pre)}(\xi)$ and $\widetilde{P}^{(\infty, \pre)}(\xi)$ have the following local growth conditions as $\xi \to 1$:
    \begin{align}
      \choice{P}^{(\infty, \pre)}_{1, k}(\xi) = {}& \bigO((\xi - 1)^{-1/4}), & \choice{P}^{(\infty, \pre)}_{2, k}(\xi) = {}& \bigO((\xi - 1)^{-1/4}), & \choice{P}^{(\infty, \pre)}_{j, k}(\xi) = {}& \bigO(1),
    \end{align}
     where $\bullet = \widetilde{\,}$ or empty, $k = 0, 1, \dotsc, \theta$ and $j = 2, \dotsc, \theta$.
    
  \item \label{enu:RHP:Pinfty_limit_2}
    $ P^{(\infty, \pre)}(\xi)$ and $\widetilde{P}^{(\infty, \pre)}(\xi)$ have the following asymptotic condition at infinity as $\xi \to \infty$
    \begin{equation}\label{eq:Pinfty_pre}
      \choice{P}^{(\infty, \pre)}(\xi) = \left( I + \bigO(\xi^{-1}) \right) \Upsilon(\xi) \Omega_{\pm} \left( I + \bigO(\xi^{-\frac{1}{\theta + 1}}) \right), \quad \text{where $\bullet = \widetilde{\,}$ or empty}. 
    \end{equation}
    where the $\bigO(\xi^{-1})$ has an asymptotic expansion as power series of $\xi^{-1}$, and the $\bigO(\xi^{-1/(\theta + 1)})$ has an asymptotic expansion as power series of $\xi^{-1/(\theta + 1)}$.
  \end{enumerate}
\end{RHP}

Then like \eqref{eq:T_from_T^pre}, we find that there exist lower triangular matrices $C'_0$ and $\widetilde{C}'_0$ that are independent of $\xi$ whose diagonal entries are all $1$, such that if we define
\begin{align}
  P^{(\infty)}(\xi) = {}& C'_0 P^{(\infty, \pre)}(\xi), & \widetilde{P}^{(\infty)}(\xi) = \widetilde{C}'_0 P^{(\infty, \pre)}(\xi),
\end{align}
then $P^{(\infty)}(x)$ and $\widetilde{P}^{(\infty)}(x)$ satisfy the following RH problem:
\begin{RHP} \label{RHP:P^infty_Ptilde^infty}
  \begin{enumerate}
  \item 
    $P^{(\infty)}(\xi)$ and $\widetilde{P}^{(\infty)}(\xi)$ satisfy part \ref{enu:RHP:Pinfty_limit_1} of RH problem \ref{RHP:Pinfty_limit}, with $P^{(\infty, \pre)}(\xi)$ replaced by $P^{(\infty)}(\xi)$ in \eqref{eq:PinftyJump} and $\widetilde{P}^{(\infty, \pre)}(\xi)$ replaced by $\widetilde{P}^{(\infty)}(\xi)$ in \eqref{eq:PtildeinftyJump}.
  \item
    $P^{(\infty)}(\xi)$ and $\widetilde{P}^{(\infty)}(\xi)$ satisfy part \ref{enu:RHP:Pinfty_limit_3} of RH problem \ref{RHP:Pinfty_limit}, with $P^{(\infty, \pre)}(\xi)$ replaced by $P^{(\infty)}(\xi)$ in \eqref{eq:Pzero} and $P^{(\infty, \pre)}_0(\xi)$ replaced by $P^{(\infty)}_0(\xi) := C'_0 P^{(\infty, \pre)}_0(\xi)$ in \eqref{eq:Ptildezero}.
  \item
    $P^{(\infty)}(\xi)$ and $\widetilde{P}^{(\infty)}(\xi)$ satisfy a stronger version of part \ref{enu:RHP:Pinfty_limit_2} of RH problem \ref{RHP:Pinfty_limit}, such that
    \begin{align} 
      P^{(\infty)}(\xi) = {}& \left( I + \bigO(\xi^{-1}) \right) \Upsilon(\xi) \Omega_{\pm}, & \widetilde{P}^{(\infty)}(\xi) = {}& \left( I + \bigO(\xi^{-1}) \right) \Upsilon(\xi) \Omega_{\pm},
    \end{align}
  \end{enumerate}
\end{RHP}

We note that by an argument analogous to the proof of the uniqueness part of Proposition \ref{lem:existence_H_to_S}, $P^{(\infty)}(\xi)$ and $\widetilde{P}^{(\infty)}(\xi)$ are uniquely determined by RH problem \ref{RHP:P^infty_Ptilde^infty}.

At last, we define
\begin{equation}
  \check{P}^{(\infty)}(\xi) = \frac{1}{\theta + 1} e^{\frac{\theta}{\theta + 1} \pi i} \xi^{-\frac{\theta}{\theta + 1}} J_{\theta + 1} \widetilde{P}^{(\infty)}(\xi) \Xicheck_{\pm},
\end{equation}
and find that $\check{P}^{(\infty)}(\xi)$ satisfies the RH problem that $(P^{(\infty)}(\xi)^{-1})^T$ satisfies. Hence, $\check{P}^{(\infty)}(\xi) = (P^{(\infty)}(\xi)^{-1})^T$, or equivalently,
\begin{equation} \label{eq:relation_Pintfy_Pinfty_tilde}
  \widetilde{P}^{(\infty)}(\xi) = (\theta + 1) e^{-\frac{\theta}{\theta + 1} \pi i} \xi^{\frac{\theta}{\theta + 1}} J_{\theta + 1} P^{(\infty)}(\xi) \Xicheck^{-1}_{\pm}.
\end{equation}
   
\subsection{Local parametrices}

Recall the analytic function $f(z)$ defined in a neighbourhood of $1$, as defined in \eqref{eq:defn_f_Airy}. Let $\delta > 0$ be a small enough constant, and we assume that $\{ \tau^{4/3} f(\xi) : \xi \in \Gamma \}$ overlaps with $\Gamma_{\Ai}$ in the neighbourhood $D(1,\delta)$. In $D(1,\delta)$, we construct a local parametrix $P^{(1)}(\xi)$ satisfying the following RH problem as follows.
\begin{equation}
  P^{(1)}(\xi)=E(\xi)\left( \left(\sigma_1 \Phi^{(\Ai)}(\tau^{\frac{4}{3}}f(\xi))\sigma_1 \diag(e^{\frac{1}{2}(1\pm 1)\pi i}e^{\mp \frac{\pi i \beta}{2}}, e^{\pm \frac{\pi i \beta}{2}}) e^{\frac{1}{2}\tau^2(g_{1}(\xi)-g_{0}(\xi))\sigma_3}\right)\oplus I_{\theta-1}\right),
\end{equation} 
where the sign follows $\xi \in \compC_{\pm}$, and with $\Phi^{(\Ai)}_{\infty}$ defined in \eqref{eq:defn_Psi^Ai_infty},
\begin{equation}
  E(\xi)= P^{(\infty)}(\xi) \left( \left(\sigma_1 \Phi^{(\Ai)}_{\infty}(\tau^{\frac{4}{3}}f(\xi))\sigma_1 \diag(e^{\frac{1}{2}(1\pm 1)\pi i}e^{\mp \frac{\pi i \beta}{2}}, e^{\pm \frac{\pi i \beta}{2}})\right)^{-1}\oplus I_{\theta-1}\right),
\end{equation} 
Also we define the local parametrix $P^{(0)}(\xi) = (P^{(0)}_{j, k}(\xi))^{\theta}_{j, k = 0}$ in $D(0, \delta) \setminus \realR$ as 
\begin{equation}
  P^{(0)}_{j, k}(\xi) = P^{(\infty)}_{j, k}(\xi) +
  \begin{cases}
    \frac{-\xi^{\frac{\beta + 1}{2}}}{2\pi i} \int^{2\delta}_0 \eta^{-\frac{\beta + 1}{2}} e^{\tau^2(g_1(\eta)-g_0(\eta))} P^{(\infty)}_{j, 1-}(\eta) \frac{d\eta}{\eta - \xi}, & k = 0,\ j = 0, 1, \dotsc, \theta, \\
    0, & \text{otherwise},
  \end{cases}
\end{equation}
 where $\arg \xi \in(0, 2\pi)$. 
\subsection{Small norm argument}

Finally, let
\begin{equation}
  R(\xi) = T(\xi) \times
  \begin{cases}
    P^{(0)}(\xi)^{-1}, & \lvert \xi \rvert < \delta, \\
    P^{(1)}(\xi)^{-1}, & \lvert \xi - 1 \rvert < \delta, \\
    P^{(\infty)}(\xi)^{-1}, & \text{otherwise}.
  \end{cases}
\end{equation}
We find that $R(\xi)$ is well defined and analytic in $D(0, \delta)$ and $D(1, \delta)$, and satisfies the following RH problem:
\begin{RHP}
  $R(\xi)$ is a  matrix-valued function defined on $\compC \setminus \Gamma^{(R)}$, where $\Gamma^{(R)} = (\delta, 1 - \delta) \cup \partial D(0, \delta) \cup \partial D(1, \delta) \cup \Gamma_{1, \delta} \cup \Gamma_{2, \delta}$, with $\Gamma_{i, \delta} = \Gamma_i \setminus \overline{D(1, \delta)}$ for $i = 1, 2$. Let $\partial D(0, \delta)$ and $\partial D(1, \delta)$ be oriented clockwise, and other contours have the same orientation as $\Gamma$ in RH problems \ref{RHP:T^pre} and \ref{enu:RHP:T_2}.
  \begin{itemize}
  \item
    For $\xi \in \Gamma^{(R)}$,
    \begin{multline}
      R_+(\xi) = R_-(\xi) J_R(\xi), \quad \text{where} \quad
      J_R(\xi) = \\
      \begin{cases}
        P^{(\infty)}(\xi) \left(
        \begin{pmatrix}
          1 & e^{\pi i\beta}  e^{\tau^2(g_0(\xi)-g_1(\xi))}  \\
          0 & 1  \\
        \end{pmatrix}\oplus I_{\theta-1} \right) P^{(\infty)}(\xi)^{-1}, & \xi \in \Gamma_{1, \delta}, \\
        P^{(\infty)}(\xi) \left(
        \begin{pmatrix}
          1 & -e^{-\pi i\beta}  e^{\tau^2(g_0(\xi)-g_1(\xi))} \\
          0 & 1  \\
        \end{pmatrix}\oplus I_{\theta-1} \right) P^{(\infty)}(\xi)^{-1}, & \xi \in \Gamma_{2, \delta}, \\
        P^{(\infty)}_-(\xi) \left(
        \begin{pmatrix}
          1  &0 \\
          -e^{\pi i \beta} e^{\tau^2(g_1(\xi)-g_0(\xi))} & 1 \\
        \end{pmatrix}\oplus I_{\theta-1} \right) P^{(\infty)}_-(\xi)^{-1},
        & \begin{aligned} \xi \in {}& (\delta,1 - \delta)  \end{aligned} \\
        P^{(\infty)}(\xi) P^{(0)}(\xi)^{-1}, & \lvert \xi \rvert = \delta, \\
        P^{(\infty)}(\xi) P^{(1)}(\xi)^{-1}, & \lvert \xi - 1 \rvert = \delta.
      \end{cases}
    \end{multline}
  \item
    As $\xi \to \infty$,
    \begin{equation}
      R(\xi) = I + \bigO(\xi^{-1}).
    \end{equation} 
  \end{itemize}
\end{RHP}
Since as $\tau \to +\infty$, we have that $J_{R}(\xi) = \bigO(\tau^{-2})$, we use the standard small norm argument to find that $\lvert R(\xi) - I \rvert = \bigO(\tau^{-2})$ everywhere in $\compC \setminus \Gamma^{(R)}$. Below we use this estimate in two regions.

\paragraph{Approximation in $D(1, \delta)$}

In this region, $T(\xi) = R(\xi) P^{(1)}(\xi)$, and by \eqref{eq:Phi_to_T}, we have that for $\xi \in D(1, \delta) \cap \Omega_2$, if we denote $Q(\xi) = R(\xi) P^{(\infty)}(\xi)$ there, then
\begin{multline} \label{eq:first_est_region_D1_delta}
  \Phi_{0, 1}((c_1\tau)^{\theta+1}\xi) + e^{\beta \pi i} \Phi_{0, 0}((c_1\tau)^{\theta+1}\xi) = T_{0, 1}(\xi) e^{\tau^2 g_0(\xi)} \\
  = e^{\frac{\pi i \beta}{2}} \sqrt{\pi} e^{\frac{\tau^2}{2} (g_0(\xi) + g_1(\xi))}
  \times \left\{ \tau^{\frac{1}{3}} f(\xi)^{\frac{1}{4}} \left[ (i e^{\frac{\pi i \beta}{2}})  Q(\xi)_{0, 0} + (e^{-\frac{\pi i \beta}{2}}) Q(\xi)_{0, 1} \right] \Ai(\tau^{\frac{4}{3}} f(\xi)) \right. \\
    \left. - \tau^{-\frac{1}{3}} f(\xi)^{-\frac{1}{4}} \left[ (-i e^{\frac{\pi i \beta}{2}}) Q(\xi)_{0, 0} + (e^{-\frac{\pi i \beta}{2}}) Q(\xi)_{0, 1} \right] \Ai'(\tau^{\frac{4}{3}} f(\xi)) \right\},
\end{multline}
and
\begin{multline}
     e^{-\frac{1}{2}(1 + \beta) \pi i} (\Phi((c_1\tau)^{\theta+1}\xi))^{-1}_{0, \theta} + e^{\frac{1}{2}(1 + \beta) \pi i} (\Phi((c_1\tau)^{\theta+1}\xi))^{-1}_{1, \theta} = e^{-\frac{1}{2}(1 + \beta) \pi i} (c_1\tau)^{-\theta} T^{-1}_{0, \theta}(\xi) e^{-\tau^2 g_0(\xi)} \\
    = (c_1\tau)^{-\theta} \sqrt{\pi} e^{-\frac{\tau^2}{2} (g_0(\xi) + g_1(\xi))}
  \times \left\{ \tau^{\frac{1}{3}} f(\xi)^{\frac{1}{4}} \left[ (-i e^{-\frac{\pi i \beta}{2}})  Q(\xi)^{-1}_{0, \theta} + (-e^{\frac{\pi i \beta}{2}}) Q(\xi)^{-1}_{1, \theta} \right] \Ai(\tau^{\frac{4}{3}} f(\xi)) \right. \\
    \left. - \tau^{-\frac{1}{3}} f(\xi)^{-\frac{1}{4}} \left[ (i e^{-\frac{\pi i \beta}{2}}) Q(\xi)^{-1}_{0, \theta} + (-e^{\frac{\pi i \beta}{2}}) Q(\xi)^{-1}_{1, \theta} \right] \Ai'(\tau^{\frac{4}{3}} f(\xi)) \right\}.
\end{multline}

By direct calculation using \eqref{eq:Iinvhat_1_at_1} and \eqref{eq:Iinvhat_2_at_1}, we have, as $\xi \to 1$ in $\compC_+$ and $\tau \to +\infty$,
\begin{align}
  Q_{0, 0}(\xi) = {}& -i e^{-\frac{\pi i \beta}{2}} \theta^{-\frac{\beta + 1}{2}} \left( \frac{\theta + 1}{2} \right)^{\frac{1}{4}} (1 - \xi^{\frac{1}{\theta}})^{-\frac{1}{4}} (1 + \bigO(\tau^{-2}) + \bigO(\xi - 1)), \\
  Q_{0, 1}(\xi) = {}& e^{\frac{\pi i \beta}{2}} \theta^{-\frac{\beta + 1}{2}} \left( \frac{\theta + 1}{2} \right)^{\frac{1}{4}} (1 - \xi^{\frac{1}{\theta}})^{-\frac{1}{4}} (1 + \bigO(\tau^{-2}) + \bigO(\xi - 1)), \\
  \intertext{and by relation \eqref{eq:relation_Pintfy_Pinfty_tilde}, we also have}
  Q^{-1}_{0, \theta}(\xi) = {}& \frac{1}{\theta + 1} e^{\frac{\theta}{\theta + 1} \pi i} e^{\frac{\pi i \betatilde}{2}} \theta^{\frac{\betatilde}{2}} \left( \frac{\theta + 1}{2} \right)^{\frac{1}{4}} (1 - \xi^{\frac{1}{\theta}})^{-\frac{1}{4}} (1 + \bigO(\tau^{-2}) + \bigO(\xi - 1)), \\
  Q^{-1}_{1, \theta}(\xi) = {}& \frac{-i}{\theta + 1} e^{\frac{\theta}{\theta + 1} \pi i} e^{-\frac{\pi i \betatilde}{2}} \theta^{\frac{\betatilde}{2}} \left( \frac{\theta + 1}{2} \right)^{\frac{1}{4}} (1 - \xi^{\frac{1}{\theta}})^{-\frac{1}{4}} (1 + \bigO(\tau^{-2}) + \bigO(\xi - 1)).
\end{align}

We then have, if $\lvert \xi - 1 \rvert < C \tau^{-4/3}$ for some $C$ and $\xi \in \Omega_2$,
\begin{multline} \label{eq:Phi_Omega_2}
  \Phi_{0, 1}((c_1\tau)^{\theta+1}\xi) + e^{\beta \pi i} \Phi_{0, 0}((c_1\tau)^{\theta+1}\xi) = 2 e^{\frac{\pi i \beta}{2}} \sqrt{\pi} \theta^{-\frac{\beta + 1}{2}} e^{\frac{\tau^2}{2} (g_0(\xi) + g_1(\xi))} \\
  \times \left( \frac{\theta(\theta + 1)}{2} \right)^{\frac{1}{4}} [-f'(1)]^{\frac{1}{4}} \tau^{\frac{1}{3}} \left( \Ai \left( \frac{2^{\frac{1}{3}} \theta}{(\theta - 1)^{\frac{2}{3}} (\theta + 1)^{\frac{5}{3}}} \tau^{\frac{4}{3}} (1 - \xi) \right) + \bigO(\tau^{-\frac{2}{3}}) \right),
\end{multline}
\begin{multline} \label{eq:Phi_inverse_Omega_2}
  e^{-\frac{1}{2}(1 + \beta) \pi i} (\Phi((c_1\tau)^{\theta+1}\xi))^{-1}_{0, \theta} + e^{\frac{1}{2}(1 + \beta) \pi i} (\Phi((c_1\tau)^{\theta+1}\xi))^{-1}_{1, \theta} = (c_1\tau)^{-\theta}\frac{2}{\theta + 1} \sqrt{\pi} \theta^{\frac{\betatilde}{2}} e^{-\frac{\tau^2}{2} (g_0(\xi) + g_1(\xi))} \\
  \times \left( \frac{\theta(\theta + 1)}{2} \right)^{\frac{1}{4}} [-f'(1)]^{\frac{1}{4}} \tau^{\frac{1}{3}} \left( \Ai \left( \frac{2^{\frac{1}{3}} \theta}{(\theta - 1)^{\frac{2}{3}} (\theta + 1)^{\frac{5}{3}}} \tau^{\frac{4}{3}} (1 - \xi) \right) + \bigO(\tau^{-\frac{2}{3}}) \right).
\end{multline}
Similar calculations show that \eqref{eq:Phi_Omega_2} and \eqref{eq:Phi_inverse_Omega_2} hold if  $\lvert \xi - 1 \rvert < C \tau^{-4/3}$ and $\xi \in \Omega_1 \cup \Omega_3 \cup \Omega_4$. 

\paragraph{Approximation in $D(0, \delta)$ and $(\delta, 1 - \delta)$}

For $\xi \in D(0, \delta)$, $T(\xi) = R(\xi) P^{(0)}(\xi)$, and for $\xi \in (\delta, 1 - \delta)$, $T(\xi) = R_+(\xi) P^{(\infty)}_+(\xi) = R_-(\xi) P^{(\infty)}_-(\xi)$. Similar to the derivations in \eqref{eq:first_est_region_D1_delta} -- \eqref{eq:Phi_inverse_Omega_2}, we have the approximation formula for $\phi^{(\tau)}((c_1 \tau)^{\theta + 1} \xi)$ and $\phitilde^{(\tau)}((c_1 \tau)^{\theta + 1} \xi)$ when $\lvert \xi \rvert < \delta$ or $\xi \in [\delta, 1 - \delta]$. We are not going to give all the details, but only state the result that for all $\xi, \eta \in D(0, 1 - \delta) \cup (\delta, 1 - \delta)$, there exists $\epsilon, C > 0$, such that if $\lvert \xi - \eta \rvert < \epsilon$, then
\begin{align} \label{eq:estimate_phi_phitilde_prod}
  \lvert \phi^{(\tau)}((c_1 \tau)^{\theta + 1} \xi) \phitilde^{(\tau)}((c_1 \tau)^{\theta + 1} \eta) \rvert < {}& e^{-\epsilon \tau^2}, && \text{for all $\tau > C$}.
\end{align}

\subsection{Proof of part \ref{enu:lem:phi_phitilde_sy:2} of Lemma \ref{lem:phi_phitilde_sy} and part \ref{enu:thm:kernel_asy:2} of Theorem \ref{thm:kernel_asy}}

From \eqref{eq:Phi_Omega_2} and \eqref{eq:Phi_inverse_Omega_2}, and their counterparts for  $\lvert \xi - 1 \rvert < C \tau^{-4/3}$ and $\xi \in \Omega_1 \cup \Omega_3 \cup \Omega_4$, we derive \eqref{eq:phi_asy_airy} and \eqref{eq:phitilde_asy_airy}. Thus we prove part \ref{enu:lem:phi_phitilde_sy:2} of Lemma \ref{lem:phi_phitilde_sy}.

To prove part \ref{enu:thm:kernel_asy:2} of Theorem \ref{thm:kernel_asy}, we divide the integral domain on the right-hand side of \eqref{eq:kernel_limit_formula} that defines $K^{(\tau)}(x, y)$ into $[\tau, \tau + \epsilon]$ and $(\tau + \epsilon, +\infty)$. For $\sigma \in [\tau, \tau + \epsilon]$, we use \eqref{eq:Phi_Omega_2} and \eqref{eq:Phi_inverse_Omega_2} to estimate the integrand $\phi^{(\sigma)}((c_1 \tau)^{\frac{\theta + 1}{\theta}} (1 - c_2 \tau^{-\frac{4}{3}})x) \phitilde^{(\sigma)}((c_1 \tau)^{\frac{\theta + 1}{\theta}} (1 - c_2 \tau^{-\frac{4}{3}})y)$, and for $\sigma \in (\tau + \epsilon, +\infty)$, we use \eqref{eq:estimate_phi_phitilde_prod} to estimate the same integrand. In this way, we prove \eqref{eq:thm:kernel_asy:2}.
\appendix

\section{The Airy parametrix} \label{app:Airy}

In this subsection, let $y_0$, $y_1$ and $y_2$ be the functions defined by
\begin{equation}
  y_0(\zeta) = \sqrt{2\pi}e^{-\frac{\pi i}{4}} \Ai(\zeta), \quad y_1(\zeta) = \sqrt{2\pi}e^{-\frac{\pi i}{4}} \omega\Ai(\omega \zeta), \quad y_2(\zeta) = \sqrt{2\pi}e^{-\frac{\pi i}{4}} \omega^2\Ai(\omega^2 \zeta),
\end{equation}
where $\Ai$ is the usual Airy function (cf. \cite[Chapter 9]{Boisvert-Clark-Lozier-Olver10}) and $\omega=e^{2\pi i/3}$. We then define a $2\times 2$ matrix-valued function $\Psi^{(\Ai)}$ by
\begin{equation} \label{eq:defn_Psi_Ai}
  \Psi^{(\Ai)}(\zeta)
  = \left\{
    \begin{array}{ll}
      \begin{pmatrix}
        y_0(\zeta) &  -y_2(\zeta) \\
        y_0'(\zeta) & -y_2'(\zeta)
      \end{pmatrix}, & \hbox{$\arg \zeta \in (0,\frac{2\pi}{3})$,} \\
      \begin{pmatrix}
        -y_1(\zeta) &  -y_2(\zeta) \\
        -y_1'(\zeta) & -y_2'(\zeta)
      \end{pmatrix}, & \hbox{$\arg \zeta \in (\frac{2\pi}{3},\pi)$,} \\
      \begin{pmatrix}
        -y_2(\zeta) &  y_1(\zeta) \\
        -y_2'(\zeta) & y_1'(\zeta)
      \end{pmatrix}, & \hbox{$\arg \zeta \in (-\pi,-\frac{2\pi}{3})$,} \\
      \begin{pmatrix}
        y_0(\zeta) &  y_1(\zeta) \\
        y_0'(\zeta) & y_1'(\zeta)
      \end{pmatrix}, & \hbox{$\arg \zeta \in  (-\frac{2\pi}{3},0)$.}
    \end{array}
  \right.
\end{equation}
It is well-known that $\det (\Psi^{(\Ai)}(z))=1$ and $\Psi^{(\Ai)}(\zeta)$ is the unique solution of the following $2 \times 2$ RH problem; cf. \cite[Section 7.6]{Deift99}.
\begin{RHP} \hfill\label{rhp:Ai}
\begin{enumerate}
\item
  $ \Psi^{(\Ai)}(\zeta)$ is analytic in $\mathbb{C} \setminus \Gamma_{\Ai}$, where the contour $\Gamma_{\Ai}$ is defined in
    \begin{equation} \label{def:AiryContour}
      \Gamma_{\Ai}:=e^{-\frac{2\pi i}{3}}[0,+\infty) \cup \mathbb{R} \cup e^{\frac{2\pi i}{3}}[0,+\infty)
    \end{equation}
 with all rays oriented from left to right.
\item
  For $z \in \Gamma_{\Ai}$, we have
  \begin{equation} \label{eq:Ai_jump}
     \Psi^{(\Ai)}_+(\zeta) =  \Psi^{(\Ai)}_-(\zeta)
    \begin{cases}
      \begin{pmatrix}
        1 & 1 \\
        0 & 1
      \end{pmatrix},
      & \arg \zeta =0, \\
      \begin{pmatrix}
        1 & 0 \\
        1 & 1
      \end{pmatrix},
      & \arg \zeta = \pm \frac{2\pi }{3}, \\
      \begin{pmatrix}
        0 & 1  \\
        -1 & 0
      \end{pmatrix},
      & \arg \zeta = \pi.
    \end{cases}
  \end{equation}
\item
  As $\zeta \to \infty$, we have
    \begin{equation} \label{eq:defn_Psi^Ai_infty}
      \Psi^{(\Ai)}(\zeta) = \Psi^{(\Ai)}_{\infty}(\zeta) (I+\bigO(\zeta^{-\frac32}))e^{-\frac23 \zeta^{\frac32}\sigma_3}. \quad \Psi^{(\Ai)}_{\infty}(\zeta) = \zeta^{-\frac{1}{4} \sigma_3} \frac{1}{\sqrt{2}}
      \begin{pmatrix}
        1 & 1 \\
        -1 & 1
      \end{pmatrix}
      e^{-\frac{\pi i}{4} \sigma_3}.
    \end{equation}
\end{enumerate}
\end{RHP}
%

\section{Hard edge local parametrix} \label{sec:Phi^Mei}

Recall the $(\theta + 1) \times (\theta + 1)$ matrix-valued function $\Psi^{(\Mei)}(\xi)$ defined in \cite[RH Problem 3.12]{Wang-Zhang21} (where the variable is denoted $\zeta$). Let 
\begin{multline} \label{eq:Phi^Mei_pre_and_Psi_Mei}
  \Phi^{(\Mei, \pre)}(\xi) = (2\pi)^{1 - \frac{\theta}{2}} \frac{\sqrt{\theta + 1}}{\theta} \xi^{-\frac{\beta}{2\theta}} \Psi^{(\Mei)}(\xi) \\
  \times
  \begin{cases}
    \diag \left( \theta e^{-\frac{\beta}{2} \pi i} \xi^{\frac{\theta - \alpha - 1}{\theta}}, e^{\frac{\beta}{2} \pi i} \right) \oplus \diag \left( e^{-\beta \left( \frac{k - 1}{\theta} - \frac{1}{2} \right) \pi i}\right)^{\theta}_{k = 2}, & \xi \in \compC_+, \\
    \diag \left( -\theta e^{\frac{\beta}{2} \pi i} \xi^{\frac{\theta - \alpha - 1}{\theta}}, e^{-\frac{\beta}{2} \pi i} \right) \oplus \diag \left( e^{-\beta \left( \frac{k - 1}{\theta} - \frac{1}{2} \right) \pi i}\right)^{\theta}_{k = 2}, & \xi \in \compC_-.
  \end{cases}
\end{multline}
We find that $\Phi^{(\Mei, \pre)}(\xi)$ satisfies the asymptotic condition at infinity as $\xi \to \infty$
\begin{equation} \label{eq:Phi^Mei_pro_asy_exp}
   \Phi^{(\Mei, \pre)}(\xi) = \Upsilon(\xi) \Omega_{\pm} \left( I + \sum^{\infty}_{n = 1} A_n \xi^{-\frac{n}{\theta + 1}} \right) e^{-\Lambda(\xi)},
\end{equation}
where $\Upsilon(\xi)$ and $\Omega_{\pm}$ are the same as in \eqref{eq:asyNew}, $A_n$ are constant matrices, and
\begin{equation} \label{eq:defn_Lambda}
  \Lambda(\xi) = (\theta + 1) \xi^{\frac{1}{\theta + 1}} \times
  \begin{cases}
    \begin{pmatrix}
      e^{-\frac{\pi i}{\theta + 1}} & 0 \\
      0 & e^{\frac{\pi i}{\theta + 1}}
    \end{pmatrix}
    \oplus \diag \left( e^{\frac{2j - 1}{\theta + 1} \pi i} \right)^{\theta}_{j = 2}, & \xi \in \compC_+, \\
    \begin{pmatrix}
      e^{\frac{\pi i}{\theta + 1}} & 0 \\
      0 & e^{-\frac{\pi i}{\theta + 1}}
    \end{pmatrix}
    \oplus \diag \left( e^{\frac{2j - 1}{\theta + 1} \pi i} \right)^{\theta}_{j = 2}, & \xi \in \compC_-.
  \end{cases}
\end{equation}
We note that the series in \eqref{eq:Phi^Mei_pro_asy_exp} is an asymptotic expansion and does not converge.

Entries of $\Psi^{(\Mei)}$ and $\Phi^{(\Mei, \pre)}$ are expressed by Meijer G-functions. The asymptotic formula \eqref{eq:Phi^Mei_pro_asy_exp} is due to the asymptotic expansion of Meijer G-functions in \cite[Theorem 5 in Section 5.7 and Theorem 2 in Section 5.10]{Luke69}. The leading terms of the asymptotic expansion are given in \cite[Equations (3.64), (3.79), (3.80) and (3.81)]{Wang-Zhang21}.

We have the following lemma whose proof is by term-by-term calculation:
\begin{lemma} \label{lem:RH_for_Phi_0}
  Suppose $\Phi_0(\xi)$ is a formal Puiseux series of a $(\theta + 1) \times (\theta + 1)$ matrix-valued function
  \begin{equation}
    \Phi_0(\xi) = I+\sum^{\infty}_{n = 1} B_n \xi^{-\frac{n}{\theta + 1}},
  \end{equation}
  and it satisfies, if the power functions all take the principal branch, for $x \in \realR \setminus \{ 0 \}$,
  \begin{equation} \label{eq:RH_for_Phi_0}
    \Upsilon_+(x) \Omega_+ \Phi_{0, +}(x) = \Upsilon_-(x) \Omega_- \Phi_{0, -}(x) \times
    \begin{cases}
      \begin{pmatrix}
        0 &  1 \\
        1 & 0
      \end{pmatrix}
      \oplus I_{\theta - 1}, & x > 0, \\
      M_{\cyclic}, & x < 0,
    \end{cases}
  \end{equation}
  then there exists a constant lower triangular matrix $C$ whose diagonal entries are all $1$, such that
  \begin{equation} \label{eq:lower_triangular_C}
    C\Upsilon(\xi) \Omega_{\pm} \Phi_0(\xi) = (I + \bigO(\xi^{-1})) \Upsilon(\xi) \Omega_{\pm}.
  \end{equation}
\end{lemma}
It is straightforward to check that the Puiseux series in \eqref{eq:Phi^Mei_pro_asy_exp} satisfies \eqref{eq:RH_for_Phi_0} in Lemma \ref{lem:RH_for_Phi_0} with $\Phi_0$ replaced by this Puiseux series. Thus we have a constant lower triangular matrix $C$ whose diagonal entries are all $1$ as in \eqref{eq:lower_triangular_C}, and define
\begin{equation} \label{eq:defn_Phi^Mei}
  \Phi^{(\Mei)}(\xi) = C \Phi^{(\Mei, \pre)}(\xi),
\end{equation}
and it satisfies the following RH problem:
\begin{RHP} \label{RHP:Phi_Mei}
  $\Phi^{(\Mei)}(\xi)$ is a $(\theta + 1) \times (\theta + 1)$ matrix-valued function on $\compC$ except for $\realR$ and $i\realR$. It satisfies
  \begin{enumerate}
  \item $\Phi^{(\Mei)}_+(\xi) = \Phi^{(\Mei)}_-(\xi) J^{(\theta)}_{\Phi}(\xi)$, where $J^{(\theta)}_{\Phi}(\xi)$ is defined in \eqref{eq:defn_J_hard_to_soft}.
  \item
    $\Phi^{(\Mei)}(\xi)$ has the following asymptotic condition at infinity as $\xi \to \infty$
    \begin{equation}
      \Phi^{(\Mei)}(\xi) = \left(I+ \bigO(\xi^{-1} ) \right) \diag \left( e^{-\frac{k}{\theta + 1} \pi i} \xi^{\frac{k}{\theta + 1}} \right)^{\theta}_{k = 0} \Omega_{\pm} e^{-\Lambda(\xi)},
    \end{equation}
    where $\Lambda(\xi)$ is defined in \eqref{eq:defn_Lambda}, and $\Omega_{\pm}$ are defined in \eqref{def:Lpm}
  \item 
    $\Phi^{(\Mei)}(\xi)$ has the following local growth condition as $z \to 0$
    \begin{equation}
      \Phi(\xi) =  \xi^{-\frac{\beta}{2\theta}} N^{(\Mei)}(\xi) \diag \left( \xi^{-\frac{\alpha + 1 - \theta}{\theta}}, 1, \xi^{\frac{1}{\theta}}, \xi^{\frac{2}{\theta}}, \dotsc, \xi^{\frac{\theta - 1}{\theta}} \right)E,
    \end{equation}
    where $N^{(\Mei)}(\xi)$ is analytic at $0$, and $E$ is the same as in RH problem \ref{RHP:general_model}. In the sector $\arg \xi \in (\pi/2, \pi)$, $E$ is defined by \eqref{eq:defn_E_region_II} and \eqref{eq:E_log}.
  \end{enumerate}
\end{RHP}

Similar to the argument above, we also define, from the $(\theta + 1) \times (\theta + 1)$ matrix-valued function $\Psitilde^{(\Mei)}(\xi)$ defined in \cite[RH Problem 4.11]{Wang-Zhang21}, that
\begin{multline} \label{eq:Phi^Mei_inverse_pre_and_Psitilde_Mei}
  \Phitilde^{(\Mei, \pre)}(\xi) = (2\pi)^{\frac{\theta}{2}} \sqrt{\theta + 1} e^{\left( \frac{\alpha}{\theta} - \frac{1}{2\theta}(\betatilde + 1) \right) \pi i} \xi^{\frac{\betatilde + 1}{2\theta}} \Psitilde^{(\Mei)}(\xi) \\
  \times
  \begin{cases}
    \diag \left( e^{-\frac{2\alpha}{\theta} \pi i} \xi^{\frac{\alpha}{\theta}}, e^{\frac{1}{\theta}(\betatilde + 1) \pi i} \right) \oplus \diag \left( e^{\frac{k}{\theta}(\betatilde + 1) \pi i}\right)^{\theta}_{k = 2}, & \xi \in \compC_+, \\
    \diag \left( e^{\frac{1}{\theta}(\betatilde + 1) \pi i} \xi^{\frac{\alpha}{\theta}}, -e^{-\frac{2\alpha}{\theta} \pi i} \right) \oplus \diag \left( e^{\frac{k}{\theta}(\betatilde + 1) \pi i}\right)^{\theta}_{k = 2}, & \xi \in \compC_-.
  \end{cases}
\end{multline}
Then there exists a lower triangular matrix $\Ctilde$ whose diagonal entries are all $1$, such that
\begin{equation} \label{eq:defn_Phitilde^Mei}
  \Phitilde^{(\Mei)}(\xi) = \Ctilde \Phitilde^{(\Mei, \pre)}(\xi),
\end{equation}
and we find that it satisfies the following RH problem:
\begin{RHP}
  $\Phitilde^{(\Mei)}(\xi)$ is a $(\theta + 1) \times (\theta + 1)$ matrix-valued function on $\compC$ except for $\realR$ and $i\realR$. It satisfies
  \begin{enumerate}
  \item $\Phitilde^{(\Mei)}_+(\xi) = \Phitilde^{(\Mei)}_-(\xi) J^{(\theta)}_{\Phitilde}(\xi)$, where $J^{(\theta)}_{\Phitilde}(\xi)$ is defined in \eqref{eq:defn_J_hard_to_soft_q}.
  \item
    $\Phitilde^{(\Mei)}(\xi)$ has the following asymptotic condition at infinity as $\xi \to \infty$
    \begin{equation}
      \Phitilde^{(\Mei)}(\xi) = \left(I+ \bigO(\xi^{-1} ) \right) \diag \left( e^{-\frac{k}{\theta + 1} \pi i} \xi^{\frac{k}{\theta + 1}} \right)^{\theta}_{k = 0} \Omega_{\pm} e^{\Lambda(\xi)},
    \end{equation}
    where $\Lambda(\xi)$ is defined in \eqref{eq:defn_Lambda}, and $\Omega_{\pm}$ are defined in \eqref{def:Lpm}
  \item 
    $\Phitilde^{(\Mei)}(\xi)$ has the following local growth condition as $z \to 0$
    \begin{equation}
      \Phitilde^{(\Mei)}(\xi) = \N^{(\Mei)}(\xi) \diag \left( \xi^{-\frac{\betatilde}{2}}, \xi^{\frac{\betatilde + 2\theta - 1}{2\theta}}, \dotsc, \xi^{\frac{\betatilde + 5}{2\theta}} \xi^{\frac{\betatilde + 3}{2\theta}}, \xi^{\frac{\betatilde + 1}{2\theta}}, \right) (E^{-1})^T \Xicheck^{-1}_{\pm},
    \end{equation}
    where $\N^{(\Mei)}(\xi)$ is analytic at $0$, and $E$ is the same as in RH problem \ref{RHP:general_model}. In the sector $\arg \xi \in (\pi/2, \pi)$, $E$ is defined by \eqref{eq:defn_E_region_II} and \eqref{eq:E_log}.
  \end{enumerate}
\end{RHP}

We note that the first row of $\Phi^{(\Mei)}(\xi)$ (resp.~$\Phitilde^{(\Mei)}(\xi)$) is the same as that of $\Psi^{(\Mei)}(\xi)(\xi)$ (resp.~$\Psi^{(\Mei)}(\xi)$).

Also, by arguments similar to those in Section \ref{subsubsec:Phi_q}, we have that
\begin{align} \label{eq:Phitilde^Mei_and_Phi^Mei_inverse}
  \Phitilde^{(\Mei)}(\xi) = (\theta + 1) e^{-\frac{\theta}{\theta + 1} \pi i} \xi^{\frac{\theta}{\theta + 1}} J_{\theta + 1} (\Phi^{(\Mei)}(\xi)^{-1})^T \Xicheck^{-1}_{\pm}.
\end{align}

\section{Limiting mapping $J^{(\pre)}$} \label{sec:limiting_J}

We consider the mapping
\begin{equation}\label{eq:J_mapping}
  J^{(\pre)}(\sigma) = -(-\sigma)^{\frac{\theta + 1}{\theta}} + (\theta + 1) \theta^{-\frac{\theta}{\theta + 1}} (-\sigma)^{\frac{1}{\theta}},
\end{equation}
We find that
\begin{align}
  J^{(\pre)}(\sigma_0) = {}& 1, & \frac{d}{d\sigma} J^{(\pre)}(\sigma_0) = {}& 0, & \frac{d^2}{d\sigma^2} J^{(\pre)}(\sigma_0) = {}& -(\theta + 1) \theta^{-\frac{2}{\theta + 1}}, & \sigma_0 = {}& -\theta^{-\frac{\theta}{\theta +1}}.
\end{align}
Also we have a contour $\gammahat^{(\pre)}$ that is from $e^{-\pi i/(\theta + 1)} \cdot \infty$, passing through $\sigma_0$, and to $e^{\pi i/(\theta + 1)} \cdot \infty$, such that $\Im J^{(\pre)}(\sigma) = 0$ for $\sigma \in \gammahat^{(\pre)}$. We denote $\gammahat^{(\pre)}_{\pm} := \gammahat^{(\pre)} \cap \compC_{\pm}$. We also denote the region to the right of $\gammahat^{(\pre)}$ as $\D^{(\pre)}$. It is straightforward to check that $J^{(\pre)}(\sigma)$ maps the region $\compC \setminus \overline{\D^{(\pre)}}$ univalently to $\compC \setminus (1, +\infty)$, and it maps the region $\D^{(\pre)} \setminus (0, +\infty)$ univalently to $\halfH \setminus (1, +\infty)$. We denote the inverse mapping of $J^{(\pre)}(\sigma)$ in these two regions $\Iinvhat^{(\pre)}_1$ and $\Iinvhat^{(\pre)}_2$, respectively. It follows from the  definition of the mappings $\Iinvhat^{(\pre)}_1$ and $\Iinvhat^{(\pre)}_2$ that  
\begin{align} \label{eq:PreImage}
  \Iinvhat^{(\pre)}_{1,\pm }((1,\infty)) = {}& \gammahat^{(\pre)}_{\pm}, & \Iinvhat^{(\pre)}_{2,\pm }((1,\infty)) = {}& \gammahat^{(\pre)}_{\mp}.
\end{align}
We also have, as $z \to 1$,
\begin{align}
  \Iinvhat^{(\pre)}_1(z) = {}& \sigma_0 - (1 - z)^{\frac{1}{2}} \theta^{\frac{1}{\theta + 1}} \sqrt{\frac{2}{\theta + 1}} \left( 1 + \frac{\theta - 1}{3} \sqrt{\frac{2}{\theta + 1}} (1 - z)^{\frac{1}{2}} + \bigO(1 - z) \right), \label{eq:Iinvhat_1_at_1} \\
  \Iinvhat^{(\pre)}_2(z) = {}& \sigma_0 + (1 - z)^{\frac{1}{2}} \theta^{\frac{1}{\theta + 1}} \sqrt{\frac{2}{\theta + 1}} \left( 1 - \frac{\theta - 1}{3} \sqrt{\frac{2}{\theta + 1}} (1 - z)^{\frac{1}{2}} + \bigO(1 - z) \right). \label{eq:Iinvhat_2_at_1}
\end{align}

\end{document}